\documentclass[12pt]{article}

\usepackage[title]{appendix}
\usepackage{lscape}
\usepackage{amsfonts}
\usepackage{amssymb}
\usepackage{amsmath}
\usepackage{amsthm}
\usepackage{caption}
\usepackage{color}
\usepackage{epstopdf}
\usepackage{float}
\usepackage[multiple]{footmisc}
\usepackage{longtable}
\usepackage{multirow}
\usepackage{natbib}
\usepackage[margin=.75in]{geometry}
\usepackage[pdftex]{graphicx}
\usepackage{psfrag}
\usepackage{rotating}
\usepackage{setspace}
\usepackage{url}

\usepackage{bbm}

\newcommand{\Perp}{\perp\!\!\!\perp}
\renewcommand{\qed}{\hfill $\blacksquare$}
\newcommand{\sign}{\text{sign}}
\allowdisplaybreaks

\newtheorem{theorem}{Theorem}
\renewenvironment{proof}[1][Proof]{\begin{trivlist}
\item[\hskip \labelsep {\bfseries #1.}]}{\end{trivlist}}

\newtheorem{lemma}[theorem]{Lemma}

\newtheorem{hypothesis}{Hypothesis}

\begin{document}

\title{Warp Speed Price Moves: \\ Jumps after Earnings Announcements\thanks{We thank an anonymous reviewer for insightful and constructive referee reports during the revision process. We also thank Andrew Patton, Federico Bandi, Hashem Pesaran, Marcelo Moreira, Pete Kyle, Peter Hansen, Roberto Ren\`{o}, Shuping Shi, Thomas Chemmanur, Tim Bollerslev, Torben G. Andersen, Viktor Torodov, Yunus Topbas, and participants at the 33rd (EC)$^{2}$ conference at ESSEC, France, the Advances in Financial Econometrics conference at Copenhagen Business School, Denmark, the 15th annual SoFiE conference at Sungkyunkwan University, Republic of Korea, the Volatility conference at Singapore Management University, Singapore, the 32nd annual conference on Pacific Basin Finance, Economics, Accounting, and Management at Rutgers Business School, USA, as well as in invited seminars at the Advanced Financial Technologies Laboratory (AFTLab) at Stanford University, University of Houston, University of Maryland, University of Nottingham, and University of Southern California for helpful comments. We are grateful to research assistance from Hanbo Wang, who calculated the constituent members' weights in the S\&P 500 index. Christensen appreciates funding from the Independent Research Fund Denmark (DFF 1028--00030B). This work was also supported by the Center for Research in Econometric Analysis of Time Series (CREATES). In addition, Christensen and Veliyev are affiliated with the Danish Finance Institute (DFI), which financially supports researchers in Finance.}
}
\author{Kim Christensen\thanks{Department of Economics and Business Economics, Aarhus University, Fuglesangs All\'{e} 4, 8210 Aarhus V, Denmark.} \thanks{Research fellow at Danish Finance Institute (DFI).} \and Allan Timmermann\thanks{Rady School of Management and Department of Economics, UCSD, 9500 Gilman Drive, La Jolla, CA 92037, USA.} \and Bezirgen Veliyev\footnotemark[2] \footnotemark[3]}

\date{November, 2024}

\maketitle

\begin{abstract}
Corporate earnings announcements unpack large bundles of public information that should, in efficient markets, trigger jumps in stock prices. Testing this implication is difficult in practice, as it requires noisy high-frequency data from after-hours markets, where most earnings announcements are released. Using a unique dataset and a new microstructure noise-robust jump test, we show that earnings announcements almost always induce jumps in the stock price of announcing firms. They also significantly raise the probability of price co-jumps in non-announcing firms and the market. We find that returns from a post-announcement trading strategy are consistent with efficient price formation after 2016.

\vspace*{0.5cm}


\end{abstract}

\setlength{\baselineskip}{18pt}\setlength{\abovedisplayskip}{10pt} \belowdisplayskip \abovedisplayskip \setlength{\abovedisplayshortskip }{5pt} \abovedisplayshortskip \belowdisplayshortskip \setlength{\abovedisplayskip}{8pt} \belowdisplayskip \abovedisplayskip \setlength{\abovedisplayshortskip }{4pt}

\vfill

\thispagestyle{empty}

\pagebreak

\section{Introduction} \setcounter{page}{1}

Earnings announcements constitute perhaps the single most important source of firm-level information. They are the key occasions on which firms release audited results and communicate their views on economic prospects to stock market investors. With several headline numbers being released simultaneously, earnings announcements can be thought of as points in time where draws from the public news process have a very high variance which, in all likelihood, leads investors to materially revise their estimates of firm values. In an efficient market where trading frictions are not prohibitively large, the widespread prevalence of high-frequency trading for the most liquid stocks therefore implies that earnings announcements should---almost surely---trigger jumps in prices. Effectively, a very high propensity of jumps in stock prices immediately after earnings announcements is a \emph{necessary} condition for markets to efficiently incorporate the new information.\footnote{The condition that prices jump after earnings announcements is, however, not sufficient because the price could systematically over- or undershoot, thus introducing predictable dynamics in the post-announcement returns. We explore this issue in Section \ref{section:price-discovery}.} Building on this idea, we propose to use jump tests based on high-frequency tick-by-tick price data as a way to study market efficiency.

In recent years, the vast majority of publicly traded U.S. companies have made their earnings announcements in the after-hours market outside the regular trading session (9:30am--4:00pm).\footnote{\citet*{jiang-likitapiwat-mcinish:12a}, \citet*{lyle-stephan-yohn:21a}, and \citet*{michaely-rubin-vedrashko:13a} report that more than 95\% of U.S. firms announce earnings outside the regular exchange trading session. This fraction is as high as 99.1\% in the recent analysis by \citet*{gregoire-martineau:22a} which studies companies belonging to the S\&P 1500 stock index between 2011 and 2015.} After-hours trading on days with earnings announcements are therefore central to the price discovery process and it is important to understand price dynamics in these markets. We do so by studying price discovery and return predictability at a much higher tick-by-tick frequency than in previous studies of earnings announcements. All told, we examine more than 89 billion after-hours quotes and nearly 8 billion after-hours transactions for 50 stocks over a 12-year sample. Our analysis provides the most complete analysis of price formation in the after-hours market for these stocks to date.

Unfortunately, high-frequency data from the after-hours trading session are severely affected by market microstructure noise. The irregular format of such data poses important challenges as it distorts standard jump test procedures. Empirical work has therefore resorted to testing for jumps during the regular trading session, which does not include most earnings announcements.\footnote{See, e.g., \citet{ait-sahalia-jacod:09b, ait-sahalia-jacod-li:12a, andersen-bollerslev-dobrev:07a, caporin-kolokolov-reno:17a, christensen-oomen-podolskij:14a, corsi-pirino-reno:10a, jiang-oomen:08a}.} To overcome this shortcoming, we develop a new jump test that is suitable for illiquid and noisy markets by generalizing the classical bipower variation-based jump test of \citet*{barndorff-nielsen-shephard:06a}. We construct a noise-robust version of their test statistic that relies on a pre-averaging technique to lessen the detrimental impact of the noise \citep*[see, e.g.,][]{ait-sahalia-jacod-li:12a, jacod-li-mykland-podolskij-vetter:09a, lee-mykland:12a, podolskij-vetter:09a, podolskij-vetter:09b} allowing for both heteroscedasticity and serial dependence in the microstructure noise process. We also show how to estimate the asymptotic variance of the test statistic, which is itself a non-trivial problem.\footnote{We propose a jump- and noise-robust version of the subsampling approach studied in \citet*{christensen-podolskij-thamrongrat-veliyev:17a}, which delivers a consistent estimator with desirable finite sample properties.} Through simulations and empirical results, we demonstrate that the conventional jump test is severely distorted by the type of market microstructure noise that is prevalent in after-hours markets, leading both to too many false positives and false negatives.

Using our new noise-robust test statistic, we find strong empirical evidence in support of the hypothesis that earnings announcements trigger jumps in prices. Specifically, the probability that the stock price of announcing firms jumps in after-hours sessions with earnings announcements exceeds 90\% while the corresponding jump probabilities for the regular trading session and after-hours sessions without earnings announcements are 2.95\% and 3.67\%. In a nutshell, prices almost always jump after earnings announcements, but rarely jump otherwise.

If earnings announcements contain industry- and economy-wide news components, we should expect an increase in the jump probability of non-announcing firms' stock price as well as in the market index on announcement days \citep{patton-verardo:12a, savor-wilson:16a}. We test these implications and find compelling support for their presence. Jump spillover effects are strongly related to industry proximity and after-hours market liquidity with the probability of a price jump more than doubling for a high-liquidity, non-announcing firm in the same industry as the announcing firm compared to a low-liquidity, non-announcing firm from other industries. The spillover effect is more potent early versus late in the earnings season, consistent with \citet{savor-wilson:16a}. Moreover, there is a significant increase in the jump probability of the market index when many firms announce together after controlling for the mechanical price effect resulting from the announcing firms often being constituent members of the market index.\footnote{ This finding suggests that the ``information spillover'' effect dominates the “distraction'' effect identified by \citet{hirshleifer-lim-teoh:09a} who find that the price response to earnings surprises is significantly weaker on days when multiple firms announce their earnings.}

We next study whether prices adjust \textit{sufficiently} fast and by an appropriate amount to efficiently incorporate the new information. Our analysis inspects all trades and quotes, using a more detailed and complete data set than previously examined. This means that we can pinpoint exactly how long it takes for the price discovery process to impound the news as well as identify the factors---including earnings surprises, market liquidity and analyst coverage---determining price movements in the immediate aftermath of earnings announcements.

We follow the approach of \citet{patell-wolfson:84a} by examining a simple trading rule based on the earnings response model of \citet{ball-brown:68a}. That is, we compute return forecasts off the surprise component in earnings announcements to decide whether to take a long or short position in the stock. Over the 12-year sample (2008--2020) we earn a highly significant mean return of 1.80\% per transaction in a ``no-friction'' scenario where investors enter at the first transaction price after the announcement. Executing trades instead at the midquote of the bid-ask spread, mean returns drop to (a highly significant) 1.50\% per trade. This finding continues to hold for trades executed at the actual spreads obtained from the bid-ask quotes after the announcement (mean return of 0.72\%). Delaying trades by five seconds reduces mean returns to 0.41\% per trade, which remains significant. Further delays beyond five seconds lead to mean returns that are no longer statistically significant.

To track how market efficiency has evolved over time, we split our data into two sub-samples chosen to reflect structural changes in market design, as explained in Section \ref{section:trading-strategy}. In the early sub-sample (2008--2015), mean returns are a highly significant 2.30\% and 2.00\% per trade in the frictionless and midquote scenarios, respectively, and remain significant and large (0.80\% per trade) even with a trading delay of 10 seconds. Conversely, in the second sub-sample (2016--2020), mean returns are only significantly positive in the frictionless and midquote scenarios, becoming small and insignificant in the presence of bid-ask spreads or even negative with timing delays. This suggests that the post-earnings announcement price discovery process in the after-hours market has become extremely fast and increasingly efficient over time for the most liquid U.S. stocks.\footnote{This conclusion agrees with \citet{chordia-green-kottimukkalur:18a}, who study price formation in the composite equity market around macroeconomic news. As in our paper, they show that the speed of information processing has increased in recent years and that profits from proprietary trading are economically small. \citet{brogaard-hendershott-riordan:14a} document how high-frequency traders aid with price discovery through their trading.}

The theoretical foundation based on jump testing with high-frequency data is fundamentally different from conventional approaches that examine how news releases get processed. Existing work typically looks at return predictability, usually in terms of serial correlation, in one- or five-minute intervals after news announcements, see, e.g., \cite{ederington-lee:93a}. While these time intervals are relatively short in an absolute sense, they are excessively long in modern markets with fast-paced, low-latency trading that can process thousands of orders within a few milliseconds after news releases (see Figure \ref{figure:warpspeed} for an illustrative example). In markets where a minute can be an eternity, jump tests conducted at the tick-by-tick level are arguably better at revealing whether prices react efficiently to news announcements.

Our paper contributes to a burgeoning literature on how stock prices react to earnings announcements, starting with \cite{ball-brown:68a} and \cite{beaver:68a}. The vast majority of the research in this field has refrained from studying the instantaneous price change after an earnings announcement and instead uncovers evidence of post-earnings announcement drift (PEAD), see \citet{bernard-thomas:89a, bernard-thomas:90a}. The PEAD can be interpreted as an indirect test for efficient price responses to earnings announcements because over- or underreaction in the immediate post-announcement price change induces a subsequent drift in the stock price. Tests for PEAD tend to exhibit low power and lead to conflicting evidence, however, since it is difficult to accurately estimate mean returns over a long post-event window that often spans several days (or even months). Indeed, while \cite{martineau:22a} concludes that the PEAD is notably weaker after the early 2000s, \cite{fink:21a} reaches the opposite conclusion that it has not vanished.\footnote{More recently, \cite{chan-marsh:24a} argue that although there is scant evidence of a daily (close-to-close) PEAD, a significant overnight effect remains, particularly for firms with extreme earnings misses.}

Far fewer studies analyze ultra short-term stock price dynamics after earnings announcements and none comes close to the level of resolution used here. The most closely related article is \cite{gregoire-martineau:22a}, who examine how earnings announcement news get transmitted to stock prices in the after-hours market. They examine a broader set of stocks than we do but look at the dynamics of price formation for each firm in far less detail, aggregating price changes into one-minute intervals. They conclude that bid-ask spreads are so wide that investors cannot cover their round-trip trading costs even if they knew the post-announcement closing price in advance.\footnote{Specifically, \cite{gregoire-martineau:22a} write that ``pre-announcement bid-ask spreads are wide enough to include the post-announcement closing price, eliminating the profits of informed liquidity takers.''} We arrive at an entirely different result for our sample of liquid firms: The post-earnings transaction price at 6:30pm (at the time where we close out our trading strategy) is outside the pre-announcement quoted spread for 92.66\% of the earnings announcements. So the post-earnings announcement price dynamics are rather distinctive for the most liquid stocks that exhibit the lowest bid-ask spreads and account for the vast majority of after-hours trading volume.

Our analysis is also connected to a more general stream of work that examines price discovery and trading in the after-hours market, notably \cite{barclay-hendershott:03a} who find evidence of inefficient price formation in those markets. Our results show that after-hours trading activity (and, hence, price discovery) is really a tale of two markets. Even for the vastly liquid stocks examined here, after-hours trading volume is almost entirely dormant on days without earnings announcements (or other important company-specific news releases). In contrast, on days with earnings announcements, after-hours trading activity explodes. Average after-hours trading volume on days with earnings announcements exceeds its counterpart on days without earnings announcement by almost two orders of magnitude.\footnote{Averaged over firms and years, this ratio equals 93.57 over our sample.}

Other articles also study the link between jumps in stock prices and news effects. \citet{andersen-bollerslev-diebold-vega:03a} study price discovery in foreign exchange markets around the dissemination of macroeconomic news. As in our paper, they show that changes in exchange rates are determined by the sign and magnitude of the standardized announcement surprise but restrict attention to high-frequency data sampled at 5-minute frequency. This is further documented in \citet{jeon-mccurdy-zhao:22a}, who base their analysis on general news and daily data. \citet{ait-sahalia-li-li:24a} relate stock price jumps extracted from 5-second intraday log-returns to firm-level, industry and macroeconomic news. In contrast, we examine earnings announcements and high-frequency data sampled at the tick-by-tick level, which is the highest possible degree of granularity.

Our finding that stock prices predominantly jump on earnings announcements days and induce jump spillover effects for other stocks and the market index has important implications also for the pricing of contingent claims affected by earnings announcement risk \citep{dubinsky-johannes-kaeck-seeger:19a, bandi-fusari-reno:22a,todorov-zhang:23a}. Moreover, \citet{bollerslev-li-todorov:16a} find significantly positive risk premia for jump and overnight risk, but not for continuous variation. This suggests that investors command a risk premium for exposure to discontinuous or overnight variation in the aggregate market. In addition to this, \citet{lou-polk-skouras:19a} show that most of the equity risk premium accrues overnight. Our findings can help to shed further light on this issue, including the extent to which the equity risk premium is driven by earnings announcements.

The outline of the paper is as follows. Section \ref{section:theory} establishes the link between price jumps and market efficiency before introducing our new noise-robust jump test statistic. Section \ref{section:data} explains our data sources. Section \ref{section:empirical} implements the jump test statistic on a unique set of high-frequency data that includes information about the transaction record from outside the official opening hours of the exchange. Section \ref{section:jump-spillover} examines jump spillover effects, and Section \ref{section:price-discovery} studies price dynamics after earnings announcements and analyzes returns from a simple trading strategy. Section \ref{section:conclusion} concludes. Supplemental appendices at the back of the paper provide additional results and technical analysis.

\section{Market efficiency and price jumps} \label{section:theory}

We begin by motivating why, in efficient markets, we should expect stock prices to jump immediately after earnings announcements. Next, to test this implication we introduce a new jump test that is robust to the high levels of market microstructure noise typical in after-hours trading sessions. We finish the section by reporting results from Monte Carlo simulations comparing the performance of our new noise-robust test to existing jump tests.

\subsection{Earnings announcements} \label{section:ea}

Earnings announcements unpack large bundles of information to the public. The timing of their release is typically known in advance, giving investors ample opportunity to prepare and respond more or less instantaneously.\footnote{Most large-cap companies inform markets in advance about their intention to release a financial report either Before Market Open (BMO), during Regular Trading Hours (RTH), or After Market Close (AMC), see, e.g., \path{https://www.nasdaq.com/market-activity/earnings}.} Since the main research question of our paper is to analyze the process by which stock market participants incorporate such public information into prices, we begin by defining what is meant by efficient price and price discovery processes. To this end, let $\mathcal{I}_{t}^{ \text{public}}$ be the public information set at time $t$. The \textit{efficient price} at time $t$, $P_{t}$, is defined as the expectation of some equilibrium price, $v_{t}$, conditional on this information:
\begin{equation} \label{equation:efficient-price-Q}
P_{t} = \mathbb{E}^{ \mathbb{Q}} \big[v_{t} \mid \mathcal{I}_{t}^{ \text{public}} \big],
\end{equation}
where $\mathbb{Q}$ is the risk-neutral probability measure (we ignore the time-discounting effect, or assume interest rates are zero, for notational convenience). Here, $v_{t}$ can be thought of as the liquidation value of the stock or, more broadly, the conditional expectation of its future cash flow ($\text{FCF}_{t}$) given the collective private information of investors, $\mathcal{I}_{t}^{ \text{private}}$, that is, $v_{t} = \mathbb{E}^{ \mathbb{Q}} \big[ \text{FCF}_{t} \mid \mathcal{I}_{t}^{ \text{private}} \big]$.  Since $\mathcal{I}_{t}^{ \text{public}} \subseteq \mathcal{I}_{t}^{ \text{private}}$, it follows by the dominance of the coarser filtration that $P_{t} = \mathbb{E}^{ \mathbb{Q}} \big[ \text{FCF}_{t} \mid \mathcal{I}_{t}^{ \text{public}} \big]$.

If, at each point in time, the efficient price is determined in accordance with \eqref{equation:efficient-price-Q}, we expect the release of large amounts of new public information to trigger large, and essentially instantaneous, price changes $\Delta P_{t} \neq 0$ (i.e., a jump). We take \textit{price discovery} to represent the process by which new public information gets impounded into the efficient price; a common definition in the literature \citep[e.g.][]{andersen-bollerslev-diebold-vega:03a, hasbrouck:91a}.

In a second strand of the literature, price discovery refers to how private information gets incorporated into the efficient price \citep[e.g.][]{biais-hillion-spatt:99a, vives:95a} through a tatonnement process in which the trades of informed parties reveal their private information \citep[e.g.][]{glosten-milgrom:85a, kyle:85a}. In that setting, $P_{t}$ converges toward $v_{t}$ (a moving target in practice). This definition is ill-suited for studying the pricing implications of earnings announcements as we do not expect the release of previously-private now-public information to alter the equilibrium price which depends on the firm's underlying business activity, the state of the economy, etc. In contrast, the efficient price determined according to \eqref{equation:efficient-price-Q} is highly likely to change.\footnote{Although we are not measuring how far the efficient price is from the equilibrium price, we do expect the post-announcement price to be closer to the equilibrium price than the pre-announcement price.}

Changing to the physical probability measure, \eqref{equation:efficient-price-Q} can be rewritten as
\begin{equation} \label{equation:efficient-price}
P_{t} = \mathbb{E} \big[v_{t} \mid \mathcal{I}_{t}^{ \text{public}} \big] + \text{cov} \big[v_{t}, m_{t} \mid \mathcal{I}_{t}^{ \text{public}} \big],
\end{equation}
where $m_{t}$ is the stochastic discount factor.

The covariance term in \eqref{equation:efficient-price} represents a risk premium that allows prices to jump even for earnings announcements that are equal to or close to expectations and do not change the first term much. Whether this effect arises depends on the underlying asset pricing model. Specifically, \cite{ai-bansal:18a} show that asset pricing models cannot generate an announcement premium, defined as the return from buying an asset prior to the announcement and selling it in its immediate aftermath, under time-separable expected utility. However, a class of generalized risk sensitivity preferences, including robust-control (uncertainty aversion), implies concavity of the certainty equivalent function. This causes additional risk-discounting for payoffs that correlate positively with utility, which generates a non-negative announcement risk premium \citep[see also][]{hansen:21a}. This also holds under \citet{epstein-zin:89a} recursive utility when these lead to a preference for early resolution of uncertainty. Earnings announcements in line with analyst expectations are therefore more likely to cause the stock price of the announcing firm to rise than decline due to a reduced risk premium, that is $\Delta \text{cov} \big[v_{t}, m_{t} \mid \mathcal{I}_{t}^{ \text{public}} \big] \geq 0$.

This discussion suggests the following hypothesis.

\begin{hypothesis} \label{hypothesis:necessary}
Necessary conditions for market efficiency after an earnings announcement include that, for the most liquid stocks,
\begin{enumerate}
\item[1.] The stock price of the announcing firm almost always jumps immediately after the earnings announcement.
\item[2.] Due to uncertainty resolution, the stock price of the announcing firm is more likely to jump up than down following an earnings announcement in line with analyst expectations.
\end{enumerate}
\end{hypothesis}

Earnings announcements may contain firm-, industry-, and market-wide components. Therefore, investors revise their expectations not only about the prospects of the announcing firm but also of other firms and the broader economy  \citep{patton-verardo:12a, savor-wilson:16a}. We therefore expect earnings announcements to induce spillover effects that increase the likelihood of jumps in the prices of non-announcing firms and the market index. This effect is arguably stronger for non-announcing firms with a closer industry proximity and high after-hours market trading activity, where the latter is required to facilitate the price discovery process. We also expect the information effect---and hence the jump spillover---to be enhanced for companies that announce their fiscal results early in the earnings season, since investors revise their beliefs the most in the beginning of the cycle, as documented in \citet{savor-wilson:16a}.

\begin{hypothesis} \label{hypothesis:spillover}
Information and learning-induced spillover effects in jump probabilities from the announcing to non-announcing firms and the broader market index imply that:
\begin{enumerate}
\item[1.] The prices of other (non-announcing) firms are more likely to jump on days with earnings announcements.
\item[2.] The increase in the jump probability for a non-announcing firm is higher if (i) the non-announcing firm is closer to the announcing firm (in terms of industry proximity); (ii) the announcing firm is early in the news announcement cycle; and (iii) the announcing and non-announcing firm are more liquid with higher after-hours market trading activity.
\item[3.] The market index is more likely to jump on days with earnings announcements and the jump probability is higher on days where multiple firms announce their earnings.
\end{enumerate}
\end{hypothesis}

Higher jump probabilities in stock prices after an earnings announcement are a weak implication in the sense that it is a necessary but not sufficient condition for market efficiency. This holds because stock prices could over- or undershoot. For example, prices could jump too far, giving rise to a mean-reverting component that drifts toward the pre-announcement price in the aftermath of the jump. If instead prices do not jump far enough, we should expect to see post-announcement prices drift in the direction of the jump, as evident in the post-earnings announcement drift (PEAD) literature \citep{bernard-thomas:89a, bernard-thomas:90a}.

These examples suggest that a sufficient condition for market efficiency after earnings announcements is that investors cannot detect return predictability at any point on the post-announcement price path and exploit such predictability after accounting for trading costs:\footnote{Efficient prices are only determined up to the marginal cost of acquiring and processing information and executing transactions \citep*{grossman-stiglitz:80a, pedersen:15a}. In principle, we should also account for risk premia, but since the returns from our proposed trading strategy are accrued over very short horizons, we do not expect such a correction to alter the results much.}

\begin{hypothesis} \label{hypothesis:sufficient}
A sufficient condition for market efficiency after an earnings announcement is that, at each point in time of the post-announcement price path:
\begin{enumerate}
\item[1.] Post-announcement returns of the announcing firm do not display any predictive patterns large enough to allow investors to earn abnormal returns after accounting for transaction costs.
\item[2.] Post-announcement returns of the non-announcing firms and the market index do not display any predictive patterns large enough to allow investors to earn abnormal returns after accounting for transaction costs.
\end{enumerate}
\end{hypothesis}

Hypothesis \ref{hypothesis:sufficient} does not mean that the price discovery process is settled instantaneously. We expect elevated volatility of the price process to last several minutes after the release of the news as information gets processed and trading positions are adjusted. However, the price discovery process should not leave any (local) biases allowing investors to predict future price movements.

The condition that return predictability should not be exploitable is important and goes back to \citet{jensen:78a}.\footnote{He defines market efficiency as follows: ``A market is efficient with respect to information set $\theta_{t}$ if it is impossible to make economic profits by trading on the basis of information set $\theta_{t}$.''} Trading costs and bid-ask spreads may induce negative autocorrelation in returns, but this should not generate abnormal profits after accounting for round-trip transaction costs. Time-varying risk premia may induce predictability in the price process but they are unlikely to matter at the high resolution studied here.

Testing the sufficient conditions in Hypothesis \ref{hypothesis:sufficient} is challenging because it has to hold for all possible trading rules that condition on the prevailing information. Short of inspecting the profitability of all possible trading rules based on the released information, it is difficult to conduct an exhaustive test of the condition. In practice, analysis typically examines simple strategies based on news measures such as earnings surprises. Conclusions about market efficiency are therefore limited to the class of trading rules being examined.

\subsection{Comparison to existing tests of market efficiency}
Our approach to testing market efficiency by examining the presence of jumps after the disclosure of public information is fundamentally different from common practice in the finance literature on testing market efficiency in the aftermath of macroeconomic news announcements \citep[e.g.,][]{ederington-lee:93a, andersen-bollerslev-diebold-vega:03a} or corporate earnings announcements \citep[e.g.,][]{beaver:68a, chambers-penman:84a, jiang-likitapiwat-mcinish:12a, gregoire-martineau:22a, lyle-stephan-yohn:21a}.

This literature typically studies predictability in post-announcement returns over fixed time intervals such as one or five minutes. For example, \cite{ederington-lee:93a} conclude that ``\textit{Following an announcement, traders with immediate access to the market apparently form an estimate of the release's implication for market prices almost immediately, and the actual price adjusts to this level within one minute. The price level at the end of one minute of trading is a relatively unbiased estimate of the final equilibrium price.}'' \citet{andersen-bollerslev-diebold-vega:03a} observe that ``\textit{The general pattern is one of very quick exchange-rate conditional mean adjustment, characterized by a jump immediately following the announcement, and little movement thereafter.}'' Importantly, the evidence of jumps presented by the latter relies on inspection of price movements over five-minute intervals and their sampling scheme does not allow them to formally test for instantaneous jumps in exchange rates.

In contrast, our hypotheses explicitly do \textit{not} reference time intervals, such as one or five minutes. Rather, jumps must occur as close to instantaneously as trading technology allows (Hypotheses \ref{hypothesis:necessary} and \ref{hypothesis:spillover}) and the sufficient conditions in Hypothesis \ref{hypothesis:sufficient} must hold at each point in time, i.e., on the entire price path, regardless of whether we are studying fixed or varying intervals in calendar or trading time. This is an important advantage because any choice of time interval used in conventional market efficiency tests is arbitrary.

More importantly, in the current fast-paced markets transactions can occur at high frequency where it is possible to execute thousands of trades at very short intervals measured in milliseconds or less. The absence of serial correlation in, e.g., one-minute returns is therefore no guarantee that prices move efficiently. In a world dominated by algorithmic trading with latency delays shrinking ever closer to the lower bound imposed by the speed of light, only jump tests can reveal whether markets react efficiently to big news events.

To illustrate the extreme speed with which transactions occur in the immediate aftermath of an announcement, Panel A of Figure \ref{figure:warpspeed} shows tick-by-tick evidence on the price movement and transaction count from the announcement (time 0) to 60 seconds after Apple (AAPL) released its earnings report on 07/30/2020; Panel B zooms further in on the first 100 milliseconds. Transactions  begin 15 milliseconds after the announcement and their toll exceeds more than 800 in the first post-announcement second. The fastest trades can possibly be attributed to algorithmic traders sniping stale quotes in the limit order book, while the flat spots can appear when a large amount of passive resting volume (e.g, an iceberg order) is executed against a sequence of smaller trades. Trading activity then lies dormant for a few seconds before picking up significant steam again around the time where slower financial intermediaries are likely to enter the market. On average, in the first ten seconds 300 trades execute per second, while in the one-minute post-announcement window more than 165 trades get executed per second.\footnote{toward the end of our sample, trading activity in the after-hours market is much higher than in the early years with high-speed electronic trading often exceeding 200 trades per second right after an earnings announcement.}

This example highlights the speed with which prices of the most liquid stocks adjust after earnings announcements. It also motivates the need for new market efficiency tests that account for such extremely fast trading, yet possess robustness to the high levels of noise typically encountered in after-hours markets. This is the topic we next turn to.

\subsection{A noise-robust jump test}

The vast majority of corporate earnings announcements occur in the after-hours market.\footnote{Perhaps for this reason, existing empirical work based on high-frequency data from the regular trading session is somewhat inconclusive about the importance of jumps relative to the diffusive volatility (as often modeled by the inclusion of a Brownian motion in the price process), see, e.g., \citet*{christensen-oomen-podolskij:14a}.} Here, the trading volume is typically far lower and bid-ask spreads much wider compared to the regular trading session. As we demonstrate below, in the presence of such microstructure frictions, conventional jump tests tend to spuriously identify false jumps, while often failing to detect true jumps. This shortcoming makes it important to have tests that are robust to the noised-up features of high-frequency data from the after-hours market.\footnote{\citet*{ait-sahalia-jacod-li:12a} construct a noise-robust jump test from power variations sampled at different frequencies, but it is not evident how one should implement such sampling in the after-hours market that is typically illiquid with limited trading. In the simulation analysis, we show that our noise-robust bipower variation-based jump test is on par with or even outperforms their test statistic.}

Given the popularity in empirical work of the jump test of \citet*{barndorff-nielsen-shephard:06a}, we develop a noise-robust generalization which retains the intuitive nature of their test. We state the main results below with detailed definitions and formal proofs provided in Appendix \ref{appendix:proof}.

\subsubsection{The efficient price}

We start by introducing a general continuous-time setup for modeling noisy high-frequency data. Consider a security price observed on the time interval [0,1] which we interpret as a complete daily trading session as described in Section \ref{section:data}. The fundamental theorem of asset pricing implies that in an efficient market without frictions and no arbitrage the price process, $P = (P_{t})_{t \geq 0}$, must be a semimartingale \citep*[e.g.,][]{delbaen-schachermayer:94a}. Starting with the log-price process $p = (p_{t})_{t \geq 0}$, where $p_{t} = \log(P_{t})$, we can decompose the cumulative intraday log-return at time $t$, $r_{t} = p_{t} - p_{0}$, into a continuous part and a jump, or discontinuous, part:
\begin{equation} \label{equation:X}
r_{t} = r_{t}^{c} + r_{t}^{d}.
\end{equation}
In Appendix \ref{appendix:proof}, we introduce a general but succinct mathematical definition of these two separate sources of risk. Our framework is nonparametric and model-free, since we do not constrain the dynamics of the efficient price process, and is compatible with salient features of high-frequency financial data such as time-varying expected returns, stochastic volatility, leverage effects, and jumps in the price and volatility. Furthermore, the price jump component can be infinitely active and of infinite variation.

The null hypothesis that we wish to test is that $P$ is continuous. This can be formulated as $\mathcal{H}_{0} : \omega \in \Omega_{0}$, where $\Omega_{0} \subseteq \Omega$ is the subset:
\begin{equation}
\Omega_{0} = \{ \omega \in \Omega : t \longmapsto P_{t}( \omega) \text{ is continuous on } [0,1] \}.
\end{equation}
We test this null, which is equivalent to the restriction $r_{t}^{d} = 0$ for all $t$, against the alternative $\mathcal{H}_{a} : \omega \in \Omega_{1}$ with $\Omega_{1} = \Omega_{0}^{ \complement}$, the complement of $\Omega_{0}$.

To develop a test of $\mathcal{H}_{0}$, we examine the composition of the quadratic return variation:
\begin{equation} \label{equation:qv}
[r]_{1} = \int_{0}^{1} \sigma_{s}^{2} \text{d}s + \sum_{0 \leq s \leq 1} ( \Delta r_{s}^{d})^{2},
\end{equation}
where $\sigma_{s}$ is the instantaneous volatility at time $s$, $\Delta r_{s}^{d} = r_{s}^{d} - r_{s-}^{d}$ is the associated jump size, and $r_{s-}^{d} = \lim_{t \uparrow s}r_{t}^{d}$. The first term in \eqref{equation:qv} is the quadratic variation of the continuous log-return, also known as the integrated variance. The second term represents the quadratic variation of the jump process, which is zero under $\mathcal{H}_{0}$.

Quadratic variation can equivalently be defined as follows:
\begin{equation}
[r]_{1} = \underset{n \rightarrow \infty}{ \text{plim}} \sum_{i=1}^{n} r_{i}^{2},
\end{equation}
where $r_{i} = p_{t_{i}} - p_{t_{i-1}}$ is the log-return over $[t_{i-1}, t_{i}]$ and the convergence in probability holds for every partition $0 = t_{0} < t_{1} < \ldots < t_{n} = 1$ such that $\underset{1 \leq i \leq n}{ \max_{i}}(t_{i} - t_{i-1}) \overset{p}{ \longrightarrow} 0$ as $n \rightarrow \infty$, motivating why tick-by-tick high-frequency data is used to estimate quadratic variation.

\subsubsection{The observed price}

In practice, the efficient market hypothesis is violated because of market frictions---or microstructure noise---such as price discreteness and bid--ask spreads. In addition, even though the price process evolves in continuous-time, trading is discrete. In this section, we suppose for notational convenience that the observed log-price is recorded regularly at equidistant time points $t_{i} = i \Delta_{n}$, for $i = 0,1, \dots, n$, where $\Delta_{n} = 1/n$ is the time gap, which we collect in $p^{*} = (p_{i \Delta_{n}}^{*})_{i=0}^{n}$.\footnote{In the empirical application, we employ tick time sampling, which leads to irregularly and randomly spaced observation times. We elaborate on the details in that section. However, it is important to note that the pre-averaging theory presented below remains compatible with such data, as long as $\underset{1 \leq i \leq n}{ \max_{i}}(t_{i} - t_{i-1}) \overset{p}{ \longrightarrow} 0$ as $n \rightarrow \infty$.}

To account for microstructure noise, we follow \citet*{hasbrouck:95a} and model $p_{i \Delta_{n}}^{*}$ as
\begin{equation} \label{equation:observed-price}
p_{i \Delta_{n}}^{*} = p_{i \Delta_{n}} + \epsilon_{i \Delta_{n}},
\end{equation}
where $\epsilon_{i \Delta_{n}}$ is the microstructure effect.

The key challenge is to draw inference about the presence of jumps in the efficient price from a high-frequency record of noisy prices. To see this, note from \eqref{equation:observed-price} that the observed log-return measures the efficient log-return with error due to the microstructure component:
\begin{equation} \label{equation:noisy-return}
r_{i}^{*} = p_{i \Delta_{n}}^{*} - p_{(i - 1) \Delta_{n}}^{*} = r_{i} + \epsilon_{i \Delta_{n}} - \epsilon_{(i - 1) \Delta_{n}}.
\end{equation}

\subsubsection{The pre-averaging approach}

To mitigate the effect of microstructure noise, we adopt the pre-averaging procedure of \citet*{jacod-li-mykland-podolskij-vetter:09a, podolskij-vetter:09a, podolskij-vetter:09b}. Specifically, we replace the observed noisy log-return by a pre-averaged log-return:
\begin{equation} \label{equation:pre-averaged-return}
\bar{r}_{i}^{*} = \sum_{j=1}^{k_{n}-1} g_{j}^{n} r_{i+j}^{*}, \quad \text{for } i = 0, \ldots, n - 2 k_{n} + 1,
\end{equation}
where $g_{j}^{n} = g(j/k_{n})$ is a weight function, $k_{n} = \theta \sqrt{n} + o\bigl( n^{-1/4} \bigr)$ is the pre-averaging horizon, and $\theta > 0$ is a tuning parameter.

The averaging in \eqref{equation:pre-averaged-return} attenuates the microstructure noise and enhances the efficient log-return. If $k_{n}$ is even and $g(x) = \min(x,1-x)$, \eqref{equation:pre-averaged-return} can be rewritten as
\begin{equation}
\bar{r}_{i}^{*} = \frac{1}{k_{n}} \sum_{j=k_{n}/2+1}^{k_{n}} p_{(i+j) \Delta_{n}}^{*} - \frac{1}{k_{n}} \sum_{j=1}^{k_{n}/2} p_{(i+j) \Delta_{n}}^{*}.
\end{equation}
Thus, $( \bar{r}_{i}^{*})_{i=1}^{n-k_{n}+2}$ constitutes a sequence of log-returns constructed by simple averaging of the noisy log-price. With minor modifications, standard estimators of quadratic variation can therefore be tweaked to exploit the pre-averaged log-return series. We next define the pre-averaged realized variance and pre-averaged bipower variation as follows:
\begin{align} \label{equation:rv}
\begin{split}
RV_{n}^{*} &= c_{1}^{n} \sum_{i = 0}^{n - 2 k_{n} + 1} | \bar{r}_{i}^{*}|^{2} - c_{2}^{n} \hat{ \omega}_{n}^{2}, \\
BV_{n}^{*} &= c_{1}^{n} \frac{ \pi}{2} \sum_{i = 0}^{n - 2k_{n} + 1} | \bar{r}_{i}^{*}| | \bar{r}_{i+k_{n}}^{*}| - c_{2}^{n} \hat{ \omega}_{n}^{2},
\end{split}
\end{align}
where
\begin{equation} \label{equation:noise-variance}
\hat{ \omega}_{n}^{2} = \rho_{n}(0) + 2 \sum_{m = 1}^{ \ell_{n}} \rho_{n}(m),
\end{equation}
\begin{equation}
\rho_{n}(m) = \frac{1}{n-5h_{n}+1} \sum_{i=0}^{n-5h_{n}} (p^{*}_{i \Delta_{n}} - \tilde{p}_{(i + 2h_{n}) \Delta_{n}}^{*})(p^{*}_{(i+m) \Delta_{n}} - \tilde{p}_{(i + 4h_{n}) \Delta_{n}}^{*}),
\end{equation}
for $m = 0, 1, \dots, \ell_{n}$, and
\begin{equation}
\tilde{p}_{i \Delta_{n}}^{*} = \frac{1}{h_{n}} \sum_{j=0}^{h_{n}-1} p_{(i+j) \Delta_{n}}^{*}.
\end{equation}

$RV_{n}^{*}$ is the sum of squared pre-averaged returns, whereas $BV_{n}^{*}$ is based on cross-products that are $k_{n}$ terms apart. As we prove in Theorem \ref{theorem:clt}, the extra distancing ensures that $BV_{n}^{*}$ is asymptotically jump-robust, whereas $RV_{n}^{*}$ is not. The scalars $c_{1}^{n}$ and $c_{2}^{n}$ are defined in Appendix \ref{appendix:proof} and depend on the weight function, $g$, and pre-averaging window, $k_{n}$.

In the above, both estimators are normalized for the effect of pre-averaging and bias-corrected for residual microstructure noise. The bias-correction employs the estimator in \eqref{equation:noise-variance} due to \citet{jacod-li-zheng:19a}. It was designed to estimate the long-run noise variance under endogeneity, heteroscedasticity, and a mixing condition allowing the autocovariances of the noise to decay at a polynomial rate (see Assumption (N) in Appendix \ref{appendix:proof}). The calculation involves a second pre-averaging parameter, $h_{n}$. The intuition is that in order to construct the sample autocovariances of the noise, we need to purge the efficient log-price component from \eqref{equation:observed-price}. Hence, demeaning of $p_{i \Delta_{n}}^{*}$ is done with a local average of the noisy log-price process (a proxy of $p_{i \Delta_{n}}$). In the empirical application, we follow the implementation from \citet{jacod-li-zheng:19a}, adopting a data-driven choice of lag length, $\ell_{n}$.

It is worth pointing out that since the main tool for jump identification is going to be based on $RV_{n}^{*} - BV_{n}^{*}$, the bias-correction in \eqref{equation:rv} cancels out in the construction of our jump test statistic. It therefore has no impact on the testing procedure but merely affects those parts of the analysis, where a pre-averaged realized measure is employed on a stand-alone basis. Even there, however, its magnitude is negligible. In our sample, the bias-correction term is an order of magnitude smaller than the uncorrected pre-averaging statistic, which is consistent with \citet{christensen-hounyo-podolskij:18a}. This suggests that our empirical findings are not in any way driven by the concrete choice of long-run noise variance estimator adopted for the bias-correction term.

\subsubsection{Asymptotic theory}

We next develop the theoretical foundation for the construction of our noise-robust jump test procedure.

\begin{theorem} \label{theorem:clt} Assume that $r$ follows the process in \eqref{equation:X} and that Assumptions (V) and (N) in Appendix \ref{appendix:proof} hold. Then, as $n \rightarrow \infty$,
\begin{equation}
RV_{n}^{*} \overset{p}{ \longrightarrow} [r]_{1} \qquad \mathrm{and} \qquad BV_{n}^{*} \overset{p}{ \longrightarrow} \int_{0}^{1} \sigma_{s}^{2} \mathrm{d}s.
\end{equation}
Moreover, provided that $\mathcal{H}_{0}$ is true, as $n \rightarrow \infty$,
\begin{equation} \label{equation:clt}
n^{1/4}
\begin{pmatrix}
RV_{n}^{*} - \int_{0}^{1} \sigma_{s}^{2} \mathrm{d}s \\[0.10cm] BV_{n}^{*} - \int_{0}^{1} \sigma_{s}^{2} \mathrm{d}s
\end{pmatrix}
\overset{d_{s}}{ \longrightarrow} N(0, \Sigma),
\end{equation}
where ``$\overset{d_{s}}{ \longrightarrow}$'' denotes convergence in law stably \citep[e.g.][Section 2.2.1]{jacod-protter:12a} and $\Sigma$ is the $2\times 2$ random asymptotic covariance matrix defined in \eqref{asymptotic-covariance-Sigma} in Appendix \ref{appendix:proof}.
\end{theorem}

Theorem \ref{theorem:clt} establishes a weak law of large numbers for the pre-averaged realized variance and pre-averaged bipower variation in the noisy jump-diffusion model under $\mathcal{H}_{a}$. It shows that $RV_{n}^{*}$ is consistent for the quadratic variation, including the jump part, while $BV_{n}^{*}$ converges to the integrated variance, excluding the jump part. Their difference therefore estimates the quadratic jump variation, $RV_{n}^{*} - BV_{n}^{*}  \overset{p}{ \longrightarrow} \sum_{0 \leq s \leq 1} ( \Delta r_{s}^{d})^{2}$ which is zero in the absence of jumps. The last part of Theorem \ref{theorem:clt} establishes a bivariate asymptotic mixed normal distribution of $RV_{n}^{*}$ and $BV_{n}^{*}$ as estimators of integrated variance under $\mathcal{H}_{0}$.

Applying the delta rule to \eqref{equation:clt}, it follows from Theorem \ref{theorem:clt} that, under $\mathcal{H}_{0}$:
\begin{equation} \label{equation:infeasible-test-statistic}
\mathcal{J}_{n}^{ \text{inf.}} = \frac{n^{1/4} \big( RV_{n}^{*} - BV_{n}^{*} \big)}{\sqrt{v^{ \top}  \Sigma v}} \overset{d}{ \longrightarrow} N(0,1),
\end{equation}
where $v = [1,-1]^{ \top}$. In contrast, $\mathcal{J}_{n}^{ \text{inf.}} \overset{p}{ \longrightarrow} \infty$ under the alternative, $\mathcal{H}_{a}$, so large values of the test statistic point toward the presence of price jumps.

\subsubsection{A noise- and jump-robust covariance estimator}

The test statistic $\mathcal{J}_{n}^{ \text{inf.}}$ is not feasible since it depends on the covariance matrix $\Sigma$, which is a function of the latent volatility and microstructure noise processes. We therefore next develop a noise- and jump-robust plug-in estimator of $\Sigma$.\footnote{In principle, we only need to estimate $\Sigma$ given that $\mathcal{H}_{0}$ is true. However, by using a jump-robust estimator of the asymptotic covariance matrix we avoid losing power under $\mathcal{H}_{a}$.} Specifically, we propose to set
\begin{equation} \label{equation:Sigma-n}
\Sigma_{n}^{*} = \frac{1}{L} \sum_{l = 1}^{L} \Bigg( \frac{n^{1/4}}{ \sqrt{L}} \begin{bmatrix}  \check{BV}_{l}^{*}(2,0) - \check{BV}_{n}^{*}(2,0) \\[0.10cm] \check{BV}_{l}^{*}(1,1) - \check{BV}_{n}^{*}(1,1) \end{bmatrix} \Bigg) \Bigg( \frac{n^{1/4}}{ \sqrt{L}} \begin{bmatrix}  \check{BV}_{l}^{*}(2,0) - \check{BV}_{n}^{*}(2,0) \\[0.10cm] \check{BV}_{l}^{*}(1,1) - \check{BV}_{n}^{*}(1,1) \end{bmatrix} \Bigg)^{ \top},
\end{equation}
where $\check{BV}_{n}^{*}(q,r)$ is a truncated pre-averaged bipower variation of \citet*{christensen-hounyo-podolskij:18a}. The details are in Appendix \ref{appendix:subsampler}.

$\Sigma_{n}^{*}$ is the sample covariance matrix of truncated subsampled pre-averaged bipower variation estimators, $\check{BV}_{l}^{*}(q,r)$, calculated on small batches of high-frequency data covering blocks of length $pk_{n}$, for $l = 1, \dots, L$, where $p \geq 2$ is an integer, and $L$ is the number of subsamples. They are centered around the full sample statistic, $\check{BV}_{n}^{*}(q,r)$, and suitably normalized.

\begin{theorem} \label{theorem:subsampler}
Assume that $r$ follows the process in \eqref{equation:X}, that Assumptions (V) and (N) in Appendix \ref{appendix:proof} hold, as do the conditions listed in Theorem \ref{theorem:subsampler-general} therein. Then, as $n \rightarrow \infty$,
\begin{equation}
\Sigma_{n}^{*} \overset{p}{ \longrightarrow} \Sigma.
\end{equation}
\end{theorem}
Theorem \ref{theorem:subsampler} shows that $\Sigma_{n}^{*}$ is consistent both under the null and alternative. Our proof confines attention to jump processes of bounded variation, which is standard when inference is based on the continuous part of the process.

By Slutsky' theorem, we obtain the feasible jump test statistic:
\begin{equation} \label{equation:feasible-test-statistic}
\mathcal{J}_{n} = \frac{n^{1/4} \big( RV_{n}^{*} - BV_{n}^{*} \big)}{\sqrt{v^{ \top}  \Sigma_{n}^{*} v}}.
\end{equation}
The finite sample properties of $\mathcal{J}_{n}$ under the alternative can be further improved if we also replace $BV_{n}^{*}$ in the numerator of \eqref{equation:feasible-test-statistic} with a truncated version:\footnote{On the one hand, ``big'' jumps induce an upward bias in $BV_{n}^{*}$ for realistic sample sizes, causing a downward bias in the estimated jump variation and reducing the power of the jump test. However, such biases can readily be handled through truncation. On the other hand, the bipower mechanism---multiplication of adjacent log-returns---is an effective tool to get rid of ``small'' jumps that survive any truncation. This approach of cracking down twice on the jump component was proposed by \citet*{corsi-pirino-reno:10a} for bipower variation in the noise-free setting. In Appendix \ref{appendix:proof}, we show that the substitution of $BV_{n}^{*}$ with $\check{BV}_{n}^{*}$ in the numerator does not alter the asymptotic theory in Theorem \ref{theorem:clt}.}
\begin{equation} \label{equation:test-statistic}
\mathcal{J}_{n} = \frac{n^{1/4} \big( RV_{n}^{*} - \check{BV}_{n}^{*} \big)}{\sqrt{v^{ \top}  \Sigma_{n}^{*} v}},
\end{equation}
where $\check{BV}_{n}^{*} \equiv \check{BV}_{n}^{*}(1,1)$.

The final form of $\mathcal{J}_{n}$ in \eqref{equation:test-statistic} serves as our noise-robust jump test statistic.\footnote{\citet*{barndorff-nielsen-shephard:06a} advocate a log- and ratio-based transformation of the noise-free version of \eqref{equation:test-statistic} via the delta method. We implemented both these alternative versions of the noise-robust jump test statistic, which does not impact any of our conclusions.} A one-sided test procedure is appropriate here, so the decision rule is to reject $\mathcal{H}_{0}$ at significance level $\alpha$ if $\mathcal{J}_{n} > \Phi^{-1}(1 - \alpha)$, where $\Phi(x)$ is the standard normal distribution function. This implies that:
\begin{equation}
\mathbb{P} \Big( \mathcal{J}_{n} >  \Phi^{-1}(1 - \alpha) \Big) \rightarrow
\begin{cases}
\alpha, & \text{under } \mathcal{H}_{0}, \\[0.10cm]
1, & \text{under } \mathcal{H}_{a},
\end{cases}
\end{equation}
so our test is asymptotically unbiased and consistent.

\subsection{Monte Carlo study}

In Appendix \ref{appendix:simulation}, we conduct a thorough Monte Carlo investigation of the finite-sample properties of our jump test. We simulate from a general jump-diffusion process which gets affected by different forms of microstructure noise with an extreme magnitude matching what we observe empirically in the after-hours trading session. Moreover, we investigate the robustness of our test to various degrees of pre-averaging, subsampling, and jump-truncation.

We compare our approach to the noise-robust test of \citet{ait-sahalia-jacod-li:12a} and the noise-free version from \cite{barndorff-nielsen-shephard:06a}. The former is also based on pre-averaging to compress the noise but the building block in the test statistic is the ratio of power variation estimators sampled at different frequencies, in contrast to the difference of multipower variation estimators sampled at identical frequencies as advocated here.\footnote{In Appendix \ref{appendix:close-to-open}, we inspect the jump test of \citet{lee-mykland:08a, lee-mykland:12a}. Whereas our test \citep[and][]{ait-sahalia-jacod-li:12a} targets the cumulative return variation over a discrete time interval, i.e., quadratic variation and integrated variance, the latter is a point-in time test based on spot volatility. In principle, a local jump test over a short post-announcement window is more appropriate for analyzing price responses to earnings announcements because we know the announcement time (up to minimal rounding). However, the regime-switching nature of post-market high-frequency data—–in particular, the complete lack of liquidity on non-announcement days—renders their test unsuited for the empirical analysis conducted in this paper. Our reliance on integrated measures of return variation over the regular and extended trading session leads to a direct measure of incremental volatility in the after-hours market (typically large on announcement days and close to zero on non-announcement days), which can be further dissected into a continuous variation and jump component. In the appendix, we adapt the noise-free \citet{lee-mykland:08a} testing approach by designing a close-to-open return-based jump test that can be applied to the overnight period even when high-frequency data for the extended trading session is not available. We implemented this version of their test on a much larger cross-section of generally less liquid stocks. The conclusion from this exercise supports the corresponding evidence below.}

The key findings can be summarized as follows.
\begin{enumerate}
\item Our test has correct size and does not systematically over- or under-reject $\mathcal{H}_{0}$. The sole exception occurs when the noise is fat-tailed and the pre-averaging window is too short to adequately accommodate this which leads to a modest overrejection. Under $\mathcal{H}_{a}$, our test has good power properties and correctly identifies the jump component.

\item The \citet{ait-sahalia-jacod-li:12a} test is slightly more sensitive to the structure of microstructure noise. In particular, their test is very conservative under $\mathcal{H}_{0}$, leading to loss of power under $\mathcal{H}_{a}$, with conditionally autocorrelated noise. We emphasize that this setting was not allowed in their theoretical framework. We do expect that their test statistic can be modified to exhibit good properties also in that setting, though this version remains to be developed. Otherwise, the conclusion of their test is often close to ours.

\item In sharp contrast, the \citet{barndorff-nielsen-shephard:06a} test based on 5-minute realized variance and bipower variation degrades in the presence of the high level of microstructure noise typical for after-hours trading. In particular, it identifies too many spurious jumps, falsely rejecting the null when no jumps are present. Conversely, this test frequently fails to identify true jumps, falsely accepting the null when jumps are genuinely present.
\end{enumerate}

In summary, our simulations demonstrate both the need for and benefits from utilizing our noise-robust jump test on the high-frequency after-hours transaction price data examined in our empirical application, which we next proceed to.

\section{Data description} \label{section:data}

Earnings announcements should trigger jumps in stock prices if investors efficiently and instantaneously process the new earnings information. In this section, we introduce the data sources employed to study if this prediction is supported empirically.

\subsection{Transaction and quotation data}

Jump tests benefit crucially from having data at the highest possible resolution, so our analysis employs tick-by-tick transaction and quotation data on individual firms' stock prices to track their response to corporate earnings announcements.

U.S. stock exchanges are only open for trading on normal business days from 9:30am--4:00pm Eastern time.\footnote{The major national stock exchanges close at 1:00pm the day before big public holidays, such as July 4th, Thanksgiving, and Christmas.} On most days, official pre-market trading is available from 4:00am--9:30am, while the after-hours market runs from 4:00pm--8:00pm.\footnote{The trading schedule of NASDAQ- and NYSE-listed securities can be gauged at \path{https://www.nasdaq.com/stock-market-trading-hours-for-nasdaq} and \path{https://www.nyse.com/markets/nyse-arca/market-info}.} Orders can additionally be submitted and trades executed at any time via private trading systems such as electronic communication networks (ECNs) or dark pools. The consolidated tape collates such real-time tick-by-tick data, which is made available for purchase via the New York Stock Exchange (NYSE) Trade and Quote (TAQ) database. In practice, there is usually limited trading activity before 6:00am or after 6:30pm, so we exclude this part of the day.

As noted earlier, the vast majority of U.S. companies publish their financial results outside the regular trading session. They may prefer to announce at these times because these markets are mostly comprised of professional investors and informed traders, speeding up the price discovery process \citep*[see][]{jiang-likitapiwat-mcinish:12a}. In Appendix \ref{appendix:pre-market}, we examine the trading activity in the pre-market (6:00am--9:30am) and after-hours market (4:00pm--6:30pm). We find that trading volume is much lower in the pre-market, even on announcement days. This impairs the implementation of our jump test so our analysis focuses on the after-hours market. Following \citet*{dubinsky-johannes-kaeck-seeger:19a}, we look at the 50 firms with the largest after-hours trading volume on announcement days.\footnote{To identify these firms, we sort the constituents of the S\&P 500 index as of 12/31/2020 by their post-close (4:00pm--6:30pm) transaction volume on quarterly earnings announcement days, averaged over the five-year period from 01/01/2015--31/12/2019. We further require that the firms report at least 75\% of their earnings announcements after the close of the regular trading session during our full sample period 06/02/2008--12/31/2020. We exclude Xilinx (XLNX, ranked 43th) from the original list since the analyst information is missing in the I/B/E/S database and replace it by CSX corporation (CSX).} This is appropriate not only because individual transaction counts are prohibitively large, but also because liquidity evaporates fast as we descend to stocks with lower after-hours market liquidity. The selected stocks account for a sizeable fraction of trading activity in the after-hours market among all firms in the S\&P 500 index during our sample (46.89\% of the total transaction count vs. 21.31\% in the regular trading session).

We download high-frequency data from the NYSE TAQ database for the stocks over the sample period 06/02/2008 through 12/31/2020, a total of $T = 3{,}169$ business days. The data is cleaned using state-of-the-art procedures \citep*[e.g.][]{barndorff-nielsen-hansen-lunde-shephard:09a, christensen-oomen-podolskij:14a}. The main distinction here is that we retain observations with a timestamp outside the official exchange trading hours if they otherwise fulfill the conditions required to be tagged non-erroneous.\footnote{A complete list of sale conditions that can be associated with a transaction is available in the documentation of the Daily TAQ Client Specification manual available for download at: \path{http://www.nyxdata.com/Data-Products/Daily-TAQ}. In particular, the label ``T'' identifies pre-market and after-hours trades. These are virtually always ignored in the high-frequency volatility literature, since the first rule of the widely used \citet*{barndorff-nielsen-hansen-lunde-shephard:09a} filtering algorithm instructs researchers to ``\textit{P1. Delete entries with a time stamp outside the 9:30am--4:00pm window when the exchange is open.}''} In total, we preserve almost ninety billion (89,640,467,708) quotations and nearly eight billion (7,985,202,326) transactions.

We employ tick time sampling, whereby the transaction price process is sampled with every price change, as advocated by \citet{griffin-oomen:08a}. Apart from rendering the log-return series potentially more homogeneous, this approach has the convenient effect of removing transactions that merely repeat the previous price. Thus, we purge zero returns that are known to be detrimental for jump identification \citep[e.g.,][]{bandi-kolokolov-pirino-reno:23a, kolokolov-reno:24a}. Hence, our tick data are inevitably irregularly spaced, but the pre-averaging theory readily conforms with such an extension.\footnote{Whether the sampling scheme is equidistant or irregular, the effective sampling is random, as pointed out by \citet{bandi-kolokolov-pirino-reno:23a}. Pre-averaging complies with sampling random observation times that can be endogenous (i.e. dependent on the price process and microstructure noise), see \citet{jacod-li-zheng:19a, mykland-zhang:16a}. In practice, transaction times are likely endogenous \citep[e.g.][]{kolokolov-livieri-pirino:22a}. This is almost unavoidable when studying price movements after an earnings announcement. The impact of endogenous sampling times on realized variance has been studied by \citet{li-mykland-renault-zhang-zheng:14a} in the noise-free setting.} Finally, we aggregate any remaining observations with identical timestamps and replace them with an average price, a procedure that is known to further alleviate microstructure noise.

Table \ref{table:sp500-volume-pm.tex} reports descriptive statistics about the after-hours trading volume of stocks 1--25, while Table \ref{table:sp500-volume-pm-supplemental.tex} in Appendix \ref{appendix:supplemental} contains the information of stocks 26--50.\footnote{The corresponding table for the pre-market trading activity can be found in Appendix \ref{appendix:pre-market}.} We show the sample average, standard deviation, and the 25th, 50th, 75th and 99th percentiles of the transaction count distribution for days without earnings announcements (columns 2--8) and days with earnings announcements (columns 9--15). For many firms, the after-hours transaction count is two orders of magnitude larger on days with earnings announcements relative to days without them. For example, Facebook (FB) averaged 76,129 after-hours transactions on earnings announcement days, compared with only 891 transactions on days with no such announcements. It is clearly uncommon for individual firms to exhibit high transaction counts in the after-hours trading session on days without earnings announcements.

Transaction volume in the after-hours market has grown substantially between 2008 and 2020 with more than a ninefold increase for the stocks we analyze in our paper. The fraction of after-hours trading has grown even more, quadrupling from 0.17\% to 0.71\%, as a proportion of the total transaction volume (including trades from the regular trading session). On days with earnings announcements this fraction is substantially higher at the firm-level, frequently exceeding 10\%. Appendix \ref{appendix:after-hours-market} summarizes trading volume and spreads in the after-hours market over time.

\subsection{Announcement timing}

While most companies announce their financial results at a fixed time, for others the timing can vary quarter-by-quarter. Even slight inaccuracies in the recording of the announcement time can impede our calculation of the post-announcement price reaction, so accurate timing is imperative and we spend considerable effort on double checking the time stamps.

We begin by gathering announcement times from a data set provided by Wall Street Horizon (\path{https://www.wallstreethorizon.com/}) that offers a comprehensive suite of corporate events data, including detailed information on the timing of earnings announcements. Wall Street Horizon stores the timestamp included with the press release issued by the company when it publicly announces its quarterly results.\footnote{A publicly traded company with classes of securities registered in Section 12 or subject to Section 15(d) of the Securities Exchange Act of 1934 is legally required to file regular reports with the SEC, including an annual form 10-K and quarterly form 10-Q's, in addition to proxy reports and other reporting requirements. Most companies file form 10-Q either 40 or 45 days after the end of the fiscal quarter, but many large corporations choose to summarize their financial statements at an earlier date by issuing a press release and filing a form 8-K.}

As a double-check,  we hand collect a separate set of announcement times from the Factiva news archive by searching on the announcement day for press releases marked with the company ticker code and a designated ``Earnings'' subject.\footnote{A limitation of both data sets is that announcement times are rounded to the nearest minute. However, most of the earnings reports in our sample are released to the public through professional dissemination services, such as Business Wire or PR Newswire, which are under strict requirements to ensure immediate and equal access to company information as noted by the SEC's Regulation Fair Disclosure. The audit-trail information left from the high-frequency data after each announcement suggests that these vendors often publish announcements close to a whole minute such that the rounding effect is typically negligible.}

Comparing the merged set of announcement times with the audit trail of the complete tick-by-tick transaction price history suggests that occasionally both Wall Street Horizon and Factiva get it wrong. To correct seemingly erroneous timestamps, we design a conservative screening algorithm that infers the announcement time based on the price volatility. In particular, we calculate the log-return for each minute from 4:00pm--6:30pm. If the largest absolute one-minute log-return exceeds 2.5\% and precedes the announcement time from both Wall Street Horizon and Factiva, we assume the earnings were released at the earlier time. The filter affects only a very limited number of announcements (77 in total, or 3.5\%). The median change is one minute, suggesting that the algorithm primarily captures small rounding effects for some of the most important earnings reports that lead to substantial revisions of the security price.\footnote{Importantly, our noise-robust jump test statistic does not depend on the announcement time and so is unaltered by the screening algorithm. Moreover, employing only Wall Street Horizon and Factive information does not cause any discernible change in the post-announcement return regression (Section \ref{section:jump-size-determinant}). Only the returns from our trading strategy (Section \ref{section:trading-strategy}) deteriorate slightly, but remain significant with the removal of the inferred announcement time.}

The third-to-last column in Table \ref{table:sp500-volume-pm.tex} shows the modal announcement time for each firm, which ranges from as early as 4:00pm for reports released immediately after the close of the regular exchange-trading session to as late as 4:30pm, with many firms announcing around 4:05pm.

\subsection{Earnings surprises}

Earnings announcements contain a wealth of information but the most important part is arguably the headline earnings figure. To compute the surprise element in earnings announcements, we employ the Institutional Brokers Estimate System (I/B/E/S) database available via subscription to Refinitiv (\path{https://www.refinitiv.com/}), formerly a part of Thomson Reuters. I/B/E/S aggregates earnings information for over 22,000 companies. In addition to storing the actual reported earnings per share (EPS), adjusted for non-recurring items and stock option expenses, prior to each announcement I/B/E/S surveys earnings forecasts from leading professional analysts covering individual companies. We compile this information into a consensus estimate ($\mu_{ \mathrm{EPS}}$) and a standard deviation of that estimate calculated over the distribution of analyst expectations ($\sigma_{ \mathrm{EPS}}$). Following \citet*{berkman-truong:09a}, \citet*{michaely-rubin-vedrashko:13a}, among others, our measure of the standardized unexpected earnings is given by\footnote{Some studies employ the median analyst estimate instead of the sample average \citep*{dellavigna-pollet:09a, gregoire-martineau:22a, hartzmark-shue:18a}. Others replace the standard deviation with the lagged closing price of the stock, e.g. from a week before the announcement or at the end of the previous quarter \citep*{gregoire-martineau:22a, lyle-stephan-yohn:21a}. Our main findings are not altered by adopting these alternative ways of defining the standardized unexpected earnings.}
\begin{equation} \label{equation:eps-surprise-measure}
z_{ \mathrm{EPS}} = \frac{ \mathrm{Actual \ EPS} - \mu_{ \mathrm{EPS}}}{ \sigma_{ \mathrm{EPS}}}.
\end{equation}

We purge announcements with no analyst forecasts and those where $\sigma_{ \mathrm{EPS}}$ is less than 0.001 (one tenth of a cent) which result in an unstable $z_{ \mathrm{EPS}}$ measure. This removes five announcements from Ulta Beauty (ULTA) that was still a relatively overlooked small-cap stock with limited broker coverage at the beginning of our sample. A single announcement from Broadcom (AVGO) is also eliminated. We further exclude announcements where share trading was temporarily suspended by the exchange as it leads to a disconnected trade sequence.\footnote{Trading halts are registered in the daily TAQ master files that contain static information for securities traded by participants of the Consolidated Tape Association (CTA), mainly stocks with NYSE as primary exchange, and Unlisted Trading Privileges (UTP), or NASDAQ-listed issues. The file has a letter code in the ``Halt Delay Reason'' if trading was paused (blank otherwise). Resumption of trading is flagged in the transaction data with a sale condition ``5'' defined as ``Market Center Re-Opening Trade'' (CTA) or ``Re-opening Prints'' (UTP).} Finally, following \citet*{gregoire-martineau:22a} we discard announcements where $z_{ \mathrm{EPS}}$ exceeds ten in absolute value. Such outliers are often caused by extraordinary items in the profit statement.\footnote{For example, on 08/25/2020 Salesforce (CRM) posted a quarterly EPS of \$1.44 after adjusting for non-recurrent items. With a consensus estimate of \$0.66 and a standard deviation of about four cents, this yields a twenty-sigma event. The surprise was motivated by mark-to-market accounting for the company’s investments in nCino, which saw a \$617 million unrealized gain in the quarter.} The last step removes 40 observations (about 1.7\% of the sample).\footnote{The I/B/E/S database incorrectly puts the EPS of Akamai Technologies (AKAM) as 0.31, 0.29, 0.28, and 0.34 for the Q1--Q4 earnings announcement of the fiscal year 2009. This induces a -9.64, -2.12, -8.86, and -7.80 unexpected standardized earnings surprise, despite the company meeting or beating expecations at all occasions. This is possibly caused by AKAM employing non-GAAP financial metrics in their accounting system, in particular a so-called normalized net income, which eliminates the effects of events
that are either not part of the company’s core operations or are non-cash. We replaced the wrong EPS numbers with the correct ones, i.e. 0.43, 0.40, 0.38, and 0.46, which were extracted directly from the PDF files containing the financial statements, available at the company's investor relations website (\path{www.ir.akamai.com}).}

The number of quarterly earnings announcements retained for each company is shown in the fourth-to-last column in Table \ref{table:sp500-volume-pm.tex}. We report summary statistics of $z_{ \mathrm{EPS}}$ in the second-to-last column. Standardized unexpected earnings are substantially higher than zero on average, consistent with the notion that managers engage in earnings smoothing to avoid negative surprises.

\section{Empirical results} \label{section:empirical}

We next proceed with our empirical investigation. Price dynamics and trading activity in the after-hours market is not as thoroughly studied and understood as that in the regular trading session and it has evolved significantly since the seminal analysis of \citet*{barclay-hendershott:03a}. We therefore begin by briefly summarizing some of the key features of the after-hours market, emphasizing the contrast between days with and without earnings announcements. We start by analyzing the pattern in trading volumes and bid-ask spreads in a short interval around earnings announcements before turning to the evidence on the presence of price jumps.

\subsection{Trading volume and bid-ask spread}

To get a broad picture of post-announcement trades, for each five-second interval in a one-hour window centered around the announcement time, we calculate transaction counts and bid-ask spreads (in basis points) relative to the $\mathrm{midquote} = (\mathrm{bid} + \mathrm{ask})/2$ as
\begin{equation}
\text{Spread} = 10000 \times \frac{ \mathrm{ask} - \mathrm{bid}}{ \mathrm{midquote}},
\end{equation}

As a control sample, for each announcement we select a non-announcement date at random and calculate transaction counts and $\text{Spread}$ on a corresponding window.

Panel A in Figure \ref{figure:trade-spread-announcement} reports the transaction count per five-second interval on days with and without earnings announcements. The modal announcement time across companies is 4:05pm so the market typically closes five minutes prior to the announcement (labeled time $0$), as highlighted in the figure. The transaction volume right before the regular market close is notably higher on days with announcements than on days without announcements, but both drop sharply at the close. From this point onward, the graphs evolve very differently. Whereas the after-hours transaction count remains close to zero on non-announcement days, it spikes to more than five contracts shortly after the announcement time on days with earnings announcements before gradually tailoring off over the next 30 minutes.

In Panel B of Figure \ref{figure:trade-spread-announcement}, we compare bid-ask spreads in the after-hours market on days with and without earnings announcements. Again, we see distinct differences. On non-announcement days, the median quoted spread starts at 5 bps at the end of the regular trading session before rising to 30--35 bps within the first five minutes of the post-trading session. It then gradually inclines for the duration of the after-hours session as we approach the overnight period. On announcement days, the median quoted spread goes from 5 bps at close (4:00pm) before quickly rising to 60 bps at the time of the announcement, then drops to 30 bps within 10--20 minutes of the announcement.

For comparison, the average quoted spread during post-close trading is 58.1 bps, while the average spread in regular trading equals 8.4 bps, implying that the bid-ask spread on average is about seven times greater in the after-hours session compared with the regular trading session.\footnote{See Appendix \ref{appendix:after-hours-market}, which reports summary statistics on the after-hours trading volume and bid-ask spreads.}

\subsection{Jump proportion}

The theoretical analysis in Section 2 highlights that jumps in prices increase the realized variance without affecting the bipower variation. This insight suggests an intuitive way to check whether earnings announcements cause jumps.

First, since earnings announcements almost exclusively fall outside the hours of the exchange-traded session, in Panel A of Figure \ref{figure:regular-versus-extended-all} we plot the pre-averaged realized variance computed for the regular trading session (9:30am--4:00pm) against its corresponding value computed for the extended trading session (9:30am--6:30pm), both converted to annualized standard deviations. Their difference measures the incremental volatility observed in the after-hours market (4:00pm--6:30pm). On non-announcement days (indicated by a red dot) the extra volatility during after-hours trading is typically minuscule and has a negligible impact on the cumulative pre-averaged realized variance measure. Except for a few outliers, the vast majority of these observations form a cloud close to the 45-degree line. Conversely, on days with a quarterly earnings announcement (indicated by a blue cross), differences are much greater with the extended trading session pre-averaged realized variance substantially exceeding its corresponding value from the regular trading session. Earnings releases clearly trigger substantial price volatility in the after-hours market, which is otherwise dormant.

Second, and more directly related to our jump test, in Panel B of Figure \ref{figure:regular-versus-extended-all} we plot the pre-averaged bipower variation against the pre-averaged realized variance, both based on high-frequency data from the extended trading session (9:30am--6:30pm). Realized variance less bipower variation isolates the contribution of jumps to return variation, which is zero (up to sampling error) if the price path is continuous but is strictly positive if there are price jumps. Hence, the farther to the right of the 45-degree line a data point is, the more it indicates the presence of jumps. Consistent with this interpretation, the pre-averaged realized variance and bipower variation are more or less perfectly aligned on non-announcement days. Conversely, on days with earnings announcements we observe large differences indicating jumps in price movements at these times. Table \ref{table:rv-descriptive} provides further stock-level details on realized variance and bipower variation for the regular and extended trading session broken into days with and without earnings announcements. Appendix \ref{appendix:apple} provides an even more detailed analysis for Apple which is perhaps the most prominent stock in our sample.

Third, an alternative approach to gauge the relative importance of jumps is to examine the jump proportion which estimates the fraction of the return variation originating from the jump component. Following \citet*{christensen-oomen-podolskij:14a}, this is defined as\footnote{The jump proportion is computed without pre-averaging if sampling at a 5-minute frequency or with pre-averaging if sampling at the tick-by-tick frequency.}
\begin{equation} \label{equation:jump-proportion}
\text{Jump proportion} = 1 - \frac{\text{Bipower variation}}{ \text{Realized variance}}.
\end{equation}
Figure \ref{figure:jump-proportion} plots kernel density estimates of the distribution of the jump proportion. As indicated by the peak near zero for the graphs based on our noise-robust pre-averaging estimators, on most days without announcements (Panel A) the jump component accounts for none of the variation in the stock price. Note that small negative numbers can be observed due to estimation error. Conversely, on days with earnings announcements (Panel B) more than fifty percent of the total variation in price movements stems from the jump component.\footnote{The slight attenuation observed with $\theta = 1$ is consistent with our Monte Carlo evidence in Appendix \ref{appendix:simulation}, where we note that excessive pre-averaging tends to deflate the power of our microstructure noise-robust jump test.}

In contrast, the jump proportion extracted with five-minute sampling peaks near 10\%, regardless of whether there is an announcement or not. Although the jump proportion for the five-minute measure is slightly more right-skewed on days with announcements, in general it overestimates the importance of jumps on days with no announcements and underestimates it on days with announcements. These findings are consistent with  the biases in the conventional jump measure documented in our Monte Carlo simulation results.

\subsection{Jump frequency} \label{section:jump-frequency}

Figures \ref{figure:regular-versus-extended-all}--\ref{figure:jump-proportion} shows that the jump component is important in after-hours trading sessions on days with earnings announcements but largely absent on days without announcements. To examine this point more closely and formally test Hypothesis \ref{hypothesis:necessary}.1, we compute the jump indicator:
\begin{equation} \label{equation:jump-indicator}
J_{it} =
\begin{cases}
  1, & \mathcal{J}_{n,it} > q, \\
  0, & \text{otherwise.}
\end{cases}
\end{equation}
where $\mathcal{J}_{n,it}$ is the noise-robust jump test statistic in \eqref{equation:test-statistic} for company $i$ on day $t$. We set the critical value $q$ in a way that reduces the likelihood of identifying false jumps by applying a Bonferroni correction to account for multiple testing and keep the family-wise error rate at 1\%.\footnote{\citet{bajgrowicz-scaillet-treccani:16a} report that up to 90\% of jumps identified with the five-minute jump test are spurious due to multiple hypothesis testing.}

Table \ref{table:rv-descriptive} reports the jump frequency for stock 1--25 in our sample, while Table \ref{table:rv-descriptive-supplemental} in Appendix \ref{appendix:supplemental} contains the information for stock 26--50. Panel A is for the regular trading session, while Panel B is for the extended trading session. The sample is further split into days with and without an earnings announcement. For the regular trading session, we estimate a jump frequency of 2.95\%. While jumps do occur at a higher frequency than expected by random chance, they are relatively rare during the regular trading session. This result is consistent with \citet*{christensen-oomen-podolskij:14a}. For after-hours trading sessions without earnings announcements, we obtain a slightly higher jump frequency of 3.67\%. Finally, for after-hours trading sessions with earnings announcements, our jump frequency estimate is 91.06\% which is far higher than the previous values. Moreover, though very high on average, the jump frequency is even higher after 2016 (95.84\%) versus before (87.49\%).

In essence, while it is rare to identify a jump during regular trading sessions or in extended trading sessions with no announcements, it is rare \emph{not} to find a jump in extended sessions with earnings announcements. This finding provides strong empirical support for Hypothesis \ref{hypothesis:necessary}.1.

The corresponding jump frequency rates for the five-minute jump test statistic of \citet*{barndorff-nielsen-shephard:06a} are 5.46\% in regular trading sessions, 14.38\% in extended trading sessions without earnings announcements, and 40.40\% in extended trading sessions with earnings announcements. Recall that our simulations in Appendix \ref{appendix:simulation} show that the five-minute jump test overrejects in the presence of little or no microstructure noise (as in the regular trading sessions) but lacks power (underrejects) in the presence of high levels of noise contamination (as in the after-hours market). Hence, our estimates illustrate that the conventional jump test identifies too many jumps (false positives) during regular trading sessions and extended trading sessions without earnings announcements but too few jumps (false negatives) during extended trading with earnings announcements.

Table \ref{table:rv-descriptive} suggests that failure to find jumps after earnings announcements is associated with the trading activity in the after-hours market. As this has increased over time, the lower average jump frequency prior to 2016 is therefore largely a result of fewer trades and weaker power in detecting jumps. If trading is too sparse, there are often relatively few post-announcement observations to base the analysis on, rendering it problematic for our noise-robust test statistic to work as intended. The other main reason why we sometimes fail to detect a post-announcement price jump is that although our jump test has excellent power, it is not bullet proof. This is particularly relevant for announcements that are succeeded by severe whipsawing in the price. Intuitively, a jump is less likely to be detected when it is surrounded by extremely high diffusive volatility which reduces the relative contribution of the jump to return variation.

\subsection{Determinants of the jump probability} \label{section:logit}

To more extensively test Hypothesis \ref{hypothesis:necessary}, we next develop a regression framework to examine the determinants of price jumps. We begin by inspecting the role of earnings surprises since this is the measure most directly related to earnings announcements.

To capture the relation between the jump probability versus the sign and magnitude of earnings surprises, we sort announcements by the value of the earnings surprise, form decile portfolios, and plot the corresponding average jump frequencies implied by our noise-robust jump test in Panel A of Figure \ref{figure:decile-sort}. The jump probability mostly exceeds 90\% across deciles, and does not depend systematically on the size of the earnings surprise.

Our finding that prices jump even when investors are not surprised by the reported earnings per share is consistent with Hypothesis \ref{hypothesis:necessary}.2 and can be explained as follows. First, information other than the headline earnings figure is typically released (e.g., forward guidance on future earnings) and this may trigger a jump.\footnote{For example, on 10/24/2011 Netflix (NFLX) reported third-quarter EPS of \$1.16 which beat the consensus expectation of \$0.94, corresponding to a modest 1.6 standard deviation surprise. However, the company also significantly lowered its forward guidance, and the stock crashed more than 20\% in after-hours trading.} While this helps explain why prices \textit{can} jump even on days where markets do not get surprised by the earnings figure, we would still expect a lower \textit{average} jump probability on such days. This points to the second explanation, namely that prices jump because of the resolution of uncertainty associated with the earnings announcement \citep*{ai-bansal:18a}. An implication of this story is that prices are more likely to move up (due to the lower post-announcement risk premium) than down following an earnings announcement that is close to investor expectations. We explore this prediction in Section \ref{section:jump-size-determinant}.

To establish a more formal connection between price jumps and earnings announcements and examine whether other factors affect jump probabilities, we next estimate logit regressions. Our analysis includes up to five regressors. To estimate the effect of an earnings announcement on the jump probability and test Hypothesis \ref{hypothesis:necessary}.1, we consider an announcement dummy variable:
\begin{equation}
EA_{it} =
\begin{cases}
  1, & \text{on announcement days,}\\
  0, & \text{otherwise.}
\end{cases}
\end{equation}
Since $z_{ \text{EPS}}$ proxies for the magnitude of the news component in the announcement, we expect this to matter for the jump probability, even though the evidence from Panel A of Figure \ref{figure:decile-sort} is a bit weak. As our second and third covariates, we therefore include the absolute value of the earnings surprise, measured separately for positive and negative surprises, i.e., we define $z_{ \mathrm{EPS},it}^{+} = z_{ \mathrm{EPS},it}EA_{it}^{+}$ and $z_{ \mathrm{EPS},it}^{-} = z_{ \mathrm{EPS},it}EA_{it}^{-}$, where $EA_{it}^{+} = \mathbbm{1}(z_{ \mathrm{EPS},it} > 0)$ and $EA_{it}^{-} = \mathbbm{1}(z_{ \mathrm{EPS},it} < 0)$ are indicator variables that capture the direction of the earnings surprise.

Fourth, to see whether high volatility in prices prior to the announcement makes it more or less likely to find jumps, we include the pre-averaged realized volatility, $\sqrt{RV_{n}^{*}}$, computed over the regular trading session from 9:30am--4:00pm on the day of the announcement. We expect to find a negative effect since high volatility makes it harder for our noise-robust jump test to detect jumps.

Fifth, as a measure of the resources available to facilitate the price discovery process, we include the number of analysts covering the stock, $N_{A}$.\footnote{We also considered the dispersion in EPS forecasts across analysts, $\sigma_{ \mathrm{EPS}}$, to get a measure of subjective uncertainty surrounding a particular earnings announcement. However, this measure fails to be significant in our logit regression.} We would expect to find a positive coefficient since the price discovery process, and thus the speed of price adjustments, should be more efficient for stocks with the greatest analyst coverage.

In summary, we estimate the logit model using observations for stock $i$ on day $t$:
\begin{equation} \label{equation:logit}
P(J_{it} = 1) = F( a + b_{1} EA_{it} + b_{2} |z_{ \mathrm{EPS},it}^{+}| + b_{3} |z_{ \mathrm{EPS},it}^{-}| + b_{4} \sqrt{RV_{n,it}^{*}} + b_{5} N_{A,it}),
\end{equation}
where $F$ is the logistic distribution function.\footnote{The coefficients in \eqref{equation:logit} are estimated on the entire sample of days with and without earnings announcements. Hence, the extra explanatory variables in the model should implicitly be understood as being interacted with the $EA$ variable and assume a value of zero on non-announcement days. With this design, all covariates measure the marginal effect (or the logarithm of the odds ratio) on the jump probability \textit{conditional} on an announcement.}

Table \ref{table:logit} reports estimation results for the logit regression applied to the extended trading session.\footnote{The results of our logistic jump probability regressions and the linear return regression in the next subsection are robust to the inclusion of stock fixed effects.} The estimated intercept in column (2) with only the announcement dummy yields a jump probability on days without announcements ($EA = 0$) of 3.67\% compared to 91.06\% for announcement days ($EA = 1$), which agrees with the results in Section \ref{section:jump-frequency}. Further, and consistent with Hypothesis \ref{hypothesis:necessary}.1 and Figure \ref{figure:decile-sort}, the earnings announcement dummy is highly statistically significant and positive, showing that the likelihood of a jump increases by a huge amount and is near unity in the wake of an earnings announcement.

As reported in column (3), the estimated coefficients on both positive ($|z_{ \text{EPS}}^{+}|$) and negative ($|z_{ \text{EPS}}^{-}|$) earnings surprises are positive. While the former is small and insignificant, the much larger coefficient estimate for negative surprises is significant, suggesting that disappointing earnings news is particularly likely to trigger a jump.

The effect of $\sqrt{RV_{n}^{*}}$ is significantly negative so higher pre-announcement return volatility reduces the jump probability. For example, setting the pre-averaged realized volatility at its 10th and 90th percentiles while keeping the number of analysts fixed at its average value our coefficient estimates imply a conditional jump probability of 98.17\% and 97.51\% (93.43\% and 89.87\%) for an average negative (positive) announcement. A plausible explanation is that it is harder to detect a post-announcement price jump amid high levels of volatility which weakens the testing power. Moreover, a high level of volatility can act as a substitute for price jumps, since it widens the dispersion of the return distribution and enables the price to move rapidly, albeit continuously, in response to earnings announcements.

The effect of the number of analysts covering a stock, $N_{A}$, is significantly positive, so jumps are more likely on earnings announcement days for those stocks covered by the largest number of analysts. For example, the coefficient estimates imply a conditional jump probability of 93.96\% and 99.37\% (77.90\% and 98.01\%) for an average negative (positive) announcement when the number of analysts is at its 10th and 90th percentiles while keeping pre-averaged realized volatility fixed at its average value. Broader analyst coverage may speed up the price discovery process and increase the chance of observing jumps.

\section{Jump spillover effects} \label{section:jump-spillover}

Earnings announcements reveal important information about the cash flow prospects of the announcing firm, but it may also affect other firms or even the broader economy. \citet{patton-verardo:12a} and \citet{savor-wilson:16a} develop theoretical models in which investors use earnings of announcing firms (AFs) to revise their expectations of earnings for non-announcing firms (NAFs), generating a covariance between firm-specific and market-wide cash flow news. As suggested in Hypothesis \ref{hypothesis:spillover}, it is natural to expect that such learning effects can induce spillovers in the form of an earnings announcement from the AF triggering co-jumps in the prices of NAFs and the market portfolio.\footnote{Existing empirical work on co-jumps include \citet{bollerslev-law-tauchen:08a}, \citet{caporin-kolokolov-reno:17a}, and \citet{li-todorov-tauchen:17b}. However, these studies mainly examine jump dependence between individual equity prices and the stock index around macroeconomic announcements, where systematic news are expected to arrive.} In this section, we examine if these predictions hold true and test Hypothesis \ref{hypothesis:spillover}.

\subsection{Jump spillover frequency}

We begin by matching announcement days for the AFs with non-announcement days in the NAFs. Any evidence of co-jump spillover effects should show up as a higher conditional jump probability for the NAFs, given an earnings release from the AF, compared to the jump probability when no firms announce.

Table \ref{table:co-jump} reports the proportion of extended trading sessions with a jump in the NAF (listed in columns) when the individual AFs (in rows) report their financial results. To preserve space, we narrow down our examination to the 25 most liquid stocks. This suffices to highlight the importance of industry proximity and differences in the size of firms' idiosyncratic and market-wide earnings components.

We find compelling evidence that industry proximity plays a crucial role in determining the co-jump probability. For example, conditional on a Facebook (FB) announcement, the price of Apple (AAPL), Amazon (AMZN), and Google (GOOG) jump on 13.3\%, 25.8\%, and 28.1\% of their non-announcement days. Thus, conditional on the AF being a tech firm the jump frequency generally rises for companies within that sector and is much lower for companies from other sectors. This is consistent with the presence of a common jump factor for the tech industry. Moreover, on days where Intel (INTC) is the AF, we record a conditional jump frequency of 20.0\%, 44.0\%, and 30.0\% in the price of Nvidia (NVDA), Advanced Micro Devices (AMD), and Micron Technology (MU), pointing to a common jump component tracking the microconductor industry. These numbers are much higher than the common non-announcement jump frequency of 4.18\%, calculated as the average jump frequency on days where there are no announcements among any of our firms.\footnote{The conditional jump frequency on common non-announcement days is slightly lower than the conditional non-announcement day jump frequency reported in Table \ref{table:rv-descriptive}, since the latter does not discriminate between days where other firms are or are not announcing.} Our analysis therefore provides strong support for Hypothesis \ref{hypothesis:spillover}.1.

Conversely, there is scant evidence of spillover effects from some of the other companies in our sample without obvious industry peers. For example, conditional on an announcement from Walt Disney (DIS), the jump frequency in the other stocks is always less than 10\% and often close to zero, which suggests that the jump component is tracking idiosyncratic news.\footnote{Disney may of course exert influence on the conditional jump probability of related companies not included in our sample such as Hasbro (HAS) or Mattel (MAT).}

In the right-most column of Table \ref{table:co-jump}, we compute the jump frequency in the S\&P 500 index (proxied by the SPDR S\&P 500 ETF Trust, SPY) over the extended trading session in which the firm announces its earnings. This can be interpreted as the conditional jump probability in the market index given that a particular firm releases its financial statement. If earnings announcements contain systematic news, we should expect this number to be higher than the comparable jump frequency in the S\&P 500 in the absence of any earnings announcement.\footnote{To rule out that a jump in the S\&P 500 is a mechanical effect following from the AF being an index constituent, we remove the AF's price effect from the index by subtracting its price times its weight in the S\&P 500 index from the SPY before we implement the jump test. We are grateful to Andrew Patton for raising this point.} Furthermore, we expect the effect to be stronger for businesses that are more central to the economy.

The table offers strong backing for the presence of a systematic ``market'' factor in individual stocks' price jumps. For instance, the frequency of jumps in the S\&P 500 index is very high conditional on an announcement from important stocks such as Apple (AAPL), Facebook (FB), Amazon (AMZN), Microsoft (MSFT), Oracle (ORCL), and Starbucks (SBUX). Notably, the market index jumps 30\% of the time Apple makes its earnings announcements, which is an order of magnitude higher than the jump probability in the market index on common non-announcement days (3.60\%). By contrast, the market jump probability is much lower when firms such as Tesla (TSLA) and Salesforce (CRM) announce earnings, suggesting that price movements in both of these stocks are to a larger extent driven by idiosyncratic news unrelated to the broader economy.

In our sample, there are between a minimum of zero and a maximum of nine earnings announcements on any given day. The conditional jump frequency in the market index is increasing in the number of announcements. On average, the rate at which the market index jumps rises from the 3.60\% on common non-announcement days [0 announcements] to 4.82\% [1--2 announcements], 3.02\% [3--4 announcements], 12.28\% [5--6 announcements], and 21.05\% [7--9 announcements]. These observations are consistent with Hypothesis \ref{hypothesis:spillover}.3. Furthermore, this finding suggests that the ``information spillover'' effect dominates the ``distraction'' effect, as identified by \citet{hirshleifer-lim-teoh:09a}, who show that the immediate price response to earnings surprises is significantly weaker on days when multiple firms announce their earnings.

\subsection{Determinants of the jump spillover probability}

Our empirical results on co-jump frequencies are particularly interesting because it is often difficult to detect jumps in the market index even when the price of the underlying stocks jump. This is because individual firms typically only take up a small part of the index basket, so their price jumps are attenuated. Our findings point toward an information channel that leads investors to revise their earnings expectations for the broader economy after individual firms' earnings announcements. We next examine key determinants of this transmission mechanism and estimate their importance for jump spillover effects with a view toward Hypothesis \ref{hypothesis:spillover}.2.

We let $J_{jt}^{ \text{NAF}}$ be an indicator variable that equals one if the price of NAF $j$ jumps on day $t$, zero otherwise. We then estimate a logit model for the conditional jump probability of stock $j$'s price, given that AF $i$ releases its fiscal results on day $t$ ($i \neq j$ by construction):
\begin{equation} \label{equation:logit-naf}
P \left(J_{jt}^{ \text{NAF}} = 1 \mid EA_{it} = 1 \right) = F \left( a + b_{1} IP_{ij} + b_{2} \log(V_{jt}^{ \text{NAF}}) + b_{3} \log(V_{it}^{ \text{AF}}) + b_{4} D^{ \text{EARLY}}_{it} + b_{5} D^{ \text{LATE}}_{it} \right).
\end{equation}
We include an intercept, a measure of industry proximity, $IP$, the natural logarithm of the transaction volume of the AF and NAF, $\log(V^{ \text{NAF}})$ and $\log(V^{ \text{AF}})$, and indicator variables for whether the AF is disclosing its fiscal statement early or late in the earnings season, $D^{ \text{EARLY}}$ and $D^{ \text{LATE}}$.\footnote{We also included the jump indicator of the AF, $J_{it}^{ \text{AF}}$, as a covariate. However, since the probability of a jump for an AF is near unity and varies little over our sample, this variable has little-to-no explanatory power.} Following \citet*{wang-zajac:07a}, we define industry proximity as the proportion of common consecutive digits (starting with leading digit and going down) in their SIC code shown in Table \ref{table:sp500-volume-pm.tex}.\footnote{$IP_{ij}$ is discrete with range 0.00, 0.25, 0.50, 0.75, 1.00, where $IP_{ij} = 0$ means no industry overlap (leading digit in SIC code differs), while $IP_{ij} = 1$ is perfect overlap (SIC codes are identical).} $V^{ \text{NAF}}$ and $V^{ \text{AF}}$ are the average number of transactions in the NAF and AF in the after-hours market in the month prior to $t$. We log-transform the latter to make them less right-skewed and closer to Gaussian. Following \citet{savor-wilson:16a}, we define early (late) announcers as those firms that release their financial results in the 1st (4th) quartile of the earnings season. The latter is assumed to start two weeks and stop two months after the end of a quarter (e.g., the earnings season for Q1 starts on April 15 and ends on May 31).

The results are reported in Panel A of Table \ref{table:logit-co-jump}. We expect a positive relationship with the first three covariates since greater industry proximity should induce a positive jump spillover effect, higher liquidity facilitates a bigger and more rapid transmission of information from the AF to the NAF, and companies that announce early in the cycle are more likely to disclose new information relevant to other companies. This is confirmed by the fitted model. Industry proximity is associated with a positive and highly significant marginal effect on the jump spillover probability (t-statistic of 7.16). The liquidity of both the NAF and AF in the after-hours market are also positive and significant with the coefficient estimate on $V^{ \text{NAF}}$ being twice as large and significant as that on $V^{ \text{AF}}$ (t-statistics of 7.16 and 14.56). In addition, the spillover from early announcers is positive and significant (t-statistic of 3.28), and vice versa for late announcers (t-statistic of -10.11). These findings are all consistent with Hypothesis \ref{hypothesis:spillover}.2.

To convert these estimates into an implied co-jump probability, we again  change the covariates one at a time, while holding the others fixed. For example, early in the earnings cycle ($D^{ \text{EARLY}} = 1$) and for firms with no industry overlap ($IP = 0$), the coefficient estimates imply that the conditional jump probability for the NAF is 4.02\% and 7.19\% (4.96\% and 5.80\%) when $V^{ \text{NAF}}$ ($V^{ \text{AF}}$) is assumed to be at its 10th and 90th percentile while holding $V^{ \text{AF}}$ ($V^{ \text{NAF}}$) fixed at its average value. Late in the earnings cycle ($D^{ \text{LATE}} = 1$), those numbers drop notably to 2.39\% and 4.34\% (2.97\% and 3.48\%). Furthermore, early in the earnings cycle and for firms with identical industry classification ($IP = 1$), the conditional jump probability for the NAF is shifted up to 5.70\% and 10.06\% (7.01\% and 8.17\%) when $V^{ \text{NAF}}$ ($V^{ \text{AF}}$) is at its 10th and 90th percentile while holding $V^{ \text{AF}}$ ($V^{ \text{NAF}}$) fixed at its average value. Those numbers again drop late in the earnings cycle to 3.42\% and 6.15\% (4.23\% and 4.96\%).

To uncover any jump spillover effects on the market portfolio, we estimate the following model:
\begin{equation} \label{equation:logit-spy}
P \left(J_{t}^{ \text{SPY}} = 1 \mid \# \text{AF}_{t} \geq 1 \right) = F \left( a + b_{1} D_{t}^{\# \text{AF}} + b_{2} \log \left(CV_{t}^{ \text{AF}} \right) \right),
\end{equation}
where
\begin{equation}
D_{t}^{\# \text{AF}} =
\begin{cases}
  1, & \text{if } \# \text{AF}_{t} \geq 5, \\
  0, & \text{otherwise},
\end{cases}
\end{equation}
and $\# \text{AF}_{t} = \sum_{i} EA_{it}$ is the number of AFs on a given day. The dummy variable separates days with ``few'' or ``many'' announcements. We further include the natural logarithm of the cumulative average number of transactions in the AFs in the after-hours market in the month prior to day $t$, $CV_{t}^{ \text{AF}} = \sum_{i \, : \, EA_{it} = 1} V_{it}^{ \text{AF}}$.

Panel B of Table \ref{table:logit-co-jump} shows that the jump probability for the market portfolio is significantly higher on days where many firms announce their results. Moreover, there is also a positive relationship with the after-hours liquidity of the AFs, $CV_{t}^{ \text{AF}}$. These findings make sense because a larger pool of company announcements implies that more market-related information is being released to investors and larger liquidity among the AFs facilitates a faster price discovery process, which feeds into the market index, consistent with Hypothesis \ref{hypothesis:spillover}.3.\footnote{Again, we account for the mechanical effect of a constituent members' price jump on the market index by subtracting its price component from the SPY.}

\section{Price dynamics after earnings announcements} \label{section:price-discovery}

Our results so far provide insights into price jump dynamics in the immediate aftermath of earnings announcements. However, they do not address if the jumps in prices over- or undershoot. In this section, we explore whether post-announcement returns are predictable and whether this can be exploited in simple trading strategies to generate abnormal returns. Both points help us better understand the efficiency of the post-announcement price formation and test the ``sufficiency'' part of conditions for market efficiency in Hypothesis \ref{hypothesis:sufficient}.

\subsection{Determinants of the jump size} \label{section:jump-size-determinant}

While our noise-robust jump test statistic speaks to the presence (or absence) of a jump, it does not reveal by how much the price jumps. We therefore next construct a proxy for the jump size based on returns measured over a small time interval following the earnings announcement. In particular, we compute the one-minute post-announcement return ($r_{1m}^{ \text{EA}}$) as the log-price a minute after the earnings release less the latest log-price prior to the announcement.

The last column of Table \ref{table:sp500-volume-pm.tex} reports summary statistics for the one-minute post-announcement returns for each company. While the cross-sectional sample average of 0.15\% is rather large, the standard deviation of 3.50\% is even bigger, indicating that one-minute post-announcement returns, while positive on average, are highly dispersed. One-minute returns are slightly negatively skewed (-0.48), although this is mainly driven by a few extremely negative return observations in our sample as evidenced by a large coefficient of kurtosis (5.76).

In Panel B of Figure \ref{figure:decile-sort}, we plot the average one-minute post-announcement return against the standardized earnings surprise. We observe a near-monotonic relationship between the sign and magnitude of earnings surprises and subsequent returns. Specifically, large negative (positive) earnings surprises are associated with large negative (positive) post-announcement returns.

More than 85\% of earnings announcements surprise positively in our sample, so firms are expected to beat consensus estimates. The average value of the standardized unexpected earnings is $\bar{z}_{ \mathrm{EPS}} = 1.73$ (median value of 1.37), which is located in the sixth decile in Figure \ref{figure:decile-sort}. The average one-minute post-announcement return in that decile is 0.46\% (median value of 0.37\%), which is consistent with Hypothesis \ref{hypothesis:necessary}.2, that earnings announcements in line with expectations generate positive returns due to uncertainty resolution.

We next examine the determinants of the change in stock prices in the aftermath of earnings announcements. To understand which factors drive post-announcement returns, we estimate a regression model that includes the predictor variables from the logit model in \eqref{equation:logit} except the $EA$ dummy variable, which is captured by the intercept here. In addition, we add the cumulative net order imbalance, which we construct following \citet*{gregoire-martineau:22a}:
\begin{equation} \label{equation:oi}
OI_{it} = \frac{B_{it}-S_{it}}{B_{it}+S_{it}},
\end{equation}
where $B_{it}$ and $S_{it}$ correspond to the buyer- and seller-initiated trading volume over the one-minute post-announcement window. To get signed trade volume, we assign an aggressiveness indicator to each transaction based on whether it was buyer- or seller-initiated. We employ a ``level-1 algorithm'' known as the tick rule \citep*[e.g.][]{chakrabarty-pascual-shkilko:15a}.\footnote{Quoting from their paper: \textit{"The tick rule classifies a trade as buyer-initiated if the trade price is above the preceding trade price (an uptick trade) and as seller-initiated if the trade price is below the preceding trade price (a downtick trade). If the trade price is identical to the previous trade price (a zero-tick trade), the rule looks for the closest prior price that differs from the current trade price. Zero-uptick trades are classified as buys, and zero-downtick trades are classified as sells."}} Note that in contrast to the other variables, the net order imbalance is not known ahead of the announcement.

We then estimate a pooled panel regression for the one-minute post-announcement return for stock $i$ on day $t$, $r_{1m,it}^{ \text{EA}}$:
\begin{equation} \label{equation:return-regression}
r_{1m,it}^{ \text{EA}} = a + b_{1} z_{ \text{EPS},it}^{+} + b_{2} z_{ \text{EPS},it}^{-} + b_{3} \sqrt{RV_{n,it}^{*}} D_{it} + b_{4} N_{A,it} D_{it} + b_{5} OI_{it} + \epsilon_{it},
\end{equation}
where $D_{it} = \sign(z_{ \text{EPS},it})$ is the sign function. We interact the sign of $z_{ \text{EPS}}$ with pre-averaged realized volatility and the number of analysts because the latter are non-negative and merely capture the speed with which information gets impounded into stock prices whereas earnings surprises are directional. The bigger the return volatility and number of analysts, the more we expect positive news (a positive sign indicator) to move returns upwards and negative news (a negative sign indicator) to move returns downwards. A priori, we therefore expect all coefficients to be positive in \eqref{equation:return-regression}.

Table \ref{table:return-regression} reports our estimation results. In the restricted model in column (2) with only $z_{ \text{EPS}}^{+}$ and $z_{ \text{EPS},it}^{-}$ but no additional controls, the estimates of $b_{1}$ and $b_{2}$ are both positive, implying that larger earnings surprises translate into larger announcement returns. The marginal effect of negative earnings surprises is more than twice as large as that of positive earnings surprises. Negative earnings surprises thus trigger a much larger drop in the stock price than the rise in prices for comparable positive surprises. In this regression, the coefficients are also significantly different from each other ($P$-value of 0.0016). Adding controls in column (3) and (4), the estimates of $b_{1}$ and $b_{2}$ remain positive and $b_{2} > b_{1}$, but they are nearly identical and the hypothesis $H_{0}: b_{1} = b_{2}$ is no longer rejected.

The slope estimate on the pre-averaged realized volatility, $\sqrt{RV_{n}^{*}}$, is positive and borderline statistically significant. Higher volatility is therefore weakly associated with wider dispersion in returns as post-announcement prices move further away from their initial value during periods of higher volatility.

The coefficient estimate on the number of analysts, $N_{A}$, is positive and significant. Hence, a larger number of analysts monitoring a company leads to larger one-minute post-announcement returns following positive surprises and smaller (more negative) returns for negative surprises.

Finally, we consider the order imbalance variable, $OI$. While the other variables in the regression are known ex-ante (at or before the earnings announcement), order imbalance is only known ex-post. Bearing this in mind, our estimates on $OI$ are positive and highly significant for both positive and negative earnings surprises, implying that larger net order imbalances are associated with bigger movements in returns. Intuitively, buying pressure (positive OI) pushes prices up, while selling pressure (negative OI) reduces prices.\footnote{As a robustness check, we reestimated the return regression in \eqref{equation:return-regression} based on 5-, 10-, 30-, and 60-minute post-announcement returns. Although some of the marginal effects were found to be larger simply because the longer-horizon returns tend to be larger than one-minute returns, none of the conclusions regarding the sign and significance of the coefficient estimates changed. \label{footnote:return-regression}}

We conclude from this evidence that both small and large earnings surprises almost always trigger a jump in the stock price. However, while small earnings surprises are associated with small returns, large earnings surprises tend to move prices by a bigger amount in a direction that is consistent with the sign of the earnings surprise. This effect is stronger for negative earnings surprises. Moreover, prices move by more, the more analysts cover a given stock and the magnitude of price movements is also affected by the pre-announcement price volatility.\footnote{We also examined our data for a pre-earnings announcement drift but failed to identify any significant effect.}

\subsection{Trading strategy} \label{section:trading-strategy}

We next examine if the one-minute post-announcement return can be predicted and, if so, whether the resulting forecasts can be exploited in a simple trading strategy. We employ the parsimonious earnings response model of \citet*{ball-brown:68a} by using the pooled panel regression in \eqref{equation:return-regression} based on the standardized earnings surprise as the only conditioning information:
\begin{equation} \label{equation:trading-strategy}
r_{1m,it}^{ \text{EA}} = a + b_{1} z_{ \text{EPS},it}^{+} + b_{2} z_{ \text{EPS},it}^{-} + \epsilon_{it}.
\end{equation}

We recursively estimate the parameters in \eqref{equation:trading-strategy} using an initial warm-up sample of one month of announcement days, adding new data as it becomes available. This ensures our forecasts are available in real time and do not suffer from look-ahead bias.

Our trading strategy follows \citet{patell-wolfson:84a}. It opens a long (short) position if the predicted return in \eqref{equation:trading-strategy} is positive (negative). In our baseline scenario the position is held until the end of the trading day (6:30pm, or EOD) and closed at the last observation. In contrast to \cite{patell-wolfson:84a}, there is little evidence of elevated returns during the overnight period following the release or at the opening of trading on the next day, see Appendix \ref{section:weighted-price-contribution}. Thus, our results are robust to leaving the position on the book until the exchange opens the next morning. To avoid chasing too small profits, we set an entrance barrier of 0.75\% in absolute value for the predicted return, which is a 25\% markup of the immediate post-announcement bid-ask spread observed in Panel B of Figure \ref{figure:trade-spread-announcement}. However, our results are robust to using a range of other threshold values. With a 0.75\% required return, the strategy generates 826 trading signals from 2,193 earnings announcements, of which 664 are long and 162 are short.

To gauge the importance of transaction costs, we implement four versions of our trading strategy. The first sets the entry level of the trade equal to the first post-announcement transaction price (Trade). Our second approach employs the midquote (Midquote). We expect the transaction price to be in the direction of the announcement surprise, but it is possible that it is on the other side of the strategy (e.g., a sale following a positive announcement surprise, or vice versa). Moreover, whereas the Trade strategy can snipe stale quotes resting in the limit order book prior to the announcement, the Midquote strategy employs an updated quote that incorporates the information from the announcement.\footnote{\cite{baldauf-mollner:22a} document the importance of stale quotes in generating jumps in the transaction price.} Trading at the midquote is often possible for large institutional traders, but may not be feasible for other entities. Our third strategy therefore incorporates the quoted bid-ask spread (Best Bid and Offer, or BBO). It initiates a buy order at the prevailing best ask, whereas a sell order is initiated at the prevailing best bid. The latter therefore incorporates a full spread in the returns generated from the trading strategy.\footnote{We do not control for market impact since our data is not rich enough to include such information. Hence, our results are mostly relevant for small trading sizes.} Finally, our fourth approach makes the odds even more unfavorable by introducing a latency delay, which is the minimum number of seconds an investor must wait before entering a position (BBO+Xs). Such trading delays cause investors to miss out on the initial reaction in the stock price. In principle, it could be beneficial to wait a few seconds before starting a trade because spreads narrow following an announcement.

Figure \ref{figure:cumulative-return} shows the evolution over time in the cumulative returns based on payoffs from individual trades. The slope of the graph reflects the rate at which returns accrue: The steeper the curve, the larger the returns. Conversely, flat spots or a declining curve indicate that the trading strategy fails to earn, or  outright loses, money.

The figure highlights that the transaction price and midquote approaches yield positive returns throughout our sample with no notable periods of inferior performance (full blue and dashed red lines). As expected, introducing execution costs or latency delays reduces the rate of accumulation. In particular, from 2016 onward the strategy that pays the bid-ask spread (dashed-dotted line) is flatlining. Adding a further 5- or 10-second latency delay (dashed lines) reduces returns even more and the strategy loses money after 2016. This indicates that the speed with which earnings announcement information gets incorporated into prices has increased in the second part of our sample.

To inspect the profitability of our trading strategy, Figure \ref{figure:cumulative-return} also reports (in parenthesis) the sample average return and a test statistic examining if the mean return is different from zero. In frictionless markets where investors can enter the position at the first post-announcement transaction price, a highly significant return of 1.80\% per trade (t-statistic of 8.88) is earned over the full sample. Trading at the midquote, the average return drops to a highly significant 1.50\% per trade (t-statistic of 7.67). Even if trading is executed at the bid-ask spreads, the average return of 0.72\% continues to be significant (t-statistic of 3.68). A 5-second latency delay reduces the return to 0.41\% per trade, which is only borderline significant (t-statistic of 2.18). Forcing a 10-second delay on trade entrance, the average return is still positive at 0.28\% per trade, which is now insignificant (t-statistic of 1.51).\footnote{We can estimate how much money investors leave on the table after an earnings announcement using transaction information between the announcement and 6:30pm. Denote by $P_{t,i}^{j}$ and $n_{t,i}^{j}$ the price and number of shares of the $i$th transaction in stock $j$ after the announcement at time $t$. Also, let $P_{t,6:30}^{j}$ be the price at 6:30pm. The dollar profit made by an investor buying in the $i$th post-announcement trade and closing it out at 6:30pm is then $\pi_{t,i}^{j} = n_{t,i}^{j} (P_{t,6:30}^{j} - P_{t,i}^{j})$. Cumulative dollar profits across the whole sample are $\pi^{\$} = \sum_{j} \sum_{ \#EA_{j}} \sum_{i} \pi_{t,i}^{j} \mathcal{I}_{t,i}^{j}$, where $\mathcal{I}_{t,i}^{j}$ is the signal from the trading strategy (1 = long, -1 = short, 0 = no position). We estimate a total profit of \$418,783,007, or 0.3209\% of turnover. The average profit is \$507,001 with a standard deviation of \$189,115, or a z-score of 2.68. We thank Lasse Heje Pedersen for suggesting this calculation.}

Panel A of Table \ref{table:trading-strategy} reports a more comprehensive set of summary statistics. On top of the full sample (2008--2020) results, we also do the analysis separately for the period 2008--2015 and 2016--2020. In this way, we can examine if the price discovery process has changed over time. We choose the subsample splits based on major structural changes in market design. First, the last electronic communication network (ECN) in the U.S., as defined by Rule 600(b)(23) of Regulation NMS, closed in 2015. The disappearance of ECNs implies that during our sample an increased amount of liquidity has been rerouted to alternative after-hours trading venues, such as dark pools.\footnote{Traditional floor trading is available, but it is not anonymous and probably also not fast enough.} The latter provide anonymous trading and, in general, do not display quotes to the public. On the one hand, this can impede price discovery, because order flow contains valuable information.\footnote{\citet*{hendershott-jones:03a} find evidence of a deterioration in price discovery, when the Island ECN transitioned into a dark pool following stricter display requirements pursuant of Regulation ATS.} On the other hand, less pre-trade transparency can facilitate price discovery because investors can conceal their trading activity, which increases the incentive to collect private information. Moreover, as earnings announcement news is arguably short-lived, it encourages faster and more aggressive trading on such information \citep*[e.g.][]{brogaard-pan:21a, flood-huisman-koedijk-mahieu:99a}. Second, stricter legislation approved by the SEC means that from April 4, 2016 a greater number of broker-dealers are required to disclose stock-by-stock aggregate volume and transaction count information to the Financial Industry Regulatory Authority (FINRA) on a weekly basis.\footnote{See, for example, \path{https://www.finra.org/sites/default/files/NoticeDocument/p446087.pdf} and \path{https://www.finra.org/sites/default/files/Regulatory-Notice-15-48.pdf}.} Post-trade transparency helps investors locate more liquid platforms, which accelerates the matching process. Furthermore, the reduction in the number of trading venues makes the market less fragmented, which further enhances this process.

Across the board, average returns are both larger and more significant in the first subsample (2008--2015) compared to the full sample. For example, the average return per trade in the transaction price setting is 2.30\% (t-statistic of 8.77) over 2008--2015 but only 1.11\% (t-statistic of 3.49) over 2016--2020. Moreover, the average return inferred from the midquote strategy drops to 0.80\% (t-statistic of 2.65) in the second subsample -- more than a full percentage point lower than its early subsample counterpart (2.00\%). Trading at the bid-ask spread further reduces the average return and adding a latency delay even renders them slightly negative, though not significant, in the second subsample. These findings are consistent with a more efficient price discovery process over time with average returns declining to the point where they are no longer significantly different from zero after accounting for bid-ask spreads or latency delays. They also concur with \citet{martineau:22a}, who finds that most stock prices are significantly more responsive to earnings surprises during the latest subsample (2016--2019) compared to the earliest (1984--1990).

\cite{patell-wolfson:84a} found that excess returns in their sample ceased to exist within 10- to 15-minutes following an earnings announcement. In Panel B of Table \ref{table:trading-strategy}, we show the effect of changing the investment horizon with more recent data. Instead of closing positions at 6:30pm, the duration of the position is now measured either in minutes (physical time) or by the number of tick updates (tick time) after the announcement.\footnote{We also estimated the predictive return regression with the post-announcement return (dependent variable) measured over a 5-, 10-, 30-, and 60-minute horizon. Average trading returns decline monotonically as the length of the measurement window is extended, suggesting that the price response to earnings surprises is very rapid and the ``signal'' is strongest for the one-minute post-announcement return. For example, compared to the number in Table \ref{table:trading-strategy}, moving from a one-minute to a 60-minute horizon causes a decline of 48 basis points in average return for the transaction price approach.}

In Panels B.1--B.5, we maintain the position from 30 seconds to 5 minutes. In the full sample, average returns from the no-friction strategy based on transaction prices increase monotonically from 0.74\% (t-statistic of 8.28) for positions closed after 30 seconds to 1.58\% (t-statistic of 10.82) per trade for positions closed after five minutes. In the midquote strategy, the average return is slightly compressed but continues to be positive and statistically significant. Trades executed at the BBO reduce the average return which becomes slightly negative at -0.47\% (t-statistic of -4.75) for the 30-second termination rule but rises to 0.41\% at the 5-minute mark (t-statistic of 2.81). Latency delays lead to further reductions in trading performance as none of the approaches generate positive \textit{and} significant average returns. Once again the results change meaningfully between the first and second subsamples. As the price discovery process improves, average returns become lower and less significant in the second subsample compared to the earlier subsample. Interestingly, the average return over the 2016--2020 sample for the baseline scenario in Panel A is close to the corresponding value for the 5-minute stopping rule in Panel B.5. This suggests that, in the most recent subsample, the price discovery process has more or less been completed at that time.

Next, we consider what happens to trading performance in transaction time if we close the position after between 50 and 1,000 tick updates. As shown in Panels B.6--B.10, we continue to generate positive and significant average returns for trades executed at the transaction price or midquote. However, executing at the BBO again produces insignificant mean returns in the latest subsample. Interestingly, trades stopped after 1,000 tick updates in Panel B.10 yield mean returns nearly identical to those from the positions held until 6:30pm (Panel A) in the first subsample, but not in the second. This reflects the lower trading activity in the early parts of our sample, which means that 1,000 tick updates cover a much longer time interval.

Figure \ref{figure:trading-result} summarizes these results by plotting the cumulative return from trading strategies that close out the position after a fixed number of tick updates (Panels A--C) or after a fixed number of seconds (Panels D--F), employing either the transaction price, midquote, or BBO. The figure illustrates both how fast the price adjusts in the very short period after an earnings announcement (both in trade time and in physical time) and also that the price settles at its new level even faster in the second subsample.

Overall, our results demonstrate that price discovery in the post-announcement window has gotten far speedier during our twelve-year sample. Prior to 2016, investors were able to reap significant profits from post-earnings announcement trades executed at the BBO even with a latency delay of 10 seconds. In contrast, during the last subsample, after-hours prices move so fast after earnings announcements that it is no longer possible to generate significant outperformance after accounting for bid-ask spreads. This suggests that the positive effects of decreased pre-trade transparency and improved post-trade transparency outweigh the negative effects of decreased pre-trade transparency, leading to better price discovery around corporate earnings announcements.

We can gain additional insights by comparing our results with \citet{gregoire-martineau:22a}. They find that that for ``most'' announcements in their sample, the post-announcement closing price at 4:00pm the following day is within the bid-ask spread prior to the announcement, so even a liquidity taker who was informed about the upcoming unexpected earnings surprise could not trade profitably on the information. We achieve the reverse finding for our sample: The post-earnings transaction price at 6:30pm---where we close our trading strategy---is outside the pre-announcement quoted spread for 92.66\% of the earnings announcements.

This, and several other points, help to explain the differences between our findings and those in \citet{gregoire-martineau:22a}. First, the data used in their paper cover the shorter and earlier period 2012--2015. This is important, because we document a shift in after-hours price dynamics after 2016 following the change in market design. Second, and most importantly, while we analyze the 50 most liquid constituents of the S\&P 100 index, they study the after-hours behavior of the 1,500 stocks included in the S\&P 500, S\&P 400, or S\&P 600 indexes corresponding to the large-, medium- and small-cap segment of the U.S. stock market. Many of the firms in their sample are very illiquid: The median number of trades in the after-hours market on earnings announcement days is 70, 10, and 7 for S\&P 500, S\&P 400, and S\&P 600 stocks, and 20\% of the included stocks do not trade \textit{at all}. In contrast, as seen in Table \ref{table:sp500-volume-pm-supplemental.tex} in Appendix \ref{appendix:supplemental}, for our selection of stocks the smallest average number of trades on announcement days is 1,072, and even on non-announcement days many of our firms are more actively traded than those studied by \citet{gregoire-martineau:22a} are on announcement days. Third, the firms in their sample are often quoted with wide spreads in excess of 20\% while the spreads for our firms are narrow(er) and often less than 50 basis points, on average, as shown in Figure \ref{figure:trade-spread-announcement}. Finally, the time-aggregation is different: We sample at tick frequency, while they sample at a 1-minute frequency.

Proving that the sufficiency condition for market efficiency holds ultimately requires us to inspect a much broader class of trading strategies. However, the evidence presented here for the most recent subsample is consistent with Hypothesis \ref{hypothesis:sufficient}.1 and with trading in the after-hours market becoming notably more efficient over time, so we do not pursue this point further here.

\subsection{Spillover effects on returns of non-announcing firms}

Our jump spillover analysis shows that the release of financial statements by an announcing firm (AF) increases the likelihood of jumps in a non-announcing firm (NAF), particularly when the two firms have a close industry proximity and high after-hours trading volumes. These results do not, however, show the extent to which earnings surprises of AFs affect the returns of NAFs and whether the information dissemination is fully efficient as stipulated by Hypothesis \ref{hypothesis:sufficient}.2.

To shed light on this, we estimate a return regression equivalent to \eqref{equation:return-regression}:
\begin{equation}
r_{1m,jt}^{ \text{NAF}} = a + b_{1} z_{ \text{EPS},it}^{+} IP_{ij} + b_{2} z_{ \text{EPS},it}^{-} IP_{ij} + b_{3} D_{it} \log(V_{jt}^{ \text{NAF}}) + b_{4} D_{it} \log(V_{it}^{ \text{AF}}) + \epsilon_{jt}, \label{spillover_returns}
\end{equation}
where $r_{1m,jt}^{ \text{NAF}}$ is the one-minute post-announcement return of NAF $j$ after the announcement of AF $i$ on day $t$. The other variables are unchanged. The industry proximity interaction terms allow us to attenuate the effect of earnings surprises on the post-announcement return of NAFs in more distant industries. Similarly, we interact the sign of $z_{ \text{EPS}}$ with $\log(V_{jt}^{ \text{NAF}})$ and $\log(V_{it}^{ \text{AF}})$, allowing the liquidity measures to reinforce an expected negative price movement after a bad announcement. We find that none of the estimated coefficients in \eqref{spillover_returns} are statistically significant and their explanatory power is close to zero.

Figure \ref{figure:return-distribution} explains this finding. In the left panel, we show a box plot of the distribution of one-minute returns of NAFs in three settings: (i) NAFs on common non-announcement days; (ii) NAFs in other industries ($IP_{ij} = 0$); and (iii) NAFs in the same industry ($IP_{ij} = 1$). Results in (ii) and (iii) are based on days with at least one announcement.\footnote{We construct (i) by taking for each announcement of an AF a common non-announcement day at random and measuring the stock price movement of the AF over the one-minute window aligned with its associated announcement time.}

While the median one-minute return is close to zero and nearly identical across the three distributions, there is a clear ranking of their spread. The dispersion is smallest on common non-announcement days, where the one-minute return mostly reflects the bid-ask spread since no information is being released. The spread is slightly wider on announcement days for NAFs in other industries and much wider for NAFs in the same industry as the AF. However, the spillover effects on the NAFs pale in comparison to the impact of the AF on itself shown in the right panel. Thus, although we can detect co-jump spillover effects, the magnitude of the price impact is much smaller than the impact on the announcing firm itself.

We next implement the earlier trading strategies on the post-announcement returns of NAFs. Absent any trading frictions, and assuming we do not impose a threshold on the magnitude of the predicted return, the trading strategy produces a small, positive, and marginally significant average return on the order of 3--4 bps per trade.\footnote{The return forecasts of NAFs rarely exceed more than a few basis points, so imposing a threshold on expected return often prevents the trading strategy from entering a position, which renders the mean return insignificant.} This is consistent with the volatility of NAF returns being minuscule. Indeed, any profits completely vanish and turn into severe losses once trading frictions and latency delays are introduced.

We also inspect whether AFs affect one-minute returns on the market index, $r_{1m,t}^{ \text{SPY}}$, again after correcting for the mechanical price impact from the announcement, through the regression:
\begin{equation}
r_{1m,t}^{ \text{SPY}} = a + b_{1} \sum_{i \, : \, EA_{it} = 1} w_{i} z_{ \text{EPS},it} + b_{2} \bar{D}_{it} \sum_{i \, : \, EA_{it} = 1} \log(V_{it}^{ \text{AF}}) + \epsilon_{t},
\end{equation}
where
\begin{equation}
\bar{D}_{it} = \sign \left( \sum_{i \, : \, EA_{it} = 1} w_{i} z_{ \text{EPS},it} \right) \quad \text{and} \quad w_{i} = \frac{ \log(V_{it}^{ \text{AF}})}{ \sum_{i \, : \, EA_{it} = 1} \log(V_{it}^{ \text{AF}})}.
\end{equation}
In this model, when more than one firm is announcing on a given day, we weight the earnings surprises by the respective after-hour trading volume of the companies.

The coefficient estimates are both significant and of the correct sign, i.e $100 \times \hat{b}_{1} = 0.15$ (t-statistic of 2.26) and $100 \times \hat{b}_{2} = 0.09$ (t-statistic of 2.49), but the predictive power of the regression is very small. Moreover, the slope estimates are economically small with a marginal effect around one-tenth of a basis point for a unit change in the explanatory variables. Consequently, the predicted market return only rarely exceeds a full basis point which cannot be exploited in practice.

Taken together, the analysis in this subsection is consistent with Hypothesis \ref{hypothesis:sufficient}.2.

\section{Conclusion} \label{section:conclusion}

High-frequency trading is widespread in today's financial markets. This suggests that, at least for the most liquid stocks, information from public announcements should be incorporated into prices almost instantaneously after the news release. In an efficient market without large trading frictions, stock prices should therefore jump, almost surely, in the immediate aftermath of releases of large bundles of information such as firms' earnings announcements. This suggests employing jump tests on highly granular tick-by-tick data as a way to test a necessary (but not sufficient) condition for market efficiency. This strategy is very different from the conventional practice of scrutinizing post-announcement returns defined over a fixed time intervals, such as one or five minutes, which can be almost an eternity in a modern fast-paced market environment.

The vast majority of earnings announcements occur outside the regular exchange trading hours, so such an analysis of the post-announcement price discovery process requires that we examine a relatively unexplored set of high-frequency data that includes quotation and transaction records from the after-hours trading session. The irregular trading activity combined with pronounced levels of microstructure noise in this market pose serious challenges to conventional jump tests that can severely distort the inference.

To address these critical shortcomings, we develop a new jump test that is robust to the unusually noisy price data observed in after-hours markets. Our noise-robust generalization extends the classical bipower variation-based jump test of \citet*{barndorff-nielsen-shephard:06a} to a pre-averaged version that can be implemented on noisy high-frequency data. Moreover, we develop a subsample estimator of the asymptotic variance-covariance matrix that is jump-robust under the alternative, extending earlier work of \citet*{christensen-podolskij-thamrongrat-veliyev:17a}. Equipped with these tools, we show that prices of the stocks with the highest after-hours trading volume almost always jump after earnings announcements. Conversely, jumps in stock prices are rare both during the regular trading session and during after-hours trading sessions without earnings announcements.

We document a strong jump spillover effect. Conditional on an announcement from a firm, there is a significantly higher probability that prices of non-announcing firms, both within and outside the industry of the announcing firm, as well as the market index, jumps in the same after-hours trading session. The effect is larger for companies with a closer industry proximity to the announcing firm. Moreover, early announcers induce a much stronger spillover effect than late announcers. This finding is consistent with investors learning about industry- and economy-wide news components from earnings announcements.

That stock prices in the after-hours market nearly always jump after the release of earnings announcements is indicative of a very rapid price discovery process and, thus, consistent with efficient markets. However, the price jump in the immediate aftermath of an earnings release may be too big or too small, potentially introducing profitable trading opportunities. We examine if in fact prices are unbiased predictors of their future steady-state values by studying the performance of a trading rule subject to different degrees of trading frictions.

In the absence of transaction costs, we find that investors could have earned highly significant, positive average returns in the early years of our sample (2008--2015). Conversely, in the second part of our sample (2016--2020), even minor trading frictions, such as bid-ask spreads, reduce average returns to the point where they are no longer statistically significant. We conclude from this analysis that the after-hours market incorporates information on the earnings announcements of the largest U.S. firms extremely fast, particularly after 2016.

\pagebreak


\bibliographystyle{rfs}
\bibliography{userref}

\clearpage


\begin{figure}[h!]
\begin{center}
\caption{Examples of price jumps after earnings announcements. \label{figure:jump-illustration}}
\begin{tabular}{cc}
\small{Panel A: MSFT 10/24/2013} & \small{Panel B: QCOM 11/02/2011.} \\
\includegraphics[height=8cm,width=0.48\textwidth]{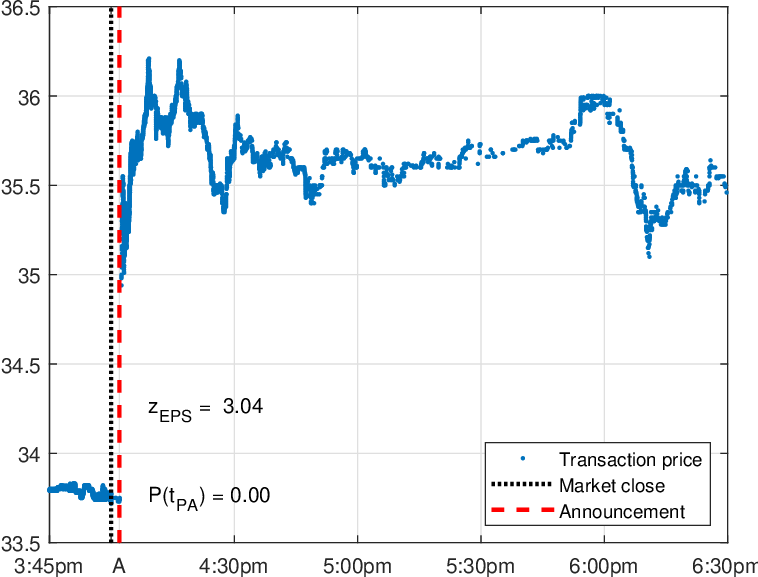} &
\includegraphics[height=8cm,width=0.48\textwidth]{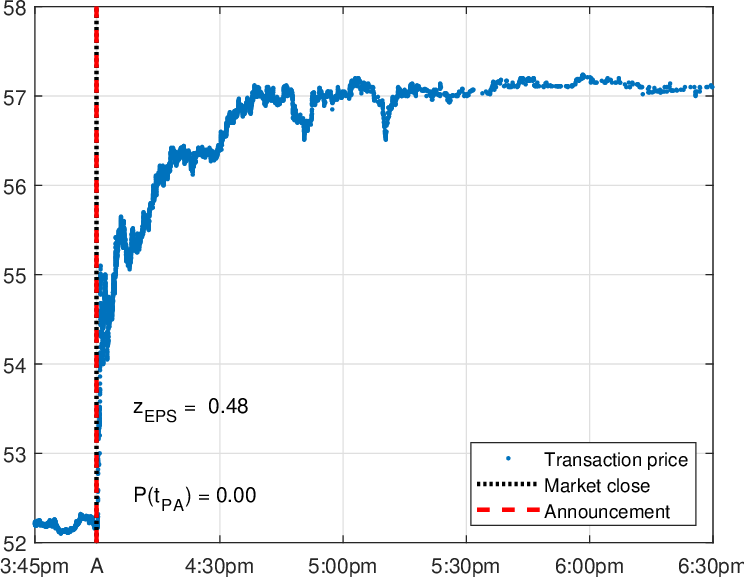}
\end{tabular}
\begin{scriptsize}
\parbox{\textwidth}{\emph{Note.} We plot the price from 3:45pm to 6:30pm for Microsoft on 10/24/2013 in Panel A, and Qualcomm on 11/02/2011 in Panel B, where the companies made an earnings announcement in the after-hours trading session. The figure also shows the market close and the announcement time. $z_{ \text{EPS}}$ is the standardized unexpected earnings, and $P($t$_{ \text{PA}})$ is the $P$-value of the pre-averaging jump test statistic.}
\end{scriptsize}
\end{center}
\end{figure}

\clearpage

\begin{figure}[ht!]
\begin{center}
\caption{Apple's third-quarter 2020 earnings announcement on 07/30/2020.}
\label{figure:warpspeed}
\begin{tabular}{cc}
\small{Panel A: 60-second horizon.} & \small{Panel B: 100-millisecond horizon.} \\
\includegraphics[height=8cm,width=0.48\textwidth]{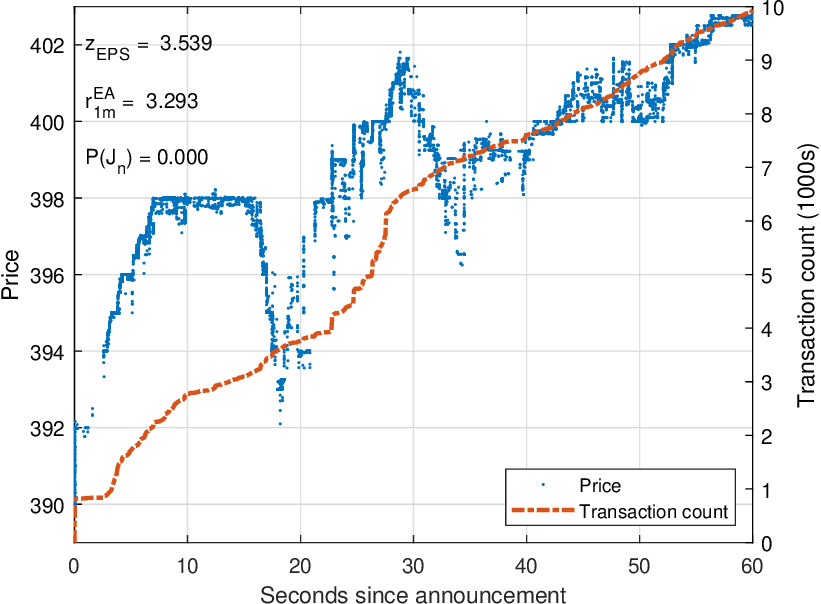} &
\includegraphics[height=8cm,width=0.48\textwidth]{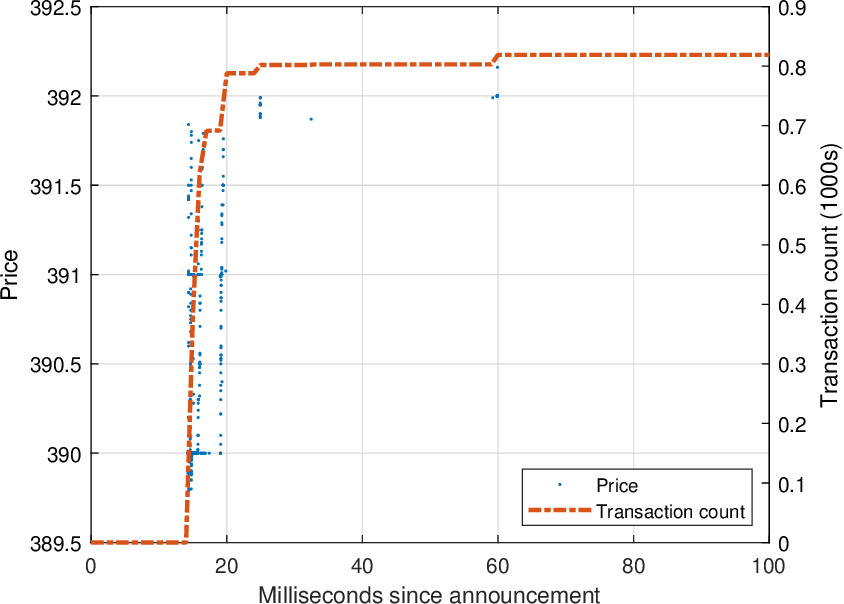} \\
\end{tabular}
\begin{scriptsize}
\parbox{\textwidth}{\emph{Note.} The figure shows the post-announcement price (left $y$-axis) and cumulative transaction count in 1000s (right $y$-axis) for Apple's third-quarter 2020 earnings announcement, which was released at 4:30pm on 07/30/2020. In Panel A, we plot the data from the first 60 seconds since the announcement, whereas Panel B zooms further in on the first 100 milliseconds since the announcement.}
\end{scriptsize}
\end{center}
\end{figure}

\clearpage

\begin{figure}[ht!]
\begin{center}
\caption{Trading volume and bid-ask spread in the after-hours market.} \label{figure:trade-spread-announcement}
\begin{tabular}{cc}
\small{Panel A: Transaction count.} & \small{Panel B: Bid-ask spread (in bps).} \\
\includegraphics[height=8cm,width=0.48\textwidth]{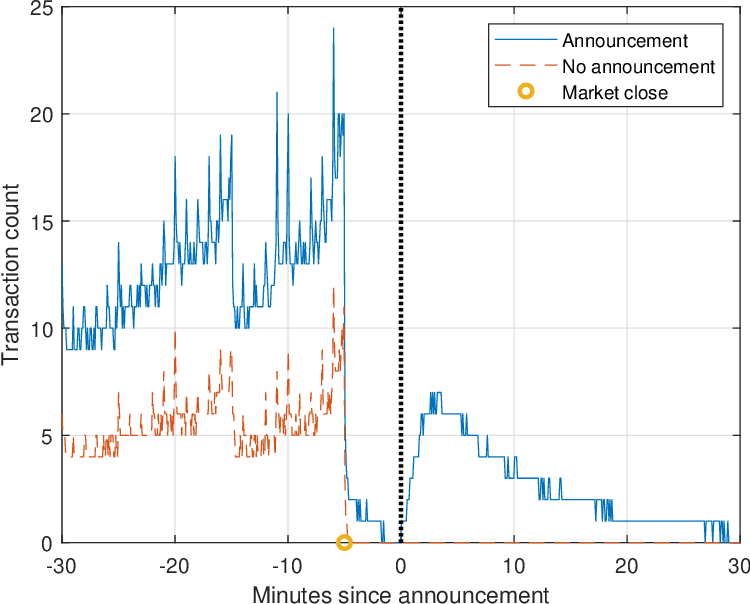} &
\includegraphics[height=8cm,width=0.48\textwidth]{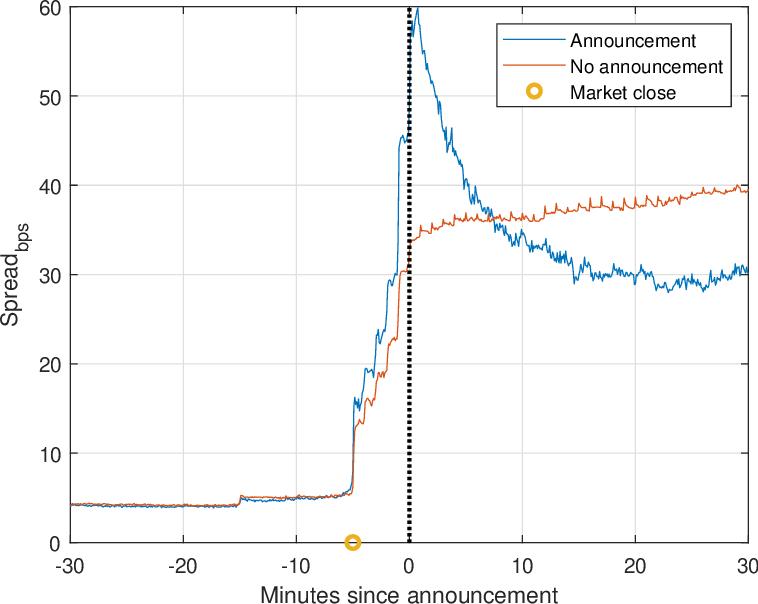}
\end{tabular}
\begin{scriptsize}
\parbox{\textwidth}{\emph{Note.} In Panel A, we show the cross-sectional sample average trading volume (measured by transaction counts) for each five-second interval in a one-hour window centered around the earnings announcement, which occurs at time 0. In Panel B, we report the associated median bid-ask spread (in basis points), which is computed as $\text{Spread}_{ \text{bps}} = 10000 \times ( \text{ask} - \text{bid})/ \text{midquote}$, where $\text{midquote} = (\text{bid} + \text{ask})/2$. As a control sample, for each announcement we select a random non-announcement date without replacement and calculate trading volume and $\text{Spread}_{ \text{bps}}$ on the corresponding time interval. The modal announcement time across companies is 4:05pm. Hence, the exchange typically closes five minutes prior to an announcement (highlighted by an orange circle).}
\end{scriptsize}
\end{center}
\end{figure}

\clearpage

\begin{figure}[ht!]
\begin{center}
\caption{Pre-averaged realized variance and bipower variation.} \label{figure:regular-versus-extended-all}
\begin{tabular}{cc}
\small{Panel A: Incremental return variation.} & \small{Panel B: Incremental jump variation.} \\
\includegraphics[height=8cm,width=0.48\textwidth]{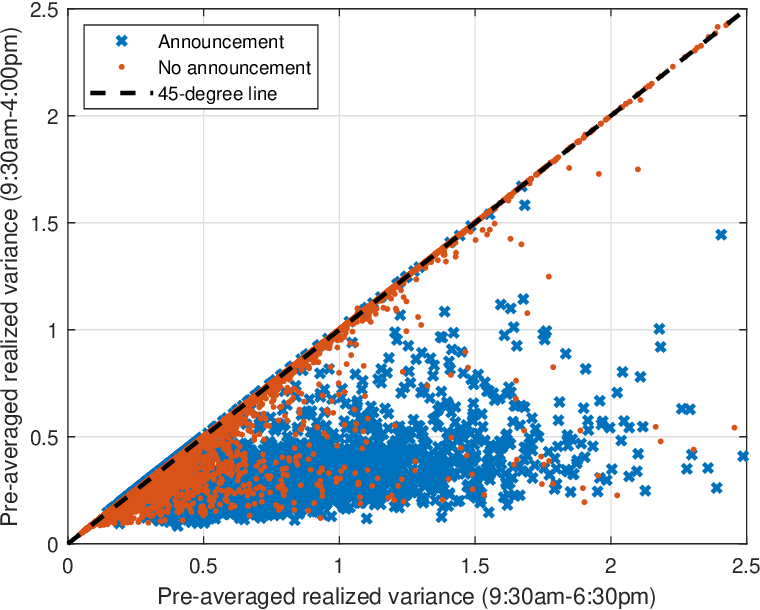} &
\includegraphics[height=8cm,width=0.48\textwidth]{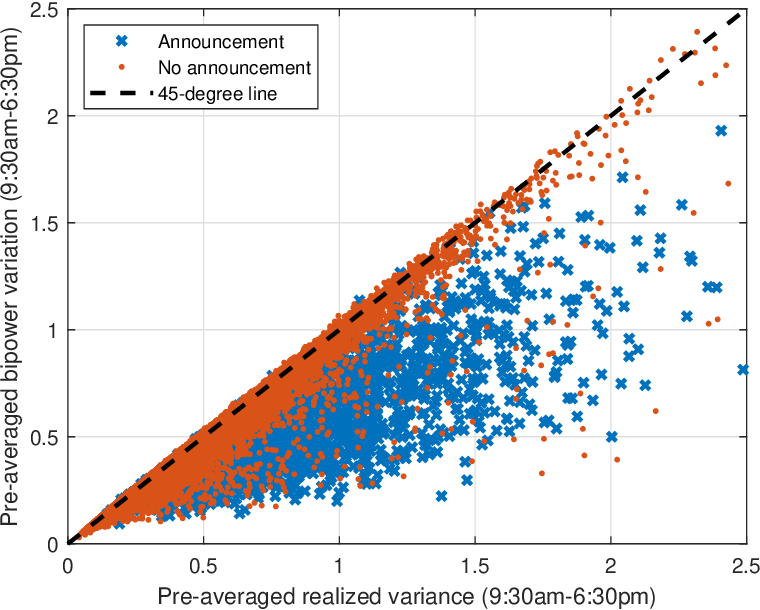}
\end{tabular}
\begin{scriptsize}
\parbox{\textwidth}{\emph{Note.} We show point estimates of the pre-averaged realized variance and pre-averaged bipower variation, converted to an annualized standard deviation. The axis label shows which estimator is plotted. In parenthesis, we further indicate for which part of the day high-frequency data are employed to calculate the estimate.}
\end{scriptsize}
\end{center}
\end{figure}

\clearpage

\begin{figure}[ht!]
\begin{center}
\caption{Distribution of the estimated jump proportion.}
\label{figure:jump-proportion}
\begin{tabular}{cc}
\small{Panel A: No announcement.} & \small{Panel B: Announcement.} \\
\includegraphics[height=8cm,width=0.48\textwidth]{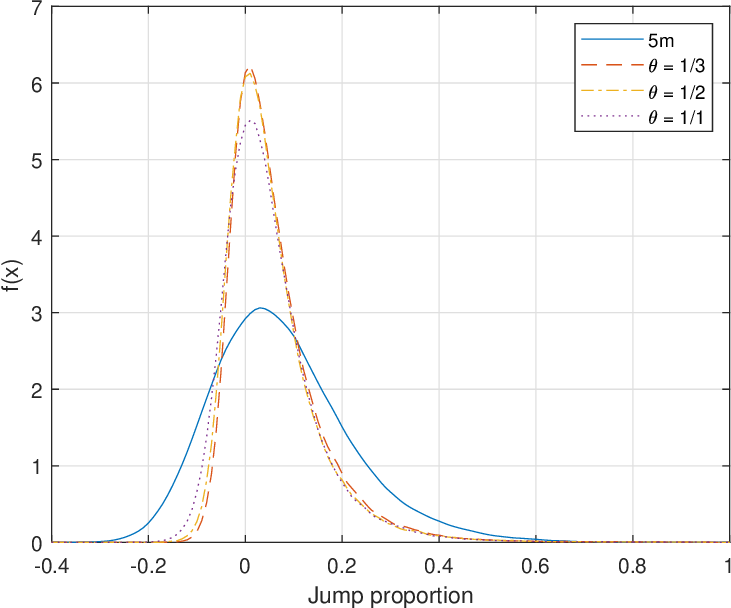} &
\includegraphics[height=8cm,width=0.48\textwidth]{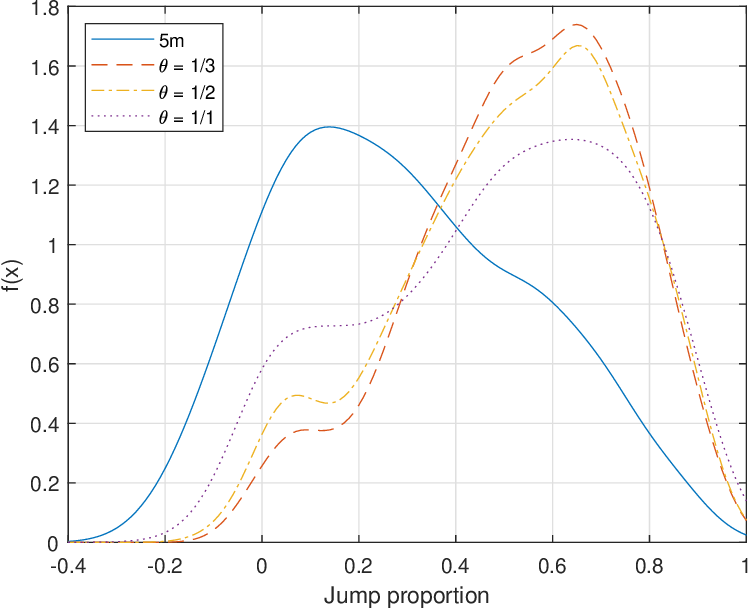}
\end{tabular}
\begin{scriptsize}
\parbox{\textwidth}{\emph{Note}. This figure shows kernel smoothed densities of the proportion of quadratic return variation from the jump component, which is estimated by: Jump proportion = 1 - Bipower variation / Realized variance. The latter are computed without pre-averaging if sampling at a 5-minute frequency (5m) or with pre-averaging if sampling at the tick-by-tick frequency (rest). $\theta$ controls the size of the pre-averaging horizon $k_{n} = \lfloor \theta \sqrt{n} \rfloor$ for calculating pre-averaged returns. $n$ is the number of tick-by-tick data. The sample is split into days without (in Panel A) and with earnings announcements (in Panel B).}
\end{scriptsize}
\end{center}
\end{figure}

\clearpage

\begin{figure}[ht!]
\begin{center}
\caption{Jump frequency in the regular and extended trading session.}
\label{figure:jump-frequency}
\begin{tabular}{cc}
\small{Panel A: Regular trading session.} & \small{Panel B: Extended trading session.} \\
\includegraphics[height=8cm,width=0.48\textwidth]{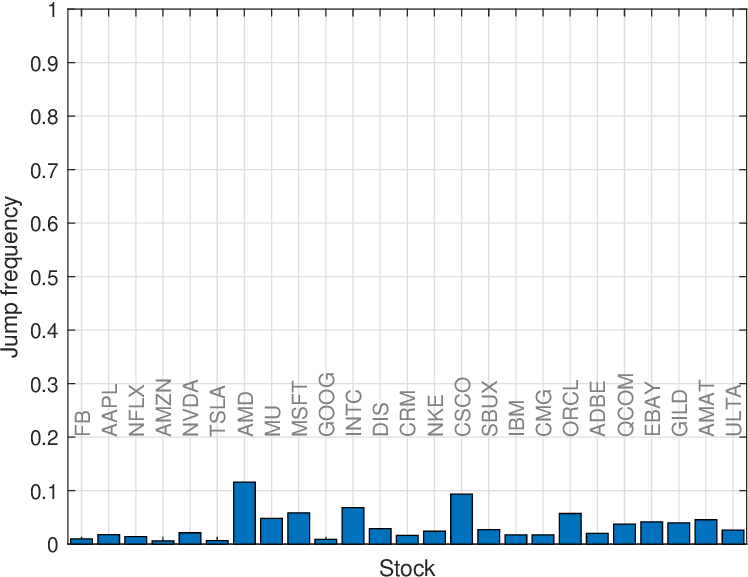} &
\includegraphics[height=8cm,width=0.48\textwidth]{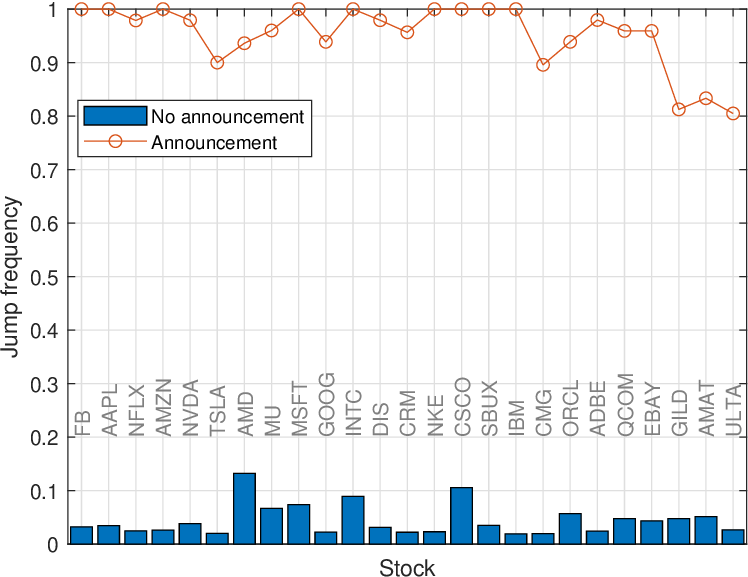}
\end{tabular}
\begin{scriptsize}
\parbox{\textwidth}{\emph{Note.} We report the proportion of days where our noise-robust jump test is significant during the regular trading session (9:30am--4:00pm) and extended trading session (9:30am--6:30pm). The latter is separated into days with and without earnings announcements.}
\end{scriptsize}
\end{center}
\end{figure}

\clearpage

\begin{figure}[ht!]
\begin{center}
\caption{Jump frequency and post-announcement return against standardized earnings surprise.}
\label{figure:decile-sort}
\begin{tabular}{cc}
\small{Panel A: Jump frequency.} & \small{Panel B: Post-announcement return.} \\
\includegraphics[height=8cm,width=0.48\textwidth]{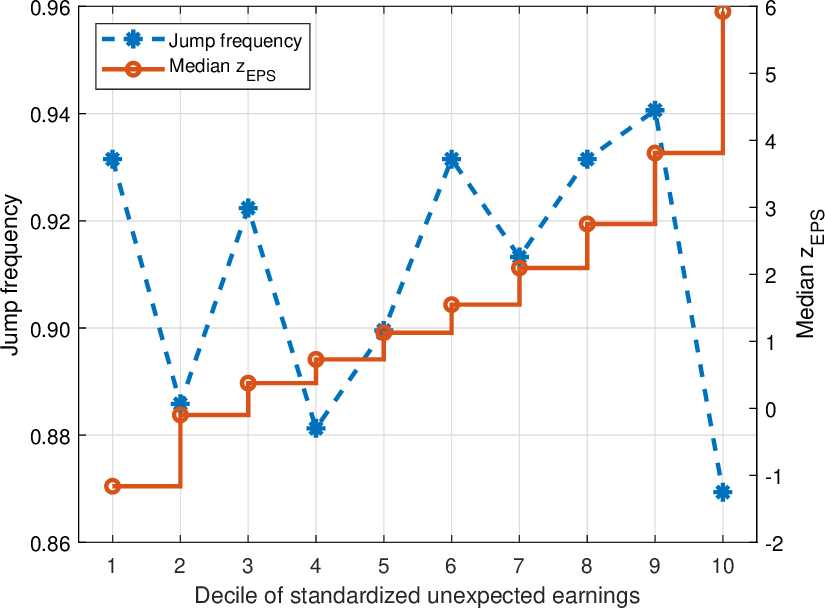} &
\includegraphics[height=8cm,width=0.48\textwidth]{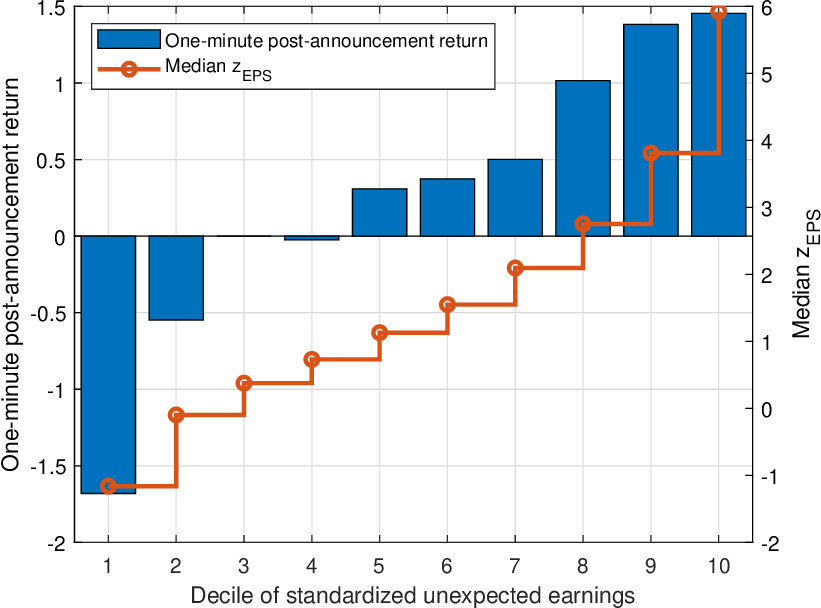}
\end{tabular}
\begin{scriptsize}
\parbox{\textwidth}{\emph{Note.} We sort our earnings announcements by the value of the standardized unexpected earnings, $z_{ \text{EPS}}$, and form decile portfolios from lowest to highest values of the earnings surprise. In Panel A, we then plot the corresponding announcement jump frequency implied by our noise-robust jump test (averaged within decile). In Panel B, we report the one-minute post-announcement return (averaged within decile). The median value of $z_{ \text{EPS}}$ within each decile is plotted against the right-hand $y$-axis.}
\end{scriptsize}
\end{center}
\end{figure}

\clearpage

\begin{figure}[ht!]
\begin{center}
\caption{Post-announcement return against standardized earnings surprise.}
\label{figure:return-regression}
\begin{tabular}{c}
\includegraphics[height=8cm,width=0.48\textwidth]{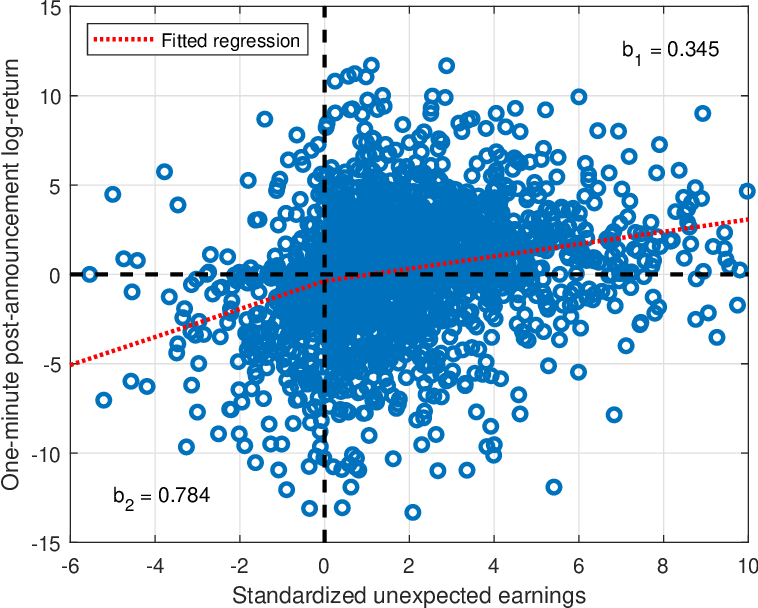}
\end{tabular}
\begin{scriptsize}
\parbox{\textwidth}{\emph{Note.} We show the one-minute post-announcement return, $r_{1m}^{ \text{EA}}$, against the standardized earnings surprise, $z_{ \text{EPS}}$. The fit from a pooled panel regression $r_{1m,it}^{ \text{EA}} = a + b_{1} z_{ \text{EPS},it}^{+} + b_{2} z_{ \text{EPS},it}^{-} + \epsilon_{it}$ is plotted as a red dotted line for visual reference.}
\end{scriptsize}
\end{center}
\end{figure}

\clearpage

\begin{figure}[ht!]
\begin{center}
\caption{Cumulative return from trading strategy.}
\label{figure:cumulative-return}
\begin{tabular}{c}
\includegraphics[height=8cm,width=0.48\textwidth]{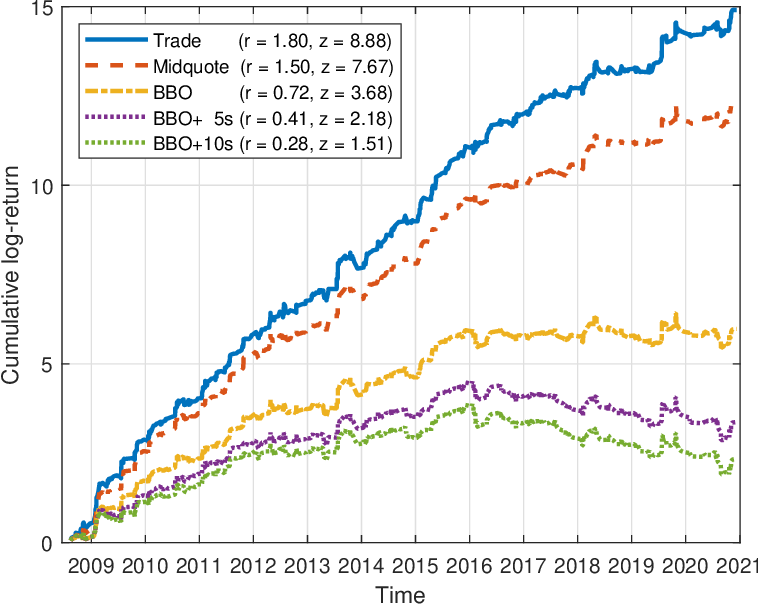}
\end{tabular}
\begin{scriptsize}
\parbox{\textwidth}{\emph{Note.} The figure shows the cumulative return from a trading strategy that employs the standardized earnings surprise to predict the one-minute post-announcement return: $r_{1m,it}^{ \text{EA}} = a + b_{1} z_{ \text{EPS},it}^{+} + b_{2} z_{ \text{EPS},it}^{-} + \epsilon_{it}$. A long (short) position in the stock is entered if the predicted return is greater (smaller) than 0.5\% (-0.5\%) and held until 6:30pm. The sample average return and test statistic for testing that the mean return is zero, based on robust standard errors, is reported in parenthesis. ``Trade'' employs the transaction price, ``Midquote'' the midquote, and ``BBO'' the best bid and offer. ``+Xs'' enforces a latency delay of X seconds before entrance.}
\end{scriptsize}
\end{center}
\end{figure}

\clearpage

\begin{sidewaysfigure}[ht!]
\begin{center}
\caption{Cumulative returns from trading strategy.}
\label{figure:trading-result}
\begin{tabular}{ccc}
\small{Panel A: Trade.} & \small{Panel B: Midquote.} & \small{Panel C: BBO.} \\
\includegraphics[height=0.35\textwidth,width=0.30\textwidth]{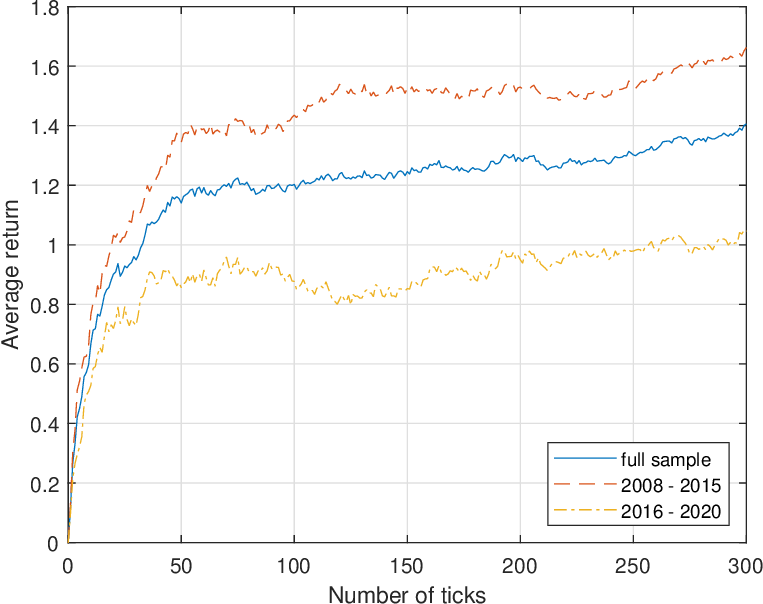} &
\includegraphics[height=0.35\textwidth,width=0.30\textwidth]{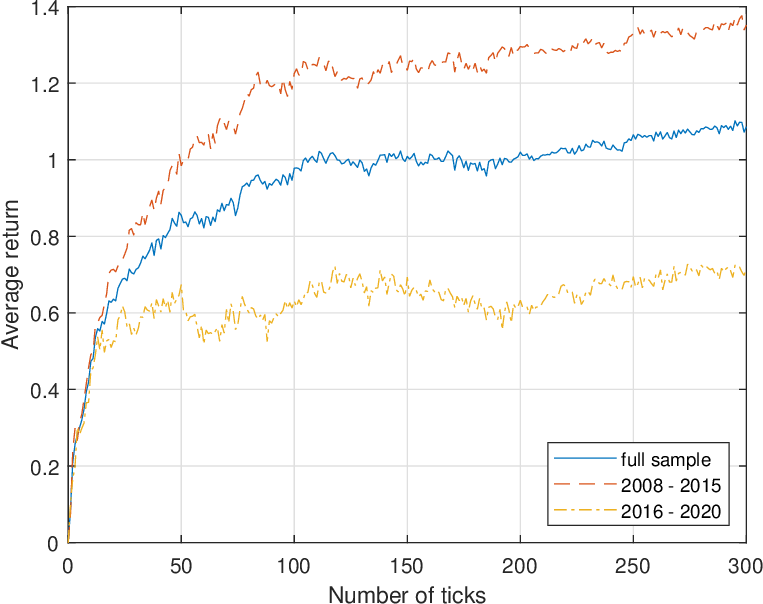} &
\includegraphics[height=0.35\textwidth,width=0.30\textwidth]{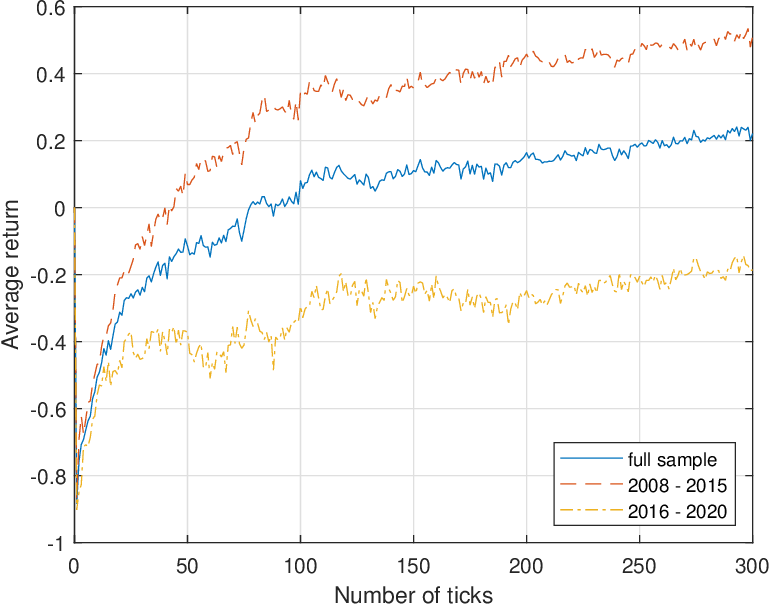} \\
\end{tabular}
\begin{scriptsize}
\parbox{\textwidth}{\emph{Note.} The figure shows the cumulative return in percent from a trading strategy that employs the standardized earnings surprise to predict the one-minute post-announcement return: $r_{1m,it}^{ \text{EA}} = a + b_{1} z_{ \text{EPS},it}^{+} + b_{2} z_{ \text{EPS},it}^{-} + \epsilon_{it}$. A long (short) position in the stock is entered if the predicted excess return is greater (smaller) than 0.5\% (-0.5\%). The cumulative return is based on closing out the position after a fixed number of tick updates. ``Trade'' employs the transaction price, ``Midquote'' the midquote, and ``BBO'' the best bid and offer.}
\end{scriptsize}
\end{center}
\end{sidewaysfigure}

\clearpage

\begin{figure}[ht!]
\begin{center}
\caption{One-minute return distribution.}
\label{figure:return-distribution}
\begin{tabular}{cc}
\small{Panel A: Excluding AFs} & \small{Panel B: Including AFs} \\
\includegraphics[height=8cm,width=0.48\textwidth]{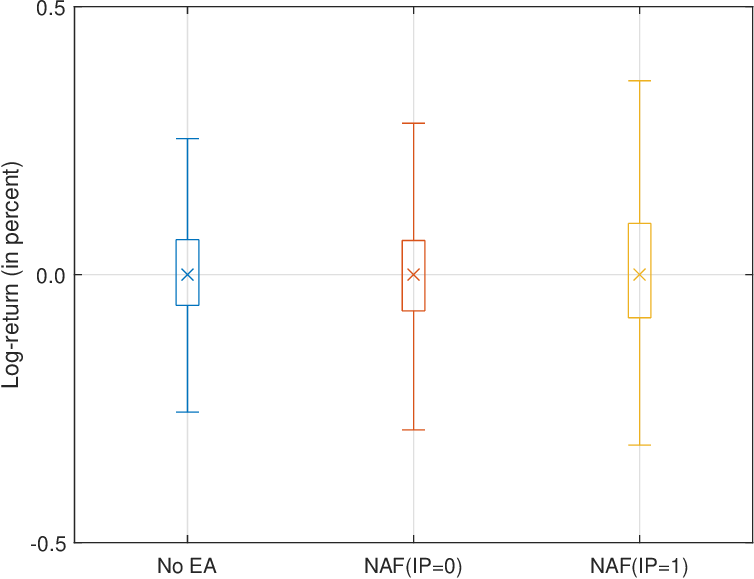} &
\includegraphics[height=8cm,width=0.48\textwidth]{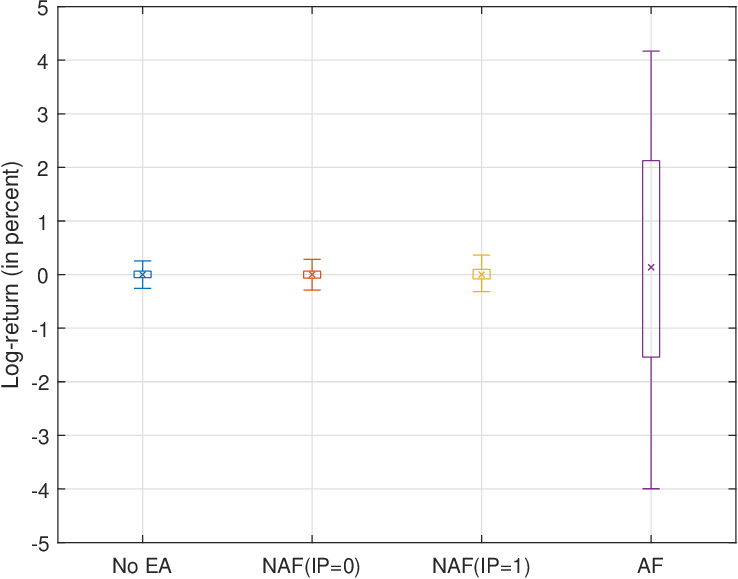}
\end{tabular}
\begin{scriptsize}
\parbox{\textwidth}{\emph{Note.} The box plots in the left panel show the one-minute return distribution of (i) non-announcing firms (NAFs) on common non-announcement days (no EA), (ii) NAFs outside the industry of the announcing firm (AF), NAF(IP = 0), and (iii) NAFs within the industry of the AF, NAF(IP = 1). The cross measures the median, the box is the interquartile range, and the whisker marks the 0.10 (lower) and 0.90 (upper) percentile. While (i) is based on days when none of our firms are announcing, (ii) and (iii) are based on days with at least one announcement. We construct (i) by taking for each announcement of an AF a common non-announcement day at random and measuring the stock price movement of the AF over the one-minute window aligned with its associated announcement time. In Panel B, we add the one-minute post-announcement return distribution of the AFs in our sample.}
\end{scriptsize}
\end{center}
\end{figure}

\clearpage

\begin{sidewaystable}[p!]
\begin{footnotesize}
\setlength{ \tabcolsep}{0.10cm}
\begin{center}
\caption{ \normalsize S\&P 500 companies by after-hours (4:00pm--6:30pm) trading activity.}
\label{table:sp500-volume-pm.tex}
\vspace*{-0.25cm}
\begin{tabular}{lcrrrrrrrrrrrrrrrrrrrr}
\hline \hline
& & \multicolumn{7}{c}{Conditional on no announcement} && \multicolumn{7}{c}{Conditional on announcement} && \multicolumn{4}{c}{Announcement information}\\
\cline{3-9} \cline{11-17} \cline{19-22} 
& & & & & \multicolumn{4}{c}{Quantile} & & & & & \multicolumn{4}{c}{Quantile}\\
\cline{6-9} \cline{14-17}
Ticker & SIC & Mean & Fraction & Std. & 0.25 & 0.50 & 0.75 & 0.99 & & Mean & Fraction & Std. & 0.25 & 0.50 & 0.75 & 0.99 & & \#EA & Time & $z_{ \text{EPS}}$ & $r_{1m}^{ \text{EA}}$\\
\hline
FB & 7370 & 891 & 0.52 & 1,258 & 307& 526& 1,000& 5,380 & & 76,129 & 20.80 & 29,991 & 59,075 & 70,024 & 88,936 & 211,431 & & 34 & 4:05pm & $\underset{(1.252)}{1.750}$ & $\underset{(4.446)}{0.336}$
\\
AAPL & 3571 & 1,552 & 0.61 & 3,950 & 386& 649& 1,235& 18,311 & & 51,346 & 15.94 & 29,927 & 32,891 & 44,509 & 56,832 & 194,220 & & 40 & 4:30pm & $\underset{(1.584)}{1.669}$ & $\underset{(3.395)}{0.806}$
\\
NFLX & 7841 & 367 & 0.55 & 606 & 52& 128& 434& 2,732 & & 33,328 & 20.89 & 26,447 & 14,460 & 25,303 & 46,619 & 122,817 & & 47 & 4:05pm & $\underset{(1.621)}{1.028}$ & $\underset{(6.769)}{-0.236}$
\\
AMZN & 5961 & 808 & 0.82 & 1,442 & 74& 152& 866& 6,751 & & 33,612 & 22.16 & 18,872 & 22,063 & 28,423 & 38,523 & 94,020 & & 50 & 4:01pm & $\underset{(1.790)}{0.687}$ & $\underset{(5.714)}{-0.369}$
\\
NVDA & 3674 & 405 & 0.36 & 860 & 39& 74& 428& 3,912 & & 16,909 & 8.45 & 20,656 & 1,903 & 4,456 & 36,972 & 68,790 & & 48 & 4:20pm & $\underset{(1.787)}{1.626}$ & $\underset{(3.988)}{0.714}$
\\
TSLA & 3711 & 2,156 & 0.86 & 10,399 & 25& 195& 824& 32,536 & & 27,210 & 15.41 & 35,016 & 1,287 & 15,304 & 30,270 & 139,936 & & 40 & 4:05pm & $\underset{(1.696)}{0.355}$ & $\underset{(3.422)}{1.248}$
\\
AMD & 3674 & 579 & 0.38 & 1,486 & 40& 80& 612& 5,722 & & 15,417 & 7.37 & 21,607 & 1,711 & 4,849 & 25,714 & 86,880 & & 47 & 4:15pm & $\underset{(1.342)}{0.746}$ & $\underset{(4.201)}{-0.859}$
\\
MU & 3674 & 251 & 0.21 & 480 & 53& 102& 280& 1,959 & & 13,508 & 6.89 & 13,340 & 1,990 & 11,078 & 21,868 & 59,740 & & 50 & 4:01pm & $\underset{(1.259)}{0.262}$ & $\underset{(2.927)}{0.196}$
\\
MSFT & 7372 & 609 & 0.27 & 1,425 & 112& 195& 479& 6,460 & & 20,152 & 7.31 & 17,196 & 9,919 & 16,144 & 24,788 & 104,119 & & 46 & 4:03pm & $\underset{(2.698)}{2.257}$ & $\underset{(1.787)}{-0.148}$
\\
GOOG & 7370 & 198 & 0.60 & 370 & 61& 100& 201& 1,524 & & 15,221 & 20.81 & 7,253 & 9,630 & 14,043 & 19,716 & 30,074 & & 49 & 4:01pm & $\underset{(1.849)}{0.930}$ & $\underset{(2.973)}{0.347}$
\\
INTC & 3674 & 278 & 0.20 & 549 & 92& 159& 281& 2,029 & & 22,705 & 9.72 & 17,642 & 12,198 & 17,908 & 27,180 & 102,918 & & 42 & 4:01pm & $\underset{(2.250)}{2.320}$ & $\underset{(2.425)}{-0.034}$
\\
DIS & 7990 & 328 & 0.38 & 898 & 26& 57& 276& 3,666 & & 10,043 & 7.87 & 12,298 & 2,116 & 5,545 & 11,760 & 56,136 & & 48 & 4:15pm & $\underset{(2.192)}{1.263}$ & $\underset{(2.200)}{-0.194}$
\\
CRM & 7372 & 102 & 0.24 & 296 & 18& 34& 90& 1,066 & & 10,131 & 12.71 & 13,701 & 4,216 & 6,285 & 10,431 & 86,704 & & 46 & 4:05pm & $\underset{(1.728)}{1.761}$ & $\underset{(3.484)}{0.874}$
\\
NKE & 3021 & 98 & 0.23 & 171 & 16& 33& 120& 715 & & 7,339 & 10.79 & 6,527 & 2,679 & 5,947 & 9,337 & 36,950 & & 48 & 4:15pm & $\underset{(2.632)}{2.309}$ & $\underset{(2.903)}{0.421}$
\\
CSCO & 3576 & 194 & 0.17 & 239 & 74& 124& 219& 1,210 & & 20,912 & 10.00 & 16,488 & 11,140 & 17,694 & 24,999 & 91,696 & & 50 & 4:05pm & $\underset{(1.428)}{2.331}$ & $\underset{(2.588)}{-0.109}$
\\
SBUX & 5810 & 145 & 0.24 & 314 & 39& 61& 153& 1,269 & & 6,607 & 7.46 & 4,012 & 3,670 & 5,691 & 9,331 & 16,886 & & 50 & 4:05pm & $\underset{(1.872)}{1.126}$ & $\underset{(2.998)}{-0.347}$
\\
IBM & 3570 & 79 & 0.21 & 124 & 30& 49& 87& 469 & & 8,385 & 11.85 & 5,432 & 5,161 & 7,474 & 9,924 & 30,725 & & 48 & 4:08pm & $\underset{(1.337)}{1.203}$ & $\underset{(2.058)}{0.030}$
\\
CMG & 5812 & 47 & 0.45 & 157 & 9& 24& 53& 263 & & 4,695 & 18.30 & 4,002 & 1,222 & 3,821 & 7,024 & 16,027 & & 48 & 4:10pm & $\underset{(2.324)}{1.231}$ & $\underset{(4.224)}{-0.443}$
\\
PYPL & 7389 & 205 & 0.28 & 273 & 51& 125& 216& 1,297 & & 10,836 & 8.74 & 10,650 & 4,104 & 6,680 & 13,351 & 47,021 & & 21 & 4:15pm & $\underset{(1.944)}{2.159}$ & $\underset{(2.525)}{0.116}$
\\
ORCL & 7372 & 102 & 0.11 & 554 & 39& 61& 96& 612 & & 8,994 & 6.34 & 5,863 & 4,131 & 8,979 & 11,217 & 26,529 & & 49 & 4:00pm & $\underset{(2.804)}{1.757}$ & $\underset{(2.402)}{-0.220}$
\\
ADBE & 7372 & 84 & 0.24 & 202 & 26& 40& 81& 624 & & 5,389 & 9.51 & 4,273 & 2,675 & 4,628 & 6,477 & 23,037 & & 49 & 4:05pm & $\underset{(2.084)}{2.235}$ & $\underset{(2.856)}{0.102}$
\\
QCOM & 3663 & 145 & 0.18 & 581 & 49& 74& 130& 947 & & 9,695 & 8.21 & 7,473 & 5,000 & 7,594 & 12,302 & 33,231 & & 49 & 4:00pm & $\underset{(1.444)}{1.355}$ & $\underset{(2.639)}{0.500}$
\\
EBAY & 7389 & 77 & 0.13 & 86 & 38& 55& 85& 374 & & 8,664 & 8.15 & 7,007 & 4,180 & 7,488 & 11,404 & 42,753 & & 49 & 4:15pm & $\underset{(1.518)}{1.624}$ & $\underset{(3.514)}{-0.489}$
\\
GILD & 2836 & 208 & 0.23 & 2,014 & 46& 69& 112& 1,685 & & 4,125 & 4.07 & 5,449 & 566 & 1,830 & 7,064 & 25,630 & & 48 & 4:05pm & $\underset{(1.703)}{0.671}$ & $\underset{(1.362)}{0.081}$
\\
AMAT & 3674 & 88 & 0.15 & 106 & 37& 59& 101& 476 & & 3,368 & 3.85 & 3,221 & 666 & 2,268 & 4,998 & 12,834 & & 48 & 4:01pm & $\underset{(1.878)}{1.573}$ & $\underset{(2.278)}{0.324}$
\\
\hline \hline
\end{tabular}
\smallskip
\begin{scriptsize}
\parbox{0.98\textwidth}{\emph{Note.} 
In the left-hand side, we show descriptive statistics of the most liquid companies from the S\&P 500 index based on their after-hours market (4:00pm--6:30pm) trading volume, which is separated into no announcement days (columns 2--8) and announcement days (columns 9--15).
Ticker is the stock listing symbol.
SIC is the standard industrial classification code.
Mean and Std. are the transaction count sample average and standard deviation.
Quantile are selected quantiles from the empirical transaction count distribution.
Fraction is the number of transactions executed in the after-hours session divided by the total transaction count for the extended trading session from 9:30am--6:30pm.
In the right-hand side, we report announcement information.
\#EA is the number of earnings announcements.
Time is the modal announcement time (Eastern Standard Time).
$z_{ \text{EPS}}$ and $r_{1m}^{ \text{EA}}$ are the sample average standardized unexpected earnings and one-minute post-announcement return (standard deviation across announcements reported below in parenthesis).
}
\end{scriptsize}
\end{center}
\end{footnotesize}
\end{sidewaystable}

\clearpage

\begin{sidewaystable}[p!]
\setlength{\tabcolsep}{0.25cm}
\begin{center}
\caption{Sample average of realized variance, bipower variation, and jump proportion.}
\label{table:rv-descriptive}
\vspace*{-0.25cm}
\begin{tabular}{lrrrrrrrrrrrrrrrrrrr}
\hline \hline
& \multicolumn{9}{c}{Panel A: Regular trading session (9:30am--4:00pm)} && \multicolumn{9}{c}{Panel B: Extended trading session (9:30am--6:30pm)} \\
\cline{2-10} \cline{12-20}
& \multicolumn{4}{c}{no EA} && \multicolumn{4}{c}{EA} && \multicolumn{4}{c}{no EA} && \multicolumn{4}{c}{EA} \\
\cline{2-5} \cline{7-10} \cline{12-15} \cline{17-20}
& RV & BV & JP & JF && RV & BV & JP & JF && RV & BV & JP & JF && RV & BV & JP & JF \\
FB & 16.2 & 16.1 & 0.6 & 1.0 && 26.0 & 25.9 & 1.6 & 0.0 && 16.5 & 16.2 & 1.6 & 3.1 && 94.0 & 66.2 & 42.6 & 97.1 \\
AAPL & 20.3 & 20.0 & 1.6 & 1.6 && 19.8 & 19.3 & 4.1 & 5.0 && 20.6 & 20.2 & 2.2 & 3.3 && 61.1 & 37.9 & 51.7 & 100.0 \\
NFLX & 30.5 & 29.9 & 0.8 & 1.3 && 35.5 & 35.2 & 0.9 & 2.1 && 30.7 & 30.0 & 1.2 & 2.4 && 132.2 & 68.1 & 57.3 & 97.9 \\
AMZN & 23.0 & 22.9 & 0.0 & 0.5 && 27.4 & 27.1 & 1.2 & 0.0 && 23.2 & 23.0 & 0.7 & 2.4 && 103.2 & 59.1 & 50.4 & 100.0 \\
NVDA & 29.1 & 28.3 & 3.8 & 1.9 && 34.4 & 33.7 & 3.4 & 2.1 && 29.4 & 28.3 & 4.5 & 3.5 && 92.4 & 46.9 & 54.2 & 97.9 \\
TSLA & 27.9 & 27.5 & 0.2 & 0.5 && 37.6 & 36.9 & 0.2 & 5.0 && 28.4 & 27.8 & 0.9 & 2.0 && 100.0 & 69.9 & 37.7 & 90.0 \\
AMD & 38.0 & 33.1 & 14.8 & 10.9 && 40.9 & 34.8 & 16.7 & 19.1 && 38.4 & 33.1 & 15.3 & 12.6 && 95.7 & 50.3 & 49.0 & 93.6 \\
MU & 35.6 & 33.3 & 7.6 & 4.5 && 40.4 & 37.3 & 7.8 & 2.0 && 35.8 & 33.2 & 8.4 & 6.3 && 83.6 & 56.0 & 40.7 & 96.0 \\
MSFT & 18.7 & 18.1 & 6.4 & 5.3 && 20.9 & 20.2 & 6.1 & 4.3 && 18.9 & 18.1 & 6.9 & 7.0 && 54.8 & 35.1 & 49.8 & 100.0 \\
GOOG & 17.5 & 17.4 & 0.6 & 0.8 && 20.5 & 20.5 & 0.4 & 0.0 && 17.7 & 17.4 & 1.1 & 2.2 && 69.9 & 38.8 & 51.1 & 93.9 \\
INTC & 20.6 & 19.8 & 7.8 & 6.4 && 24.4 & 23.6 & 7.1 & 0.0 && 20.8 & 19.8 & 8.4 & 8.3 && 66.6 & 42.5 & 50.6 & 100.0 \\
DIS & 17.4 & 16.9 & 4.3 & 2.7 && 21.8 & 21.2 & 5.9 & 4.2 && 17.5 & 16.9 & 4.5 & 3.0 && 52.0 & 24.7 & 61.6 & 97.9 \\
CRM & 23.6 & 22.9 & 3.7 & 1.5 && 32.3 & 31.2 & 5.0 & 2.2 && 23.7 & 22.7 & 4.0 & 2.1 && 80.3 & 34.5 & 59.0 & 95.7 \\
NKE & 17.5 & 16.9 & 4.5 & 2.1 && 21.3 & 20.7 & 5.3 & 8.3 && 17.5 & 16.9 & 4.5 & 2.2 && 60.3 & 21.6 & 63.5 & 100.0 \\
CSCO & 18.9 & 17.9 & 9.9 & 8.7 && 22.1 & 21.0 & 9.9 & 14.0 && 19.0 & 17.9 & 10.3 & 9.9 && 67.1 & 44.1 & 47.5 & 100.0 \\
SBUX & 19.3 & 18.7 & 4.6 & 2.6 && 23.7 & 22.8 & 6.5 & 2.0 && 19.4 & 18.6 & 4.9 & 3.3 && 62.1 & 32.1 & 52.7 & 100.0 \\
IBM & 15.0 & 14.7 & 3.2 & 1.6 && 17.4 & 17.1 & 2.9 & 2.1 && 15.0 & 14.7 & 3.3 & 1.8 && 47.3 & 20.9 & 62.5 & 100.0 \\
CMG & 18.2 & 17.2 & 2.7 & 1.6 && 23.5 & 22.8 & 3.1 & 0.0 && 18.2 & 17.1 & 2.8 & 1.8 && 76.1 & 17.8 & 60.5 & 87.5 \\
PYPL & 10.1 & 9.9 & 3.7 & 1.1 && 25.4 & 25.1 & 3.0 & 0.0 && 10.1 & 9.9 & 4.1 & 1.6 && 80.2 & 47.6 & 53.5 & 90.5 \\
ORCL & 17.8 & 17.0 & 7.7 & 5.2 && 23.7 & 22.2 & 9.5 & 10.2 && 17.8 & 17.0 & 7.7 & 5.3 && 61.2 & 27.2 & 60.2 & 93.9 \\
ADBE & 19.8 & 19.2 & 3.8 & 1.9 && 24.3 & 23.7 & 4.0 & 2.0 && 19.9 & 19.2 & 4.0 & 2.4 && 68.7 & 31.2 & 54.2 & 98.0 \\
QCOM & 20.7 & 20.1 & 5.2 & 3.5 && 23.2 & 22.5 & 5.7 & 4.1 && 20.8 & 20.0 & 5.6 & 4.5 && 69.7 & 35.3 & 58.2 & 95.9 \\
EBAY & 21.5 & 20.6 & 7.0 & 3.9 && 27.9 & 26.7 & 7.2 & 4.1 && 21.6 & 20.5 & 7.0 & 4.0 && 83.9 & 40.6 & 57.5 & 95.9 \\
GILD & 21.0 & 20.2 & 5.1 & 3.6 && 24.8 & 23.8 & 6.0 & 4.2 && 21.1 & 20.2 & 5.4 & 4.4 && 45.3 & 25.4 & 46.7 & 81.2 \\
AMAT & 23.6 & 22.4 & 7.8 & 4.2 && 26.5 & 25.2 & 8.7 & 4.2 && 23.6 & 22.3 & 8.1 & 4.7 && 56.1 & 29.9 & 46.9 & 83.3 \\
\hline
Average & 21.7 & 20.8 & 4.7 & 3.2 && 26.6 & 25.6 & 5.3 & 4.0 && 21.8 & 20.8 & 5.1 & 4.2 && 74.5 & 40.1 & 52.8 & 95.4 \\
\hline \hline
\end{tabular}
\smallskip
\begin{scriptsize}
\parbox{0.98\textwidth}{ \emph{Note.} 
This table reports the pre-averaged realized variance (RV) and pre-averaged bipower variation (BV).
We also construct the jump proportion (JP), which is defined as Jump Proportion = 1 - Bipower variation/Realized Variance.
We further compute the jump frequency (JF), which is derived from the jump indicator in \eqref{equation:jump-indicator}.
The realized variance and bipower variation are converted to annualized standard deviation in percent for convenience.
The full sample is split into days without (no EA) and with (EA) earnings announcements.
The measures are then averaged across subsamples by company.
In Panel A, we report the analysis for the regular trading session (9:30am--4:00pm), whereas Panel B displays the associated results for the extended trading session (9:30am--6:30am).
The bottom row reports the grand mean (``Average'') over all stock-days.
}
\end{scriptsize}
\end{center}
\end{sidewaystable}

\clearpage

\begin{table}[ht!]
\setlength{\tabcolsep}{0.70cm}
\begin{center}
\caption{Logit estimates for the stock price jump probability.} \label{table:logit}
\begin{tabular}{lrrr}
\hline
\hline
Variable & \multicolumn{1}{c}{(1)} & \multicolumn{1}{c}{(2)} & \multicolumn{1}{c}{(3)}\\
\hline
Intercept & $\underset{(0.012)}{-2.970}$ & $\underset{(0.013)}{-3.268}$ & $\underset{(0.013)}{-3.268}$\\
$EA$ & & $\underset{(0.076)}{5.589}$ & $\underset{(0.328)}{3.172}$\\
$|z_{ \text{EPS}}^{+}|$ & & & $\underset{(0.042)}{0.059}$\\
$|z_{ \text{EPS}}^{-}|$ & & & $\underset{(0.233)}{0.616}$\\
$\sqrt{RV_{n}^{*}}$ & & & $\underset{(0.069)}{-0.179}$\\
$N_{A}$ & & & $\underset{(0.011)}{0.110}$\\ \\
Pseudo $R^{2}$ & & 0.1831 & 0.1853\\
\hline
\hline
\end{tabular}
\begin{scriptsize}
\parbox{\textwidth}{\emph{Note.} We estimate the logit regression: $P(J_{it} = 1) = F( a + b_{1} EA_{it} + b_{2} |z_{ \mathrm{EPS},it}^{+}| + b_{3} |z_{ \mathrm{EPS},it}^{-}| + b_{4} \sqrt{RV_{n,it}^{*}} + b_{5} N_{A,it})$, where $F$ is the logistic distribution function. $J_{it}$ is equal to one if there is a jump in the price of company $i$'s stock on day $t$, zero otherwise. $EA_{it}$ is equal to one if there is an earnings announcement for company $i$ on day $t$, zero otherwise. $z_{ \text{EPS},it}^{+}$ ($z_{ \text{EPS},it}^{-}$) is the standardized unexpected earnings for positive (negative) announcements, $\sqrt{RV_{n,it}^{*}}$ is the pre-averaged realized volatility calculated over the regular trading session on the announcement day, and $N_{A,it}$ is the number of analyst earnings forecasts for the announcement. The other covariates are interacted with $EA$ and take a value of zero on non-announcement days. The table reports parameter estimates of the full model in column (3) and of restricted versions in column (1)--(2). Standard errors are shown in parenthesis below the parameter estimate. The number of observations is 158{,}279 of which 156{,}086 are non-announcement days and 2{,}193 are announcements days. Pseudo $R^{2}$ is the likelihood ratio index, $R^{2} = 1 - L_{1}/L_{0}$, where $L_{0}$ is the log-likelihood of the constant model in (1) and $L_{1}$ is the log-likelihood of the models in (2)--(3).}
\end{scriptsize}
\end{center}
\end{table}

\clearpage

\begin{sidewaystable}[ht!]
\begin{footnotesize}
\setlength{ \tabcolsep}{0.06cm}
\begin{center}
\caption{Co-jump analysis.}
\label{table:co-jump}
\begin{tabular}{ll*{28}{r}}
\hline \hline
& & \multicolumn{24}{c}{Non-announcing firm (NAF)} \\ \cline{3-27}
& & Fb & Aapl & Nflx & Amzn & Nvda & Tsla & Amd & Mu & Msft & Goog & Intc & Dis & Crm & Nke & Csco & Sbux & Ibm & Cmg & Orcl & Adbe & Qcom & Ebay & Gild & Amat & Ulta & & MKT \\
\multirow{25}{*}{\rotatebox[origin=c]{90}{Announcing firm (AF)}}
 & Fb &  & 16.1 & 3.0 & 25.8 & 17.6 & 0.0 & 12.5 & 8.8 & 10.3 & 28.1 & 5.9 & 0.0 & 0.0 & 0.0 & 17.6 & 0.0 & 2.9 & 0.0 & 5.9 & 2.9 & 16.7 & 0.0 & 10.7 & 8.8 & 2.9 & & 14.7\\
& Aapl & 24.3 &  & 2.6 & 19.4 & 15.0 & 2.5 & 17.6 & 20.0 & 18.4 & 5.4 & 12.8 & 0.0 & 2.5 & 5.0 & 20.0 & 5.9 & 5.0 & 2.9 & 2.5 & 5.0 & 23.1 & 8.6 & 16.1 & 7.5 & 5.0 & & 30.0\\
& Nflx & 13.0 & 6.7 &  & 33.3 & 6.4 & 2.1 & 15.6 & 8.5 & 8.9 & 19.1 & 15.2 & 0.0 & 2.1 & 4.3 & 19.1 & 6.8 & 2.6 & 2.3 & 4.3 & 2.1 & 7.0 & 2.4 & 10.9 & 6.4 & 0.0 & & 6.4\\
& Amzn & 17.0 & 6.5 & 12.5 &  & 12.0 & 0.0 & 20.0 & 18.0 & 11.1 & 35.1 & 11.9 & 6.0 & 2.0 & 0.0 & 16.0 & 5.7 & 6.0 & 2.4 & 2.0 & 8.0 & 16.3 & 4.0 & 2.5 & 8.0 & 2.0 & & 12.0\\
& Nvda & 0.0 & 6.2 & 2.1 & 0.0 &  & 2.2 & 35.4 & 10.4 & 0.0 & 0.0 & 10.4 & 2.6 & 0.0 & 4.2 & 8.9 & 4.3 & 0.0 & 0.0 & 2.1 & 0.0 & 4.3 & 4.2 & 0.0 & 2.8 & 0.0 & & 2.1\\
& Tsla & 6.1 & 5.0 & 0.0 & 12.5 & 10.5 &  & 15.4 & 5.0 & 6.1 & 7.5 & 15.0 & 5.0 & 0.0 & 2.5 & 16.2 & 0.0 & 0.0 & 2.7 & 2.5 & 0.0 & 2.9 & 0.0 & 2.5 & 10.3 & 0.0 & & 5.0\\
& Amd & 13.3 & 9.8 & 2.2 & 12.8 & 19.1 & 2.2 &  & 19.1 & 15.8 & 6.7 & 19.1 & 2.1 & 0.0 & 0.0 & 10.6 & 4.9 & 0.0 & 0.0 & 6.4 & 2.1 & 6.7 & 4.7 & 4.8 & 0.0 & 6.4 & & 8.5\\
& Mu & 0.0 & 2.0 & 2.0 & 0.0 & 8.0 & 2.0 & 20.0 &  & 8.0 & 0.0 & 12.0 & 2.0 & 4.0 & 2.5 & 10.0 & 2.0 & 2.0 & 2.0 & 10.2 & 4.0 & 0.0 & 4.0 & 6.0 & 16.0 & 0.0 & & 6.0\\
& Msft & 17.1 & 6.8 & 4.5 & 31.2 & 6.5 & 0.0 & 13.5 & 15.2 &  & 20.0 & 17.1 & 6.5 & 2.2 & 0.0 & 13.0 & 2.6 & 2.3 & 5.6 & 2.2 & 2.2 & 7.0 & 2.3 & 0.0 & 13.0 & 2.2 & & 10.9\\
& Goog & 31.9 & 17.4 & 2.0 & 27.8 & 8.2 & 2.0 & 15.6 & 10.2 & 10.5 &  & 13.6 & 8.2 & 0.0 & 0.0 & 8.2 & 4.8 & 4.8 & 0.0 & 6.1 & 8.2 & 8.2 & 2.0 & 2.3 & 2.0 & 4.1 & & 8.2\\
& Intc & 9.5 & 7.3 & 4.9 & 5.9 & 11.9 & 2.4 & 38.1 & 26.2 & 24.3 & 10.8 &  & 2.4 & 4.8 & 2.4 & 9.5 & 5.9 & 2.7 & 0.0 & 2.4 & 2.4 & 10.0 & 2.6 & 5.0 & 14.3 & 0.0 & & 7.1\\
& Dis & 4.2 & 0.0 & 2.1 & 0.0 & 0.0 & 0.0 & 8.3 & 4.2 & 4.2 & 0.0 & 6.2 &  & 2.1 & 2.1 & 8.5 & 2.1 & 0.0 & 0.0 & 2.1 & 0.0 & 2.1 & 2.1 & 6.8 & 10.4 & 0.0 & & 2.1\\
& Crm & 2.2 & 2.2 & 0.0 & 0.0 & 4.3 & 4.3 & 10.9 & 8.7 & 2.2 & 0.0 & 8.7 & 6.5 &  & 0.0 & 13.3 & 2.2 & 2.2 & 0.0 & 4.3 & 0.0 & 2.2 & 6.5 & 2.2 & 0.0 & 0.0 & & 2.2\\
& Nke & 2.1 & 6.2 & 4.2 & 0.0 & 4.2 & 2.1 & 14.6 & 7.9 & 8.3 & 4.2 & 8.3 & 0.0 & 4.2 &  & 16.7 & 8.3 & 6.2 & 2.1 & 13.3 & 2.1 & 6.2 & 6.2 & 10.4 & 4.2 & 4.2 & & 6.2\\
& Csco & 0.0 & 4.0 & 0.0 & 2.0 & 0.0 & 4.3 & 18.0 & 2.0 & 8.0 & 0.0 & 4.0 & 4.1 & 0.0 & 4.0 &  & 4.0 & 4.0 & 0.0 & 6.0 & 2.0 & 2.0 & 2.0 & 4.0 & 2.4 & 0.0 & & 4.0\\
& Sbux & 13.0 & 4.5 & 4.3 & 8.6 & 10.2 & 0.0 & 22.7 & 22.0 & 9.5 & 11.6 & 9.5 & 2.0 & 4.0 & 2.0 & 8.0 &  & 0.0 & 0.0 & 0.0 & 4.0 & 15.2 & 2.3 & 4.3 & 6.0 & 2.0 & & 10.0\\
& Ibm & 4.2 & 8.3 & 2.6 & 14.6 & 6.2 & 2.1 & 7.3 & 18.8 & 15.2 & 7.3 & 16.3 & 0.0 & 4.2 & 2.1 & 18.8 & 0.0 &  & 0.0 & 2.1 & 0.0 & 8.5 & 2.3 & 2.1 & 4.2 & 4.2 & & 4.2\\
& Cmg & 13.0 & 14.0 & 11.1 & 25.6 & 8.3 & 0.0 & 21.1 & 18.8 & 18.4 & 9.5 & 19.1 & 4.5 & 0.0 & 0.0 & 16.7 & 4.4 & 4.3 &  & 4.2 & 4.2 & 4.8 & 4.5 & 0.0 & 10.4 & 2.1 & & 6.2\\
& Orcl & 2.0 & 6.1 & 2.0 & 2.0 & 0.0 & 2.0 & 8.2 & 10.4 & 4.1 & 2.0 & 10.2 & 10.2 & 4.1 & 4.3 & 8.2 & 6.1 & 8.2 & 2.0 &  & 7.3 & 6.1 & 6.1 & 10.2 & 6.1 & 0.0 & & 12.8\\
& Adbe & 4.1 & 10.2 & 4.1 & 0.0 & 0.0 & 2.0 & 6.1 & 8.2 & 6.1 & 0.0 & 10.2 & 4.1 & 0.0 & 4.1 & 8.2 & 2.0 & 2.0 & 4.1 & 4.9 &  & 2.0 & 12.2 & 2.0 & 4.1 & 2.2 & & 4.3\\
& Qcom & 0.0 & 8.3 & 4.4 & 12.5 & 6.2 & 2.3 & 10.6 & 12.2 & 10.9 & 12.2 & 12.8 & 2.1 & 4.1 & 2.0 & 4.2 & 0.0 & 2.1 & 2.3 & 2.0 & 2.0 &  & 2.9 & 4.3 & 10.2 & 2.0 & & 10.2\\
& Ebay & 9.3 & 4.5 & 2.3 & 30.6 & 4.1 & 6.4 & 17.8 & 14.3 & 17.4 & 6.1 & 10.9 & 2.0 & 4.1 & 8.2 & 14.3 & 0.0 & 4.5 & 6.7 & 4.1 & 6.1 & 5.7 &  & 4.2 & 8.2 & 2.0 & & 8.2\\
& Gild & 19.0 & 12.8 & 4.3 & 10.5 & 12.5 & 2.1 & 11.6 & 14.6 & 8.7 & 11.9 & 8.7 & 4.5 & 4.2 & 2.1 & 18.8 & 4.5 & 6.2 & 2.4 & 6.2 & 6.2 & 10.9 & 2.1 &  & 4.2 & 2.1 & & 8.3\\
& Amat & 0.0 & 4.2 & 0.0 & 0.0 & 2.8 & 2.1 & 14.6 & 10.4 & 6.2 & 0.0 & 12.5 & 0.0 & 4.5 & 4.2 & 12.5 & 6.2 & 2.1 & 0.0 & 8.3 & 2.1 & 4.2 & 4.2 & 4.2 &  & 4.2 & & 2.1\\
& Ulta & 2.4 & 0.0 & 0.0 & 2.4 & 2.4 & 2.4 & 9.8 & 4.9 & 7.3 & 2.4 & 4.9 & 9.8 & 0.0 & 4.9 & 19.5 & 9.8 & 0.0 & 0.0 & 5.1 & 0.0 & 9.8 & 2.4 & 4.9 & 7.3 &  & & 4.9\\
\hline \hline
\end{tabular}
\begin{scriptsize}
\parbox{0.98\textwidth}{\emph{Note.} We report the proportion of extended trading sessions with a jump in the non-announcing firm (NAF, in columns) when an announcing firm (AF, in rows) report financial results. The column MKT reports the jump frequency of the equity market (as proxied by SPY), conditional on the various firms' announcement days.}
\end{scriptsize}
\end{center}
\end{footnotesize}
\end{sidewaystable}

\clearpage

\begin{table}[ht!]
\setlength{\tabcolsep}{0.60cm}
\begin{center}
\caption{Logit estimates for the co-jump probability.}
\label{table:logit-co-jump}
\begin{tabular}{lrrrrr}
\hline
\hline
Variable & \multicolumn{1}{c}{(1)} & \multicolumn{1}{c}{(2)} & \multicolumn{1}{c}{(3)} & \multicolumn{1}{c}{(4)} & \multicolumn{1}{c}{(5)}\\ \hline
\multicolumn{6}{l}{ \textit{Panel A: Firm-level}} \\
Intercept & $\underset{(0.015)}{-2.960}$ & $\underset{(0.017)}{-3.020}$ & $\underset{(0.043)}{-3.595}$ & $\underset{(0.019)}{-2.900}$ & $\underset{(0.046)}{-3.611}$ \\
$IP$ & & $\underset{(0.048)}{0.354}$ & & & $\underset{(0.049)}{0.368}$ \\
log($V^{ \text{NAF}}$) & & & $\underset{(0.009)}{0.132}$ & & $\underset{(0.009)}{0.148}$ \\
log($V^{ \text{AF}}$) & & & $\underset{(0.010)}{0.063}$ & & $\underset{(0.010)}{0.046}$ \\
$D^{ \text{EARLY}}$ & & & & $\underset{(0.036)}{0.060}$ & $\underset{(0.036)}{0.117}$ \\
$D^{ \text{LATE}}$ & & & & $\underset{(0.041)}{-0.372}$ & $\underset{(0.041)}{-0.417}$ \\
\\
\multicolumn{6}{l}{ \textit{Panel B: Market-index}} \\
Intercept & $\underset{(0.137)}{-2.910}$ & $\underset{(0.153)}{-3.063}$ & $\underset{(0.363)}{-4.036}$ & & $\underset{(0.364)}{-3.941}$ \\
$D^{ \# \text{AF}}$ & & $\underset{(0.360)}{1.287}$ & & & $\underset{(0.381)}{0.958}$\\
log($CV^{ \text{AF}}$) & & & $\underset{(0.098)}{0.359}$ & & $\underset{(0.103)}{0.293}$\\
\hline
\hline
\end{tabular}
\begin{scriptsize}
\parbox{\textwidth}{\emph{Note.} In Panel A, we estimate the logit regression: $P \big(J_{jt}^{ \text{NAF}} = 1 \mid EA_{it} = 1 \big) = F \big( a + b_{1} IP_{ij} + b_{2} \log(V_{jt}^{ \text{NAF}}) + b_{3} \log(V_{it}^{ \text{AF}}) +  b_{4} D^{ \text{EARLY}}_{it} + b_{5} D^{ \text{LATE}}_{it}  \big)$, where $F$ is the logistic distribution function. $J_{jt}^{ \text{NAF}}$ is equal to one if there is a jump in the price of the non-announcing firm (NAF) $j$ on day $t$, zero otherwise. $EA_{it}$ is equal to one if the announcing firm (AF) $i$ releases its quarterly earnings report on day $t$, zero otherwise. $IP_{ij}$ is an industry proximity score between 0 (no overlap) and 1 (perfect overlap) calculated from the SIC codes of company $i$ and $j$ following \citet{wang-zajac:07a}. $V^{ \text{NAF}}$ and $V^{ \text{AF}}$ are the average number of transactions in the NAF and AF in the after-hours market in the month prior to $t$. $D^{ \text{EARLY}}_{it}$ and $D^{ \text{LATE}}_{it}$ are indicator variables of whether the AF discloses its fiscal statement early or late in the earnings cycle. In Panel B, we estimate the logit regression: $P(J_{t}^{ \text{SPY}} = 1 \mid \# \text{AF}_{t} \geq 1) = F \left( a + b_{1} D_{t}^{\# \text{AF}} + b_{2} \log( CV_{t}^{ \text{AF}}) \right)$. Here, $J_{t}^{ \text{SPY}}$ is equal to one if there is a jump in the market index (proxied by SPY) on day $t$, zero otherwise. Moreover, $D_{t}^{\# \text{AF}}$ is equal to one if $\# \text{AF}_{t} \geq 5$, zero otherwise, where $\# \text{AF}_{t} = \sum_{i} EA_{it}$ is the number of AFs on day $t$. $CV_{t}^{ \text{AF}} = \sum_{i \, : \, EA_{it} = 1} V_{it}^{ \text{AF}}$ is the cumulative average number of transactions in the AFs in the after-hours market in the month prior to $t$. The table reports parameter estimates of the full model in column (5) and of restricted versions in column (1)--(4). Standard errors are shown in parenthesis below the parameter estimate. The number of observations is 97{,}811 in Panel A and 1,084 in Panel B.}
\end{scriptsize}
\end{center}
\end{table}

\clearpage

\begin{table}[ht!]
\setlength{\tabcolsep}{0.50cm}
\begin{center}
\caption{Post-announcement return regression.}
\label{table:return-regression}
\begin{tabular}{lrrrr}
\hline
\hline
Variable & \multicolumn{1}{c}{(1)} & \multicolumn{1}{c}{(2)} & \multicolumn{1}{c}{(3)} & \multicolumn{1}{c}{(4)}\\ \hline
Intercept & $\underset{(0.075)}{0.145}$ & $\underset{(0.108)}{-0.370}$ & $\underset{(0.132)}{-0.745}$ & $\underset{(0.131)}{-0.731}$\\
$z_{ \text{EPS}}^{+}$ & & $\underset{(0.037)}{0.345}$ & $\underset{(0.037)}{0.281}$ & $\underset{(0.037)}{0.267}$\\
$z_{ \text{EPS}}^{-}$ & & $\underset{(0.157)}{0.784}$ & $\underset{(0.155)}{0.327}$ & $\underset{(0.149)}{0.330}$\\
$\sqrt{RV_{n}^{*}}$ & & & $\underset{(0.074)}{0.165}$ & $\underset{(0.074)}{0.159}$\\
$N_{A}$ & & & $\underset{(0.006)}{0.014}$ & $\underset{(0.006)}{0.013}$\\
$OI$ & & & & $\underset{(0.188)}{1.394}$\\ \\
$\bar{R}^{2}$ & & 0.0682 & 0.0795 & 0.0960\\ \\
$P$-value & & 0.0016 & 0.7682 & 0.6811\\
\hline
\hline
\end{tabular}
\begin{scriptsize}
\parbox{\textwidth}{\emph{Note.} We estimate the linear regression: $r_{1m,it}^{ \text{EA}} = a + b_{1} z_{ \text{EPS},it}^{+} + b_{2} z_{ \text{EPS},it}^{-} + b_{3} \sqrt{RV_{n,it}^{*}} D_{it} + b_{4} N_{A,it} D_{it} + b_{5} OI_{it} + \epsilon_{it}$. $r_{1m,it}^{ \text{EA}}$ is the one-minute post-announcement return, $z_{ \text{EPS},it}^{+}$ ($z_{ \text{EPS},it}^{-}$) is the standardized unexpected earnings for positive (negative) announcements, $\sqrt{RV_{n,it}^{*}}$ is the pre-averaged realized volatility calculated over the regular trading session on the announcement day, $N_{A,it}$ is the number of analyst earnings forecasts for the announcement, $OI_{it}$ is the one-minute post-announcement cumulative net order imbalance, and $D_{it} = \sign(z_{ \text{EPS},it})$ is the sign function. The table reports parameter estimates of the full model in column (4) and of restricted versions in column (1)--(3). Standard errors are shown in parenthesis below the parameter estimate. The number of observations is 2{,}193. $\bar{R}^{2}$ is the adjusted coefficient of multiple determination. The $P$-value is for testing the hypothesis $H_{0} : b_{1} = b_{2}$ against $H_{1} : b_{1} \neq b_{2}$.}
\end{scriptsize}
\end{center}
\end{table}

\clearpage

\begin{table}[ht!]
\setlength{\tabcolsep}{0.20cm}
\begin{center}
\caption{Average return from trading strategy.}
\label{table:trading-strategy}
\begin{footnotesize}
\begin{tabular}{l*{7}{c}}
\hline
\hline
\multicolumn{4}{l}{ \textit{Panel A: Baseline termination rule}} \\
& \multicolumn{3}{c}{Terminate at 6:30pm (EOD).} \\ \cline{2-4}
& Full sample & 2008--2015 & 2016--2020 \\
Trade & 1.80 (8.88) & 2.30 (8.77) & 1.11 (3.49) \\
Midquote & 1.50 (7.67) & 2.00 (7.87) & 0.80 (2.65) \\
BBO & 0.72 (3.68) & 1.23 (4.84) & 0.01 (0.03) \\
BBO+5s & 0.41 (2.18) & 0.93 (3.88) & -0.32 (-1.09) \\
BBO+10s & 0.28 (1.51) & 0.80 (3.32) & -0.44 (-1.48) \\ \\
\multicolumn{4}{l}{ \textit{Panel B: Alternative termination rules}} \\
& \multicolumn{3}{c}{Terminate after fixed number of minutes.} && \multicolumn{3}{c}{Terminate after fixed number of ticks.} \\ \cline{2-4} \cline{6-8}
& Full sample & 2008--2015 & 2016--2020 && Full sample & 2008--2015 & 2016--2020\\
& \multicolumn{3}{c}{ \textit{Panel B.1: Stop after 30 seconds}} && \multicolumn{3}{c}{ \textit{Panel B.6: Stop after 50 ticks}} \\
Trade & 0.74 (8.28) & 0.82 (7.13) & 0.62 (4.43) && 1.14 (9.84) & 1.34 (8.21) & 0.86 (5.47) \\
Midquote & 0.56 (6.31) & 0.66 (5.35) & 0.43 (3.39) && 0.85 (8.04) & 0.98 (6.93) & 0.67 (4.23) \\
BBO & -0.47 (-4.75) & -0.35 (-2.60) & -0.64 (-4.42) && -0.11 (-1.03) & 0.07 (0.47) & -0.37 (-2.17) \\
BBO+5s & -0.79 (-9.72) & -0.65 (-6.31) & -0.97 (-7.57) && -0.43 (-4.13) & -0.24 (-1.70) & -0.70 (-4.48) \\
BBO+10s & -0.91 (-12.11) & -0.78 (-8.00) & -1.09 (-9.34) && -0.56 (-5.45) & -0.37 (-2.70) & -0.82 (-5.35) \\
& \multicolumn{3}{c}{ \textit{Panel B.2: Stop after 1 minute}} && \multicolumn{3}{c}{ \textit{Panel B.7: Stop after 100 ticks}} \\
Trade & 1.05 (9.36) & 1.26 (8.53) & 0.76 (4.42) && 1.20 (9.43) & 1.44 (8.34) & 0.88 (4.68) \\
Midquote & 0.85 (7.97) & 1.04 (7.09) & 0.59 (3.86) && 0.98 (7.61) & 1.23 (7.06) & 0.63 (3.35) \\
BBO & -0.17 (-1.43) & 0.07 (0.46) & -0.50 (-2.96) && 0.08 (0.61) & 0.35 (1.99) & -0.29 (-1.51) \\
BBO+5s & -0.48 (-4.86) & -0.23 (-1.81) & -0.83 (-5.43) && -0.24 (-2.00) & 0.04 (0.28) & -0.63 (-3.43) \\
BBO+10s & -0.61 (-6.29) & -0.36 (-2.86) & -0.95 (-6.48) && -0.36 (-3.18) & -0.09 (-0.59) & -0.74 (-4.29) \\
& \multicolumn{3}{c}{ \textit{Panel B.3: Stop after 2 minutes}} && \multicolumn{3}{c}{ \textit{Panel B.8: Stop after 250 ticks}} \\
Trade & 1.25 (9.97) & 1.55 (9.73) & 0.82 (4.14) && 1.30 (9.07) & 1.53 (8.14) & 0.98 (4.44) \\
Midquote & 0.98 (8.08) & 1.25 (7.83) & 0.59 (3.23) && 1.07 (7.63) & 1.33 (7.34) & 0.69 (3.19) \\
BBO & 0.02 (0.13) & 0.32 (1.87) & -0.40 (-2.08) && 0.19 (1.33) & 0.47 (2.59) & -0.21 (-0.92) \\
BBO+5s & -0.30 (-2.64) & 0.01 (0.09) & -0.74 (-4.18) && -0.13 (-0.99) & 0.17 (1.03) & -0.54 (-2.58) \\
BBO+10s & -0.42 (-3.84) & -0.12 (-0.82) & -0.85 (-5.06) && -0.25 (-2.02) & 0.03 (0.22) & -0.65 (-3.27) \\
& \multicolumn{3}{c}{ \textit{Panel B.4: Stop after 3 minutes}} && \multicolumn{3}{c}{ \textit{Panel B.9: Stop after 500 ticks}} \\
Trade & 1.34 (9.78) & 1.64 (9.57) & 0.92 (4.12) && 1.42 (9.33) & 1.74 (8.56) & 0.96 (4.28) \\
Midquote & 1.08 (8.13) & 1.39 (8.04) & 0.66 (3.16) && 1.12 (7.56) & 1.40 (7.29) & 0.72 (3.14) \\
BBO & 0.17 (1.24) & 0.50 (2.82) & -0.29 (-1.31) && 0.29 (1.94) & 0.59 (3.03) & -0.12 (-0.53) \\
BBO+5s & -0.14 (-1.16) & 0.19 (1.27) & -0.62 (-3.03) && -0.02 (-0.18) & 0.28 (1.62) & -0.46 (-2.06) \\
BBO+10s & -0.27 (-2.20) & 0.06 (0.42) & -0.74 (-3.73) && -0.15 (-1.11) & 0.15 (0.87) & -0.57 (-2.69) \\
& \multicolumn{3}{c}{ \textit{Panel B.5: Stop after 5 minutes}} && \multicolumn{3}{c}{ \textit{Panel B.10: Stop after 1,000 ticks}}\\
Trade & 1.58 (10.82) & 1.89 (10.39) & 1.15 (4.80) && 1.46 (8.90) & 1.93 (9.03) & 0.79 (3.17) \\
Midquote & 1.27 (8.95) & 1.60 (8.70) & 0.81 (3.68) && 1.24 (7.88) & 1.56 (7.62) & 0.78 (3.22) \\
BBO & 0.41 (2.81) & 0.78 (4.18) & -0.12 (-0.52) && 0.44 (2.76) & 0.78 (3.78) & -0.04 (-0.16) \\
BBO+5s & 0.09 (0.69) & 0.48 (2.84) & -0.45 (-2.11) && 0.12 (0.84) & 0.48 (2.55) & -0.37 (-1.57) \\
BBO+10s & -0.03 (-0.25) & 0.35 (2.05) & -0.57 (-2.76) && -0.00 (-0.01) & 0.35 (1.85) & -0.49 (-2.11) \\
\hline
\hline
\end{tabular}
\end{footnotesize}
\smallskip
\begin{scriptsize}
\parbox{\textwidth}{\emph{Note.} The table shows the average return in percent from a trading strategy that employs the standardized earnings surprise to predict the one-minute post-announcement return: $r_{1m,it}^{ \text{EA}} = a + b_{1} z_{ \text{EPS},it}^{+} + b_{2} z_{ \text{EPS},it}^{-} + \epsilon_{it}$. A long (short) position in the stock is opened if the predicted excess return is greater (smaller) than 0.5\% (-0.5\%). In the baseline termination rule in Panel A, the position is held until 6:30pm (EOD). In Panel B, alternative termination rules are inspected. In Panels B.1--B.5, the position is closed after a fixed number of minutes, whereas in Panels B.6--B.10 the position is closed after a fixed number of tick updates. ``Trade'' employs the transaction price, ``Midquote'' the midquote, and ``BBO'' the best bid and offer. ``+Xs'' enforces a latency delay of X seconds before entrance. The test statistic for testing that the average return is zero, based on robust standard errors, is reported in parenthesis.}
\end{scriptsize}
\end{center}
\end{table}

\clearpage


\appendix

\section*{Web appendices}

\setcounter{table}{0}
\setcounter{figure}{0}
\renewcommand{\thetable}{\thesection.\arabic{table}}
\renewcommand{\thefigure}{\thesection.\arabic{figure}}

\section{Proof of theoretical results} \label{appendix:proof}

\setcounter{table}{0}
\setcounter{figure}{0}

In this appendix, we start with a mathematical description of the dynamic of the log-price process under investigation, $p = (p_{t})_{t \geq 0}$, which we suppose evolves over the unit time interval, $t \in [0,1]$. As usual, randomness is described by a filtered probability space $( \Omega, \mathcal{F}, ( \mathcal{F}_{t})_{t \geq 0}, \mathbb{P})$. The construction of this space in a noisy high-frequency setting is outlined in \citet*{jacod-protter:12a}. We decompose the cumulative intraday log-return at time $t$, or $r_{t}$, into a continuous part and a jump, or discontinuous, part: $r_{t} = r_{t}^{c} + r_{t}^{d}$.

Consistent with no-arbitrage, $r^{c} = (r_{t}^{c})_{t \geq 0}$ comprises a drift and a diffusion term:
\begin{equation} \label{equation:bsm}
r_{t}^{c} = \underbrace{ \int_{0}^{t} a_{s} \mathrm{d}s}_{ \text{drift}} + \underbrace{ \int_{0}^{t} \sigma_{s} \mathrm{d}W_{s}}_{ \text{volatility}},
\end{equation}
where the drift $a = (a_{t})_{t \geq 0}$ is a predictable and locally bounded process, the volatility $\sigma = ( \sigma_{t})_{t \geq 0}$ is an adapted and c\`{a}dl\`{a}g process, and $W = (W_{t})_{t \geq 0}$ is a standard Brownian motion.

We decompose $r^{d} = (r_{t}^{d})_{t \geq 0}$ into small and big jumps:\footnote{The separation into ``small'' and ``large'' jumps is not crucial in the analysis.}
\begin{equation} \label{equation:jp}
r_{t}^{d} = \underbrace{ \int_{0}^{t} \int_{ \mathbb{R}} \delta(s,x) 1_{ \{| \delta(s,x)| \leq 1 \}} (\mu - \nu) ( \text{d}s, \text{d}x)}_{ \text{``small'' jumps}} + \underbrace{ \int_{0}^{t} \int_{ \mathbb{R}} \delta(s,x) 1_{ \{| \delta(s,x)| > 1 \}} \mu( \text{d}s, \text{d}x)}_{ \text{``big'' jumps}},
\end{equation}
where $\mu$ is a Poisson random measure on $\mathbb{R}_{+} \times \mathbb{R}$ with compensator $\nu( \text{d}s, \text{d}x) = \text{d}s F( \text{d}x)$. Here  $F$ is a $\sigma$-finite measure, $\delta$ is a predictable function, and $( \tau_{k})_{k \geq 1}$ is a sequence of $\mathcal{F}_{t}$-stopping times increasing to $\infty$ such that $| \delta( \omega,s,x)| \wedge 1 \leq \psi_{k}(x)$ for all $( \omega, s, x)$ with $s \leq \tau_{k}( \omega)$ and  $\int_{ \mathbb{R}} \psi^{ \beta}_{k}(x) F( \text{d}x) < \infty$ for all $k \geq 1$ and $\beta \in [0,2]$. $\beta$ relates to the activity of the price jumps and can be interpreted as a generalized version of the \citet*{blumenthal-getoor:61a} index for a  L\'{e}vy process, see \citet*{ait-sahalia-jacod:09a}.\footnote{A higher value of $\beta$ increases the frequency of the small jumps. Figure \ref{figure:simulation-illustration} in Appendix \ref{appendix:simulation} provides an illustration of this characteristic.}

Moreover, we need a condition on the volatility process and microstructure noise.\\[-0.25cm]

\noindent \textbf{Assumption (V)}: $\quad \sigma$ is of the form:
\begin{equation}
\label{Eqn:sigma}
\sigma_{t} = \sigma_{0} + \int_{0}^{t} \tilde{a}_{s}\text{d}s + \int_{0}^{t} \tilde{ \sigma}_{s} \text{d}W_{s} + \int_{0}^{t} \tilde{v}_{s}\text{d}B_{s},
\end{equation}
where $\sigma_{0}$ is an $\mathcal{F}_{0}$-measurable random variable, while $\tilde{a} = (\tilde{a}_{t})_{t \geq 0}$, $\tilde{ \sigma} = ( \tilde{ \sigma}_{t})_{t \geq 0}$ and $\tilde{v} = ( \tilde{v}_{t})_{t \geq 0}$ are adapted, c\`{a}dl\`{a}g processes. \\[-0.25cm]

\noindent \textbf{Assumption (N)}: The microstructure noise is of the form
\begin{equation}
\epsilon_{i \Delta_{n}} = \omega_{i \Delta_{n}} \pi_{i},
\end{equation}
where \\[0.25cm]
(i) $\pi = ( \pi_{i})_{i \geq 0}$ follows an MA($\ell$) process with $\ell \in \mathbb{Z}^{+}$, i.e. $\pi_{i} = u_{i} + \sum_{j=1}^{ \ell} \theta_{j} u_{i-j}$ for some i.i.d. mean zero sequence $(u_{i})_{i \in \mathbb{Z}}$ and constants $\theta_{j}$, $j = 1, \dots, \ell$ such that $1 + \sum_{j=1}^{ \ell} \theta_{j} \neq 0$. Moreover, the distribution of $\pi_{0}$ has unit variance, is symmetric, and has moments of arbitrary order. \\[0.25cm]
(ii) $\omega = (\omega_t)_{t \geq 0}$ is of the form:
\begin{equation}
\omega_{t} = \omega_{0} + \int_{0}^{t} \bar{a}_{s}\text{d}s + \int_{0}^{t}
\bar{ \sigma}_{s}\text{d}W_{s} + \int_{0}^{t} \bar{v}_{s}\text{d}B_{s},
\end{equation}
where $\omega_{0}$ is an $\mathcal{F}_{0}$-measurable random variable, $\bar{a} = (\bar{a}_{t})_{t \geq 0}$, $\bar{ \sigma} = ( \bar{ \sigma}_{t})_{t \geq 0}$, $\bar{v} = ( \bar{v}_{t})_{t \geq 0}$ are adapted, c\`{a}dl\`{a}g processes, while $B = (B_{t})_{t \geq 0}$ is a standard Brownian motion and $B \Perp W$.\footnote{The symbol $A \Perp B$ means $A$ and $B$ are stochastically independent.} \\[0.25cm]
(iii) $\pi \Perp (p, \omega)$. \\[0.25cm]
(iv) $\limsup_{|t| \to \infty} |\chi(t)| <1$, where  $\chi$ is the characteristic function of $u_{0}$. \\[-0.25cm]

Assumption (V) and Assumption (N,ii) allow us to establish the order of various approximating terms appearing in the proofs, once we freeze the volatility and noise processes locally on small blocks of high-frequency data.\footnote{Assumption (V) can be extended to discontinuous volatility processes following \citet*{barndorff-nielsen-graversen-jacod-podolskij-shephard:06a}.}

We assume without loss of generality that $a$, $\sigma$, $\tilde{a}$, $\tilde{ \sigma}$, $\tilde{v},$ $\omega$, $\bar{a}$, $\bar{ \sigma}$, $\bar{v}$ and the jump components of $p_{t}$ are bounded. That this is innocuous follows from the localization procedure described in \citet*[][Section 4.4.1]{jacod-protter:12a}.

The conditions in Assumption (N) are weak and allow for complicated dynamics, such as heteroscedasticity, serial dependence, and endogeneity. This generalizes previous work in the pre-averaging space based on i.i.d. and exogenous noise \citep*{christensen-oomen-podolskij:14a, jacod-li-mykland-podolskij-vetter:09a, podolskij-vetter:09a, podolskij-vetter:09b}. The assumption follows \citet*{jacod-li-zheng:19a}, who assume $\pi$ is a stationary MA($\infty$) process with suitable mixing condition. Here, for technical reasons, we restrict the noise to be at most $\ell$-dependent, for a finite (but arbitrarily large) $\ell$.

The theory also imposes weak regularity conditions on the weight function $g$, namely $g: [0,1] \rightarrow \mathbb{R}$ is continuous and piecewise continuously differentiable with piecewise Lipshitz derivative $g'$, $g(0) = g(1) = 0$, and $\int_{0}^{1} g(s)^{2} \mathrm{d}s > 0$.

Next, we prove our mathematical results. We employ a generic constant $C$ to bound various terms. Its value may change from line to line without notice.

\subsection{An extended central limit theorem}

We start by presenting an extended version of Theorem \ref{theorem:clt} for a generalized pre-averaged bipower variation estimator. With a slight abuse of the notation from the main text, we define here the estimator with an alternative normalization. We further omit the bias-correction part. This delivers an estimator that is more convenient to work with in the derivations:
\begin{equation} \label{equation:pbv}
BV_{n}(q,r) = \frac{1}{n} \sum_{i=0}^{n-2 k_{n}+1} |n^{1/4} \bar{r}_{i}^{*}|^{q} | n^{1/4} \bar{r}_{i+k_{n}}^{*}|^{r},
\end{equation}
where $q,r \in \mathcal{S} \equiv \{0 \} \cup [1, \infty)$.

We introduce the following notation:
\begin{equation*}
\psi_{1}^{n} = k_{n} \sum_{j = 0}^{k_{n} - 1} (g_{j+1}^{n} - g_{j}^{n})^{2} \quad \text{and}  \quad \psi_{2}^{n} = \frac{1}{k_{n}} \sum_{j = 1}^{k_{n}} (g_{j}^{n})^{2}.
\end{equation*}
Furthermore,
\begin{equation*}
c_{1}^{n} = \frac{1}{k_{n} \psi_{2}^{n}} \quad \text{and} \quad c_{2}^{n} = \frac{ \psi_{1}^{n}}{ \psi_{2}^{n} \theta^{2}}.
\end{equation*}
After \citet*{podolskij-vetter:09a} and \citet{hautsch-podolskij:13a}, we expect that under $\mathcal{H}_{0}$
\begin{equation*}
BV_{n}(q,r) \overset{p}{ \longrightarrow} V(q,r) =
\mu_{q} \mu_{r} \int_{0}^{1} \Big( \theta \psi_{2} \sigma_{s}^{2} + \frac{1}{ \theta} \psi_{1} \rho^{2} \omega^{2}_{s} \Big)^{ \frac{q + r}{2}} \text{d}s,
\end{equation*}
where $\mu_{s} = \mathbb{E} \big[|N(0,1)|^{s} \big]$ and $\rho^{2}$ is the long-run variance of the noise process $\pi$, which is given by $\rho^{2} = \rho(0) + 2 \sum_{m = 1}^{ \ell} \rho(m)$ , where $\rho(m) = \mathbb{E}[ \pi_{0} \pi_{m}]$ is the $m$th autocovariance. Moreover,
\begin{equation*}
\psi_{1}^{n} \rightarrow \psi_{1} = \int_{0}^{1} [g'(s)]^{2} \text{d}s \quad \text{and} \quad \psi_{2}^{n} \rightarrow \psi_{2} = \int_{0}^{1} g(s)^{2} \text{d}s.
\end{equation*}

\begin{theorem} \label{appendix:clt}
Assume that $r$ follows the process in \eqref{equation:X} with $r_{t}^{d} \equiv 0$ (for all $t$), and that Assumptions (V) and (N) hold. As $n \rightarrow 0$, for $q_{1}, q_{2}, r_{1}, r_{2} \in \mathcal{S}$,
\begin{equation*}
n^{1/4} \begin{pmatrix}
BV_{n}(q_{1},r_{1}) - V(q_{1},r_{1}) \\[0.10cm]
BV_{n}(q_{2},r_{2}) - V(q_{2},r_{2})
\end{pmatrix} \overset{ \mathcal{D}_{s}}{ \longrightarrow} \text{MN}(0, \tilde{ \Sigma}),
\end{equation*}
where $\tilde{ \Sigma}$ is defined in \eqref{equation:asymptotic-covariance-tilde-Sigma} and $\overset{ \mathcal{D}_{s}}{ \longrightarrow}$ denotes stable convergence in law.
\end{theorem}

To present an expression for the asymptotic covariance matrix in the above central limit theorem, we need to introduce some additional notation. In particular, for $i, j \in \{1, 2 \}$, we set
\begin{equation} \label{equation:h-function-covariance}
h_{ij}(x,y,z) = \mathrm{cov}(|H_{1}|^{q_{i}} |H_{2}|^{r_{i}}, |H_{3}|^{q_{j}} |H_{4}|^{r_{j}}),
\end{equation}
where $x,y$ are two-dimensional vectors, whereas $z$ is a four-dimensional vector, and $(H_{1}, H_{2}, H_{3}, H_{4})$ are centered multivariate normal distributed random variables with covariance structure:
\begin{itemize}
\item[(i)] $\mathbb{E} \big[|H_{i}|^{2}] = y_{1} x_{1}^{2} + y_{2} x_{2}^{2}, \quad i \in \{1, 2, 3, 4 \}$,
\item[(ii)] $H_{1} \perp H_{2}, \quad H_{1} \perp H_{4}, \quad H_{3} \perp H_{4}$,
\item[(iii)] $\mathrm{cov}(H_{1},H_{3}) = \mathrm{cov}(H_{2},H_{4}) = z_{1}x_{1}^{2} + z_{2}x_{2}^{2}$,
\item[(iv)] $\mathrm{cov}(H_{2},H_{3}) = z_{3}x_{1}^{2} + z_{4}x_{2}^{2}$.
\end{itemize}
We also introduce the following functions:
\begin{equation*}
f_{1}(s) = \frac{1}{ \theta} \phi_{1}(s), \quad f_{2}(s) = \theta \phi_{2}(s), \quad f_{3}(s) = \frac{1}{ \theta} \phi_{3}(s), \quad f_{4}(s) = \theta \phi_{4}(s),
\end{equation*}
for $s \in [0, 2]$, where
\begin{align*}
\phi_{1}(s) &= \int_{0}^{1-s} g'(u)g'(u+s) \mathrm{d}s, \quad \phi_{2}(s) = \int_{0}^{1-s} g(u)g(u+s) \mathrm{d}s, \\[0.10cm]
\phi_{3}(2-s) &= \int_{0}^{1-s} g'(u)g'(u+s-1) \mathrm{d}s, \quad \phi_{4}(s) = \int_{0}^{2-s} g(u)g(u+s-1) \mathrm{d}s.
\end{align*}
We compactly write these in vectorized form as $ f(s) = (f_{1}(s), \dots, f_{4}(s))$. The $(i,j)$-th component of $\tilde{ \Sigma}$ is then given by
\begin{equation} \label{equation:asymptotic-covariance-tilde-Sigma}
\tilde{ \Sigma}_{ij} = 2 \theta \int_{0}^{1} \int_{0}^{2} h_{ij} \big(( \rho \omega_{u}, \sigma_{u}), (\psi_{1}/ \theta, \theta \psi_{2}), f(s) \big) \mathrm{d}s \mathrm{d}u.
\end{equation}

Compared to the estimators introduced in this appendix, the pre-averaged realized variance and bipower variation in \eqref{equation:rv} are merely rescaled and bias-corrected versions of \eqref{equation:pbv} with $(q,r) = (2,0)$ and $(q,r) = (1,1)$, since (see also Lemma \ref{lemma:noise})
\begin{equation*}
RV_{n}^{*} = c_{1} BV_{n}(2,0) + o_{p} \big(n^{-1/4} \big) \quad \mathrm{and} \quad BV_{n}^{*} = c_{1} \frac{ \pi}{2} BV_{n}(1,1) + o_{p} \big(n^{-1/4} \big),
\end{equation*}
where
\begin{equation*}
c_{1} = \frac{1}{ \theta \psi_{2}}.
\end{equation*}
Hence, the relation between $\tilde{ \Sigma}$ and $\Sigma$ is as follows:
\begin{equation} \label{asymptotic-covariance-Sigma}
\Sigma_{ij} = c_{i}^{2} \left( \frac{ \pi}{2} \right)^{i+j-2} \tilde{ \Sigma}_{ij}.
\end{equation}

\subsection*{Proof of Theorem \ref{appendix:clt}}

For any $m \leq i$, we define
\begin{equation}
\bar{r}_{i,m}^{*} = \sum_{j=1}^{k_{n}} g(j/k_{n}) \big( \sigma_{m \Delta_{n}} (W_{(i+j) \Delta_{n}}-W_{(i+j-1) \Delta_{n}})+ \omega_{m \Delta_{n}} ( \pi_{i+j}- \pi_{i+j-1}) \big).
\end{equation}
We note that $\bar{r}_{i,m}^{*}$ serves to approximate $\bar{r}_{i}^{*}$ by freezing the processes $\sigma$ and $\omega$ locally at the time point $m \Delta_{n}$.

We also denote
\begin{align*}
\eta(q,r)_{i}^{n} &= |n^{1/4} \bar{r}_{i}^{*}|^{q}   |n^{1/4} \bar{r}_{i+k_{n}}^{*}|^{r}, \\
\eta(q,r)_{i,m}^{n} &= |n^{1/4} \bar{r}_{i,m}^{*}|^{q}   |n^{1/4}\bar{r}_{i+k_{n},m}^{*}|^{r}.
\end{align*}
Set $\mathcal{G}_{i} = \sigma(u_{j} : j \leq i)$ to be the $\sigma$-field generated by the i.i.d. sequence $(u_{i})$ used to build the noise process. We then define the product $\sigma$-field $\mathcal{H}_{i}^{n} = \mathcal{F}_{i \Delta_{n}} \otimes \mathcal{G}_{i-2 k_{n}}$, where $n$ is chosen large enough such that $\ell \leq 2 k_{n}$, so that the lagged form $\mathcal{G}_{i-2 k_{n}}$ deals with the $\ell$-dependence in $\pi$.

\begin{lemma} \label{lemma:moment}
Under the maintained assumptions of Theorem \ref{appendix:clt}, for any $q, r \in \mathcal{S}$, and real number $v \geq 1$, it holds that
\begin{equation*}
\max \Big( \mathbb{E} \big[|n^{1/4} \bar{r}_{i}^{*}|^{v} \mid \mathcal{H}_{m}^{n} \big], \mathbb{E} \big[|n^{1/4} \bar{r}_{i,m}^{*}|^{v} \mid \mathcal{H}_{m}^{n} \big] \Big) \leq C,
\end{equation*}
and
\begin{equation*}
\max \Big( \mathbb{E} \big[| \eta(q,r)_{i}^{n}|^{v} \mid \mathcal{H}_{m}^{n} \big], \mathbb{E} \big[| \eta(q,r)_{i,m}^{n}|^{v} \mid \mathcal{H}_{m}^{n} \big] \Big)  \leq C,
\end{equation*}
uniformly in $i$ and $m$.
\end{lemma}
\begin{proof}
The first bound is a standard result for the pre-averaged returns, which follows from the boundedness of $g$, $\sigma$ and $\omega$. Then, the second bound is due to the Cauchy-Schwarz inequality. \qed
\end{proof}

\begin{lemma} \label{lemma:edgeworth}
Under the maintained assumptions of Theorem \ref{appendix:clt}, for any $q, r \in \mathcal{S},$ it holds that
\begin{equation*}
\mathbb{E} \big[ \eta(q,r)_{i,m}^{n} \mid \mathcal{H}_{m}^{n} \big] = \mu_{q} \mu_{r} \Big( \theta \psi_{2} \sigma_{m \Delta_{n}}^{2} + \frac{1}{ \theta} \psi_{1} \rho^{2} \omega^{2}_{m \Delta_{n}}\Big)^{ \frac{q + r}{2}} + o_{p}(n^{-1/4}),
\end{equation*}
uniformly in $i$ and $m$.
\end{lemma}
\begin{proof}

We adapt the proof of Lemma 4 in \citet*{podolskij-vetter:09a} to $\ell$-dependent data by applying Theorem 6.3 from \citet{lahiri:03a}. Hence, the job is to verify conditions (C1)--(C6) therein, which are essentially those introduced adopted in \citet{gotze-hipp:83a}. We first prove the result with $r=0$, i.e the pre-averaged power variation estimator, which is then extended to the bipower case with $r>0$.

We recall that
\begin{align*}
n^{1/4} \bar{r}_{i,m}^{*} &= n^{1/4} \sum_{j=0}^{k_{n}-1} \left( \sigma_{m \Delta_{n}} g \left( \frac{j}{k_{n}} \right)(W_{(i+j) \Delta_{n}} - W_{(i+j-1) \Delta_{n}}) + \omega_{m \Delta_{n}} \left(g \left( \frac{j}{k_{n}} \right) - g \left( \frac{j+1}{k_{n}} \right) \right) \pi_{i+j} \right) \\
&= k_{n}^{-1/2} \sum_{j=0}^{k_{n}-1} \left( \sigma_{m \Delta_{n}} g \left( \frac{j}{k_{n}} \right) \left( \frac{k_{n}}{n^{1/2}} \right)^{1/2} N_{i+j} + \omega_{m \Delta_{n}} \left(g \left( \frac{j}{k_{n}} \right) - g \left( \frac{j+1}{k_{n}} \right) \right) n^{1/4} k_{n}^{1/2} \pi_{i+j} \right).
\end{align*}
Note that the subsequent analysis is conditional on $\mathcal{H}_{m}^{n}$, so that $\sigma_{m \Delta}$, $\omega_{m \Delta}$ can be treated as non-random constants. As a result, the above is a scaled sum of $k_{n}$ terms that are $\ell$-dependent, where $(N_{l})_{ l \geq 1}$ is a i.i.d. sequence of standard normal variables. We set $\mathcal{D}_{i} = \sigma(N_{i}) \otimes \sigma(u_{i})$.

Since $n^{1/4} \bar{r}_{i,m}^{*}$ is mean zero, we denote
\begin{equation*}
\tau_{m,n}^{2} = \mathbb{E} \left[ \big(n^{1/4} \bar{r}_{i,m}^{*} \big)^{2} \mid \mathcal{H}_{m}^{n} \right]
\end{equation*}
as its conditional variance and consider the standardized version
\begin{equation*}
U_{m,n} = \tau_{m,n}^{-1} n^{1/4} \bar{r}_{i,m}^{*} \equiv k_{n}^{-1/2} \sum_{j=0}^{k_{n}-1} \Gamma_{i+j}^{m,n}.
\end{equation*}
The first condition (C1) is about the existence of moments. Due to Assumption (N,i), the terms $\Gamma_{i+j}^{m,n}$ possess $v$ moments for any finite $v>0$. Condition (C2) is verified, since $\Gamma_{i+j}^{m,n}$ has mean zero and unit variance. Condition (C3) is about the strong mixing condition and weak dependence structure of the summands, which we show by choosing the approximate variable as $\Gamma_{i+j}^{m,n}$ itself. Condition (C4) involves the definition of the strong mixing condition. This is obviously fulfilled in our setting, since $( \Gamma_{i+j}^{m,n})_{j=0}^{k_{n}-1}$ are $\ell$-dependent. Moreover, condition (C5) is an approximate Markov-type property and follows immediately due to the finiteness of $\ell$. This leaves condition (C6), which is a Cramer-type condition for weakly dependent data, which requires a bit more work. To this end, we note that (conditional further on $\mathcal{H}_{m}^{n}$ in each expectation):
\begin{equation*}
\mathbb{E} \left[ \left| \mathbb{E} \left[ \exp \big( \mathrm{i} t \big( \Gamma_{i+j - k}^{m,n} + \ldots + \Gamma_{i+j}^{m,n} + \ldots + \Gamma_{i+j+k}^{m,n} \big) \big) \mid \mathcal{D}_{i+s}, s \neq j \right] \right| \right] = \left| \mathbb{E} \left[ \exp \big( \mathrm{i} t \Gamma_{i+j}^{m,n} \big) \right] \right|,
\end{equation*}
for each $r> \ell$.

With the independence between $N_{i}$ and $\pi_{i}$, this can be written as
\begin{equation*}
\exp \left(- \sigma_{m \Delta_{n}}^{2} g^{2} \left( \frac{j}{k_{n}} \right) \frac{k_{n}}{n^{1/2}} \frac{t^{2}}{2 \tau_{n,m}^{2}} \right) \prod_{k=0}^{ \ell} \left| \mathbb{E} \left[ \exp \Big( \mathrm{i} t \omega_{m \Delta_{n}} \left(g \left( \frac{j}{k_{n}} \right) - \left( \frac{j+1}{k_{n}} \right) \right) n^{1/4} k_{n}^{1/2} \theta_{k} u_{i+j-k} \Big) \right] \right|
\end{equation*}
with $\theta_{0} = 1$. Then, this condition holds due to Assumption (N,iv).

Setting $f(x) = |x|^{q}$, the following Edgeworth expansion:
\begin{align*}
\left| \frac{1}{f( \tau_{n,m})} \mathbb{E} \left[ f \big(n^{1/4} \bar{r}_{i,m}^{*} \big) \mid \mathcal{H}_{m}^{n} \right] - \int f(x) \mathrm{d} \Psi_{n}(x) \right| = o \big(k_{n}^{-1/2} \big),
\end{align*}
holds uniformly in $i$ and $m$, where $\Psi_{n}$ is the signed measure with the density
\begin{equation}
\psi_{n}(x) = \phi(x) \left(1 + \frac{1}{ \sqrt{k_{n}}} \frac{x^{3} - 3x}{6} \bar{ \kappa}_{3}^{n} \right),
\end{equation}
where $\bar{ \kappa}_{3}^{n}$ is the third cumulant of $U_{m,n}$.

As $f(x)=|x|^{q}$ is an even function and $\phi(x)(x^{3} - 3x)$ is an odd function,
\begin{equation*}
\int f(x) \phi(x) (x^{3} - 3x) \mathrm{d}x = 0.
\end{equation*}
In view of this result and the bound
\begin{equation*}
\tau_{n,m}^{2} - \left( \theta \psi_{2} \sigma_{m \Delta_{n}}^{2} + \frac{1}{ \theta} \psi_{1} \rho^{2} \omega^{2}_{m \Delta_{n}} \right) = O \big(k_{n}^{-1} \big),
\end{equation*}
we obtain the following version of the expansion:
\begin{equation} \label{equation:case-rzero}
\mathbb{E} \big[ \eta(q,0)_{i,m}^{n} \mid \mathcal{H}_{m}^{n} \big] = \mu_{q} \Big( \theta \psi_{2} \sigma_{m \Delta_{n}}^{2} + \frac{1}{ \theta} \psi_{1} \rho^{2} \omega^{2}_{m \Delta_{n}} \Big)^{ \frac{q}{2}} + o_{p}(n^{-1/4}).
\end{equation}

Next, we look at the case $r \neq 0$. We introduce the notation:
\begin{align*}
\hat{r}_{i,m}^{*} &= \sum_{j= \ell}^{k_{n}-1} \left( \sigma_{m \Delta_{n}} g \left( \frac{j}{k_{n}} \right)(W_{(i+j) \Delta_{n}} - W_{(i+j-1) \Delta_{n}}) + \omega_{m \Delta_{n}} \left(g \left( \frac{j}{k_{n}} \right) - g \left( \frac{j+1}{k_{n}} \right) \right) \pi_{i+j} \right), \\
\check{r}_{i,m}^{*} &= \sum_{j=0}^{k_{n}-1- \ell} \left( \sigma_{m \Delta_{n}} g \left( \frac{j}{k_{n}} \right)(W_{(i+j) \Delta_{n}} - W_{(i+j-1) \Delta_{n}}) + \omega_{m \Delta_{n}} \left(g \left( \frac{j}{k_{n}} \right) - g \left( \frac{j+1}{k_{n}} \right) \right) \pi_{i+j} \right).
\end{align*}
To derive the statement in the lemma on $\mathbb{E} \big[ \eta(q,r)_{i,m}^{n} \mid \mathcal{H}_{m}^{n} \big]$, it suffices to show that:
\begin{equation*}
\mathbb{E} \left[|n^{1/4} \bar{r}_{i,m}^{*}|^{q}   |n^{1/4} \hat{r}_{i+k_{n},m}^{*}|^{r} \mid \mathcal{H}_{m}^{n} \right] = \mu_{q} \mu_{r} \Big( \theta \psi_{2} \sigma_{m \Delta_{n}}^{2} + \frac{1}{ \theta} \psi_{1} \rho^{2} \omega^{2}_{m \Delta_{n}} \Big)^{ \frac{q+r}{2}} + o_{p}(n^{-1/4}),
\end{equation*}
\begin{equation*}
\mathbb{E} \left[|n^{1/4} \check{r}_{i,m}^{*}|^{q} \big(  |n^{1/4} \bar{r}_{i+k_{n},m}^{*}|^{r} -|n^{1/4} \hat{r}_{i+k_{n},m}^{*}|^{r} \big) \mid \mathcal{H}_{m}^{n} \right] = o_{p}(n^{-1/4}),
\end{equation*}
and
\begin{equation*}
\mathbb{E} \left[ \big||n^{1/4} \bar{r}_{i,m}^{*}|^{q} -|n^{1/4} \check{r}_{i,m}^{*}|^{q} \big| \big||n^{1/4} \bar{r}_{i+k_{n},m}^{*}|^{r} - n^{1/4} \hat{r}_{i+k_{n},m}^{*}|^{r} \big| \mid \mathcal{H}_{m}^{n} \right] = O_{p}(n^{-1/2}).
\end{equation*}
The first and second bounds exploit that $\bar{r}_{i,m}^{*}$ and $\hat{r}_{i+k_{n},m}^{*}$, and also $\check{r}_{i,m}^{*}$ and $( \bar{r}_{i+k_{n},m}^{*}, \hat{r}_{i+k_{n},m}^{*})$ are independent, conditional on $\mathcal{H}_{m}^{n}$, combined with an Edgeworth expansion as in \eqref{equation:case-rzero} for $\check{r}_{i,m}^{*}$ and $\hat{r}_{i+k_{n},m}^{*}$, and  Lemma \ref{lemma:moment}. The third bound relies on the inequality $||x|^{s} - |y|^{s}| \leq C |x-y| \max(|x|^{s-1}, |y|^{s-1})$ for $s \in \{q,r\}$, together with Lemma \ref{lemma:moment} and the Cauchy-Schwarz inequality. \qed
\end{proof}

The next preliminary result concerns the error of $\alpha_{i,m}^{n} \equiv n^{1/4} ( \bar{r}_{i}^{*} - \bar{r}_{i,m}^{*})$.
\begin{lemma} \label{lemma:beta}
Under the maintained assumptions of Theorem \ref{appendix:clt}, for any real number $v \geq 2$, it holds that
\begin{align*}
\mathbb{E} \big[| \alpha_{i,m}^{n}|^{v} \mid \mathcal{H}_{m}^{n} \big] &\leq C \Big(n^{-v/4} + \big( \mathbb{E} \big[ \Gamma( \sigma,m,i)^{n} + \Gamma( \omega, m,i)^{n} \mid \mathcal{H}_{m}^{n} \big] \big)^{v} \Big), \\[0.10cm]
\mathbb{E} \big[| \eta(q,r)_{i}^{n} - \eta(q,r)_{i,m}^{n}|^{v} \mid \mathcal{H}_{m}^{n} \big] &\leq C \Big(n^{-v/4} + \big( \mathbb{E} \big[ \Gamma( \sigma,m,i)^{n} + \Gamma( \omega, m,i)^{n} \mid \mathcal{H}_{m}^{n} \big] \big)^{v} \Big),
\end{align*}
where for any process $\gamma$ we set
\begin{equation*}
\Gamma( \gamma,m,i)^{n} = \sup_{m \Delta_{n} \leq s \leq (i+2 k_{n}) \Delta_{n}} | \gamma_{s} - \gamma_{m \Delta_{n}}|.
\end{equation*}
\end{lemma}

\begin{proof}
By construction
\begin{align*}
\alpha_{i,m}^{n} &= n^{1/4} \int_{i \Delta_{n}}^{(i+k_{n}) \Delta_{n}} g_{n}(s-i \Delta_{n}) \big[a_{s} \mathrm{d}s + ( \sigma_{s} - \sigma_{m \Delta_{n}}) \mathrm{d}W_{s} \big] \\[0.10cm]
&+ n^{1/4} \sum_{j=0}^{k_{n}-1} \big[g(j/k_{n})-g(j+1/k_{n}) \big] ( \omega_{(i+j) \Delta_{n}} - \omega_{m \Delta_{n}}) \pi_{i+j},
\end{align*}
where $g_{n}(s) = \sum_{j=1}^{k_{n}} g(j/k_{n}) 1_{\{(j-1) \Delta_{n}, j \Delta_{n} \}}(s)$. Note that $g$, $a$ and $\sigma$ are bounded and $|g(j/k_{n})-g(j+1/k_{n})| \leq C/k_{n}$. Furthermore, $\omega$ is bounded and independent of $\pi$. The first result then follows from the Burkholder-Davis-Gundy inequality and $\Gamma$ being bounded. The second applies H\"{o}lder's inequality and the inequalities $||y|^q-|x|^q| \leq C |x-y| \max(|x|^{q-1}, |y|^{q-1})$ and  $|y_{1} y_{2}-x_{1} x_{2}| \leq |y_{1}-x_{1}| |y_{2}|+|x_{1}||y_{2}-x_{2}|$. \qed
\end{proof}
We proceed by applying the ``big blocks--small blocks'' technique of \cite{jacod-li-mykland-podolskij-vetter:09a}. To this end, we fix an integer $p \geq 2$, which in the later stages of the proof tends to infinity, and introduce the notation
\begin{align*}
a_{j}(p) &= 2j(p+2) k_{n}, \quad b_{j}(p)= 2j(p+2)k_{n} + 2pk_{n}, \\[0.10cm]
A_{j}(p) &= \mathbb{Z} \cap [a_{j}(p), b_{j}(p)), \quad B_{j}(p) = \mathbb{Z} \cap [b_{j}(p), a_{j+1}(p)).
\end{align*}
$A_{j}(p)$ is the big block, which has size $2 pk_{n}$, while the small block $B_{j}(p)$ only has a size of $4 k_{n}$. The small blocks are going to separate the big blocks to account for the dependence introduced by pre-averaging and the MA($\ell$) structure in the noise process $\pi$. This means that some important terms (defined later) are conditionally independent, since the summands in the pre-averaged bipower estimator employ $2 k_{n}$ underlying high-frequency returns. We let $j_{n}(p) \equiv \lfloor n/2 (p+2) k_{n} \rfloor$ be the number of such big block--small block pairs. We utilize these pairwise blocks until observation $i_{n}(p) \equiv j_{n}(p) 2(p+2)k_{n}$, while leaving some residual unused data at the end.

Now, we define
\begin{equation*}
\zeta(p,q,r,1)_{j}^{n} = \sum_{u = a_{j}(p)}^{b_{j}(p)-1} \tilde{Y}_{u}^{n}(q,r) \qquad \text{and} \qquad \zeta(p,q,r,2)_{j}^{n} = \sum_{u = b_{j}(p)}^{a_{j+1}(p)-1} \tilde{Y}_{u}^{n}(q,r).
\end{equation*}
with
\begin{equation*}
\tilde{Y}_{u}^{n}(q,r) = \begin{cases}
n^{-1/2} \left( \eta(q,r)_{u,a_{j}(p)}^{n} - \mathbb{E} \big[ \eta(q,r)_{u,a_{j}(p)}^{n} \mid \mathcal{H}_{a_{j}(p)}^{n} \big] \right), \quad u \in A_{j}(p), \\[0.25cm]
n^{-1/2} \left( \eta(q,r)_{u,b_{j}(p)}^{n} - \mathbb{E} \big[ \eta(q,r)_{u,b_{j}(p)}^{n} \mid \mathcal{H}_{b_{j}(p)}^{n} \big] \right), \quad u \in B_{j}(p), \\[0.25cm]
n^{-1/2} \left( \eta(q,r)_{u,i_{n}(p)}^{n} - \mathbb{E} \big[ \eta(q,r)_{u,i_{n}(p)}^{n} \mid \mathcal{H}_{i_{n}(p)}^{n} \big] \right), \quad u \in i_{n}(p),
\end{cases}
\end{equation*}
and denote
\begin{align*}
M(p,q,r)^{n} &= n^{-1/2} \sum_{j=0}^{j_{n}(p)-1} \zeta(p,q,r,1)_{j}^{n}, \\
N(p,q,r)^{n} &= n^{-1/2} \sum_{j=0}^{j_{n}(p)-1} \zeta(p,q,r,2)_{j}^{n}, \\
C(p,q,r)^{n} &= n^{-1/2} \sum_{u=i_{n}(p)}^{n} \tilde{Y}_{u}^{n} (q,r).
\end{align*}
Next, we introduce the decomposition
\begin{equation} \label{equation:decomposition}
n^{1/4}(BV_{n}(q,r)-V(q,r)) = n^{1/4} \left[M(p,q,r)^{n} + N(p,q,r)^{n} + C(p,q,r)^{n} \right] + F(p,q,r)^{n},
\end{equation}
which implicitly defines $F(p,q,r)^{n}$ as a remainder term. In a series of lemmas, we will show that $M(p,q,r)^{n}$ is the leading term, while the remaining parts are asymptotically negligible.

\begin{lemma} \label{lemma:N}
Let $p \geq 2$ be fixed. Under the maintained assumptions of Theorem \ref{appendix:clt}, it holds that
\begin{equation*}
\mathbb{E} \left[ \left(n^{1/4} N(p,q,r)^{n} \right)^{2}  \right] \leq C/p.
\end{equation*}
\end{lemma}
\begin{proof}
Write
\begin{equation*}
L_{k}^{n} = n^{-1/2} \sum_{j=0}^{k} \zeta(p,q,r,2)_{j}^{n}.
\end{equation*}
We observe that the process $(L_{k}^{n})_{k=0}^{j_{n}(p)-1}$ is a discrete-time martingale under the filtration $( \mathcal{H}_{b_{j}(p) \Delta_{n}})_{j=0}^{j_{n}(p)-1}$. Then, Doob's inequality yields that
\begin{equation}
\mathbb{E} \left[ \left(n^{1/4} N(p,q,r)^{n} \right)^{2} \right] \leq C n^{-1/2} \sum_{j=0}^{j_{n}(p)-1} \mathbb{E} \left[ \left( \zeta(p,q,r,2)_{j}^{n} \right)^{2} \right].
\end{equation}
Lemma \ref{lemma:moment} implies that $\mathbb{E} \left[ \big( \tilde{Y}_{u}^{n}(q,r) \big)^{2} \right] \leq C/n$ uniformly in $j$, $p$, $q$, and $r$. Consequently, we deduce that $\mathbb{E} \left[ \left( \zeta(p,q,r,2)_{j}^{n} \right)^{2} \right] \leq C$. Then, using $j_{n}(p) \leq C n^{1/2}/p$, ends the proof. \qed
\end{proof}

\begin{lemma}\label{lemma-C}
Let $p \geq 2$ be fixed. Under the maintained assumptions of Theorem \ref{appendix:clt}, it holds that
\begin{equation*}
\mathbb{E} \left[ \left|n^{1/4} C(p,q,r)^{n} \right| \right] \leq Cpn^{-1/4}.
\end{equation*}
\end{lemma}
\begin{proof}
Note that $n-i_{n}(p) \leq Cpn^{1/2}$ and $\mathbb{E} \left[ \big| \tilde{Y}_{u}^{n}(q,r) \big| \right] \leq C n^{-1/2}$. This immediately leads to the conclusion that $\mathbb{E} \left[ \big|n^{1/4} C(p,q,r)^{n} \big| \right] \leq Cpn^{-1/4}$. \qed
\end{proof}

\begin{lemma}
Under the maintained assumptions of Theorem \ref{appendix:clt}, for any $\delta > 0$, it holds that
\begin{equation*}
\lim_{p \rightarrow \infty} \limsup_{n \rightarrow \infty} \mathbb{P} \left[|F(p,q,r)^{n}|> \delta \right] = 0.
\end{equation*}
\end{lemma}
\begin{proof}
Recalling the decomposition in \eqref{equation:decomposition}, we write that
\begin{equation}
F(p,q,r)^{n} = A(p,q,r)^{n} + B(p,q,r)^{n} + D(p,q,r)^{n} - n^{1/4} C(p,q,r)^{n},
\end{equation}
where
\begin{align*}
A(p,q,r)^{n} &= n^{1/4} BV_{n}(q,r) - n^{1/4} \frac{1}{n} \sum_{j=0}^{j_{n}(p)-1} \sum_{u=0}^{2(p+1)k_{n}-1} \eta(q,r)_{a_{j}(p)+ u, a_{j}(p)}^{n}, \\[0.10cm]
B(p,q,r)^{n} &= n^{1/4} \frac{1}{n} \sum_{j=0}^{j_{n}(p)-1} \sum_{u=0}^{2(p+1)k_{n}-1} \mathbb{E} \left[ \eta(q,r)_{a_{j}(p)+ u, a_{j}(p)}^{n} \mid \mathcal{H}_{a_{j}(p)}^{n} \right] -n^{1/4} V(q,r), \\[0.10cm]
D(p,q,r)^{n} &= n^{-3/4} \sum_{j=0}^{j_{n}(p)-1} \sum_{u=0}^{2k_{n}-1} \eta(q,r)_{b_{j}(p)+ u, a_{j}(p)}^{n}- \mathbb{E} \left[ \eta(q,r)_{b_{j}(p)+ u, a_{j}(p)}^{n} \mid \mathcal{H}_{a_{j}(p)}^{n} \right] \\[0.10cm]
&- n^{-3/4} \sum_{j=0}^{j_{n}(p)-1} \sum_{u=0}^{2k_{n}-1} \eta(q,r)_{b_{j}(p)+ u, b_{j}(p)}^{n}- \mathbb{E} \left[ \eta(q,r)_{b_{j}(p)+ u, b_{j}(p)}^{n} \mid \mathcal{H}_{b_{j}p)}^{n} \right]
\end{align*}
The $C(p,q,r)^{n}$ term was covered in Lemma \ref{lemma-C}, while the $D(p,q,r)^{n}$ term is dealt with following Lemma \ref{lemma:N}. We proceed to $A(p,q,r)^{n}$. Ignoring asymptotically negligible remainder terms, and employing the shorthand notation $\tilde{ \eta}(q,r)_{i,m}^{n} \equiv \eta(q,r)_{i}^{n} - \eta(q,r)_{i,m}^{n}$, it is enough to study
\begin{align*}
A_{1}^{n} &= n^{-3/4} \sum_{j=0}^{j_{n}(p)-1} \sum_{u=0}^{2(p+1)k_{n}-1} \mathbb{E} \left[ \tilde{ \eta}(q,r)_{a_{j}(p)+ u, a_{j}(p)}^{n} \mid \mathcal{H}_{a_{j}(p)}^{n} \right], \\[0.10cm]
A_{2}^{n} &= n^{-3/4} \sum_{j=0}^{j_{n}(p)-1} \sum_{u=0}^{2(p+1)k_{n}-1} \tilde{ \eta}(q,r)_{a_{j}(p)+ u, a_{j}(p)}^{n}- \mathbb{E} \left[ \tilde{ \eta}(q,r)_{a_{j}(p)+ u, a_{j}(p)}^{n} \mid \mathcal{H}_{a_{j}(p)}^{n} \right].
\end{align*}
Proceeding as in the proof of Lemma
\ref{lemma:N}, applying Doob's inequality and Lemma \ref{lemma:beta} yields
\begin{align*}
\mathbb{E} \left[ \big( A_{2}^{n} \big)^{2} \right] &\leq  C n^{-3/2} pk_{n} \sum_{j=0}^{j_{n}(p)-1} \sum_{u=0}^{2(p+1)k_{n}-1} \mathbb{E} \left[ \big( \tilde{ \eta}(q,r)_{a_{j}(p)+ u, a_{j}(p)}^{n} \big)^{2} \right] \\[0.10cm]
&\leq C n^{-3/2} pk_{n} \sum_{j=0}^{j_{n}(p)-1} \sum_{u=0}^{2(p+1)k_{n}-1} \left(n^{-1/2} + \mathbb{E} \left[ \big( \Gamma( \sigma,a_{j}(p),a_{j}(p)+ u)^{n} + \Gamma( \omega,a_{j}(p), a_{j}(p)+ u)^{n} \big)^{2} \right] \right) \\[0.10cm]
&\leq C pn^{-1/2} + C n^{-1/2} p^{2} \sum_{j=0}^{j_{n}(p)-1} \mathbb{E} \left[ \big( \Gamma( \sigma,a_{j}(p), a_{j+1}(p))^{n} + \Gamma( \omega,a_{j}(p), a_{j+1}(p))^{n} \big)^{2} \right],
\end{align*}
which is negligible by Lemma 5.4 in \citet*{jacod-li-mykland-podolskij-vetter:09a}.

Moving to the $A_{1}^{n}$ term we observe that, from the inequality
$|y_{1}y_{2} - x_{1}x_{2}| \leq |y_{1}-x_{1}| |y_{2}|+|x_{1}||y_{2}-x_{2}|$ and Lemmas \ref{lemma:moment} and \ref{lemma:beta}, it suffices to set $r=0$ and look at the term
\begin{equation*}
n^{-3/4} \sum_{j=0}^{j_{n}(p)-1} \sum_{u=0}^{2(p+1)k_{n}-1} \mathbb{E} \left[ \eta(q)_{a_{j}(p)+ u}^{n} - \eta(q)_{a_{j}(p)+ u, a_{j}(p)}^{n} \mid \mathcal{H}_{a_{j}(p)}^{n} \right].
\end{equation*}
Now, take $f(x)=|x|^{q}$ and recall that $\alpha_{i,m}^{n} = n^{1/4} \left( \bar{r}_{i}^{*} - \bar{r}_{i,m}^{*} \right)$. By Taylor's theorem and Lemmas \ref{lemma:moment} and \ref{lemma:beta}, we can reduce the problem to studying the term
\begin{equation} \label{equation:f-alpha}
n^{-3/4} \sum_{j=0}^{j_{n}(p)-1} \sum_{u=0}^{2(p+1) k_{n}-1} \mathbb{E} \left[ f'( \eta(q)_{a_{j}(p)+ u, a_{j}(p)}) \alpha_{a_{j}(p)+u,a_{j}(p)}^{n} \mid \mathcal{H}_{a_{j}(p)}^{n} \right].
\end{equation}
At this stage, we need an extra decomposition:
\begin{equation*}
\alpha_{i,m}^{n} = \alpha_{i,m}^{n}(1) + \alpha_{i,m}^{n}(2),
\end{equation*}
where
\begin{align*}
\alpha_{i,m}^{n}(1) &= n^{1/4} \sum_{j=1}^{k_{n}-1} g(j/k_{n}) \left[ \Delta_{n} a_{m \Delta_{n}} + \int_{(i+j-1) \Delta_{n}}^{(i+j) \Delta_{n}} \big[ \tilde{ \sigma}_{m \Delta_{n}}(W_{s}-W_{m \Delta_{n}}) + \tilde{v}_{m \Delta_{n}}(B_{s}-B_{m \Delta_{n}}) \big] \mathrm{d}W_{s} \right] \\[0.10cm]
&+ n^{1/4} \sum_{j=0}^{k_{n}-1} \big[g(j/k_{n})-g(j+1/k_{n}) \big]
\left[ \bar{ \sigma}_{m \Delta_{n}}(W_{(i+j) \Delta_{n}} - W_{m \Delta_{n}})+ \bar{v}_{m \Delta_{n}} (B_{(i+j) \Delta_{n}} - B_{m \Delta_{n}}) \right] \pi_{i+j}, \\[0.10cm]
\alpha_{i,m}^{n}(2) &= n^{1/4} \sum_{j=1}^{k_{n}-1} g(j/k_{n}) \left[ \int_{(i+j-1) \Delta_{n}}^{(i+j) \Delta_{n}} (a_{s}-a_{m \Delta_{n}}) \mathrm{d}s + \int_{(i+j-1) \Delta_{n}}^{(i+j) \Delta_{n}} \int_{m \Delta_{n}}^{s} \tilde{a}_{u} \mathrm{d}u \mathrm{d}W_{s} \right] \\[0.10cm]
&+ n^{1/4} \sum_{j=1}^{k_{n}-1} g(j/k_{n}) \left[ \int_{(i+j-1) \Delta_{n}}^{(i+j) \Delta_{n}} \left[ \int_{m \Delta_{n}}^{s} ( \tilde{ \sigma}_{u}- \tilde{ \sigma}_{m \Delta_{n}}) \mathrm{d}W_{u} + \int_{m \Delta_{n}}^{s} ( \tilde{v}_{u} - \tilde{v}_{m \Delta_{n}}) \mathrm{d}B_{u} \right] \mathrm{d}W_{s} \right] \\[0.10cm]
&+ n^{1/4} \sum_{j=0}^{k_{n}-1} \big[g(j/k_{n})-g(j+1/k_{n}) \big] \left[ \int_{m \Delta_{n}}^{(i+j) \Delta_{n}} \bar{a}_{s} \text{d}s + ( \bar{ \sigma}_{s} - \bar{ \sigma}_{m \Delta_{n}}) \mathrm{d}W_{s} + ( \bar{v}_{s} - \bar{v}_{m \Delta_{n}}) \mathrm{d}B_{s} \right] \pi_{i+j}.
\end{align*}
Returning to \eqref{equation:f-alpha}, we first note that
\begin{equation*}
\mathbb{E} \left[ f'( \eta(q)_{a_{j}(p)+ u, a_{j}(p)}) \alpha_{a_{j}(p)+u,a_{j}(p)}^{n}(1) \mid \mathcal{H}_{a_{j}(p)}^{n} \right] = 0,
\end{equation*}
because $(W,B, \pi) \overset{d}{=} -(W,B, \pi)$, and since $f'( \eta(q)_{a_{j}(p)+ u, a_{j}(p)})$ is odd and $\alpha_{a_{j}(p)+u,a_{j}(p)}^{n}(1)$ is even in $(W,B, \pi)$. Furthermore,
\begin{equation*}
\mathbb{E} \left[ \big( \alpha_{a_{j}(p)+ u, a_{j}(p)}^{n}(2) \big)^{2} \right] \leq C (pk_{n} \Delta_{n})^{2} + Cpn^{-1/2} \mathbb{E} \left[ ( \gamma_{j}^{n}(p))^{2} \right],
\end{equation*}
where
\begin{equation*}
\gamma_{j}^{n}(p) = \sup_{a_{j}(p) \Delta_{n} \leq s \leq a_{j+1}(p) \Delta_{n}} \Big[ |a_{s} - a_{a_{j(p)}}| + | \tilde{ \sigma}_{s} - \tilde{ \sigma}_{a_{j(p)}}| + | \tilde{v}_{s} - \tilde{v}_{a_{j(p)}}| + | \bar{ \sigma}_{s}- \bar{ \sigma}_{a_{j(p)}}| + | \bar{v}_{s} - \bar{v}_{a_{j(p)}}| \Big].
\end{equation*}
Then, for any fixed $p \geq 2$, it follows from the Cauchy-Schwarz inequality that
\begin{align*}
\mathbb{E} \big[|A_{1}^{n}| \big] &\leq Cn^{-3/4} \sum_{j=0}^{j_{n}(p)-1} \sum_{u=0}^{2(p+1) k_{n}-1} \mathbb{E} \left[ \big( \alpha_{a_{j}(p)+ u,a_{j}(p)}^{n}(2) \big)^{2} \right]^{1/2} \\[0.10cm]
&\leq C n^{-3/4} n (pk_{n} \Delta_{n})+ Cpn^{-1/2} \sum_{j=0}^{j_{n}(p)-1} \mathbb{E} \big[ ( \gamma_{j}^{n}(p))^{2} \big]^{1/2} \rightarrow 0,
\end{align*}
as $n \rightarrow \infty$, which follows from Lemma 5.4 in \citet*{jacod-li-mykland-podolskij-vetter:09a}.

Now, we proceed to the last term, $B(p,q,r)^{n}$. This can be further divided into
\begin{equation*}
B(p,q,r)^{n} = B(p,q,r)_{1}^{n} + B(p,q,r)_{2}^{n},
\end{equation*}
with
\begin{align*}
B(p,q,r)_{1}^{n} &= n^{1/4} \frac{1}{n} \sum_{j=0}^{j_{n}(p)-1} \sum_{u=0}^{2(p+1)k_{n}-1} \mathbb{E} \left[ \eta(q,r)_{a_{j}(p)+ u, a_{j}(p)} \mid \mathcal{H}_{a_{j}(p)}^{n} \right] - n^{1/4} \mu_{q} \mu_{r} \sum_{j=0}^{j_{n}(p)-1} \int_{a_{j}(p) \Delta_{n}}^{a_{j+1}(p) \Delta_{n}} \upsilon_{a_{j}(p)} \mathrm{d}t, \\[0.10cm]
B(p,q,r)_{2}^{n} &= n^{1/4} \mu_{q} \mu_{r} \sum_{j=0}^{j_{n}(p)-1} \int_{a_{j}(p) \Delta_{n}}^{a_{j+1}(p) \Delta_{n}} \big( \upsilon_{a_{j}(p)}- \upsilon_{t} \big) \mathrm{d}t - n^{1/4} \mu_{q} \mu_{r} \int_{i_{n}(p) \Delta_{n}}^{1} \upsilon_{t} \mathrm{d}t,
\end{align*}
and where $\upsilon_{t} = \theta \psi_{2} \sigma_{t}^{2} + \frac{1}{ \theta} \psi_{1} \rho^{2} \omega_{t}^{2}$.

In view of Lemma \ref{lemma:edgeworth},
\begin{equation*}
\lim_{p \rightarrow \infty} \limsup_{n \rightarrow \infty} \mathbb{P} \big[|B(p,q,r)_{1}^{n}| > \delta| \big] = 0.
\end{equation*}
Concerning the term $B(p,q,r)_{2}^{n}$, we remark that $\upsilon_{t}$ is a continuous It\^{o} process due to Assumptions (V) and (N). Hence, a known result on the error of Riemann integration (see, e.g., the proof of Lemma A.1 (iv) in \citet*{christensen-podolskij-thamrongrat-veliyev:17a}) states that
\begin{equation*}
\mathbb{E} \left[|B(p,q,r)_{2}^{n}| \right] \leq C n^{1/4}(pk_{n} \Delta_{n}+p n^{-1/2}) \leq Cpn^{-1/4},
\end{equation*}
which completes the proof. \qed
\end{proof}

\begin{lemma} \label{M-CLT}
Let $p \geq 2$ be fixed. Under the maintained assumptions of Theorem \ref{appendix:clt}, it holds that
\begin{equation*}
n^{1/4} (M(p,q_{1},r_{1})^{n},M(p,q_{2},r_{2})^{n}) \overset{ \mathcal{D}_{s}}{ \longrightarrow} MN(0, \Sigma(p)),
\end{equation*}
where the $(i,j)$-th component of the 2x2 covariance matrix $\Sigma(p)$ is defined as
\begin{equation*}
\Sigma(p)_{ij} = \theta \frac{p}{p+1} \int_{0}^{1} \int_{0}^{2} \left(2-\frac{s}{p} \right) h_{i j} \big(( \rho \omega_{u}, \sigma_{u}), ( \psi_{1}/ \theta, \theta \psi_{2}), f(s) \big) \mathrm{d}s \mathrm{d}u.
\end{equation*}
\end{lemma}
\begin{proof}
We set $\zeta(p)_{j}^{n} = \big( \zeta(p,q_{1},r_{1},1)_{j}^{n}, \zeta(p,q_{2}, r_{2},1)_{j}^{n} \big)^{ \top}$ for notational convenience. In view of Theorem IX.7.28 in \cite{jacod-shiryaev:03a}, it suffices to verify the following four conditions:
\begin{align*}
&n^{-1/2} \sum_{j=0}^{j_{n}(p)-1} \mathbb{E} \left[ \zeta(p)_{j}^{n} ( \zeta(p)_{j}^{n})^{ \top} \mid \mathcal{H}_{a_{j}(p)}^{n} \right] \overset{p}{ \longrightarrow} \Sigma(p), \\[0.10cm]
&n^{-1/2} \sum_{j=0}^{j_{n}(p)-1} \mathbb{E} \left[ \lVert \zeta(p)_{j}^{n} \rVert^{4} \mid \mathcal{H}_{a_{j}(p)}^{n} \right] \overset{p}{ \longrightarrow} 0, \\[0.10cm]
&n^{-1/2} \sum_{j=0}^{j_{n}(p)-1} \mathbb{E} \left[ \zeta(p)_{j}^{n} (W_{b_{j}(p) \Delta_{n}} - W_{a_{j}(p) \Delta_{n}}) \mid \mathcal{H}_{a_{j}(p)}^{n} \right] \overset{p}{ \longrightarrow} 0, \\[0.10cm]
&n^{-1/2} \sum_{j=0}^{j_{n}(p)-1} \mathbb{E} \left[ \zeta(p)_{j}^{n} (N_{b_{j}(p) \Delta_{n}} - N_{a_{j}(p) \Delta_{n}}) \mid \mathcal{H}_{a_{j}(p)}^{n} \right] \overset{p}{ \longrightarrow} 0,
\end{align*}
for any bounded martingale $N$ that is orthogonal to $W$.

The last three convergences are omitted, since they can be proved exactly as in \citet*{podolskij-vetter:09a} and \citet*{jacod-li-mykland-podolskij-vetter:09a}. So we only prove the convergence in the first condition, which is done separately for each component of $\tilde{ \Sigma}$. We only spell out the details for the first term.

To this end, for any $u, v \in A_{j}(p)$ and $u \leq v$, we note that
\begin{equation*}
\mathbb{E} \left[ \tilde{Y}_{u}^{n} \tilde{Y}_{v}^{n} \mid \mathcal{H}_{a_{j}(p)}^{n} \right] = \frac{1}{n} \left( \mathbb{E} \big[ \eta_{u,a_{j}(p)}^{n} \eta_{v,a_{j}(p)}^{n} \mid \mathcal{H}_{a_{j}(p)}^{n} \big] - \mathbb{E} \big[ \eta_{u,a_{j}(p)}^{n} \mid \mathcal{H}_{a_{j}(p)}^{n} \big] \mathbb{E} \big[ \eta_{v,a_{j}(p)}^{n} \mid \mathcal{H}_{a_{j}(p)}^{n} \big] \right).
\end{equation*}
We suppress the dependence on $q$ and $r$ for brevity.
Let $\tilde{ \eta}_{i,m}^{n}$ be a version of the variable $\eta_{i,m}^{n}$, where the noise variable $u$ is replaced by a Gaussian distributed one. In view of Lemma \ref{lemma:edgeworth},
\begin{equation*}
\mathbb{E} \left[ \eta_{u,a_{j}(p)}^{n} \mid \mathcal{H}_{a_{j}(p)}^{n} \right] = \mathbb{E} \left[ \tilde{ \eta}_{u,a_{j}(p)}^{n} \mid \mathcal{H}_{a_{j}(p)}^{n} \right] + o_{p}(n^{-1/4}).
\end{equation*}
This, together with Lemma \ref{lemma:moment} and the Cauchy-Schwarz inequality, implies that
\begin{align*}
\left| \mathbb{E} \left[ \eta_{u,a_{j}(p)}^{n} \eta_{v,a_{j}(p)}^{n} \mid \mathcal{F}_{a_{j}(p) \Delta_{n}} \right] - \mathbb{E} \left[ \tilde{ \eta}_{u,a_{j}(p)}^{n} \tilde{ \eta}_{v,a_{j}(p)}^{n} \mid \mathcal{H}_{a_{j}(p)}^{n} \right] \right| &\leq
\left| \mathbb{E} \left[( \eta_{u,a_{j}(p)}^{n} - \tilde{ \eta}_{u,a_{j}(p)}^{n}) \eta_{v,a_{j}(p)}^{n} \mid \mathcal{H}_{a_{j}(p)}^{n} \right] \right| \\[0.10cm]
&+ \left| \mathbb{E} \left[( \eta_{v,a_{j}(p)}^{n} - \tilde{ \eta}_{v,a_{j}(p)}^{n}) \tilde{ \eta}_{u,a_{j}(p)}^{n} \mid \mathcal{H}_{a_{j}(p)}^{n} \right] \right| \\[0.10cm]
&= o_{p}(n^{-1/8}).
\end{align*}
Further, recalling the notation in \eqref{equation:h-function-covariance}, we have that
\begin{align*}
\mathbb{E} \left[ \tilde{ \eta}_{u,a_{j}(p)}^{n} \tilde{ \eta}_{v,a_{j}(p)}^{n} \mid \mathcal{H}_{a_{j}(p)}^{n} \right] &- \mathbb{E} \left[ \tilde{ \eta}_{u,a_{j}(p)}^{n} \mid \mathcal{H}_{a_{j}(p)}^{n} \right] \mathbb{E} \left[ \tilde{ \eta}_{v,a_{j}(p)}^{n} \mid \mathcal{H}_{a_{j}(p)}^{n} \right] \\[0.10cm] &= h_{11} \big(( \omega_{a_{j}(p) \Delta_{n}}, \sigma_{a_{j}(p) \Delta_{n}}), \psi^{n}, f^{n}((v-u)/k_{n}) \big),
\end{align*}
where $\psi^{n} = (n^{1/2} \rho^{2} \psi_{1}^{n}/k_{n}, k_{n} \psi_{2}^{n}/n^{1/2})$ and
\begin{equation*}
f^{n}(s) = \bigg( \sum_{ m =-\ell}^{\ell} \rho(|m|) f_{1}^{n}(s+ m/k_{n}), f_{2}^{n}(s), \sum_{m=-\ell}^{\ell} \rho(|m|) f_{3}^{n}(s+ m/k_{n}), f_{4}^{n}(s) \bigg)
\end{equation*}
with
\begin{align*}
f_{1}^{n}(s) &= n^{1/2} \sum_{j=1}^{k_{n}(1-s)} \left(g \left( \frac{j}{k_{n}} \right) - g \left( \frac{j-1}{k_{n}} \right) \right)
\left(g \left( \frac{j+s k_{n}}{k_{n}} \right) - g \left( \frac{j-1+s k_{n}}{k_{n}} \right) \right), \\[0.10cm]
f_{2}^{n}(s) &= n^{-1/2} \sum_{j=0}^{k_{n}(1-s)} g \left( \frac{j}{k_{n}} \right) g \left( \frac{j+s k_{n}}{k_{n}} \right), \\[0.10cm]
f_{3}^{n}(s) &= n^{1/2} \sum_{j=1}^{k_{n}(2-s)} \left(g \left( \frac{j}{k_{n}} \right) - g \left( \frac{j-1}{k_{n}} \right) \right) \left(g \left( \frac{j+s k_{n}-k_{n}}{k_{n}} \right) - g \left( \frac{j-1+s k_{n} - k_{n}}{k_{n}} \right) \right), \\[0.10cm]
f_{4}^{n}(s) &= n^{-1/2} \sum_{j=0}^{k_{n}(2-s)} g \left( \frac{j}{k_{n}} \right)  g \left( \frac{j+s k_{n} - k_{n}}{k_{n}} \right).
\end{align*}
As a consequence, for $u, v \in A_{j}(p)$ and $0 \leq v-u < 2k_{n}$, we obtain
\begin{equation*}
\mathbb{E} \left[ \tilde{Y}_{u}^{n} \tilde{Y}_{v}^{n} \mid \mathcal{H}_{a_{j}(p)}^{n} \right] = \frac{1}{n} h_{11} \Big(( \omega_{a_{j}(p) \Delta_{n}}, f^{n}((v-u)/k_{n}) \Big) + o_{p} (n^{-1}),
\end{equation*}
while this term vanishes for $v-u \geq 2k_{n}$. In turn, this implies that
\begin{equation*}
\mathbb{E} \left[( \zeta(p,1)_{j}^{n})^{2} \mid \mathcal{H}_{a_{j}(p)}^{n} \right] = R_{1}^{n} + R_{2}^{n},
\end{equation*}
where
\begin{align*}
R_{1}^{n} &= 2 \sum_{u=a_{j}(p)}^{b_{j}(p)-2k_{n}-1} \sum_{v=u}^{u+2 k_{n}-1} \mathbb{E} \left[ \tilde{Y}_{u}^{n} \tilde{Y}_{v}^{n} \mid \mathcal{H}_{a_{j}(p)}^{n} \right] - \sum_{u=a_{j}(p)}^{b_{j}(p)-2 k_{n}-1} \mathbb{E} \left[( \tilde{Y}_{u}^{n})^{2} \mid \mathcal{H}_{a_{j}(p)}^{n} \right], \\[0.10cm]
R_{2}^{n} &= 2 \sum_{u=b_{j}(p)-2 k_{n}}^{b_{j}(p)-1} \sum_{v=u}^{b_{j}(p)-1} \mathbb{E} \left[ \tilde{Y}_{u}^{n} \tilde{Y}_{v}^{n} \mid \mathcal{H}_{a_{j}(p)}^{n} \right] - \sum_{u=b_{j}(p)-2k_{n}}^{b_{j}(p)-1} \mathbb{E} \left[( \tilde{Y}_{u}^{n})^{2} \mid \mathcal{H}_{a_{j}(p)}^{n} \right].
\end{align*}
Consider the first term,
\begin{align*}
n^{-1/2} \sum_{j=0}^{j_{n}(p)-1} R_{1}^{n} &= n^{-3/2} 2(2pk_{n} - 2k_{n}) \sum_{j=0}^{j_{n}(p)-1} \sum_{v=0}^{2k_{n}-1} h_{11} \big(( \omega_{a_{j}(p) \Delta_{n}}, \sigma_{a_{j}(p) \Delta_{n}}), \psi^{n}, f^{n}(v/k_{n}) \big) + o_{p}(1) \\[0.10cm]
& \overset{p}{ \longrightarrow} 2 \theta \frac{p-1}{p+1} \int_{0}^{1} \int_{0}^{2} h_{11} \big(( \rho \omega_{u}, \sigma_{u}), ( \psi_{1} / \theta, \theta \psi_{2}), f(s) \big) \mathrm{d}s \mathrm{d}u,
\end{align*}
where the argument is as in \citet*{podolskij-vetter:09a} via continuity of the function $h_{11}$ and Lebesgue's theorem.

Furthermore, we deduce that
\begin{align*}
n^{-1/2} \sum_{j=0}^{j_{n}(p)-1} n^{-1/2} R_{2}^{n} &= n^{-3/2} 2 \sum_{j=0}^{j_{n}(p)-1} \sum_{v=0}^{2 k_{n}-1} (2 k_{n}-v) h_{11} \big(( \omega_{a_{j}(p) \Delta_{n}}, \sigma_{a_{j}(p) \Delta_{n}}), \psi^{n}, f^{n}(v/k_{n}) \big) + o_{p}(1) \\[0.10cm]
& \overset{p}{ \longrightarrow} \frac{ \theta}{p+1} \int_{0}^{1} \int_{0}^{2} (2-s) h_{1 1} \big(( \rho \omega_{u}, \sigma_{u}), ( \psi_{1} / \theta, \theta \psi_{2}), f(s) \big) \mathrm{d}s \mathrm{d}u.
\end{align*}
The statement then follows by summing these two terms. \qed
\end{proof}
The proof of Theorem \ref{appendix:clt} now follows directly from \eqref{equation:decomposition} and Lemmas \ref{lemma:N} -- \ref{M-CLT}.

\subsection*{Proof of Theorem \ref{theorem:clt}}

The proof is mostly based on Theorem \ref{appendix:clt} above, after we take into account that the error in the noise variance estimator introduced in \eqref{equation:noise-variance}. We recall that the long-run variance of the process $\pi$ is defined as $\rho^{2} = \rho(0) + 2 \sum_{m = 1}^{ \ell} \rho(m)$, where $\rho(m) = \mathbb{E}[ \pi_{0} \pi_{m}]$ is the $m$th autocovariance.

\begin{lemma} \label{lemma:noise}
Assume that $r$ follows the process in \eqref{equation:X} and that Assumptions (V) and (N) hold. Then, as $n \rightarrow \infty$, such that $h_{n} \asymp  n^{1/5}$ and $\ell_{n} \asymp  n^{1/8}$, it holds that
\begin{equation} \label{equation:noise-bound}
\hat{ \omega}_{n}^{2} - \int_{0}^{1} \rho^{2} \omega_{s}^{2} \mathrm{d}s  = o_{p}(n^{-1/4}).
\end{equation}
\end{lemma}

\begin{proof}
For each $h \geq 0$, we observe that
\begin{align*}
p_{i \Delta_{n}}^{*} - \tilde{p}_{(i+h) \Delta_{n}}^{*} = \alpha_{i}^{n}(h) + \beta_{i}^{n}(h),
\end{align*}
with
\begin{equation*}
\alpha_{i}^{n}(h) = \omega_{i \Delta_{n}} (\pi_{i}-\tilde{\pi}_{i+h}) \quad \text{and} \quad \beta_{i}^{n}(h) = p_{i \Delta_{n}} - \tilde{p}_{(i+h) \Delta_{n}} + \frac{1}{h_{n}} \sum_{j=0}^{h_{n}-1} \pi_{i+h+j}( \omega_{i \Delta_{n}} - \omega_{(i+h+j) \Delta_{n}}).
\end{equation*}
It suffices to prove that for each $m \in \{0, \dots, \ell_{n} \}$:
\begin{equation*}
\rho_{n}(m) - \int_{0}^{1} \rho(m) \omega_{s}^{2} \mathrm{d}s = O_{p} \left(n^{-2/5} \right).
\end{equation*}
Then, based on the required order of $h_{n}$ and $\ell_{n}$, summing up these terms results in an $O_{p}(n^{-11/40})$ bound on the error of $\hat{ \omega}_{n}^{2}$.

To show this result, we suppose (without loss of generality) that $n$ is chosen so large that $h_{n} > \ell$ and note that $\rho(k) = 0$ for $k > \ell$. We proceed with the following decomposition:
\begin{equation*}
\rho_{n}(m) - \int_{0}^{1} \rho(m) \omega_{s}^{2} \mathrm{d}s = N_{1}^{n}(m) + N_{2}^{n}(m) + N_{3}^{n}(m),
\end{equation*}
where
\begin{align*}
N_{1}^{n}(m) &= \frac{1}{n-5h_{n}+1} \sum_{i=0}^{n-5h_{n}} \omega_{i \Delta_{n}} \omega_{(i+m) \Delta_{n}} \big[ ( \pi_{i}- \tilde{ \pi}_{i+2h_{n}})( \pi_{i+m}- \tilde{ \pi}_{i+4h_{n}}) - \rho(m) \big], \\
N_{2}^{n}(m) &= - \rho(m) \int_{(n-5h_{n}+1) \Delta_{n}}^{1} \omega_{s}^{2} \mathrm{d}s + \frac{ \rho(m)}{n-5h_{n}+1} \sum_{i=0}^{n-5h_{n}} \int_{i \Delta_{n}}^{(i+1) \Delta_{n}} \left( \omega_{i \Delta_{n}} \omega_{(i+m) \Delta_{n}} -\omega_{s}^{2} \right) \mathrm{d}s, \\
N_{3}^{n}(m) &= \frac{1}{n-5h_{n}+1} \sum_{i=0}^{n-5h_{n}} \Big[ \alpha_{i}^{n}(2h_{n}) \beta_{i+m}^{n}(4h_{n}-m) + \alpha_{i+m}^{n}(4h_{n}-m) \beta_{i+m}^{n}(4h_{n}) \\
& \qquad \qquad \qquad \qquad \ \, + \beta_{i}^{n}(2h_{n}) \beta_{i+m}^{n}(4h_{n}-m) \Big].
\end{align*}
We proceed to show that each of these terms is asymptotically negligible. In view of (2.1.44) in \citet*{jacod-protter:12a}, we deduce that $\sup_{0 \leq i \leq n-h} \mathbb{E} \left[(p_{i \Delta_{n}} - p_{(i+h) \Delta_{n}})^{2} \right] \leq C h/n.$ This order also holds if $p$ is replaced by $\omega$. Then, combining these bounds with $\mathbb{E}[ \pi_{i}^{2}] \leq C$ yields:
\begin{equation*}
\mathbb{E} \left[ \alpha_{i}^{n}(h)^{2} \right] \leq C \quad \text{and} \quad \mathbb{E} \left[ \beta_{i}^{n} (h)^{2} \right] \leq C \frac{h}{n}.
\end{equation*}
Consequently,
\begin{equation*}
\mathbb{E} \left[|N_{n}^{3}(m)| \right] \leq C \sqrt{ \frac{h_{n}}{n}} \leq C n^{-2/5}.
\end{equation*}
Next, we deal with the term $N_{2}^{n}(m)$, which is non-zero for $m \leq \ell$, which we suppose is true. A standard Riemann integration error argument then yields $N_{2}^{n}(m) = O_{p}(n^{-1/2})$.

Concerning the term $N_{1}^{n}(m)$, due to Assumption (N,i+iv) the summands are correlated only up to order $\ell_{n}$, and hence
\begin{equation*}
\mathbb{E} \left[(N_{n}^{1}(m))^{2} \right] \leq C \frac{ \ell_{n}}{n},
\end{equation*}
so $N_{1}^{n}(m) = O_{p} \big(( \ell_{n}/n)^{1/2} \big)$. Adding up terms leads to the statement in \eqref{equation:noise-bound}. \qed
\end{proof}

We now turn to the proof of Theorem \ref{theorem:clt}. For the convergence in probability part it follows from the above results that
\begin{align*}
BV_{n}(2,0) & \overset{p}{ \longrightarrow}  \int_{0}^{1} \Big( \theta \psi_{2} \sigma_{s}^{2} + \frac{1}{ \theta} \psi_{1} \rho^{2} \omega_{s}^{2} \Big) \mathrm{d}s + \theta \psi_{2} \sum_{0 \leq s \leq 1} ( \Delta r_{s})^{2}, \\[0.10cm]
BV_{n}(1,1) & \overset{p}{ \longrightarrow}  \int_{0}^{1} \Big( \theta \psi_{2} \sigma_{s}^{2} + \frac{1}{ \theta} \psi_{1} \rho^{2} \omega_{s}^{2} \Big) \mathrm{d}s.
\end{align*}
This can be shown by proceeding as in the proof of Theorem 2 in \citet*{podolskij-vetter:09a}. Now, the consistency in Theorem \ref{theorem:clt} follows from Lemma \ref{lemma:noise}. The central limit theorem is derived under the further assumption that $p$ is continuous. Hence, it follows from Theorem \ref{appendix:clt} [choosing the exponents $(q_{1},r_{1}) = (2,0)$ and $(q_{2},r_{2}) = (1,1)$] and Lemma \ref{lemma:noise}. \qed

\subsection{A jump- and noise-robust estimator of $\tilde{\Sigma}$} \label{appendix:subsampler}

Here, we present an estimator of $\tilde{ \Sigma}$, which appears in the extended central limit theorem in Theorem \ref{appendix:clt}. Its connection with $\Sigma$ is explained in \eqref{asymptotic-covariance-Sigma}.

We build on previous work of \citet*{christensen-podolskij-thamrongrat-veliyev:17a}, who propose a subsampling estimator of $\tilde{ \Sigma}$ \citep*[see also][]{politis-romano-wolf:99a, kalnina:11a, mykland-zhang:17a}. An appealing feature of their estimator is that it is positive semi-definite and has good small sample properties.\footnote{\citet*{podolskij-vetter:09a} develop an element-by-element estimator of $\tilde{ \Sigma}$ in the i.i.d noise setting. However, this estimator is not positive semi-definite and is often ill-conditioned in practice. Moreover, it is not consistent for $\tilde{ \Sigma}$ under the jump alternative.} In that paper, the subsampling estimator is shown to be consistent if there is either price jumps or microstructure noise, but not in the presence of both. Microstructure noise is further restricted to be either heteroscedastic or dependent, but not both. Here, we extend their framework to account for all of these features at once.

We propose the following jump- and noise-robust covariance matrix estimator:
\begin{equation*}
\tilde{ \Sigma}_{n} = \frac{1}{L} \sum_{l=1}^{L} \Bigg( \frac{n^{1/4}}{ \sqrt{L}} \begin{bmatrix} \check{BV}_{l}(q_{1},r_{1}) - \check{BV}_{n}(q_{1},r_{1}) \\[0.10cm] \check{BV}_{l}(q_{2},r_{2}) - \check{BV}_{n}(q_{2},r_{2}) \end{bmatrix} \Bigg) \Bigg( \frac{n^{1/4}}{ \sqrt{L}} \begin{bmatrix}  \check{BV}_{l}(q_{1},r_{1}) - \check{BV}_{n}(q_{1},r_{1}) \\[0.10cm] \check{BV}_{l}(q_{2},r_{2}) - \check{BV}_{n}(q_{2},r_{2}) \end{bmatrix} \Bigg)^{ \top},
\end{equation*}
with
\begin{align*}
\begin{split}
\check{BV}_{n}(q,r) &= \frac{1}{n} \sum_{i=0}^{n-2k_{n}+1} | n^{1/4} \check{r}_{i}^{*}|^{q} | n^{1/4} \check{r}_{i+k_{n}}^{*}|^{r}, \qquad
\check{BV}_{l}(q,r) = \frac{Lpk_{n}}{n} \sum_{i=1}^{n/Lpk_{n}} v_{(i-1)L+l}(q,r)^{n}, \\[0.10cm]
v_{i}(q,r)^{n} &= \frac{1}{pk_{n}-2k_{n}+2} \sum_{j,j+2 k_{n}-1 \in B_{i}(p)} | n^{1/4} \check{r}_{j}^{*}|^{q} | n^{1/4} \check{r}_{j+k_{n}}^{*}|^{r},
\end{split}
\end{align*}
and
\begin{equation} \label{equation:truncated-return}
\check{r}_{i}^{*} = \bar{r}_{i}^{*} \mathbbm{1}( |\bar{r}_{i}^{*}| \leq u_{n}),
\end{equation}
where $\mathbbm{1}(A)$ is an indicator function that equals one if $A$ is true, zero otherwise, and $u_{n} = \alpha n^{- \bar{ \omega}}$ with $\alpha > 0$ and $\bar{ \omega} \in (0, 1/4)$.

In the main text, we present an estimator of $\Sigma$, $\Sigma_{n}^{*}$, for the particular setting with $(q_{1},r_{1}) = (2,0)$ and $(q_{2},r_{2}) = (1,1)$. Its connection with $\tilde{ \Sigma}_{n}$ is given by:
\begin{equation*}
\Sigma_{n,ij}^{*} = \left( \sqrt{n} c_{1}^{n} \right)^{2} \left(  \frac{ \pi}{2} \right)^{i+j-2} \tilde{ \Sigma}_{n,ij}.
\end{equation*}

To deal with microstructure noise, $\tilde{ \Sigma}_{n}$ is based on the pre-averaged bipower variation. To further robustify it to jumps, we exploit the truncation device of \citet*{mancini:09a} in \eqref{equation:truncated-return}. This sets large negative or positive pre-averaged log-returns to zero, since they are most likely dominated by the jump component. The threshold, $u_{n}$, is adapted to an estimate of the square-root integrated variance by setting $\alpha = c \sqrt{ BV_{n}(1,1)}$. With this configuration, we can interpret $c$ as the number of local diffusive standard deviations a pre-averaged log-return must exceed to be labelled a jump. The rate parameter, $\bar{ \omega}$, ensures that, asymptotically, continuous pre-averaged log-returns are unaffected by the truncation. In this way, $\tilde{ \Sigma}_{n}$ is also made jump-robust.

\begin{theorem} \label{theorem:subsampler-general}
Assume that $r$ follows \eqref{equation:X} with $\beta \leq 1$ and that Assumptions (V) and (N) hold. Moreover, for each $s \in \{ q_{1}, q_{2}, r_{1}, r_{2} \} \cap [1, \infty)$, we require that
\begin{equation*}
\frac{q+ \delta-1/2}{4 q- \beta} < \bar{ \omega}< \frac{1}{4}- \delta
\end{equation*}
where $L \asymp n^{(1- \delta)/2}$. Then, as $n \rightarrow \infty$, $\rightarrow \infty$, $L/p \rightarrow \infty$, $\sqrt{n}/L p^{2} \rightarrow \infty$ and $\delta< 1/16$, it holds that
\begin{equation*}
\tilde{ \Sigma}_{n} \overset{p}{ \longrightarrow} \tilde{\Sigma}.
\end{equation*}
\end{theorem}

\subsection*{Proof of Theorem \ref{theorem:subsampler-general}}

By the polarization identity, it suffices to show the result $\tilde{ \Sigma}_{n} - \tilde{ \Sigma} \overset{p}{ \longrightarrow} 0$ in the univariate setting. To this end, we denote by $\bar{r}_{i}'$ the pre-averaged return based on the continuous part of $p_{t}$ and $\epsilon_{t}$. The corresponding subsampling estimator is denoted as
\begin{equation*}
\Sigma_{n}' = \frac{1}{L} \sum_{l=1}^{L} \Biggl( \frac{n^{1/4}}{ \sqrt{L}} \Bigl(BV_{l}'(q,r) - BV_{n}'(q,r) \Bigr) \Biggr)^{2},
\end{equation*}
where
\begin{align*}
BV_{n}'(q,r) &= \frac{1}{n} \sum_{i=0}^{n-2k_{n}+1} |n^{1/4} \bar{r}_{i}'|^{q} |n^{1/4}  \bar{r}_{i+k_{n}}'|^{r}, \qquad
BV_{l}'(q,r) = \frac{Lpk_{n}}{n} \sum_{i=1}^{n/Lpk_{n}} v'_{(i-1)L+l}(q,r)^{n}, \\[0.10cm]
v'_{i}(q,r)^{n} &= \frac{1}{pk_{n}-2k_{n}+2} \sum_{j,j+2 k_{n}-1 \in B_{i}(p)} |n^{1/4} \bar{r}_{j}'|^{q} |n^{1/4} \bar{r}_{j+k_{n}}'|^{r}.
\end{align*}
Proceeding exactly as in the proof of Theorem 3.8 in \citet*{christensen-podolskij-thamrongrat-veliyev:17a}, it immediately follows that $\Sigma_{n}' - \tilde{ \Sigma} \overset{p}{ \longrightarrow} 0$. Hence, it suffices to show that
\begin{equation*}
\tilde{ \Sigma}_{n} - \Sigma_{n}' \overset{p}{ \longrightarrow} 0.
\end{equation*}
To establish this, we note that:
\begin{align*}
\tilde{\Sigma}_{n} - \Sigma_{n}' = \frac{1}{L} \sum_{l=1}^{L} & \Bigg( \frac{n^{1/4}}{ \sqrt{L}} \bigg( BV_{l}(q,r) - BV_{l}'(q,r) + BV_{n}'(q,r) - BV_{n}(q,r) \bigg) \Bigg) \\[0.10cm]
\times & \Bigg( \frac{n^{1/4}}{ \sqrt{L}} \bigg( BV_{l}(q,r) - BV_{n}(q,r) + BV_{l}'(q,r) - BV_{n}'(q,r) \bigg) \Bigg),
\end{align*}
which is an average of $L$ terms. The main idea is to deduce that each of these $L$ terms converges in probability to zero.
In particular, we show uniform convergence in mean square:
\begin{equation} \label{equation:difference-to-zero}
\sup_{1 \leq l \leq L} \mathbb{E} \bigg[ \Big| \frac{n^{1/4}}{ \sqrt{L}} \big( BV_{l}(q,r) - BV_{l}'(q,r) \big) \Big|^{2} \bigg] \rightarrow 0 \quad \text{and} \quad \sup_{1 \leq l \leq L} \mathbb{E} \bigg[ \Big| \frac{n^{1/4}}{ \sqrt{L}} \big(BV_{n}(q,r) - BV_{n}'(q,r) \big) \Big|^{2} \bigg] \rightarrow 0.
\end{equation}
Following the arguments in the proof of Lemma A.6 in \citet*{christensen-podolskij-thamrongrat-veliyev:17a} implies that
\begin{equation} \label{equation:difference-bounded}
\sup_{1 \leq l \leq L} \mathbb{E} \left[ \Big| \frac{n^{1/4}}{ \sqrt{L}} \big( BV_{l}'(q,r) - BV_{n}'(q,r) \big) \Big|^{2} \right] \leq C.
\end{equation}
Then, equation \eqref{equation:difference-to-zero} -- \eqref{equation:difference-bounded} lead to:
\begin{equation*}
\sup_{1 \leq l \leq L} \mathbb{E} \bigg[ \Big| \frac{n^{1/4}}{ \sqrt{L}} \big( BV_{l}(q,r) - BV_{n}(q,r) \big) \Big|^{2} \bigg] \leq C.
\end{equation*}
The combination of the last three results and the Cauchy-Schwarz inequality implies:
\begin{equation*}
\mathbb{E} \big[ | \tilde{ \Sigma}_{n} - \Sigma_{n}'| \big] \rightarrow 0,
\end{equation*}
so that $\tilde{ \Sigma}_{n} - \Sigma_{n}' \overset{p}{ \longrightarrow} 0$. Thus, it is enough to show \eqref{equation:difference-to-zero}.

For each $j$, we define
\begin{align*}
\lambda_{j}^{n}(q,r) = |n^{1/4} \bar{r}_{j}^{*}|^{q} |n^{1/4}  \bar{r}_{j+k_{n}}^{*}|^{r} \mathbbm{1}_{ \{|\bar{r}_{j}^{*}| \leq u_{n} \cap | \bar{r}_{j+k_{n}}^{*}| \leq u_{n} \}} - |n^{1/4} \bar{r}_{j}'|^{q} |n^{1/4} \bar{r}_{j+k_{n}}'|^{r}.
\end{align*}
Note that
\begin{equation*}
BV_{n}(q,r) - BV_{n}'(q,r) = \frac{1}{n} \sum_{i=0}^{n-2k_{n}+1} \lambda_{i}^{n}(q,r),
\end{equation*}
and
\begin{equation*}
BV_{l}(q,r) - BV_{l}'(q,r) = \frac{Lpk_{n}}{n} \frac{1}{pk_{n}-2k_{n}+2} \sum_{i=1}^{n/Lpk_{n}} \sum_{j, j+k_{n}-1 \in B_{(i-1)L+l}(p)} \lambda_{j}^{n}(q,r).
\end{equation*}
As a result of Cauchy-Schwarz, showing \eqref{equation:difference-to-zero} is reduced to proving:
\begin{equation*}
\sup_{1 \leq j \leq n-2k_{n}+2} \mathbb{E} \bigg[ \Big| \frac{n^{1/4}}{ \sqrt{L}} \lambda_{j}^{n}(q,r) \Big|^{2} \bigg] \rightarrow 0.
\end{equation*}
Furthermore, the identity
\begin{equation*}
\lambda_{j}^{n}(q,r) =
\lambda_{j}^{n}(q,0) |n^{1/4}  \bar{r}_{j+k_{n}}'|^{r} + \lambda_{j}^{n}(0,r) |n^{1/4} \bar{r}_{j}'|^{q} + \lambda_{j}^{n}(q,0) \lambda_{j}^{n}(0,r),
\end{equation*}
and Lemma \ref{lemma:moment} combined with Cauchy-Schwarz inequality yields:
\begin{equation*}
\mathbb{E} \big[| \lambda_{j}^{n}(q,r)|^{2} \big] \leq C \Big( \mathbb{E} \big[| \lambda_{j}^{n}(q,0)|^{4} \big]^{1/2} + \mathbb{E} \big[| \lambda_{j}^{n}(0,r)|^{4} \big]^{1/2} + \mathbb{E} \big[| \lambda_{j}^{n}(q,0)|^{4} \big]^{1/2} \mathbb{E} \big[| \lambda_{j}^{n}(0,r)|^{4} \big]^{1/2} \Big).
\end{equation*}
Then, in view of the rate condition $n^{1/4}/ \sqrt{L} \rightarrow \infty,$ it is enough that
\begin{equation*}
\sup_{1 \leq j \leq n-2k_{n}+2} \mathbb{E} \bigg[ \Big| \frac{n^{1/4}}{ \sqrt{L}} \lambda_{j}^{n}(q,0) \Big|^{4} \bigg] \rightarrow 0 \quad \text{and} \quad \sup_{1 \leq j \leq n-2k_{n}+2} \mathbb{E} \bigg[ \Big| \frac{n^{1/4}}{ \sqrt{L}} \lambda_{j}^{n}(0,r) \Big|^{4} \bigg] \rightarrow 0.
\end{equation*}
We only look at the term $\lambda_{j}^{n}(q,0)$ with $q > 0$, for which
\begin{equation*}
\lambda_{j}^{n}(q,0) = n^{q/4} |\bar{r}_{j}' +  \bar{r}_{j}^{d}|^{q} \mathbbm{1}_{ \{| \bar{r}_{j}^{*}| \leq u_{n} \}} - n^{q/4} | \bar{r}_{j}'|^{q}.
\end{equation*}
Now, for any $u > 0$, $a \geq 0$ and $b\geq 0$, it holds that
\begin{equation*}
\big| |x+y|^{q} \mathbbm{1}_{ \{|x+y| \leq u \}} - |x|^{q} \big| \leq C
\bigg( \frac{|x|^{q+a}}{u^{a}} + |x|^{q} \frac{|y|^{b}}{u^{b}} + (|y| \wedge u)^{q} + \mathbbm{1}_{ \{q>1 \}}|x|^{q-1} (|y| \wedge u) \bigg),
\end{equation*}
which can be verified by looking at the cases:
\begin{equation*}
(1) \ |x| \geq u/2, \qquad (2) \ |x| < u/2 \text{ and } |x+y| > u, \qquad (3) \ |x| < u/2 \text{ and } |x+y| \leq u.
\end{equation*}
If we exploit this result with $x = \bar{r}_{j}'$, $y = \bar{r}_{j}^{d}$ and $u = u_{n}$ ($a$ and $b$ are selected below) in combination with a second inequality $|c-d|^{4} \leq |c^{4} - d^{4}|$, we find that:
\begin{equation*}
| \lambda_{j}^{n}(q,0)|^{4} \leq C n^{q} \Bigg( \frac{| \bar{r}_{j}'|^{4q+a}}{u_{n}^{a}} + | \bar{r}_{j}'|^{4q} \frac{| \bar{r}_{j}^{d}|^{b}}{u_{n}^{b}} + \big(| \bar{r}_{j}^{d}| \wedge u_{n} \big)^{4q} + | \bar{r}_{j}'|^{4q-1} \big(| \bar{r}_{j}^{d}| \wedge u_{n} \big) \Bigg).
\end{equation*}
Note that we discard the indicator function, because it holds for any $q > 1/4$. Now, arguing as in \citet*[][p. 529]{jacod-protter:12a}, we deduce that
\begin{equation*}
\mathbb{E} \left[(| \bar{r}_{j}^{d}| \wedge u_{n})^{2} \right] \leq C \frac{ \rho_{n}}{n^{1/2+(2- \beta) \bar{ \omega}}}.
\end{equation*}
In view of the fact that $(|x| \wedge u)^{q} \leq u^{q-2} (|x| \wedge u)^{2}$ for $q \geq 2$ and recalling that $u_{n} = \alpha n^{- \bar{ \omega}}$:
\begin{equation*}
\mathbb{E} \left[(| \bar{r}_{j}^{d}| \wedge u_{n})^{q} \right] \leq C \times
\begin{cases}
\displaystyle \frac{ \rho_{n}^{q/2}}{n^{q/4+(2- \beta) \bar{ \omega} q/2}} & \text{if } q \leq 2, \\[0.10cm]
\displaystyle \frac{u_{n}^{q-2} \rho_{n}}{n^{1/2+(2- \beta) \bar{ \omega}}} \leq \frac{ \rho_{n}}{n^{1/2+(q- \beta) \bar{ \omega}}} & \text{if } q>2.
\end{cases}
\end{equation*}
Moreover, for any $q>0$:
\begin{equation*}
\mathbb{E} \big[| \bar{r}_{j}'|^{q} \big] \leq C \frac{1}{n^{q/4}} \quad \text{and} \quad \mathbb{E} \big[| \bar{r}_{j}^{d}|^{2} \big] \leq C \frac{1}{n^{1/2}}.
\end{equation*}
Combining these inequalities with $a=4$, $b=1$ and
and recalling $L \asymp n^{ (1-\delta)/2}$ leads to
\begin{equation*}
\mathbb{E} \bigg[ \Big| \frac{n^{1/4}}{ \sqrt{L}} \lambda_{j}^{n}(q,0) \Big|^{4} \bigg] \leq C
\bigg(\frac{1}{n^{1-\delta-4 \bar{ \omega}}}+
 \frac{1}{n^{1/4-\delta- \bar{ \omega}}}+ \frac{ \rho_{n}}{n^{1/2-\delta-q+(4 q- \beta) \bar{ \omega}}}
+ \frac{\rho_{n}^{1/2}}{n^{-\delta+(1- \beta/2)\bar{ \omega}}} \bigg)
\end{equation*}
The first and second error terms converge to 0 due to the condition $\bar{ \omega}<1/4- \delta$. The third error term converges to 0 due to the assumption $\bar{ \omega}>(q+ \delta-1/2)/(4 q- \beta)$. This also implies $\bar{ \omega}>1/8$, which combined with $\beta \leq 1$ and $\delta < 1/16$ deals with the fourth term. \qed

\subsection*{Proof of Theorem \ref{theorem:subsampler}}

The result follows directly from the proof of Theorem \ref{theorem:subsampler-general} by taking $(q_{1},r_{1}) = (2,0)$, $(q_{2},r_{2}) = (1,1)$ and adjusting by the constants $c_{1}$ and $c_{2}$. It is worth highlighting that the constraint $\bar{ \omega} > (q+ \delta-1/2)/(4 q-\beta)$ is more binding for $q=2$ compared to $q=1$. \qed

\subsection{Irrelevance of truncation}

To show that the truncation does not influence the limiting distribution in Theorem \ref{appendix:clt} and, hence, does not affect Theorem \ref{theorem:clt}, we note that
\begin{equation*}
n^{1/4} \begin{pmatrix} BV_{n}(q_{1},r_{1}) - V(q_{1},r_{1}) \\[0.10cm] BV_{n}(q_{2},r_{2}) - V(q_{2},r_{2})
\end{pmatrix} = n^{1/4} \begin{pmatrix} BV_{n}(q_{1},r_{1}) - V(q_{1}, r_{1}) \\[0.10cm] \check{BV}_{n}(q_{2},r_{2}) - V(q_{2},r_{2}) \end{pmatrix} + n^{1/4} \begin{pmatrix} 0 \\[0.10cm] BV_{n}(q_{2},r_{2}) - \check{BV}_{n}(q_{2}, r_{2}) \end{pmatrix}.
\end{equation*}
It is thus enough to show that $n^{1/4} \big(BV_{n}(q_{2},r_{2}) - \check{BV}_{n}(q_{2},r_{2}) \big) = o_{p}(1)$. An application of Boole's and Markov's inequalities and Lemma \ref{lemma:moment} in the proof of Theorem \ref{appendix:clt} yields that
\begin{equation*}
\mathbb{P} \big[ BV_{n}(q_{2},r_{2}) \neq \check{BV}_{n}(q_{2},r_{2}) \big] \leq \sum_{i=1}^{n-k_{n}+2} \mathbb{P} \big[ | \bar{r}_{i}^{*}| > u_{n} \big] \leq \frac{C}{n^{M(1/4 - \bar{ \omega})-1}},
\end{equation*}
where $M>0$. As $M$ can be chosen arbitrarily large, the result follows. \qed

\clearpage

\section{Monte Carlo Simulations} \label{appendix:simulation}

\setcounter{table}{0}
\setcounter{figure}{0}

The jump-testing theory in Section \ref{section:theory} is derived using an infill asymptotic setup, which assumes that $\Delta_{n} \rightarrow 0$. This is unattainable in practice. In this appendix, we therefore conduct a series of Monte Carlo simulations to examine the finite sample properties of the test statistic in a more realistic setting.

We simulate from the following model:
\begin{equation}
r_{t}^{*} = r_{t} + \epsilon_{t},
\end{equation}
where $r_{t} = r_{t}^{c} + r_{t}^{d}$, and
\begin{equation} \label{equation:simulation-return}
\text{d}r_{t}^{c} = \sigma_{t} \text{d}W_{t},
\end{equation}
with $r_{0} \equiv 0$.

The drift term is omitted (i.e., $a_{t} \equiv 0$) so the continuous part of $r_{t}$ evolves as a local martingale. The contribution of the drift is small over short horizons and the terms involving it---including cross products---are asymptotically negligible. It therefore tends to play a minor role in the high-frequency setting.\footnote{A notable exception is \citet*{christensen-oomen-reno:22a}.} Similarly, adding a realistic expected rate of return does not affect the conclusions in any material way.

The instantaneous variance is assumed to be driven by a square-root process \citep*[e.g.,][]{cox-ingersoll-ross:85a, heston:93a}:
\begin{equation} \label{equation:simulation-variance}
\text{d} \sigma_{t}^{2} = \kappa (\sigma^{2} - \sigma_{t}^{2}) \text{d}t + \xi \sigma_{t} \text{d}B_{t},
\end{equation}
with $\sigma_{0}^{2} \sim \text{Gamma}(2 \kappa \sigma^{2} \xi^{-2}, 2\kappa \xi^{-2})$. Our choice of parameters is based on previous work such as \citet*{ait-sahalia-kimmel:07a}. Specifically, we set $\kappa = 5$, $\sigma = 0.4$, $\xi = 0.5$, and $\rho = -\sqrt{0.5}$, where $\rho$ is the leverage correlation between the standard Brownian motions $\mathbb{E}[\text{d}W_{t} \text{d}B_{t}]= \rho \text{d}t$. This yields a stationary and strictly positive variance in continuous-time (i.e., the Feller condition $\xi^{2} < 2 \kappa \sigma^{2}$ holds) with an annualized mean volatility of 40\%.

$r_{t}^{d}$ is an infinite-activity tempered stable process with L\'{e}vy measure
\begin{equation}
\nu( \text{d}x) = \tau \frac{e^{- \lambda x}}{x^{1 + \beta}} \text{d}x,
\end{equation}
where $\tau > 0$ and $\lambda > 0$.

The activity of the jump process is controlled by $\beta$. As $\beta$ increases, the density of the small jumps is enlarged, and the realizations of $r_{t}^{d}$ become more vibrant and start to resemble those of a Brownian motion (see Figure \ref{figure:simulation-illustration} for an illustration). As noted by \citet*{bollerslev-todorov:11a}, this renders the decomposition of diffusive and jump risk meaningless in practice and also makes it harder to separate the null from the jump alternative (\citet{ait-sahalia-jacod-li:12a}, ). A priori, we therefore expect larger values of $\beta$ to be detrimental to the rejection rate of the jump test statistic under the alternative. To gauge the impact of changing $\beta$, we follow \citet*{ait-sahalia-jacod-li:12a} and examine $\beta = 0.50$, $1.00$, $1.50$ and $1.75$. Moreover, we set $\lambda = 3$ and calibrate $\tau$ so that $r_{t}^{d}$ induces 20\% of the quadratic return variation, on average.

We simulate for $t \in [0,1]$ with a time step of $\Delta_{n} = 1/23{,}400$. This can be interpreted as adding a new observation every second during a 6.5 hour trading session. A standard Euler discretization is employed for the continuous part.\footnote{There is a small probability that the variance in \eqref{equation:simulation-variance} goes negative in finite samples. We therefore enforce a full truncation floor at zero to avoid this \cite*[e.g.,][]{andersen:08a}.} As described in \citet*{todorov-tauchen-grynkiv:14a}, the jump component is computed as the difference between two spectrally positive tempered stable processes that are generated using the acceptance-rejection algorithm of \citet*{baeumer-meerschaert:10a}. The total number of Monte Carlo replications is $T = 10{,}000$. Figure \ref{figure:simulation-illustration} illustrates a single sample path of $r_{t}^{c}$ and $r_{t}^{d}$.

\begin{figure}[ht!]
\begin{center}
\caption{A single simulation of the efficient log-return.}
\label{figure:simulation-illustration}
\begin{tabular}{cc}
\small{Panel A: Sample path of $r_{t}^{c}$.} & \small{Panel B: Sample path of $r_{t}^{d}$.}\\
\includegraphics[height=8cm,width=0.48\textwidth]{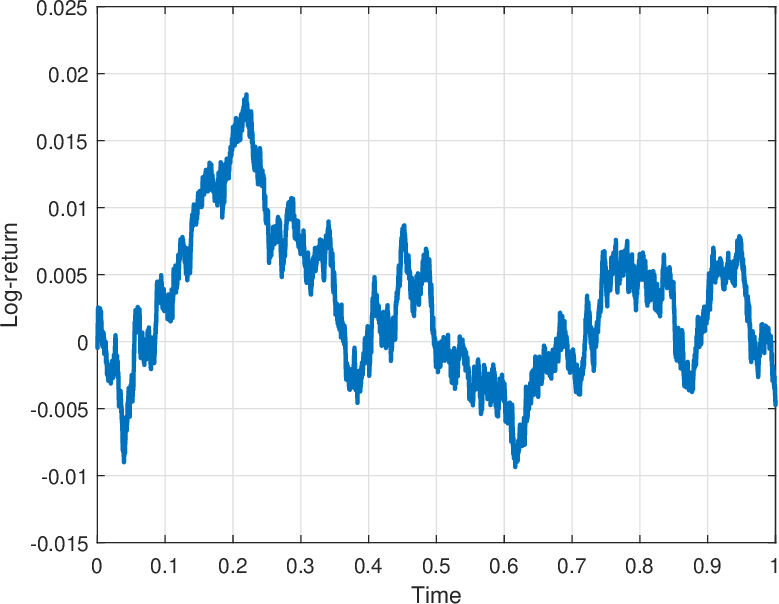} &
\includegraphics[height=8cm,width=0.48\textwidth]{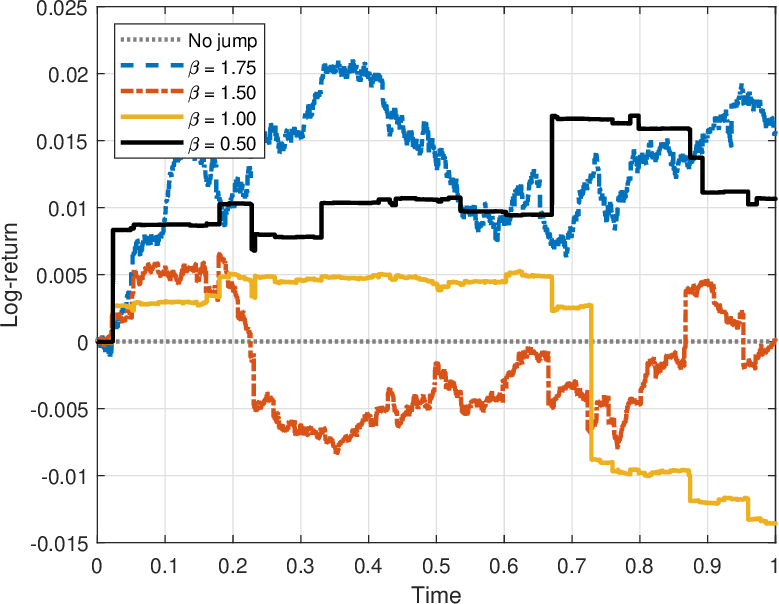}\\
\end{tabular}
\begin{scriptsize}
\parbox{\textwidth}{\emph{Note.} We show a realization of the components of the efficient log-return $r_{t}$, which is a superposition of a continuous sample path \citet*{heston:93a}-type stochastic volatility model ($r_{t}^{c}$ in Panel A) and a pure jump process of tempered stable-type ($r_{t}^{d}$ in Panel B). The latter is plotted as a function of the activity index $\beta$ using a common random seed.}
\end{scriptsize}
\end{center}
\end{figure}

We next explain how $r_{t}$ is disrupted with microstructure noise, $\epsilon_{t}$. In particular, we study four types of additive measurement error:
\begin{equation}
\epsilon_{t} =
\begin{cases}
\displaystyle \gamma \sqrt{ \text{IV} \Delta_{n}} \epsilon_{t}^{N}, & \text{``Gaussian'',} \\[0.10cm]
\displaystyle \gamma \sqrt{ \text{IV} \Delta_{n}} \epsilon_{t}^{T} \sqrt{ \frac{ \eta - 2}{\eta}}, & \text{``T-distributed'',}  \\[0.10cm]
\displaystyle \gamma \sqrt{ \text{IV} \Delta_{n}} \frac{\epsilon_{t}^{N} + \phi \epsilon_{t-1}^{N}}{1 + \phi^{2}}, & \text{``Autocorrelated'',} \\[0.10cm]
\displaystyle \gamma \sigma_{t} \sqrt{ \Delta_{n}} \epsilon_{t}^{N}, & \text{``Heteroscedastic''}.
\end{cases}
\end{equation}
where $\epsilon_{t}^{N}$ and $\epsilon_{t}^{T}$ are i.i.d. sequences of Gaussian and $t$-distributed (with $\eta$ degrees of freedom) random variables.

While the Gaussian i.i.d. draws are standard, $t$-distributed noise is less common. The latter, however, can generate notable outliers if $\eta$ is small. This leads to infrequent---but large---bouncebacks in the noisy log-price series which is a common trait of raw high-frequency data. We set $\eta = 2.5$ to be consistent with this observation.

We also examine cases with autocorrelated and heteroscedastic noise. The former scenario assumes $\epsilon_{t}$ is MA(1) with degree of memory determined by the parameter $\phi$ which is fixed at $\phi = -0.77$. Throughout, the variance of the noise (on a per increment basis) changes with the level of volatility, which is a well-documented feature in practice \citep*[e.g.,][]{bandi-russell:08a,kalnina-linton:08a}. In scenario 1 -- 3, the noise scales with the square root of the integrated variance $\text{IV}  = \int_{0}^{1} \sigma_{s}^{2} \text{d}s$, while in the last scenario it is a function of $\sigma_{t}$ and thus time-varying within each simulation.

The $\gamma$ parameter is the noise-to-volatility ratio of \citet*{oomen:06a} which determines the relative strength of the microstructure component. We assume $\gamma = 5$, which corresponds to heavy noise pertubation \citep[e.g.,][]{ait-sahalia-jacod-li:12a, christensen-oomen-podolskij:14a}.

We follow standard practice in the pre-averaging literature by setting $g(x) = \min(x,1-x)$ and $k_{n} = [ \theta \sqrt{n}]$. We fix the tuning parameter at $\theta = [1,2,3]/3$, which is in line with prior work \citep*[e.g.,][]{ait-sahalia-jacod-li:12a, christensen-kinnebrock-podolskij:10a, christensen-oomen-podolskij:14a}. The threshold for jump-truncation is implemented with $\bar{ \omega} = 0.24$ and $c = 4, 5, 6$ as a robustness check.\footnote{The choice of minimal value for $c$ follows \citet*{li-todorov-tauchen:17b}. If this parameter is chosen too small, the threshold becomes rather narrow and the truncation device starts to eliminate continuous log-returns (drawn from states of high stochastic volatility), which instills a downward bias in the truncated pre-averaged bipower variation, as opposed to the non-truncated pre-averaged realized variance. This leads to a widening of their difference, a measure of the quadratic variation of the jump component, and hence an overrejection under the null.}

We compute $\Sigma_{n}^{*}$ with $L = 10,15,20$ and $p = 10,15,20$ following the guidelines laid out in \citet*{christensen-podolskij-thamrongrat-veliyev:17a}. However, the rejection rates of the jump test statistic are not sensitive to the concrete choice of tuning parameters. We therefore restrict attention to $L = 10$ and $p = 10$. The corresponding tables for other combinations of tuning parameters are available on request.

\begin{sidewaystable}[ht!]
\setlength{ \tabcolsep}{0.30cm}
\begin{center}
\caption{Rejection rate of jump test statistic ($L = 10$ and $p = 10$). \label{table:Jtest-L=10-p=10.tex}}
\smallskip
\begin{tabular}{lccccccccccccccc}
\hline \hline
 && \multicolumn{11}{c}{pre-averaging} && \multicolumn{2}{c}{5-minute sampling} \\
\cline{3-13} \cline{15-16}
 && \multicolumn{3}{c}{$\theta = 1/3$} && \multicolumn{3}{c}{$\theta = 1/2$}  && \multicolumn{3}{c}{$\theta = 1/1$}\\
 && $c=4$ & $c=5$ & $c=6$ && $c=4$ & $c=5$ & $c=6$ && $c=4$ & $c=5$ & $c=6$ && $p$ & $p^{*}$ \\
\cline{3-5} \cline{7-9} \cline{11-13} \cline{15-16}
\multicolumn{3}{l}{\textit{Panel A: Gaussian}}\\
No jump && 0.050 & 0.049 & 0.049 && 0.050 & 0.048 & 0.048 && 0.053 & 0.051 & 0.051 && 0.092 & 0.087 \\
$\beta = $ & 1.75 & 0.476 & 0.393 & 0.336 && 0.283 & 0.221 & 0.181 && 0.210 & 0.162 & 0.135 && 0.121 & 0.111 \\
& 1.50 & 0.832 & 0.754 & 0.687 && 0.610 & 0.507 & 0.422 && 0.454 & 0.354 & 0.283 && 0.169 & 0.153 \\
& 1.00 & 0.983 & 0.966 & 0.943 && 0.909 & 0.847 & 0.766 && 0.800 & 0.698 & 0.589 && 0.292 & 0.260 \\
& 0.50 & 0.997 & 0.994 & 0.990 && 0.976 & 0.955 & 0.914 && 0.928 & 0.872 & 0.785 && 0.398 & 0.342 \\
\multicolumn{3}{l}{\textit{Panel B: T-distributed}}\\
No jump && 0.119 & 0.097 & 0.087 && 0.059 & 0.055 & 0.054 && 0.053 & 0.051 & 0.051 && 0.092 & 0.089 \\
$\beta = $ & 1.75 & 0.522 & 0.436 & 0.377 && 0.287 & 0.222 & 0.183 && 0.211 & 0.164 & 0.134 && 0.121 & 0.112 \\
& 1.50 & 0.848 & 0.780 & 0.713 && 0.614 & 0.509 & 0.427 && 0.455 & 0.354 & 0.283 && 0.169 & 0.153 \\
& 1.00 & 0.989 & 0.972 & 0.951 && 0.911 & 0.846 & 0.769 && 0.802 & 0.699 & 0.588 && 0.292 & 0.268 \\
& 0.50 & 0.998 & 0.995 & 0.991 && 0.976 & 0.954 & 0.914 && 0.926 & 0.871 & 0.786 && 0.398 & 0.349 \\
\multicolumn{3}{l}{\textit{Panel C: Autocorrelated}}\\
No jump && 0.047 & 0.046 & 0.046 && 0.052 & 0.051 & 0.051 && 0.054 & 0.052 & 0.052 && 0.092 & 0.085 \\
$\beta = $ & 1.75 & 0.508 & 0.422 & 0.363 && 0.289 & 0.226 & 0.184 && 0.209 & 0.163 & 0.134 && 0.121 & 0.116 \\
& 1.50 & 0.855 & 0.785 & 0.721 && 0.618 & 0.513 & 0.431 && 0.455 & 0.356 & 0.286 && 0.169 & 0.147 \\
& 1.00 & 0.989 & 0.976 & 0.955 && 0.915 & 0.854 & 0.774 && 0.803 & 0.702 & 0.593 && 0.292 & 0.257 \\
& 0.50 & 0.998 & 0.996 & 0.993 && 0.977 & 0.958 & 0.919 && 0.928 & 0.875 & 0.789 && 0.398 & 0.330 \\
\multicolumn{3}{l}{\textit{Panel D: Heteroscedastic}}\\
No jump && 0.047 & 0.043 & 0.043 && 0.048 & 0.047 & 0.047 && 0.050 & 0.048 & 0.048 && 0.092 & 0.072 \\
$\beta = $ & 1.75 & 0.434 & 0.354 & 0.299 && 0.277 & 0.215 & 0.176 && 0.205 & 0.159 & 0.131 && 0.121 & 0.089 \\
& 1.50 & 0.793 & 0.713 & 0.635 && 0.597 & 0.492 & 0.405 && 0.450 & 0.349 & 0.281 && 0.169 & 0.118 \\
& 1.00 & 0.980 & 0.961 & 0.927 && 0.904 & 0.838 & 0.747 && 0.798 & 0.695 & 0.580 && 0.292 & 0.189 \\
& 0.50 & 0.996 & 0.993 & 0.986 && 0.975 & 0.950 & 0.905 && 0.926 & 0.865 & 0.780 && 0.398 & 0.246 \\
\hline \hline
\end{tabular}
\smallskip
\begin{scriptsize}
\parbox{0.98\textwidth}{\emph{Note.}
In the left-hand side of the table, we report the rejection rates of the feasible noise-robust pre-averaging jump test statistic in \eqref{equation:feasible-test-statistic}, i.e. the fraction of times (out of 10,000) $\mathcal{J}_{n}$ exceeds $\Phi^{-1}(1- \alpha)$,
where $\Phi( \cdot)$ is the standard normal distribution function and the nominal level of significance is $\alpha = 0.05$.
$\theta$ controls the length of the pre-averaging window $k_{n} = [\theta \sqrt{n}]$.
$c$ is the number of local diffusive standard deviations a pre-averaged log-return must exceed to be classified as a jump and truncated to zero.
$L$ is the number of subsamples in $\Sigma_{n}^{*}$, whereas $p$ is the block length (in multiples of $k_{n}$).
``No jump'' is the null hypothesis, while $\beta$ governs the activity level of the jump process under the alternative (larger values of $\beta$ imply jumps are smaller and resemble Brownian motion closer).
Each separate panel describes the structure of the microstructure noise.
In the right-hand side, we report the rejection rates of the \citet*{barndorff-nielsen-shephard:06a} jump test computed at a 5-minute sampling frequency (or $n = 78$).
The latter is not resistant to noise, so as a comparison we implement it both on the inaccessible noise-free log-price ($p$) and the observable contaminated log-price ($p^{*}$).
Further details are available in Appendix \ref{appendix:simulation}.
}
\end{scriptsize}
\end{center}
\end{sidewaystable}

\begin{sidewaystable}[ht!]
\setlength{ \tabcolsep}{0.30cm}
\begin{center}
\caption{Rejection rate of jump test statistic ($L = 10$ and $p = 10$).}
\label{table:Jtest-L=10-p=10-AJL.tex}
\smallskip
\begin{tabular}{lccccccccccccccc}
\hline \hline
 && \multicolumn{11}{c}{pre-averaging} && \multicolumn{2}{c}{5-minute sampling} \\
\cline{3-13} \cline{15-16}
 && \multicolumn{3}{c}{$\theta = 1/3$} && \multicolumn{3}{c}{$\theta = 1/2$}  && \multicolumn{3}{c}{$\theta = 1/1$}\\
 && $c=4$ & $c=5$ & $c=6$ && $c=4$ & $c=5$ & $c=6$ && $c=4$ & $c=5$ & $c=6$ && $p$ & $p^{*}$ \\
\cline{3-5} \cline{7-9} \cline{11-13} \cline{15-16}
\multicolumn{3}{l}{\textit{Panel A: Gaussian}}\\
No jump && 0.054 & 0.052 & 0.052 && 0.055 & 0.055 & 0.055 && 0.063 & 0.063 & 0.063 && 0.092 & 0.087 \\
$\beta = $ & 1.75 & 0.599 & 0.523 & 0.443 && 0.388 & 0.310 & 0.244 && 0.278 & 0.215 & 0.170 && 0.121 & 0.111 \\
& 1.50 & 0.905 & 0.851 & 0.769 && 0.727 & 0.617 & 0.480 && 0.571 & 0.437 & 0.316 && 0.169 & 0.153 \\
& 1.00 & 0.994 & 0.988 & 0.968 && 0.959 & 0.910 & 0.794 && 0.887 & 0.771 & 0.566 && 0.292 & 0.260 \\
& 0.50 & 0.998 & 0.998 & 0.996 && 0.993 & 0.983 & 0.940 && 0.975 & 0.924 & 0.773 && 0.398 & 0.342 \\
\multicolumn{3}{l}{\textit{Panel B: T-distributed}}\\
No jump && 0.196 & 0.190 & 0.188 && 0.049 & 0.048 & 0.047 && 0.060 & 0.059 & 0.059 && 0.092 & 0.089 \\
$\beta = $ & 1.75 & 0.348 & 0.315 & 0.276 && 0.349 & 0.277 & 0.215 && 0.269 & 0.205 & 0.163 && 0.121 & 0.112 \\
& 1.50 & 0.624 & 0.580 & 0.521 && 0.679 & 0.572 & 0.446 && 0.560 & 0.427 & 0.306 && 0.169 & 0.153 \\
& 1.00 & 0.855 & 0.833 & 0.806 && 0.930 & 0.878 & 0.767 && 0.876 & 0.758 & 0.555 && 0.292 & 0.268 \\
& 0.50 & 0.925 & 0.916 & 0.907 && 0.979 & 0.965 & 0.923 && 0.969 & 0.918 & 0.770 && 0.398 & 0.349 \\
\multicolumn{3}{l}{\textit{Panel C: Autocorrelated}}\\
No jump && 0.000 & 0.000 & 0.000 && 0.000 & 0.000 & 0.000 && 0.004 & 0.004 & 0.004 && 0.092 & 0.085 \\
$\beta = $ & 1.75 & 0.242 & 0.222 & 0.191 && 0.220 & 0.161 & 0.104 && 0.190 & 0.129 & 0.089 && 0.121 & 0.116 \\
& 1.50 & 0.641 & 0.600 & 0.538 && 0.573 & 0.462 & 0.336 && 0.485 & 0.357 & 0.239 && 0.169 & 0.147 \\
& 1.00 & 0.960 & 0.948 & 0.919 && 0.918 & 0.857 & 0.738 && 0.854 & 0.730 & 0.523 && 0.292 & 0.257 \\
& 0.50 & 0.995 & 0.995 & 0.992 && 0.987 & 0.974 & 0.925 && 0.968 & 0.911 & 0.760 && 0.398 & 0.330 \\
\multicolumn{3}{l}{\textit{Panel D: Heteroscedastic}}\\
No jump && 0.142 & 0.131 & 0.131 && 0.050 & 0.050 & 0.050 && 0.055 & 0.055 & 0.055 && 0.092 & 0.072 \\
$\beta = $ & 1.75 & 0.583 & 0.513 & 0.440 && 0.372 & 0.294 & 0.230 && 0.270 & 0.205 & 0.160 && 0.121 & 0.089 \\
& 1.50 & 0.875 & 0.821 & 0.741 && 0.703 & 0.589 & 0.458 && 0.561 & 0.430 & 0.308 && 0.169 & 0.118 \\
& 1.00 & 0.990 & 0.983 & 0.959 && 0.953 & 0.897 & 0.783 && 0.883 & 0.760 & 0.559 && 0.292 & 0.189 \\
& 0.50 & 0.998 & 0.997 & 0.994 && 0.991 & 0.981 & 0.935 && 0.974 & 0.920 & 0.768 && 0.398 & 0.246 \\
\hline \hline
\end{tabular}
\smallskip
\begin{scriptsize}
\parbox{0.98\textwidth}{\emph{Note.}
In the left-hand side of the table, we report the rejection rates of the noise-robust pre-averaging jump test statistic from \citet{ait-sahalia-jacod-li:12a}, i.e. the fraction of times (out of 10,000) $\mathcal{J}_{n}$ subceeds $\Phi^{-1}(\alpha)$,
where $\Phi( \cdot)$ is the standard normal distribution function and the nominal level of significance is $\alpha = 0.05$.
$\theta$ controls the length of the pre-averaging window $k_{n} = [\theta \sqrt{n}]$.
$c$ is the number of local diffusive standard deviations a pre-averaged log-return must exceed to be classified as a jump and truncated to zero.
``No jump'' is the null hypothesis, while $\beta$ governs the activity level of the jump process under the alternative (larger values of $\beta$ imply jumps are smaller and resemble Brownian motion closer).
Each separate panel describes the structure of the microstructure noise.
In the right-hand side, we report the rejection rates of the \citet*{barndorff-nielsen-shephard:06a} jump test computed at a 5-minute sampling frequency (or $n = 78$).
The latter is not resistant to noise, so as a comparison we implement it both on the inaccessible noise-free log-price ($p$) and the observable contaminated log-price ($p^{*}$).
Further details are available in Appendix \ref{appendix:simulation}.
}
\end{scriptsize}
\end{center}
\end{sidewaystable} 

In Table \ref{table:Jtest-L=10-p=10.tex}, we report the properties of the jump test statistic at the 5\% significance level. Under the null $\mathcal{H}_{0}$ of no jumps, the rejection rates are close to the to the nominal level of significance. In Panel B, we observe a slight overrejection for $t$-distributed noise and minimal pre-averaging, which is insufficient to combat the large outliers generated by this distribution. However, increasing $\theta$ solves the problem.

Next, consider the results conducted under the alternative, $\mathcal{H}_{a}$. A near-perfect power close to 100\% is recorded for the lowest jump activity indexes. As $\beta$ increases, however, the rejection rate goes down. This reduction is consistent with \citet*{ait-sahalia-jacod-li:12a} and happens because the larger is $\beta$, the more the sample path of $r_{t}^{d}$ resembles a Brownian motion (with smaller and more erratic increments, cf. Figure \ref{figure:simulation-illustration}). This makes it tough for the test to discriminate between $\mathcal{H}_{0}$ and $\mathcal{H}_{a}$ at least over discrete intervals of fixed length $\Delta_{n}$. Moreover, we see that power is a decreasing function of $c$. The intuition is that as $c$ increases, $\check{BV}_{n}^{*}$ and $\Sigma_{n}^{*}$ are less jump resistant. This lack of robustness tends to reduce power, as also emphasized in the no-noise version of the test from \cite*{barndorff-nielsen-shephard:06a}.

Finally, we observe that the rejection rates are negatively related to $\theta$. The latter controls the pre-averaging horizon and too much smoothing reduces our ability to detect jumps on the trajectory of $p$. On the other hand, a larger pre-averaging window makes the estimator robust against more complicated noise structures than assumed here.

The right columns of Table \ref{table:Jtest-L=10-p=10.tex} compare our jump test to the noise-free version from \citet*{barndorff-nielsen-shephard:06a} implemented with a 5-minute realized variance and bipower variation. The latter are not robust to noise. As a consequence of this, the size and power of the standard jump test are distorted. In particular, the test is oversized under the null (rejects too often), while it lacks power under the alternative (rejects too little). This is true regardless of whether we allow the test statistic access to the latent efficient (no noise) log-price, $p$, or whether the statistic is computed from the observable noisy log-price, $p^{*}$. This finding is consistent with the empirical results in Table \ref{table:rv-descriptive}.

In Table \ref{table:Jtest-L=10-p=10-AJL.tex}, we look at the rejection rates of the noise-robust jump test statistic from \citet{ait-sahalia-jacod-li:12a}. As our testing procedure, it is based on pre-averaging. However, the main building block is a ratio of power variation estimators, as opposed to the difference of multipower variation estimators advocated here. This leads to small discrepancies in the implementation, and their test statistic appears a bit more sensitive to the structure of the noise. In particular, as with our test it rejects too often (even more so) with $t$-distributed noise, but it also does that with heteroscedastic noise. Here, increasing the pre-averaging horizon again removes the effect. Meanwhile, their test is rather conservative under $\mathcal{H}_{0}$ with conditionally autocorrelated noise. It prevails even with a larger pre-averaging window and leads to a sustained drop in the rejection rate under $\mathcal{H}_{a}$, yielding a lower power. We note that this setting was not inspected in \citet{ait-sahalia-jacod-li:12a}, since it falls outside of their theoretical framework. Apart from that, the conclusion of their test is close to that of our test.

Our empirical application follows the parameter settings above with $\theta = 1/2$ for pre-averaging, $c = 5$ for jump truncation, and $L = p = 10$ for the subsampler.\footnote{In practice, fixing $\theta$ means $k_{n}$ is time-varying since it changes with the actual number of high-frequency data available in each trading session, $n$. As a robustness check, we also experimented with fixed $k_{n} = 25$ and $k_{n} = 50$ without any noticeable change in the results.}

\clearpage

\section{Trading activity in the pre- and after-hours market} \label{appendix:pre-market}

Our analysis of earnings announcements concentrates on the after-hours market and excludes announcements that occur in the pre-market. We make this choice because the trading volume of announcing firms tend to be much smaller in the pre-market than in the after-hours market as can be seen by comparing Panel A and B in Figure \ref{figure:sp500-volume.eps}, see also the contrast between Table \ref{table:sp500-volume-pm.tex} and Table \ref{table:sp500-volume-am.tex}. In general, for the most liquid stocks trading volume on earnings announcement days is around 5--10 times larger during the after-hours session that during the pre-market session.

\setcounter{table}{0}
\setcounter{figure}{0}

\begin{figure}[ht!]
\begin{center}
\caption{S\&P 500 extended trading session transaction count.}
\label{figure:sp500-volume.eps}
\begin{tabular}{cc}
\small{Panel A: Pre-market trading (6:00am--9:30am).} & \small{Panel B: After-hours market (4:00pm--6:30pm).} \\
\includegraphics[height=8cm,width=0.48\textwidth]{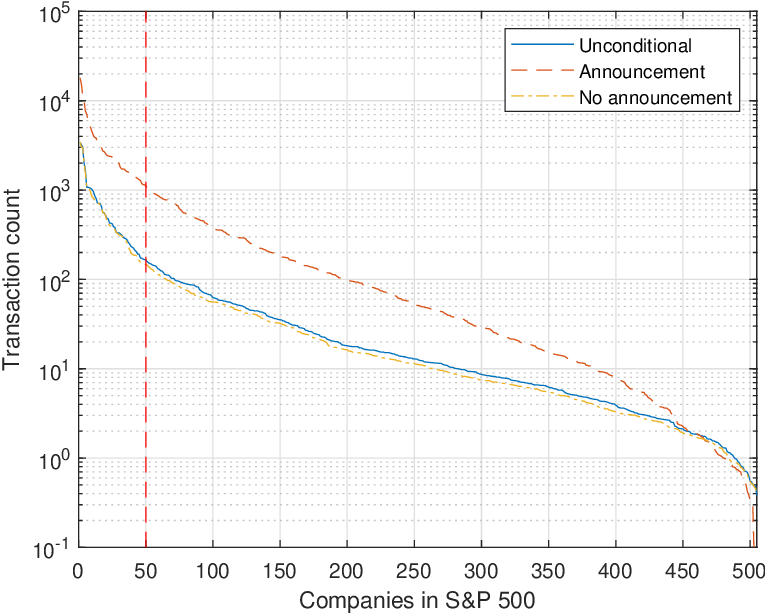} &
\includegraphics[height=8cm,width=0.48\textwidth]{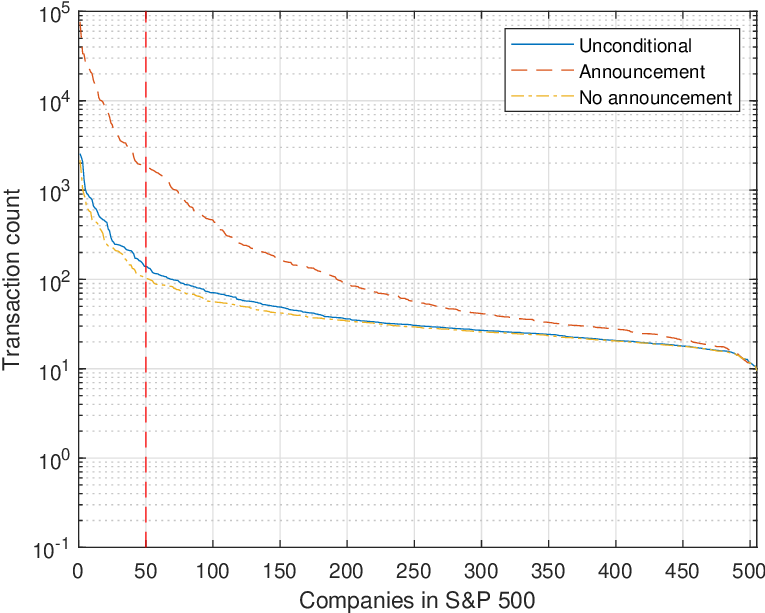}
\end{tabular}
\begin{scriptsize}
\parbox{\textwidth}{\emph{Note.} We download NYSE Trade and Quote (TAQ) high-frequency data for the constituent members of the S\&P 500 index as of 12/31/2020. We sort the companies by their average transaction counts for the sample period 06/02/2008--12/31/2020. Panel A shows the associated distribution for the pre-market session (6:00am--9:30am), whereas Panel B is for the after-hours session (4:00pm--6:30pm). Both are reported on a log-scale. We further split the unconditional transaction count distribution into days with and without earnings announcements. The vertical red dashed line indicates the transaction count for the fiftieth most liquid stock.}
\end{scriptsize}
\end{center}
\end{figure}

\begin{sidewaystable}[p!]
\begin{footnotesize}
\setlength{ \tabcolsep}{0.15cm}
\begin{center}
\caption{S\&P 500 companies by pre-market (6:00am--9:30am) trading activity. \label{table:sp500-volume-am.tex}}
\vspace*{-0.25cm}
\begin{tabular}{lrrrrrrrrrrrrrrrrrrrr}
\hline \hline
& \multicolumn{7}{c}{conditional on no announcement} && \multicolumn{7}{c}{conditional on announcement} && \multicolumn{4}{c}{announcement information}\\
\cline{2-8} \cline{10-16} \cline{18-21}
& & & & \multicolumn{4}{c}{quantile} & & & & & \multicolumn{4}{c}{quantile}\\
\cline{5-8} \cline{13-16}
ticker & mean & fraction & std. & 0.25 & 0.50 & 0.75 & 0.99 & & mean & fraction & std. & 0.25 & 0.50 & 0.75 & 0.99 & & $n_{ \text{EA}}$ & time & $r_{30m}^{ \text{EA}}$ & $z_{ \text{EPS}}$ \\
\hline
GE & 840 & 0.49 & 2,502 & 150& 334& 705& 8,395 & & 11,516 & 3.66 & 9,239 & 4,195 & 7,776 & 17,713 & 37,538 & & 49 & 6:30am & $\underset{(1.882)}{0.752}$ & $\underset{(1.648)}{1.209}$
\\
BAC & 2,990 & 1.25 & 6,085 & 656& 1,412& 3,157& 25,089 & & 18,007 & 4.86 & 14,185 & 6,815 & 11,798 & 23,330 & 52,800 & & 50 & 7:00am & $\underset{(2.258)}{-0.530}$ & $\underset{(1.772)}{0.934}$
\\
WMT & 147 & 0.18 & 544 & 12& 32& 87& 1,908 & & 4,954 & 2.99 & 7,725 & 1,237 & 2,377 & 4,912 & 46,543 & & 50 & 7:00am & $\underset{(2.220)}{0.327}$ & $\underset{(1.648)}{1.226}$
\\
TGT & 65 & 0.11 & 363 & 2& 9& 35& 792 & & 3,843 & 2.59 & 4,889 & 809 & 2,406 & 4,853 & 23,711 & & 50 & 6:30am & $\underset{(1.836)}{-0.119}$ & $\underset{(2.354)}{1.494}$
\\
BBY & 56 & 0.10 & 649 & 3& 9& 25& 560 & & 5,447 & 4.14 & 3,525 & 3,101 & 5,409 & 7,280 & 20,275 & & 34 & 7:00am & $\underset{(3.572)}{-0.872}$ & $\underset{(2.473)}{3.157}$
\\
CAT & 113 & 0.26 & 179 & 27& 61& 133& 751 & & 3,862 & 3.39 & 2,981 & 1,921 & 3,228 & 4,909 & 16,765 & & 50 & 7:30am & $\underset{(2.381)}{-0.335}$ & $\underset{(2.768)}{1.643}$
\\
KR & 33 & 0.05 & 359 & 0& 3& 14& 348 & & 2,050 & 1.41 & 2,593 & 248 & 796 & 3,392 & 9,053 & & 51 & 8:45am & $\underset{(4.313)}{-0.742}$ & $\underset{(1.869)}{1.213}$
\\
GM & 216 & 0.26 & 775 & 22& 53& 150& 2,645 & & 2,994 & 2.22 & 2,760 & 1,190 & 2,321 & 3,745 & 13,994 & & 40 & 7:30am & $\underset{(1.837)}{0.155}$ & $\underset{(1.850)}{1.541}$
\\
JPM & 387 & 0.26 & 920 & 82& 170& 373& 3,074 & & 5,993 & 2.75 & 4,624 & 3,069 & 4,374 & 7,394 & 22,071 & & 50 & 6:59am & $\underset{(0.963)}{-0.092}$ & $\underset{(1.608)}{1.596}$
\\
HD & 75 & 0.14 & 163 & 10& 26& 63& 759 & & 2,377 & 2.34 & 2,016 & 1,043 & 1,657 & 3,029 & 8,390 & & 50 & 6:00am & $\underset{(0.617)}{0.046}$ & $\underset{(1.872)}{2.012}$
\\
BA & 766 & 0.46 & 3,341 & 12& 32& 156& 14,894 & & 3,179 & 2.27 & 5,957 & 452 & 908 & 2,825 & 35,952 & & 49 & 7:30am & $\underset{(1.387)}{0.204}$ & $\underset{(1.931)}{1.271}$
\\
BIIB & 58 & 0.18 & 527 & 2& 9& 27& 748 & & 1,593 & 2.06 & 4,495 & 66 & 414 & 1,399 & 30,415 & & 50 & 7:15am & $\underset{(3.417)}{0.551}$ & $\underset{(1.993)}{1.769}$
\\
C & 1,566 & 0.74 & 4,860 & 132& 367& 1,218& 18,126 & & 10,189 & 3.31 & 14,867 & 2,714 & 5,757 & 11,118 & 87,389 & & 50 & 7:59am & $\underset{(1.605)}{0.096}$ & $\underset{(1.858)}{1.103}$
\\
LOW & 37 & 0.07 & 138 & 2& 7& 22& 462 & & 2,303 & 1.81 & 2,167 & 767 & 1,746 & 2,886 & 11,477 & & 50 & 6:00am & $\underset{(0.955)}{-0.074}$ & $\underset{(2.109)}{0.675}$
\\
MCD & 81 & 0.19 & 199 & 12& 28& 72& 829 & & 2,019 & 2.60 & 1,417 & 1,005 & 1,588 & 2,734 & 6,915 & & 50 & 7:58am & $\underset{(1.942)}{0.184}$ & $\underset{(2.213)}{1.190}$
\\
CVS & 64 & 0.09 & 374 & 1& 6& 30& 637 & & 2,007 & 1.18 & 4,258 & 127 & 390 & 1,917 & 21,491 & & 49 & 7:00am & $\underset{(0.900)}{-0.017}$ & $\underset{(2.139)}{1.856}$
\\
VZ & 150 & 0.15 & 741 & 20& 45& 105& 1,477 & & 1,649 & 1.34 & 1,363 & 777 & 1,256 & 2,244 & 6,216 & & 49 & 7:30am & $\underset{(0.680)}{-0.132}$ & $\underset{(0.965)}{0.598}$
\\
WFC & 363 & 0.18 & 1,907 & 36& 86& 221& 3,786 & & 4,484 & 1.74 & 4,485 & 1,738 & 3,165 & 4,726 & 21,858 & & 50 & 8:00am & $\underset{(2.463)}{-0.645}$ & $\underset{(1.737)}{0.316}$
\\
PG & 64 & 0.11 & 157 & 11& 29& 66& 524 & & 1,472 & 1.43 & 1,550 & 336 & 1,122 & 1,782 & 8,064 & & 48 & 7:00am & $\underset{(1.410)}{0.110}$ & $\underset{(1.842)}{1.642}$
\\
F & 756 & 0.62 & 1,566 & 97& 263& 745& 6,792 & & 7,702 & 3.16 & 10,035 & 747 & 2,526 & 10,591 & 45,122 & & 49 & 7:00am & $\underset{(4.179)}{-0.035}$ & $\underset{(1.971)}{1.098}$
\\
GS & 187 & 0.34 & 630 & 28& 72& 177& 1,639 & & 4,095 & 3.98 & 4,096 & 1,477 & 2,646 & 5,155 & 15,568 & & 45 & 7:35am & $\underset{(1.443)}{-0.447}$ & $\underset{(2.145)}{2.107}$
\\
MS & 185 & 0.15 & 1,034 & 15& 45& 132& 1,448 & & 3,340 & 2.11 & 2,721 & 1,446 & 2,449 & 4,361 & 13,841 & & 47 & 7:15am & $\underset{(2.410)}{0.236}$ & $\underset{(2.258)}{1.753}$
\\
CCL & 686 & 0.69 & 3,379 & 21& 68& 209& 11,512 & & 2,057 & 2.48 & 2,096 & 1,024 & 1,488 & 2,386 & 13,470 & & 50 & 9:15am & $\underset{(3.970)}{0.162}$ & $\underset{(1.993)}{2.626}$
\\
DLTR & 16 & 0.06 & 185 & 1& 3& 10& 177 & & 1,150 & 1.51 & 1,574 & 191 & 441 & 1,737 & 6,894 & & 50 & 7:30am & $\underset{(1.632)}{0.143}$ & $\underset{(2.314)}{1.364}$
\\
\hline \hline
\end{tabular}
\smallskip
\begin{scriptsize}
\parbox{0.98\textwidth}{\emph{Note.}
In the left-hand side, we show descriptive statistics of the most liquid companies from the S\&P 500 index based on their pre-market (6:00am--9:30am) trading volume, which is separated into no announcement days (columns 2--8) and announcement days (columns 9--15).
Ticker is the stock listing symbol.
SIC is the standard industrial classification code.
Mean and Std. are the transaction count sample average and standard deviation.
Quantile are selected quantiles from the empirical transaction count distribution.
Fraction is the number of transactions executed in the pre-market session divided by the total transaction count for the extended trading session from 6:00am--4:00pm.
In the right-hand side, we report announcement information.
\#EA is the number of earnings announcements.
Time is the modal announcement time (Eastern Standard Time).
$z_{ \text{EPS}}$ and $r_{1m}^{ \text{EA}}$ are the sample average standardized unexpected earnings and one-minute post-announcement return (standard deviation across announcements reported below in parenthesis).
}
\end{scriptsize}
\end{center}
\end{footnotesize}
\end{sidewaystable} 

\clearpage

\section{A close--to--open return-based jump test} \label{appendix:close-to-open}

\setcounter{table}{0}
\setcounter{figure}{0}

Our noise-robust jump test requires a large number of high-frequency data and is best suited for assets that are actively traded in the after-hours market. This condition is often not valid for smaller and more illiquid firms with low trading volume. In this appendix, we propose an alternative inference strategy based on close-to-open information that relies on the opening price after the announcement relative to the closing price prior to the announcement. In the parlance of Section \ref{section:ea} in the main text, the idea is that the difference between the pre-announcement closing price, i.e. $P_{t-1}^{c} = \mathbb{E}^{ \mathbb{Q}} \left[v_{t} \mid \mathcal{I}_{{t-1}^{c}}^{ \text{public}} \right]$ and the post-announcement opening price next morning, i.e. $P_{t}^{o} = \mathbb{E}^{ \mathbb{Q}} \left[v_{1} \mid \mathcal{I}_{{t}^{o}}^{ \text{public}} \right]$, reflects the pricing implications of the announcement. The test can be viewed as an adaptation of the \citet{lee-mykland:08a} jump test to low-frequency data. It can be applied to a much broader cross-section of securities regardless of their after-hours trading activity and, while it does not possess the power of our noise-robust jump test based on tick-by-tick data, in practice we find that it identifies the majority of earnings announcement days as containing price jumps. Other advantages of this approach is that it (i) is suitable for both pre-market and after-hours releases, (ii) does not require knowledge of the exact announcement time, and (iii) is not impaired if trading is halted pending the announcement.

To explain how the test is designed, we let $r_{t}^{co} = p_{t}^{o} - p_{t-1}^{c}$ denote the close-to-open log-return of a stock on day $t$, where $p_{t}^{o} = \log(P_{t}^{o})$ and $p_{t-1}^{c} = \log(P_{t-1}^{c})$ are the opening and closing log-price on day $t$ and $t-1$. We assume $r_{t}^{co}$ follows a two-state process, where the state reflects whether there is an announcement during the overnight period as captured by the indicator variable $EA_{t}$:
\begin{equation*}
r_{t}^{co} \mid EA_{t} \sim
\begin{cases}
N(0, \sigma^{2}), & EA_{t} = 0,\\[0.10cm]
N(0, \sigma^{2}) + J_{t}, & EA_{t} = 1,
\end{cases}
\end{equation*}
where $J_{t} \sim N(0, \sigma_{EA}^{2})$ is an independent jump with $\sigma_{EA}^{2} \geq 0$.

Conditional on no announcement ($EA_{t} = 0$), close-to-open returns follow a centered normal distribution with variance $\sigma^{2}$ while conditional on an announcement ($EA_{t} = 1$), returns are further affected by a jump so that $r_{t}^{co} \mid \textit{EA}_{t} = 1 \sim N(0, \sigma^{2} + \sigma_{EA}^{2})$.

To test the null hypothesis of no announcement effect, ($J_{t} = 0$ or $\sigma_{EA}^{2} = 0$), we standardize the current close-to-open return with an estimator of the non-announcement standard deviation computed from the previous $L$ close-to-open returns, $\hat{ \sigma}^{2} = L^{-1} \sum_{l=1}^{L} (r_{t-l}^{co})^{2}$:
\begin{equation} \label{equation:co-test}
Z_{t} = \frac{r_{t}^{co}}{ \hat{ \sigma}}.
\end{equation}

Under the null hypothesis $\mathcal{H}_{0}: \sigma_{EA}^{2} = 0$, the distribution of $Z_{t}$ approaches a standard normal as $L \rightarrow \infty$, i.e. $Z_{t} \overset{d}{ \longrightarrow} N(0,1)$. Under the alternative hypothesis $\mathcal{H}_{a}: \sigma_{EA}^{2} > 0$, $Z_{t} \overset{d}{ \longrightarrow} N(0,1+ \sigma_{EA}^{2}/ \sigma^{2})$, which is overdispersed compared to the standard normal. This implies that $Z_{t}$ has a higher probability of rejecting the null hypothesis than the nominal level under the alternative. However, it does not define a consistent test, because power is not going to one.\footnote{In line with \citet{lee-mykland:08a}, the time horizon has to shrink to zero to get a consistent test. In particular, assume that under the null the variance of $r_{t}^{co}$ per unit of time is $\sigma^{2}$. Then, under the alternative, $Z_{t} \sim N(0,1+ \sigma_{EA}^{2} / ( \Delta \sigma^{2}))$, where $\Delta = t_{open} - t-1_{close}$. Hence, $P( \mathrm{reject} \ \mathcal{H}_{0} \mid \mathcal{H}_{a}) \rightarrow 1$ as $\Delta \rightarrow 0$. In our context, this limit is not feasible, however, due to the fixed placement of the opening and closing.} Nonetheless, if $\sigma_{EA}^{2}$ is large compared to $\sigma^{2}$, we can get close.\footnote{The power curve is readily calculated as $P( \mathrm{reject} \ \mathcal{H}_{0} \mid \mathcal{H}_{a}) = f( \sigma_{EA}^{2}/ \sigma^{2}) =  1-P(z_{ \alpha/2} < N(0,1+ \sigma_{EA}^{2}/ \sigma^{2}) < z_{1- \alpha/2})$. For example, if $\sigma_{EA}^{2}/ \sigma^{2} = 1,2,3$, the probability of rejecting the null at the $\alpha = 1\%$ nominal level of significance is 19.78\%, 39.06\%, and 51.96\%.} Arguably, in practice the variance of the announcement effect is much bigger than the variance of the regular non-announcement news component, so the empirical rejection rate for $Z_{t}$ across many announcement days may still serve as an effective  indicator of a jump effect.

To get a consistent test, we conduct a two-sample test for equality of variance over announcement and non-announcement days. In essence, this amounts to testing whether the variance of the standardized announcement close-to-open returns is equal to one. Within our framework, this can be done with a textbook variance ratio statistic evaluated through an $F$-distribution. However, in keeping with the model-free approach, we switch to the Brown-Forsythe test in our application, which is a more robust alternative.

To illustrate the much broader usage of our close-to-open return-based jump test, we download from CRSP open and close price information from the full set of members of the S\&P 500 index as of the redefinition on March 18, 2024. This list has 503 constituents, including the stocks analyzed in the main text.\footnote{We remove GOOG since its earnings announcement dates are double-counted with GOOGL.} The series are corrected for dividend distributions, stock splits, and other corporate actions affecting the price. We also obtain the placement of an earnings announcement either ``Before Market Open'', ``During Regular Trading'', or ``After Market Close'' from Zacks Investment Research. This suffices for the calculation of the close-to-open return test statistic which does not require knowledge about the exact timestamp. In addition, we obtain the individual S\&P 500 index weights from Bloomberg. The sample period is January 1, 2012 to December 31, 2023. Our empirical analysis implements the close-to-open return-based jump test by calculating $Z_{t}$ with $L = 22$, corresponding to estimating the standard deviation based on one month of previous close-to-open returns. However, our results are very robust to this choice.\footnote{In our implementation, we exclude previous announcement days from the calculation and employ the mean absolute deviation as a more robust estimator.}

\begin{figure}[ht!]
\begin{center}
\caption{Close-to-open return-based jump test.}
\label{figure:open-to-close}
\begin{tabular}{cc}
\small{Panel A: Test statistic (AAPL).} & \small{Panel B: Jump frequency.} \\
\includegraphics[height=8cm,width=0.48\textwidth]{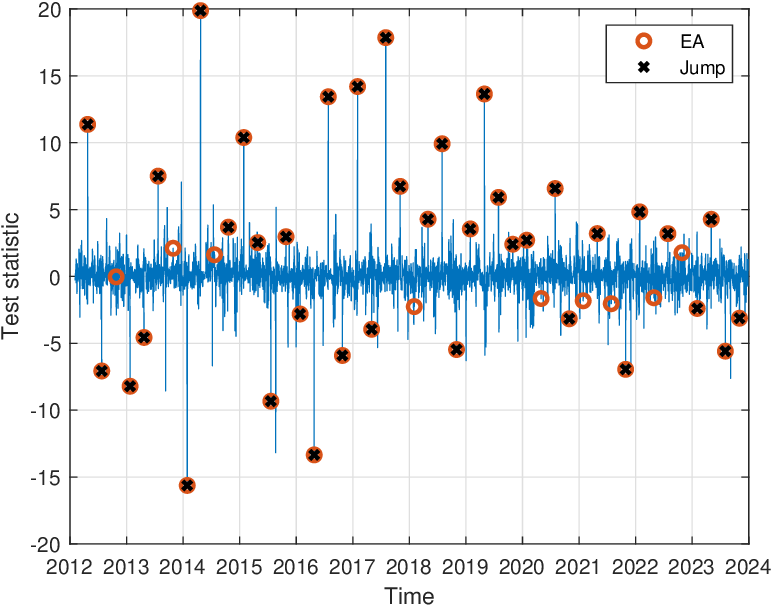} &
\includegraphics[height=8cm,width=0.48\textwidth]{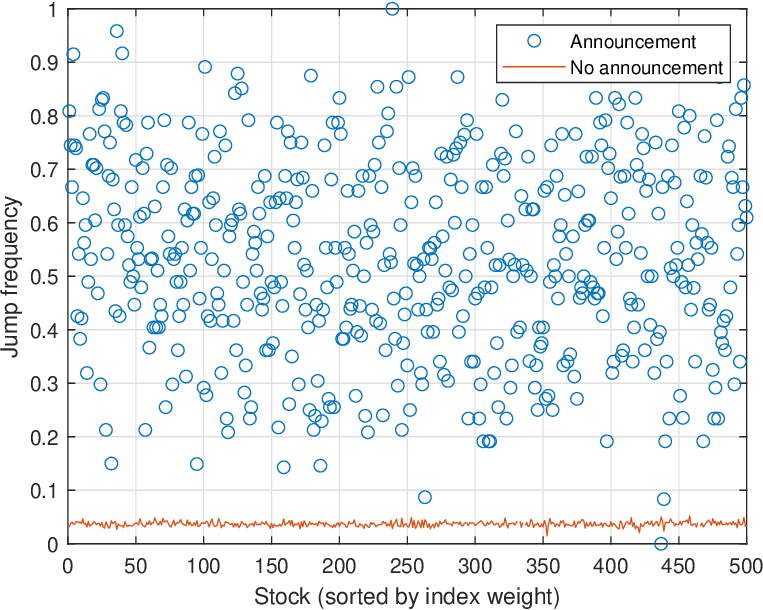}
\end{tabular}
\begin{scriptsize}
\parbox{\textwidth}{\emph{Note.} We report some properties of the close-to-open return-based jump test statistic. In Panel A, we plot the test statistic for Apple. An announcement date is highlighted with a circle (a significant test is further indicated with a cross). In Panel B, we show the rejection rate for announcement and non-announcement days across companies in the S\&P 500, sorted in descending order by index weight.}
\end{scriptsize}
\end{center}
\end{figure}

Figure \ref{figure:open-to-close} reports the outcome for our close-to-open return-based jump test. In Panel A, we plot the test statistic over time for Apple (AAPL), highlighting earnings announcement days with circles. At the 1\% level of significance, the close-to-open test flags 74.47\% of these announcements as containing a jump, as indicated with a crossed circle. This compares to 100\% detected by the noise-robust test using high-frequency data, thus demonstrating the lower power of the close-to-open jump test. In Panel B, we plot the average jump frequency across the members of the S\&P 500, sorted in descending order by index weight and separately for announcement and non-announcement days. Overall, the rejection rate on non-announcement days is 3.68\% which is indistinguishable from the 3.67\% reported for the pre-averaging jump test during after-hours trading sessions without announcements. Conversely, on announcement days the jump frequency is 52.86\%, notably lower than the 91.06\% obtained from our tick-by-tick jump test during after-hours trading sessions with announcements, but notably higher than the 40.40\% jump frequency reported by the 5-minute jump test.

The robust Brown-Forsythe test of equality of variance is firmly rejected across the board with $P$-values no larger than 0.0001. Thus, the close-to-open analysis provides compelling evidence of an announcement jump effect. We should note that there is no discernible serial correlation in the first or second moments of the close-to-open return-based jump test statistic. This is important since the variance ratio test presumes random sampling.

We next turn to the logistic regression for the stock price jump probability. As in Section \ref{section:logit}, we define an indicator variable based on our close-to-open return-based jump test statistic
\begin{equation*}
J_{it}^{oc} =
\begin{cases}
  1, & Z_{it} > q, \\
  0, & \text{otherwise.}
\end{cases}
\end{equation*}
We then estimate the logit model with $J_{it}^{oc}$ as dependent variable, but because information about after-hours trading activity is missing in the CRSP database, we exclude the variables $\log(V^{ \text{NAF}})$ and $\log(V^{ \text{AF}})$ from the analysis and only employ the remaining three covariates:
\begin{equation} \label{equation:logit-naf-co}
P \left(J_{jt}^{oc, \text{NAF}} = 1 \mid EA_{it} = 1 \right) = F \left( a + b_{1} IP_{ij} + b_{2} D^{ \text{EARLY}}_{it} + b_{3} D^{ \text{LATE}}_{it} \right).
\end{equation}
The results are reported in Table \ref{table:logit-co-jump-close-to-open}. They are entirely consistent with the results based on after-hours high-frequency data. In particular, industry proximity has a positive and significant spillover effect on the co-jump probability. Moreover, the coefficients on the early and late dummy variables are both significant with the correct positive and negative signs.

\begin{table}[ht!]
\setlength{\tabcolsep}{0.60cm}
\begin{center}
\caption{Logit estimates for the co-jump probability.}
\label{table:logit-co-jump-close-to-open}
\begin{tabular}{lrrrrr}
\hline
\hline
Variable & \multicolumn{1}{c}{(1)} & \multicolumn{1}{c}{(2)} & \multicolumn{1}{c}{(3)} & \multicolumn{1}{c}{(4)} \\ \hline
Intercept & $\underset{(0.002)}{-3.128}$ & $\underset{(0.002)}{-3.132}$ & $\underset{(0.002)}{-3.130}$ & $\underset{(0.002)}{-3.133}$ \\
$IP$ & & $\underset{(0.010)}{0.063}$ & & $\underset{(0.010)}{0.061}$ \\
$D^{ \text{EARLY}}$ & & & $\underset{(0.004)}{0.091}$ & $\underset{(0.004)}{0.091}$ \\
$D^{ \text{LATE}}$ & & & $\underset{(0.004)}{-0.092}$ & $\underset{(0.004)}{-0.092}$ \\
\hline
\hline
\end{tabular}
\begin{scriptsize}
\parbox{\textwidth}{\emph{Note.} We estimate the logit regression: $P \big(J_{jt}^{ \text{NAF}} = 1 \mid EA_{it} = 1 \big) = F \big( a + b_{1} IP_{ij} + b_{2} D^{ \text{EARLY}}_{it} + b_{3} D^{ \text{LATE}}_{it} \big)$, where $F$ is the logistic distribution function. $J_{jt}^{ \text{oc, NAF}}$ is equal to one if there is a jump in the price of the non-announcing firm (NAF) $j$ on day $t$, zero otherwise. $EA_{it}$ is equal to one if the announcing firm (AF) $i$ releases its quarterly earnings report on day $t$, zero otherwise. $IP_{ij}$ is an industry proximity score between 0 (no overlap) and 1 (perfect overlap) calculated from the SIC codes of company $i$ and $j$ following \citet{wang-zajac:07a}. $D^{ \text{EARLY}}_{it}$ and $D^{ \text{LATE}}_{it}$ are indicator variables of whether the AF discloses its fiscal statement early or late in the earnings cycle. The table reports parameter estimates of the full model in column (4) and of restricted versions in column (1)--(3). Standard errors are shown in parenthesis below the parameter estimate. The number of observations is 10,273,595.}
\end{scriptsize}
\end{center}
\end{table}

\clearpage

\section{Supplemental material for second half of stock sample} \label{appendix:supplemental}

\setcounter{table}{0}
\setcounter{figure}{0}

\begin{sidewaystable}[p!]
\begin{footnotesize}
\setlength{ \tabcolsep}{0.15cm}
\begin{center}
\caption{ \normalsize S\&P 500 companies by after-hours (4:00pm--6:30pm) trading activity.}
\label{table:sp500-volume-pm-supplemental.tex}
\vspace*{-0.25cm}
\begin{tabular}{lcrrrrrrrrrrrrrrrrrrrr}
\hline \hline
& & \multicolumn{7}{c}{Conditional on no announcement} && \multicolumn{7}{c}{Conditional on announcement} && \multicolumn{4}{c}{Announcement information}\\
\cline{3-9} \cline{11-17} \cline{19-22}
& & & & & \multicolumn{4}{c}{Quantile} & & & & & \multicolumn{4}{c}{Quantile}\\
\cline{6-9} \cline{14-17}
Ticker & SIC & Mean & Fraction & Std. & 0.25 & 0.50 & 0.75 & 0.99 & & Mean & Fraction & Std. & 0.25 & 0.50 & 0.75 & 0.99 & & \#EA & Time & $z_{ \text{EPS}}$ & $r_{1m}^{ \text{EA}}$\\
\hline
ULTA & 5990 & 33 & 0.37 & 52 & 8& 17& 41& 201 & & 3,368 & 11.32 & 4,549 & 427 & 2,576 & 4,311 & 22,330 & & 41 & 4:03pm & $\underset{(2.586)}{3.569}$ & $\underset{(5.370)}{0.773}$
\\
EA & 7372 & 57 & 0.19 & 95 & 22& 38& 71& 264 & & 3,113 & 5.63 & 3,677 & 1,014 & 1,951 & 4,140 & 19,349 & & 44 & 4:01pm & $\underset{(2.422)}{2.396}$ & $\underset{(3.016)}{-0.002}$
\\
ATVI & 7372 & 87 & 0.18 & 132 & 22& 41& 107& 535 & & 3,040 & 3.41 & 4,362 & 478 & 1,432 & 4,018 & 20,379 & & 41 & 4:05pm & $\underset{(2.358)}{3.147}$ & $\underset{(2.680)}{0.161}$
\\
AVGO & 3674 & 65 & 0.26 & 175 & 13& 33& 75& 472 & & 2,066 & 4.44 & 2,309 & 121 & 1,268 & 3,323 & 11,491 & & 43 & 4:05pm & $\underset{(1.192)}{1.051}$ & $\underset{(1.924)}{0.382}$
\\
EXPE & 4700 & 46 & 0.21 & 134 & 22& 31& 46& 219 & & 3,076 & 7.19 & 2,847 & 1,137 & 2,457 & 4,207 & 14,853 & & 43 & 4:01pm & $\underset{(1.591)}{0.339}$ & $\underset{(3.779)}{0.203}$
\\
ETSY & 7389 & 89 & 0.32 & 322 & 11& 32& 74& 658 & & 4,557 & 9.12 & 5,908 & 874 & 2,323 & 6,244 & 24,530 & & 23 & 4:05pm & $\underset{(2.276)}{0.106}$ & $\underset{(4.912)}{-0.259}$
\\
TXN & 3674 & 68 & 0.14 & 257 & 26& 40& 66& 394 & & 2,796 & 3.80 & 2,590 & 903 & 1,670 & 4,405 & 10,863 & & 47 & 4:30pm & $\underset{(2.074)}{1.724}$ & $\underset{(2.760)}{-0.435}$
\\
V & 7389 & 113 & 0.21 & 221 & 27& 49& 114& 989 & & 2,834 & 5.03 & 2,165 & 1,023 & 2,514 & 4,198 & 8,457 & & 46 & 4:05pm & $\underset{(1.501)}{1.953}$ & $\underset{(1.074)}{0.385}$
\\
GPS & 5651 & 52 & 0.14 & 189 & 15& 23& 38& 565 & & 1,884 & 2.96 & 3,531 & 300 & 575 & 1,874 & 20,942 & & 50 & 4:00pm & $\underset{(1.204)}{1.184}$ & $\underset{(4.359)}{0.136}$
\\
ADSK & 7372 & 36 & 0.19 & 31 & 20& 28& 44& 143 & & 1,754 & 4.80 & 2,006 & 446 & 1,108 & 2,228 & 8,943 & & 47 & 4:01pm & $\underset{(1.984)}{2.333}$ & $\underset{(4.168)}{-0.601}$
\\
BKNG & 4700 & 69 & 0.55 & 43 & 44& 58& 79& 233 & & 1,938 & 8.93 & 1,334 & 962 & 1,847 & 2,821 & 4,674 & & 37 & 4:01pm & $\underset{(1.830)}{2.204}$ & $\underset{(3.487)}{0.046}$
\\
NTAP & 3572 & 44 & 0.16 & 197 & 24& 33& 47& 148 & & 4,368 & 7.30 & 5,452 & 1,492 & 2,789 & 5,187 & 36,053 & & 46 & 4:01pm & $\underset{(2.429)}{1.859}$ & $\underset{(5.315)}{-0.486}$
\\
ANET & 3576 & 33 & 0.33 & 40 & 7& 23& 48& 128 & & 2,130 & 9.11 & 3,062 & 469 & 1,143 & 2,275 & 12,335 & & 25 & 4:05pm & $\underset{(2.007)}{3.993}$ & $\underset{(6.560)}{-0.561}$
\\
WYNN & 7011 & 71 & 0.30 & 145 & 24& 37& 60& 683 & & 1,779 & 5.70 & 1,638 & 735 & 1,449 & 2,389 & 8,046 & & 44 & 4:05pm & $\underset{(1.854)}{0.274}$ & $\underset{(2.710)}{-0.460}$
\\
AKAM & 7389 & 42 & 0.22 & 148 & 21& 30& 43& 179 & & 3,566 & 9.54 & 2,539 & 1,540 & 2,715 & 5,412 & 9,060 & & 50 & 4:01pm & $\underset{(2.636)}{2.513}$ & $\underset{(4.115)}{1.481}$
\\
TTWO & 7372 & 40 & 0.27 & 138 & 11& 18& 45& 223 & & 1,811 & 5.48 & 3,005 & 227 & 573 & 2,068 & 17,850 & & 14 & 4:05pm & $\underset{(2.283)}{2.176}$ & $\underset{(2.110)}{0.976}$
\\
WDC & 3572 & 57 & 0.18 & 112 & 18& 31& 62& 411 & & 1,678 & 3.25 & 2,573 & 283 & 770 & 2,108 & 14,628 & & 46 & 4:15pm & $\underset{(2.215)}{1.814}$ & $\underset{(2.780)}{0.303}$
\\
LRCX & 3559 & 51 & 0.23 & 65 & 20& 29& 59& 320 & & 1,072 & 2.80 & 1,797 & 73 & 140 & 1,599 & 7,800 & & 48 & 4:05pm & $\underset{(2.778)}{3.439}$ & $\underset{(1.618)}{0.412}$
\\
ALGN & 3842 & 27 & 0.32 & 33 & 9& 15& 39& 140 & & 1,314 & 5.59 & 2,245 & 94 & 432 & 1,406 & 10,699 & & 46 & 4:00pm & $\underset{(3.432)}{3.628}$ & $\underset{(5.337)}{-0.440}$
\\
AXP & 6199 & 65 & 0.16 & 118 & 24& 38& 62& 453 & & 1,701 & 2.95 & 1,832 & 278 & 1,224 & 1,849 & 7,815 & & 43 & 4:05pm & $\underset{(2.218)}{0.624}$ & $\underset{(1.956)}{0.057}$
\\
ISRG & 3842 & 36 & 0.54 & 89 & 14& 22& 44& 145 & & 1,780 & 13.10 & 1,334 & 888 & 1,287 & 2,194 & 7,118 & & 49 & 4:05pm & $\underset{(1.828)}{2.019}$ & $\underset{(3.707)}{1.427}$
\\
HPQ & 3570 & 84 & 0.13 & 559 & 29& 45& 71& 476 & & 5,664 & 5.39 & 5,608 & 1,195 & 3,892 & 9,440 & 22,986 & & 45 & 4:05pm & $\underset{(1.670)}{1.541}$ & $\underset{(2.362)}{0.335}$
\\
SWKS & 3674 & 54 & 0.23 & 74 & 21& 35& 59& 352 & & 1,571 & 4.01 & 1,255 & 761 & 1,353 & 1,939 & 5,865 & & 49 & 4:30pm & $\underset{(1.978)}{2.445}$ & $\underset{(3.110)}{0.753}$
\\
NOW & 7372 & 32 & 1.69 & 240 & 3& 11& 30& 210 & & 1,359 & 4.12 & 1,582 & 294 & 522 & 2,162 & 6,283 & & 34 & 4:10pm & $\underset{(2.253)}{2.485}$ & $\underset{(3.870)}{-0.500}$
\\
CSX & 4011 & 39 & 0.11 & 59 & 16& 28& 47& 190 & & 1,138 & 2.31 & 1,244 & 340 & 629 & 1,466 & 6,222 & & 48 & 4:02pm & $\underset{(1.819)}{1.728}$ & $\underset{(1.535)}{0.607}$
\\
\hline \hline
\end{tabular}
\smallskip
\begin{scriptsize}
\parbox{0.98\textwidth}{\emph{Note.}
In the left-hand side, we show descriptive statistics of the most liquid companies from the S\&P 500 index based on their after-hours market (4:00pm--6:30pm) trading volume, which is separated into no announcement days (columns 2--8) and announcement days (columns 9--15).
Ticker is the stock listing symbol.
SIC is the standard industrial classification code.
Mean and Std. are the transaction count sample average and standard deviation.
Quantile are selected quantiles from the empirical transaction count distribution.
Fraction is the number of transactions executed in the after-hours session divided by the total transaction count for the extended trading session from 9:30am--6:30pm.
In the right-hand side, we report announcement information.
\#EA is the number of earnings announcements.
Time is the modal announcement time (Eastern Standard Time).
$z_{ \text{EPS}}$ and $r_{1m}^{ \text{EA}}$ are the sample average standardized unexpected earnings and one-minute post-announcement return (standard deviation across announcements reported below in parenthesis).
}
\end{scriptsize}
\end{center}
\end{footnotesize}
\end{sidewaystable}

\clearpage

\begin{sidewaystable}[p!]
\setlength{\tabcolsep}{0.25cm}
\begin{center}
\caption{Sample average of realized variance, bipower variation, and jump proportion.}
\label{table:rv-descriptive-supplemental}
\vspace*{-0.25cm}
\begin{tabular}{lrrrrrrrrrrrrrrrrrrr}
\hline \hline
& \multicolumn{9}{c}{Panel A: Regular trading session (9:30am--4:00pm)} && \multicolumn{9}{c}{Panel B: Extended trading session (9:30am--6:30pm)} \\
\cline{2-10} \cline{12-20}
& \multicolumn{4}{c}{no EA} && \multicolumn{4}{c}{EA} && \multicolumn{4}{c}{no EA} && \multicolumn{4}{c}{EA} \\
\cline{2-5} \cline{7-10} \cline{12-15} \cline{17-20}
& RV & BV & JP & JF && RV & BV & JP & JF && RV & BV & JP & JF && RV & BV & JP & JF \\
ULTA & 17.7 & 16.4 & 5.7 & 2.4 && 29.2 & 28.1 & 5.4 & 7.3 && 17.7 & 16.4 & 5.7 & 2.6 && 80.4 & 36.9 & 50.4 & 78.0 \\
EA & 23.5 & 22.1 & 6.7 & 3.4 && 30.5 & 28.8 & 7.3 & 11.4 && 23.6 & 22.1 & 6.8 & 3.7 && 82.8 & 32.2 & 60.3 & 100.0 \\
ATVI & 23.0 & 21.2 & 9.2 & 5.6 && 28.3 & 26.3 & 10.4 & 17.1 && 23.1 & 21.2 & 9.5 & 6.4 && 69.7 & 23.6 & 58.6 & 90.2 \\
AVGO & 18.0 & 16.8 & 5.9 & 3.1 && 25.7 & 24.4 & 5.6 & 2.3 && 18.1 & 16.7 & 6.2 & 3.8 && 46.5 & 26.9 & 42.4 & 72.1 \\
EXPE & 24.0 & 22.7 & 5.4 & 2.9 && 29.9 & 28.5 & 6.4 & 4.7 && 24.0 & 22.6 & 5.6 & 3.2 && 110.8 & 41.7 & 61.6 & 93.0 \\
ETSY & 16.0 & 15.1 & 4.9 & 1.0 && 53.3 & 51.7 & 3.9 & 0.0 && 16.0 & 15.1 & 5.1 & 1.1 && 141.0 & 67.6 & 50.9 & 91.3 \\
TXN & 19.3 & 18.7 & 5.3 & 2.8 && 22.2 & 21.3 & 6.1 & 2.1 && 19.5 & 18.6 & 5.7 & 3.6 && 52.3 & 22.4 & 56.3 & 91.5 \\
V & 18.3 & 17.8 & 4.1 & 3.0 && 20.3 & 19.2 & 6.0 & 4.3 && 18.4 & 17.8 & 4.3 & 3.7 && 46.0 & 23.1 & 55.2 & 100.0 \\
GPS & 26.0 & 24.7 & 6.7 & 3.3 && 34.0 & 32.8 & 6.5 & 2.0 && 26.5 & 24.5 & 7.4 & 4.2 && 75.4 & 30.4 & 49.7 & 74.0 \\
ADSK & 21.7 & 20.6 & 5.5 & 3.4 && 29.1 & 28.1 & 5.0 & 2.1 && 21.7 & 20.6 & 5.5 & 3.5 && 74.8 & 25.4 & 58.7 & 91.5 \\
BKNG & 18.8 & 18.1 & 2.3 & 1.6 && 25.4 & 24.6 & 3.8 & 0.0 && 18.9 & 18.0 & 2.7 & 2.6 && 79.1 & 27.8 & 57.9 & 100.0 \\
NTAP & 22.0 & 20.9 & 6.1 & 3.3 && 26.8 & 25.5 & 6.8 & 4.3 && 22.1 & 20.9 & 6.2 & 3.4 && 93.2 & 25.9 & 64.7 & 95.7 \\
ANET & 12.2 & 11.8 & 2.6 & 0.4 && 29.6 & 28.6 & 3.7 & 0.0 && 12.2 & 11.8 & 2.6 & 0.5 && 86.4 & 28.0 & 57.2 & 92.0 \\
WYNN & 31.1 & 30.4 & 2.3 & 1.5 && 36.7 & 35.6 & 3.5 & 0.0 && 31.3 & 30.4 & 2.7 & 2.3 && 80.0 & 37.8 & 51.4 & 97.7 \\
AKAM & 21.5 & 20.6 & 4.8 & 2.0 && 29.2 & 28.2 & 4.9 & 0.0 && 21.6 & 20.5 & 5.0 & 2.3 && 108.2 & 49.1 & 54.7 & 100.0 \\
TTWO & 21.6 & 19.1 & 7.4 & 3.0 && 26.7 & 25.6 & 4.3 & 0.0 && 21.8 & 19.0 & 7.7 & 3.6 && 51.2 & 8.0 & 54.2 & 71.4 \\
WDC & 27.6 & 26.6 & 4.2 & 2.3 && 35.1 & 34.7 & 2.0 & 0.0 && 27.7 & 26.5 & 4.6 & 3.2 && 69.8 & 34.6 & 43.4 & 76.1 \\
LRCX & 23.0 & 22.1 & 4.1 & 2.4 && 25.7 & 23.8 & 7.9 & 4.2 && 23.0 & 22.1 & 4.2 & 2.8 && 38.0 & 23.5 & 30.6 & 45.8 \\
ALGN & 19.6 & 17.6 & 6.6 & 3.0 && 30.0 & 25.8 & 9.3 & 4.3 && 19.6 & 17.6 & 6.6 & 3.0 && 71.9 & 24.3 & 51.3 & 65.2 \\
AXP & 20.4 & 19.8 & 4.0 & 1.9 && 23.5 & 22.6 & 4.4 & 0.0 && 20.4 & 19.8 & 4.2 & 2.2 && 43.5 & 25.6 & 46.6 & 81.4 \\
ISRG & 16.2 & 15.3 & 3.9 & 2.6 && 24.6 & 23.6 & 4.4 & 0.0 && 16.3 & 15.2 & 4.0 & 2.7 && 68.3 & 27.6 & 49.6 & 93.9 \\
HPQ & 21.5 & 20.3 & 8.8 & 4.7 && 24.8 & 22.8 & 14.0 & 17.8 && 21.6 & 20.2 & 9.1 & 5.4 && 63.5 & 31.0 & 56.0 & 93.3 \\
SWKS & 26.1 & 24.9 & 4.0 & 1.7 && 31.8 & 29.5 & 7.0 & 6.1 && 26.2 & 24.6 & 4.7 & 2.9 && 70.8 & 26.3 & 58.9 & 93.9 \\
NOW & 16.6 & 15.7 & 5.1 & 1.7 && 30.0 & 28.6 & 5.7 & 2.9 && 16.6 & 15.7 & 5.2 & 1.8 && 72.5 & 24.7 & 55.3 & 79.4 \\
CSX & 21.1 & 20.2 & 6.1 & 2.9 && 25.3 & 24.1 & 6.6 & 2.1 && 21.1 & 20.2 & 6.2 & 3.0 && 45.8 & 20.1 & 53.4 & 83.3 \\
\hline
Average & 21.1 & 20.0 & 5.3 & 2.6 && 29.1 & 27.7 & 6.0 & 3.8 && 21.2 & 19.9 & 5.5 & 3.1 && 72.9 & 29.8 & 53.2 & 86.0 \\
\hline \hline
\end{tabular}
\smallskip
\begin{scriptsize}
\parbox{0.98\textwidth}{ \emph{Note.} 
This table reports the pre-averaged realized variance (RV) and pre-averaged bipower variation (BV).
We also construct the jump proportion (JP), which is defined as Jump Proportion = 1 - Bipower variation/Realized Variance.
We further compute the jump frequency (JF), which is derived from the jump indicator in \eqref{equation:jump-indicator}.
The realized variance and bipower variation are converted to annualized standard deviation in percent for convenience.
The full sample is split into days without (no EA) and with (EA) earnings announcements.
The measures are then averaged across subsamples by company.
In Panel A, we report the analysis for the regular trading session (9:30am--4:00pm), whereas Panel B displays the associated results for the extended trading session (9:30am--6:30am).
The bottom row reports the grand mean (``Average'') over all stock-days.
}
\end{scriptsize}
\end{center}
\end{sidewaystable}

\clearpage

\section{Evolution of the after-hours market}
\label{appendix:after-hours-market}

\setcounter{table}{0}
\setcounter{figure}{0}

In this appendix, we highlight key features of how the after-hours market has evolved over time. Our main finding is that trading volume has increased significantly from 2008 to 2020.

\subsection{Trading Volume}

The left panel in Figure \ref{figure:trading-volume} shows how the daily transaction counts evolved during the regular trading session and the after-hours market over our sample which runs from 06/02/2008 to 12/31/2020, using a log-scale to improve readability. The right panel shows the median number of shares traded. In both cases, we report a single daily value computed as a cross-sectional average over the 50 companies included in our empirical analysis.

While after-hours transaction volume remains smaller than the regular trading session counterpart, it increases notably over time particularly after the structural changes in the after-hours market design were introduced in 2015. The opposite shift is evident in the number of shares traded, which has trended down systematically from 2016 onward.

\begin{figure}[ht!]
\begin{center}
\caption{Transaction count and number of shares traded.}
\label{figure:trading-volume}
\begin{tabular}{cc}
\small{Panel A: Transaction count.} & \small{Panel B: Number of shares traded.} \\
\includegraphics[height=8cm,width=0.48\textwidth]{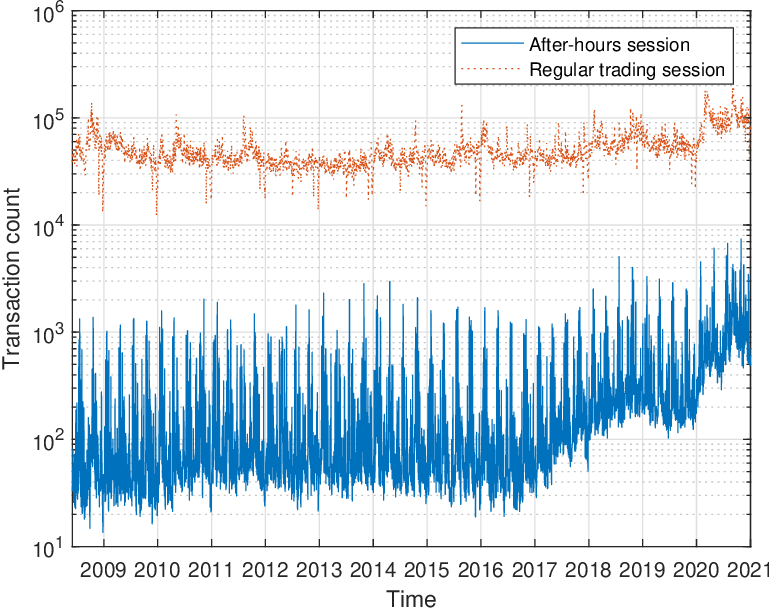} &
\includegraphics[height=8cm,width=0.48\textwidth]{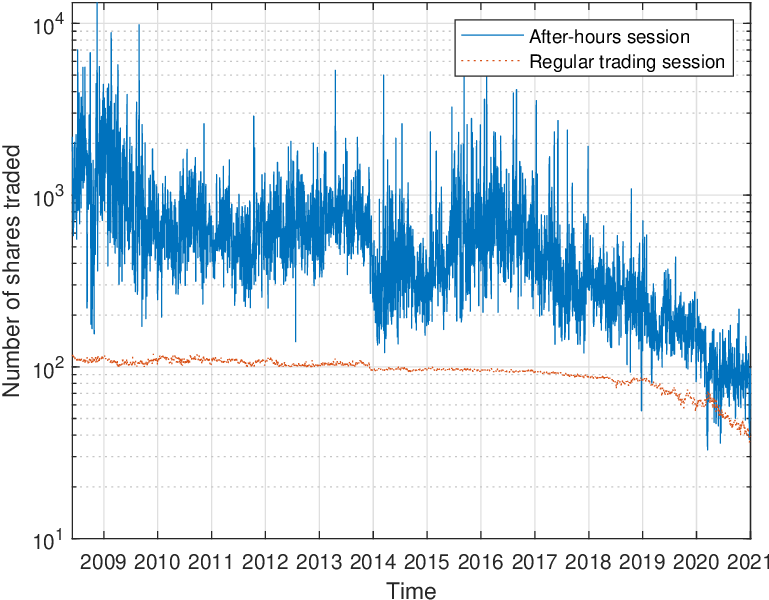} \\
\end{tabular}
\begin{scriptsize}
\parbox{\textwidth}{\emph{Note.} The figure plots transaction counts in Panel A and the number of shares traded in Panel B during the regular trading session and after-hours session. The statistics are computed daily for each firm in our sample and then averaged over the cross-section.}
\end{scriptsize}
\end{center}
\end{figure}

The left panel in Figure \ref{figure:relative-volume} shows the daily cross-sectional average transaction count in the after-hours market relative to the regular trading session. The relative trading volume increases from 0.17\% in 2008Q3--2009Q2 to 0.62\% in 2020. There are notable spikes in the series, which frequently exceeds 3\% and goes as high as 9\%. These surges occur mostly on days where large corporations publish their financial results. To corroborate this statement, the right panel in Figure \ref{figure:relative-volume} plots a kernel density estimate of the relative trading volume at the firm level on announcement days. To highlight the pronounced differences in trading activity, the figure includes the 99th percentile of the relative trading volume on non-announcement days (q$_{0.99}$), which falls far in the left tail of the distribution of relative trading volume on announcement days.

\begin{figure}[ht!]
\begin{center}
\caption{Cross-sectional and firm level relative transaction count.} \label{figure:relative-volume}
\begin{tabular}{cc}
\small{Panel A: Cross-sectional level} & \small{Panel B: Firm level} \\
\includegraphics[height=8cm,width=0.48\textwidth]{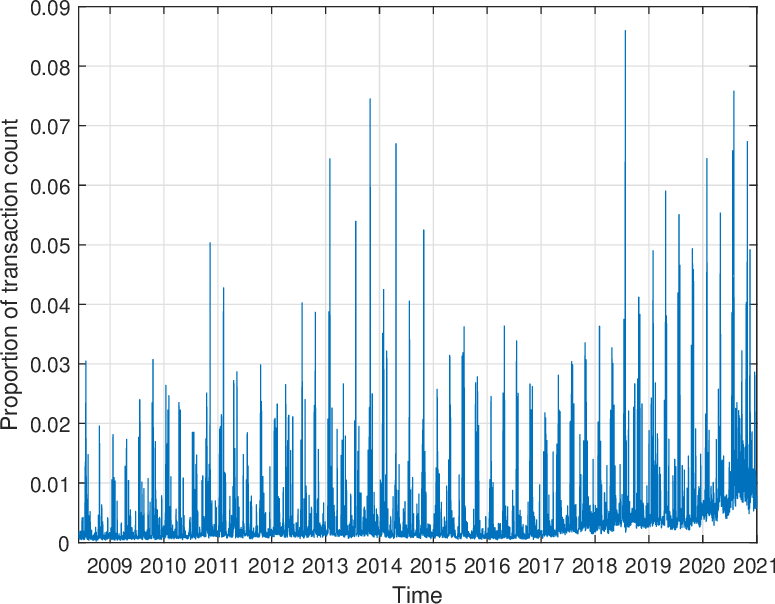} &
\includegraphics[height=8cm,width=0.48\textwidth]{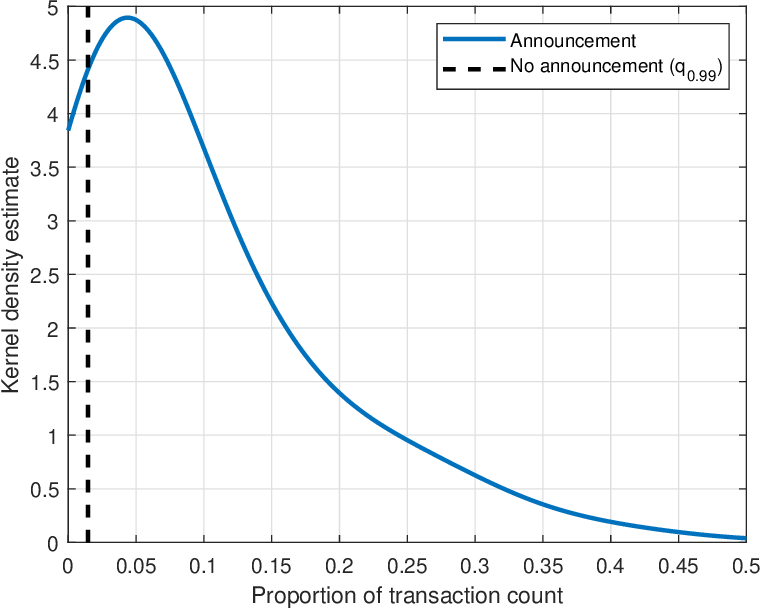}
\end{tabular}
\begin{scriptsize}
\parbox{\textwidth}{\emph{Note.} In Panel A, we plot the daily transaction count in the after-hours session relative to the transaction count in the regular trading session, both aggregated over the cross-section of firms in our sample. In Panel B, we show the relative transaction count distribution at the firm level on announcement days. The dashed line marks the 99th percentile relative transaction count on non-announcement days.}
\end{scriptsize}
\end{center}
\end{figure}

\subsection{Bid-Ask Spread}

In Panel A of Figure \ref{figure:bid-ask-spread}, we plot the daily cross-sectional average median quoted spread in basis points for the regular trading session and after-hours session, whereas Panel B reports the relative spread. The absolute spread (in basis points) is defined as
\begin{equation*}
\text{Spread} = 10000 \times \frac{ \mathrm{ask} - \mathrm{bid}}{ \mathrm{midquote}},
\end{equation*}
where $\mathrm{midquote} = (\mathrm{bid} + \mathrm{ask})/2$.

The relative spread is defined as $\text{Spread}^{ \text{after-hours}} / \text{Spread}^{ \text{regular trading}}$.

\begin{figure}[ht!]
\begin{center}
\caption{Bid-ask spread.}
\label{figure:bid-ask-spread}
\begin{tabular}{cc}
\small{Panel A: Absolute spread.} & \small{Panel B: Relative spread.} \\
\includegraphics[height=8cm,width=0.48\textwidth]{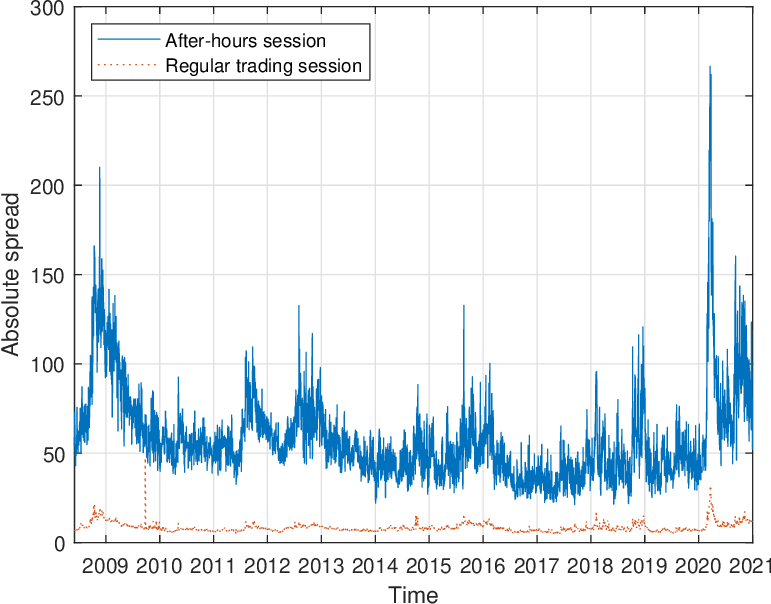} &
\includegraphics[height=8cm,width=0.48\textwidth]{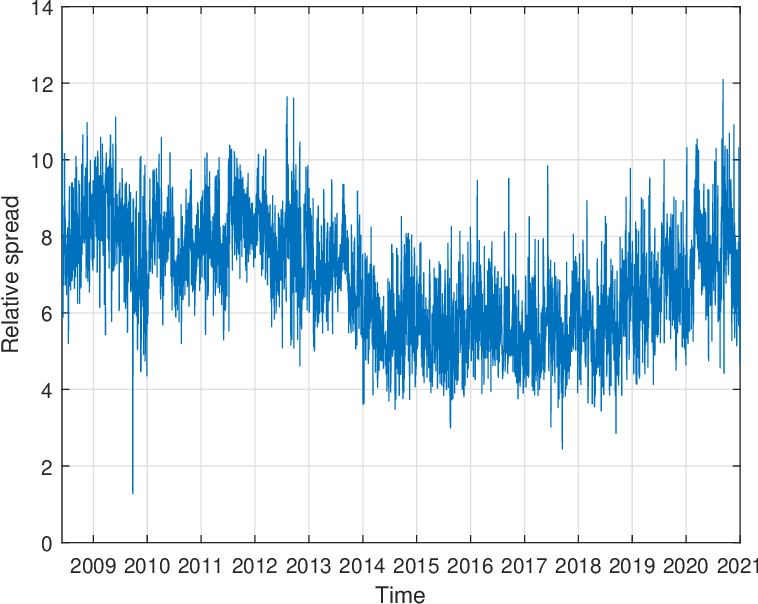} \\
\end{tabular}
\begin{scriptsize}
\parbox{\textwidth}{\emph{Note.} The figure plots the daily cross-sectional average median quoted bid-ask spread (in basis points) in the regular trading session and after-hours session in Panel A, while the relative spread is shown in Panel B.}
\end{scriptsize}
\end{center}
\end{figure}

The behavior of the quoted spread is consistent with a positive correlation between price volatility and trading costs. Spreads peak during the financial crisis in 2008--09 before starting a secular downward trend lasting until  2020 which saw spreads increase significantly and reach all-time highs after the outbreak of Covid-19. Median spreads in after-hours markets are quite stable, mostly falling in a range of four to ten times the spreads in the regular session.

\subsection{Price Discovery} \label{section:weighted-price-contribution}

To show how lengthy the price adjustment process is after an earnings announcement, we employ a standard measure of price discovery, namely the weighted price contribution (WPC) of \cite{barclay-warner:93a}. For company $i$ this is defined over small intraday intervals $t$:

\begin{equation} \label{equation:WPC}
WPC_{it} = \sum_{a=1}^{ \#EA_{i}} w_{a} \times \frac{r_{a,t}}{r_{a}}, \quad \text{for } t = 1, \dots, T_{A},
\end{equation}
where $r_{a,t}$ is the log-return for announcement $a$ over time interval $t$, $T_{A}$ is the number of time intervals, $r_{a} = \sum_{t=1}^{T_{A}} r_{a,t}$ is the total announcement log-return, $w_{a} = |r_{a}| \times \big( \sum_{a=1}^{\#EA_{i}} |r_{a}| \big)^{-1}$ is announcement $a$'s weight, and $\#EA_{i}$ is the number of earnings announcements for company $i$, as reported in Table \ref{table:sp500-volume-pm.tex}. The WPC measure starts at zero and ends at one, so the price discovery process is assumed to be completed at $T_{A}$.\footnote{We also calculated \cite{hasbrouck:95a}'s information share (IS). Consistent with our results based on the WPC, the IS measure suggests that anywhere between 60--80\% of the price discovery is accounted for in the first five minutes after the announcement.}

\begin{figure}[ht!]
\begin{center}
\caption{Weighted price contribution.}
\label{figure:wpc}
\begin{tabular}{c}
\includegraphics[height=8cm,width=0.48\textwidth]{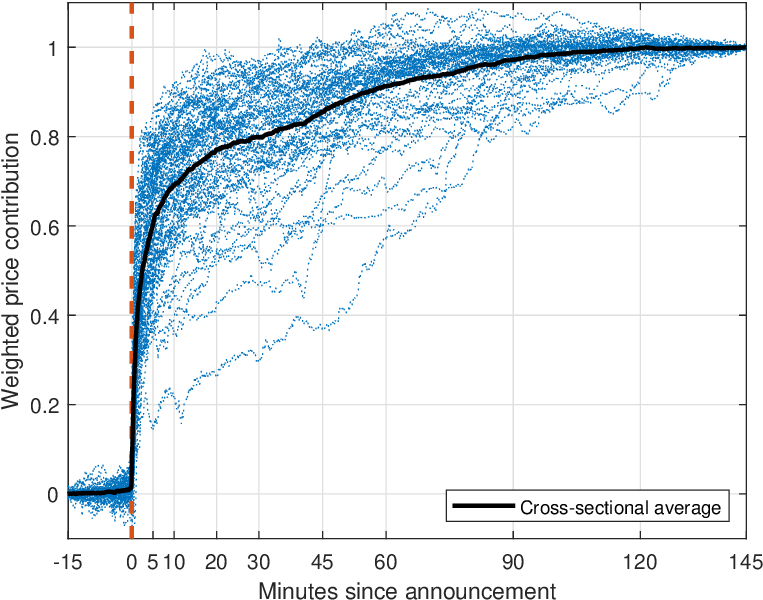}
\end{tabular}
\begin{scriptsize}
\parbox{\textwidth}{\emph{Note.} We plot the weighted price contribution (WPC) of \citet*{barclay-warner:93a}, which for company $i$ is defined over small intraday intervals $t$: $WPC_{it} = \sum_{a=1}^{ \#EA_{i}} w_{a} \times \frac{r_{a,t}}{r_{a}}$, for $t = 1, \dots, T_{A}$, where $r_{a,t}$ is the log-return for announcement $a$ over time interval $t$, $T_{A}$ is the total number of time intervals, $r_{a} = \sum_{t=1}^{T_{A}} r_{a,t}$ is the cumulative announcement log-return, $w_{a} = |r_{a}| \times \big( \sum_{a=1}^{ \#EA_{i}} |r_{a}| \big)^{-1}$ is the weight for announcement $a$, and $\#EA_{i}$ is the number of earnings announcements for company $i$, as reported in Table \ref{table:sp500-volume-pm.tex}. The returns cover a time interval from 15 minutes before to 145 minutes after each announcement with the event window broken into 5-second subintervals. The WPC is shown for each company as a blue dotted line, while the full black line is the cross-sectional average.}
\end{scriptsize}
\end{center}
\end{figure}

In Figure \ref{figure:wpc}, we plot the cumulative WPC measure for each of the 50 stocks in our analysis along with their cross-sectional average. The return covers a period from 15 minutes before to 145 minutes after each announcement with the event window broken into 5-second subintervals.\footnote{We also constructed the WPC measure over a longer window that begins 15 minutes prior to each announcement and extends through the overnight period and into next day's pre-market trading up to the opening of the stock exchange at 9:30am. The extended WPC curve is essentially flat after the end of the announcement day, suggesting that very little additional price discovery takes place here.} Implicitly, market prices at $t  = -15$ and $t = 145$ are taken as the efficient pre- and post-announcement prices. Since the typical company announces earnings at 4:05pm, the pre-announcement price is typically determined at 3:50pm during the very active regular trading session, while the post-announcement price is measured around 6:30pm. The WPC curves in Figure \ref{figure:wpc} suggest that price discovery is very fast with close to 70\% of the post-announcement price being discovered within five minutes of an announcement. This finding is consistent with \citet*{santosh:19a}.

Next, to examine how the price discovery process has shifted over time, Table \ref{table:wpc} reports the daily weighted price contribution (WPC) divided into four time intervals that are motivated by our empirical application and closely follow those used in \cite{barclay-hendershott:03a} and \citet*{jiang-likitapiwat-mcinish:12a}, namely pre-market (6:00am--9:30am), regular trading (9:30am--4:00pm), after-hours (4:00pm--6:30pm), and the overnight period (6:30pm--6:00am). In Panel A, the WPC measure calculated over the full sample and across all days indicates that 7.49\% of the price discovery process occurs in the pre-market, 73.19\% in regular trading, 4.64\% in the after-hours market, and 14.68\% overnight. On earnings announcement days, however, these fractions shift markedly as the 4:00pm--6:30pm segment accounts for 80.07\% of the WPC, with only 10.12\%, 6.87\%, and 2.95\% stemming from the remaining segments.

In Panels B and C of Table \ref{table:wpc}, we split these results into subsamples, covering 2008--2015 and 2016--2020. Over time, the after-hours segment has gained in importance. On earnings announcement days, it accounts for 85.39\% of the WPC in the 2016--2020 sample, nearly 10\% higher than the proportion during 2008--2015 (76.37\%). This increase has primarily come at the expense of the price discovery during the regular trading session and the overnight period, which on days with an earnings announcement has declined from 8.15\% and 4.71\% in the first sub-sample to 5.02\% and 0.4\%in the second sub-sample.

\begin{table}[ht!]
\setlength{ \tabcolsep}{0.20cm}
\begin{center}
\caption{Weighted price contribution.}
\label{table:wpc}
\vspace*{-0.25cm}
\begin{tabular}{lllll}
\hline \hline
& Pre-market & Regular trading & After-hours & Overnight \\
& 6:00am--9:30am & 9:30am--4:00pm & 4:00pm--6:30pm & 6:30pm--6:00am \\
\cline{2-5}
\textit{Panel A: Full sample} \\
All              & 0.0749    & 0.7319    & 0.0464    & 0.1468    \\
No announcement  & 0.0737    & 0.7638    & 0.0101    & 0.1524    \\
Announcement     & 0.1012*** & 0.0687*** & 0.8007*** & 0.0295*** \\
\\
\textit{Panel B: 2008 -- 2015}\\
All              & 0.0682    & 0.7571    & 0.0363    & 0.1383    \\
No announcement  & 0.0664    & 0.7890    & 0.0020    & 0.1426    \\
Announcement     & 0.1077*** & 0.0815*** & 0.7637*** & 0.0471*** \\
\\
\textit{Panel C: 2016 -- 2020}\\
All              & 0.0850    & 0.6944    & 0.0618    & 0.1589    \\
No announcement  & 0.0847    & 0.7264    & 0.0224    & 0.1666    \\
Announcement     & 0.0918    & 0.0502*** & 0.8539*** & 0.0040*** \\
\hline \hline
\end{tabular}
\smallskip
\begin{scriptsize}
\parbox{0.98\textwidth}{\emph{Note.} We calculate the weighted price contribution (WPC) of of \cite{barclay-warner:93a} over the four time intervals pre-market (9:30am--4:00pm), regular trading (4:00pm--6:30pm), after-hours (6:30pm--6:00am), and overnight (6:00am--9:30am). The WPC measure is described in Section  \ref{section:weighted-price-contribution}. The table shows the unconditional WPC (All), and the conditional WPC for non-announcement and announcement days. Panel A reports the full sample estimates, whereas Panels B--C report estimates separately for the subsamples 2008--2015 and 2016--2020. *, **, and *** denote that the non-announcement WPC is statistically different from the announcement WPC at the 10\%, 5\%, and 1\% level of significance, as inferred from a nonparametric permutation test with 10,000 shuffles.}
\end{scriptsize}
\end{center}
\end{table}

\clearpage

\section{Apple case study} \label{appendix:apple}

\setcounter{table}{0}
\setcounter{figure}{0}

This appendix illustrates our broader results using Apple as a case study. Apple is among the most important stocks throughout our entire sample and it has a broad coverage of analysts.

In Panel A of Figure \ref{figure:appl-extended-volume-variance}, we plot the proportion of the total volume (measured in transaction counts) traded during the after-hours session (4:00pm--6:30pm) relative to the extended trading session (9:30am--6:30pm). Liquidity in the after-hours market is typically small with an unconditional sample average value of 1.1\% of the total volume. However, there are clear spikes on earnings announcements days where the after-hours market is highly active and the relative trading volume exceeds 18\% on average on such days.

In Panel B of Figure \ref{figure:appl-extended-volume-variance}, we show the pre-averaged realized variance. The realized variance displays regularly occurring spikes that, consistent with the peaky trading volume in Panel A, almost always occur on days with earnings announcements.

\begin{figure}[ht!]
\begin{center}
\caption{Relative trading volume and pre-averaged realized variance of Apple.}
\label{figure:appl-extended-volume-variance}
\begin{tabular}{cc}
\small{Panel A: Relative trading volume.} & \small{Panel B: Pre-averaged realized variance.} \\
\includegraphics[height=8cm,width=0.48\textwidth]{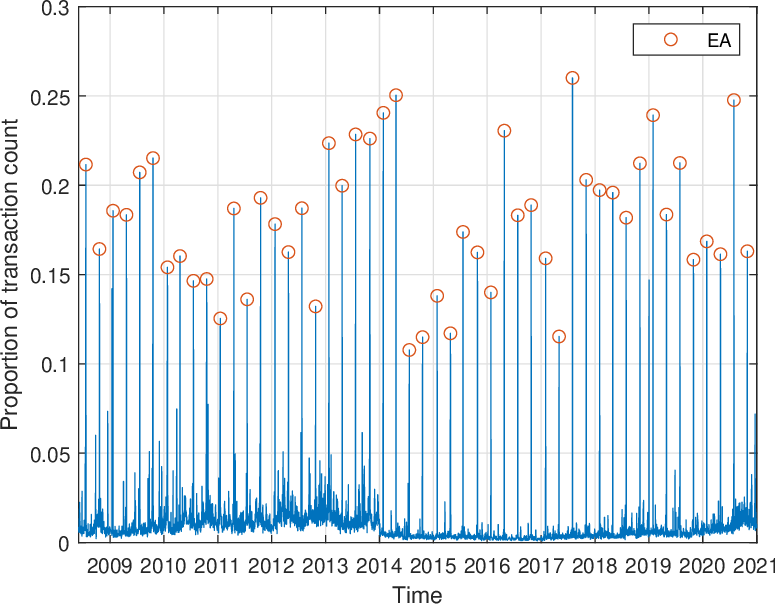} &
\includegraphics[height=8cm,width=0.48\textwidth]{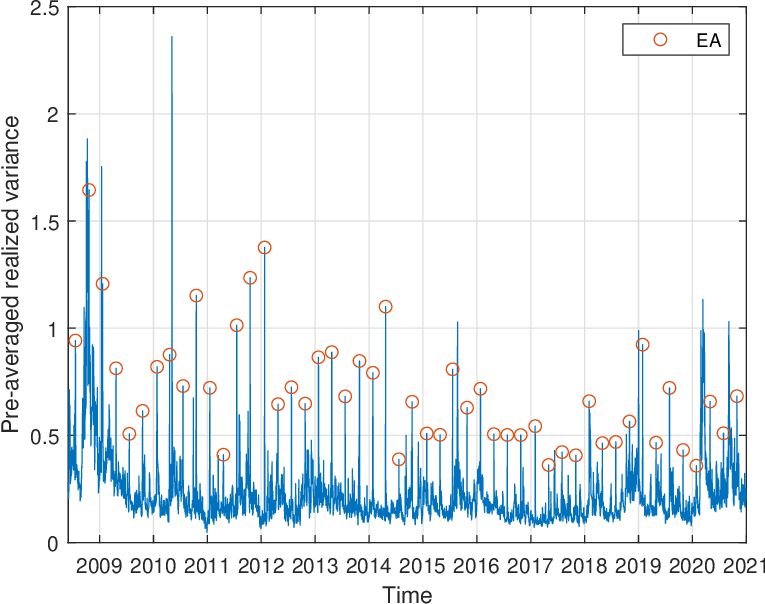}
\end{tabular}
\begin{scriptsize}
\parbox{\textwidth}{\emph{Note.} In Panel A, we show the daily after-hours market (4:00pm--6:30pm) transaction count scaled by the total transaction count in the extended trading session (9:30am--6:30pm). In Panel B, we plot the daily pre-averaged realized variance, converted to an annualized standard deviation. In both panels, a red circle indicates a day on which Apple made an earnings announcement.}
\end{scriptsize}
\end{center}
\end{figure}

The left panel in Figure \ref{figure:aapl-regular-vs-extended} plots the pre-averaged realized variance computed over the regular trading session  (9:30am--4:00pm) against its corresponding value computed over the extended trading session (9:30am--6:30pm). The figure also marks days on which Apple announced earnings. On days with no announcement, the estimates fall close to the 45-degree line. In contrast, on days with an announcement, the estimates are very different. The only notable exception in our sample is 01/17/2009, where Apple sent out an e-mail informing about Steve Jobs' medical leave, and 01/02/2019, where the company issued a profit warning. Both messages were released during the after-hours session.

To examine whether the incremental variance in after-hours markets took the form of jumps, Panel B in Figure \ref{figure:aapl-regular-vs-extended} plots the pre-averaged bipower variation against the pre-averaged realized variance, both computed over the extended trading session. Because bipower variation does not increase in the presence of jumps whereas the realized variance does, points on the 45-degree line are indicative of days without jumps, whereas days with such jumps should appear to the right of the line. Almost all days with price jumps in the extended trading session coincide with earnings announcement days. The primary exceptions are the day of the Steve Jobs medical leave e-mail, the profit warning, and the S\&P 500 Flash Crash on 05/06/2010.

\begin{figure}[ht!]
\begin{center}
\caption{Pre-averaged realized variance and bipower variation of Apple.}
\label{figure:aapl-regular-vs-extended}
\begin{tabular}{cc}
\small{Panel A: Incremental return variation.} & \small{Panel B: Incremental jump variation.} \\
\includegraphics[height=8cm,width=0.48\textwidth]{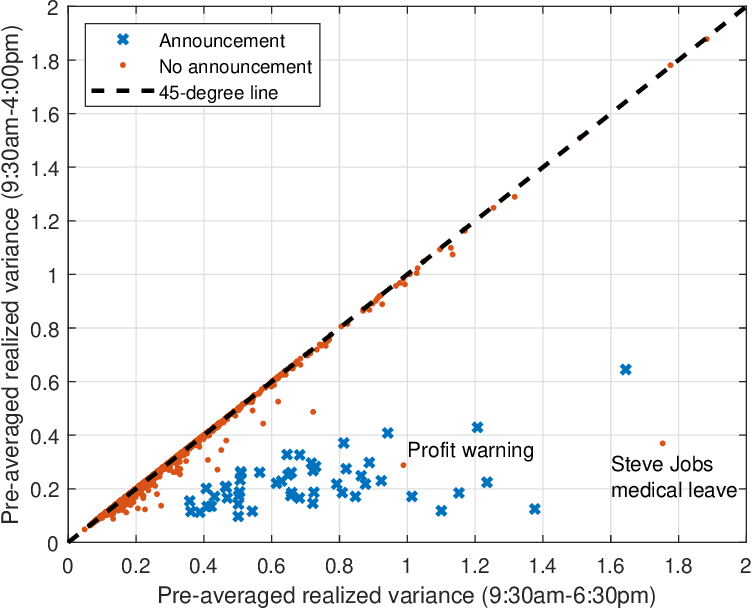} &
\includegraphics[height=8cm,width=0.48\textwidth]{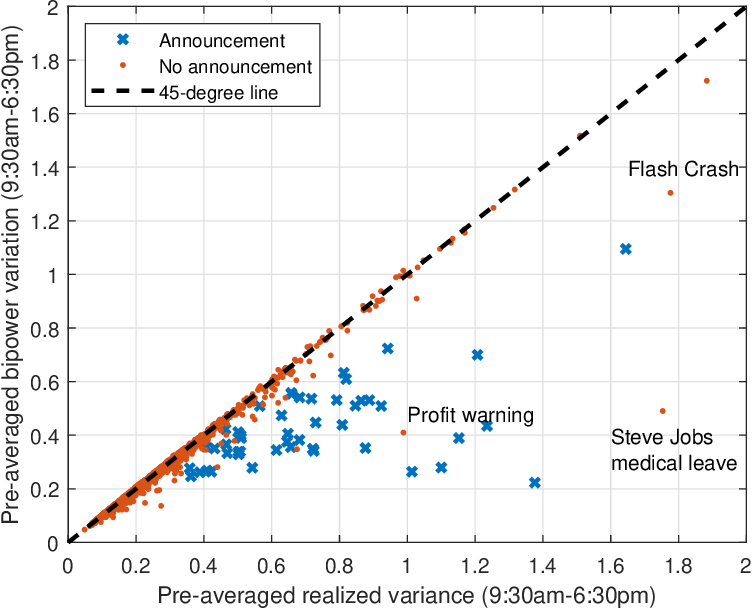}
\end{tabular}
\begin{scriptsize}
\parbox{\textwidth}{\emph{Note.} We show point estimates of the pre-averaged realized variance and pre-averaged bipower variation of Apple, converted to an annualized standard deviation. The axis label shows which estimator is plotted. In parenthesis, we further indicate for which part of the day high-frequency data are employed to calculate the estimate.}
\end{scriptsize}
\end{center}
\end{figure}

\clearpage

\section{Examples of no-jump announcements} \label{appendix:nojump}

\setcounter{table}{0}
\setcounter{figure}{0}

In this appendix, we present a few illustrative examples of corporate earnings announcements, for which our noise-robust jump test fails to be significant.

\begin{figure}[ht!]
\begin{center}
\caption{Examples of no-jump announcements.}
\label{figure:nojump}
\begin{tabular}{cc}
\small{Panel A: ULTA -- 12/03/2009.} & \small{Panel B: GILD -- 04/26/2012.}\\
\includegraphics[height=8cm,width=0.48\textwidth]{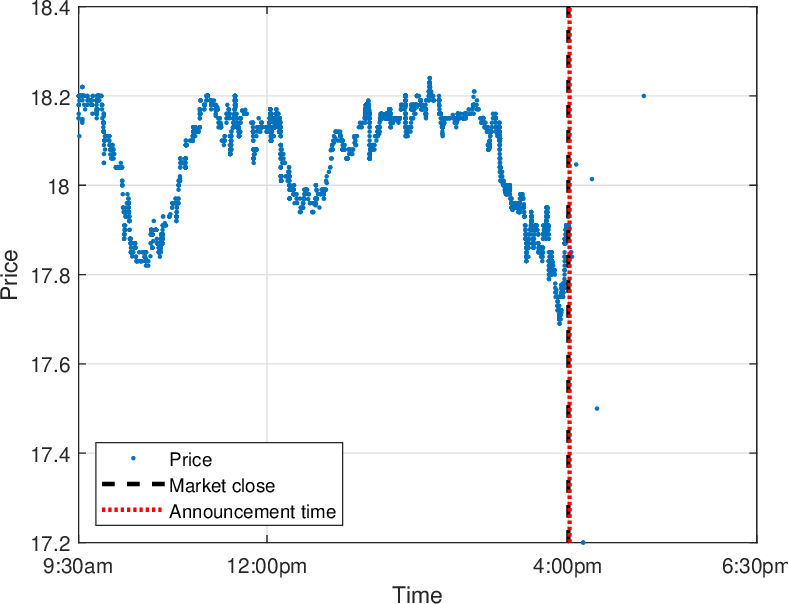} &
\includegraphics[height=8cm,width=0.48\textwidth]{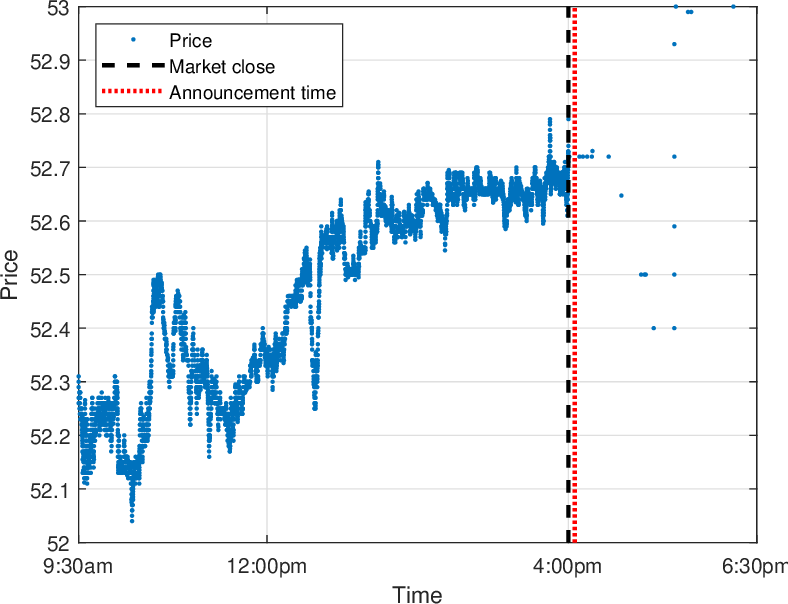}\\
\small{Panel C: TSLA -- 05/06/2015.} & \small{Panel D: GOOG -- 07/30/2020.}\\
\includegraphics[height=8cm,width=0.48\textwidth]{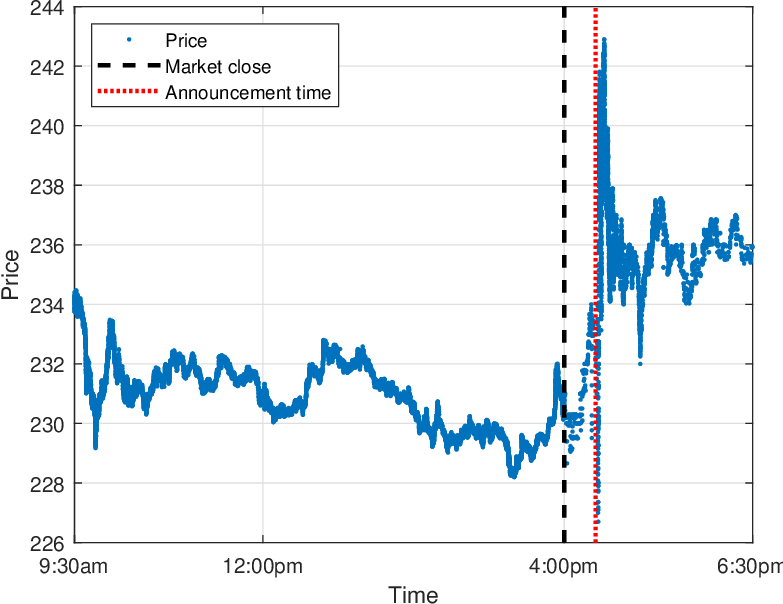} &
\includegraphics[height=8cm,width=0.48\textwidth]{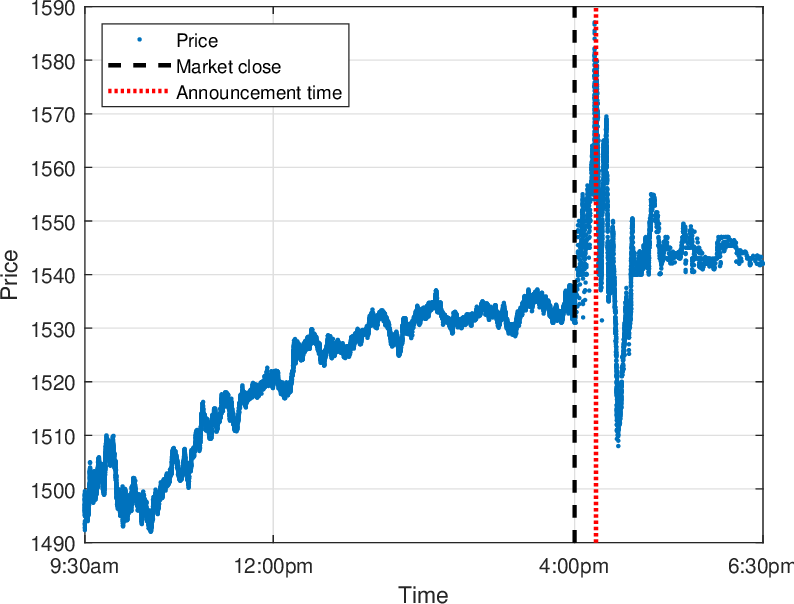}\\
\end{tabular}
\begin{scriptsize}
\parbox{\textwidth}{\emph{Note.} We plot the evolution of the stock price from 9:30am--6:30pm for selected earnings announcement days that are not flagged as significant by our noise-robust jump test. The stock ticker symbol and announcement date is indicated in the panel description.}
\end{scriptsize}
\end{center}
\end{figure}

\clearpage

\section{Missing or incorrect announcement times} \label{appendix:announcement-time}

\setcounter{table}{0}
\setcounter{figure}{0}

There are a few earnings reports, where the announcement time from Wall Street Horizon database is either missing or appears erroneous. As explained in personal e-mail exchange with Michael Raines, Director of Quantitative Data Solutions at Wall Street Horizon, a missing timestamp in their database reflects that the reporting company did not add a timestamp to the source documents (e.g. the official press release or form 8-K filing), in which case there is nothing to input. Wall Street Horizon is required to adhere to the information as provided---whether right or wrong---due to compliance needs of their institutional customers.

We therefore also hand collect the entire set of announcement times from the Factiva database. We search by company name and ticker symbol. We extract press releases that are marked as ``Earnings'' in the subject field and record the earliest time, where the actual quarterly earnings information appear in the headline of a press release.

Regardless, there are still a few announcements, where the time of the press release appears to be incorrect both in Wall Street Horizon and Factiva. An example for the stock price of Tesla (TSLA) on the announcement day 08/07/2013 is shown in Figure \ref{figure:tesla.eps}. The price rallies 8.17\% from 4:00pm to 4:01pm, suggesting that this is the actual time at which the announcement was published. Wall Street Horizon records the reporting time as 4:14pm, but we can trace the announcement back to 4:07pm in Factiva. There is no additional information in Factiva (or from popular internet search engines) suggesting the information in the report got leaked beforehand.

\begin{figure}[ht!]
\begin{center}
\caption{Transaction price of Tesla on 08/07/2013.}
\label{figure:tesla.eps}
\begin{tabular}{c}
\includegraphics[height=8cm,width=0.48\textwidth]{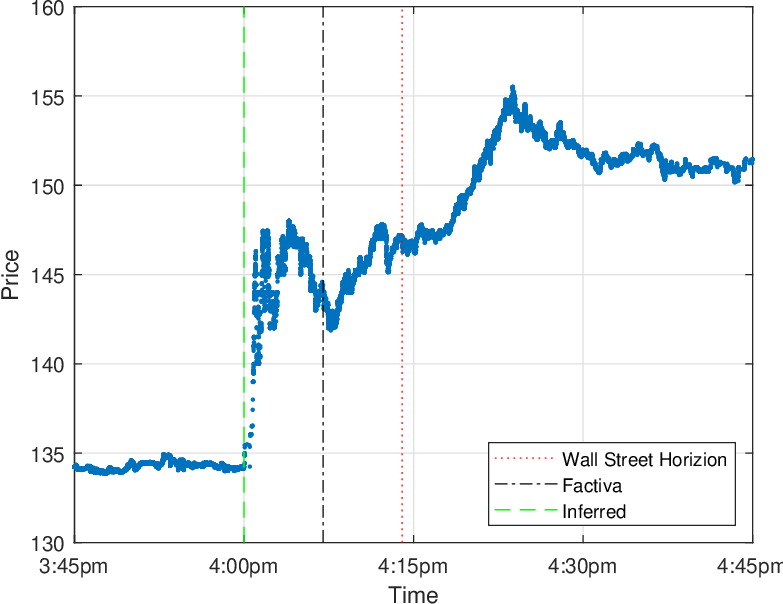} \\
\end{tabular}
\begin{scriptsize}
\parbox{\textwidth}{\emph{Note.} We plot the transaction price of Tesla's stock between 3:45pm and 4:45pm on 08/07/2013, where the company made a quarterly earnings announcement. The horizontal lines show the announcement timestamps collected from Wall Street Horizon and Factiva, which we contrast to the inferred announcement time based on our automated screening algorithm.}
\end{scriptsize}
\end{center}
\end{figure}

\clearpage

\section{Risk-adjusted return performance} \label{appendix:risk}

The large average returns on some of the \citet{patell-wolfson:84a} trading strategies may reflect a risk premium, i.e. exposure to the market or other priced risk factors. To examine this, we create a monthly return series by summing daily trading strategy log-returns within a given month. On days without active trading signals, we assume a risk-free rate is earned, which is proxied through the yield-to-maturity of a one-month T-bill converted into a continuously compounded daily interest rate. This yields the monthly log-return:
\begin{equation} \label{equation:monthly-return}
r_{m} = (N_{m} - \#\mathrm{EA}_{m}) r_{f,d} + \sum_{t=1}^{ \#\mathrm{EA}_{m}} r_{t}^{ \mathrm{EA}}, \quad \text{for } m = 1, \dots, M,
\end{equation}
where $\#\mathrm{EA}_{m}$ is the number of announcement days with trading signals in month $m$, $r_{t}^{ \mathrm{EA}}$ is the return generated from trading on day $t$, $r_{f,d}$ is the daily risk-free rate, $N_{m}$ is the number of days in month $m$, and $M$ is the number of months.

\begin{table}[ht!]
\setlength{\tabcolsep}{0.25cm}
\begin{center}
\caption{Properties of monthly trading returns.}
\label{table:excess-return}
\begin{small}
\begin{tabular}{l*{5}{c}}
\hline
\hline
& Trade & Midquote & BBO & BBO+5s & BBO+10s \\ \cline{2-6}
\\[-0.30cm]
& \multicolumn{5}{c}{Sample average return, Jensen's alpha, and Sharpe ratio}\\
\multicolumn{6}{l}{ \textit{Panel A: Full sample}}\\
\\[-0.30cm]
$\bar{r}_{m}$ & $\underset{(1.315)}{10.349^{***}}$ & $\underset{(1.305)}{8.606^{***}}$ & $\underset{(1.233)}{4.152^{***}}$ & $\underset{(1.148)}{2.342^{** }}$ & $\underset{(1.139)}{1.624^{   }}$\\
\\[-0.40cm]
$\alpha$ & $\underset{(1.306)}{11.164^{***}}$ & $\underset{(1.307)}{9.231^{***}}$ & $\underset{(1.198)}{4.809^{***}}$ & $\underset{(1.103)}{3.178^{***}}$ & $\underset{(1.074)}{2.372^{** }}$\\
\\[-0.40cm]
Sharpe ratio & 2.264 & 1.897 & 0.969 & 0.587 & 0.410\\
\\[-0.20cm]
\multicolumn{6}{l}{ \textit{Panel B: 2008 -- 2015}}\\
\\[-0.30cm]
$\bar{r}_{m}$ & $\underset{(1.648)}{12.751^{***}}$ & $\underset{(1.580)}{11.067^{***}}$ & $\underset{(1.456)}{6.836^{***}}$ & $\underset{(1.260)}{5.151^{***}}$ & $\underset{(1.244)}{4.426^{***}}$\\
\\[-0.40cm]
$\alpha$ & $\underset{(1.629)}{13.155^{***}}$ & $\underset{(1.551)}{11.130^{***}}$ & $\underset{(1.416)}{7.076^{***}}$ & $\underset{(1.272)}{5.521^{***}}$ & $\underset{(1.249)}{4.688^{***}}$\\
\\[-0.40cm]
Sharpe ratio & 2.856 & 2.586 & 1.734 & 1.510 & 1.313\\
\\[-0.20cm]
\multicolumn{6}{l}{ \textit{Panel C: 2016 -- 2020}}\\
\\[-0.30cm]
$\bar{r}_{m}$ & $\underset{(2.077)}{6.683^{***}}$ & $\underset{(2.154)}{4.851^{** }}$ & $\underset{(2.067)}{0.057^{   }}$ & $\underset{(2.043)}{-1.945^{   }}$ & $\underset{(2.035)}{-2.653^{   }}$\\
\\[-0.40cm]
$\alpha$ & $\underset{(1.992)}{7.746^{***}}$ & $\underset{(2.124)}{5.792^{***}}$ & $\underset{(1.978)}{0.723^{   }}$ & $\underset{(1.884)}{-1.106^{   }}$ & $\underset{(1.834)}{-1.946^{   }}$\\
\\[-0.40cm]
Sharpe ratio & 1.463 & 1.024 & 0.012 & -0.433 & -0.593\\
\\[-0.20cm]
& \multicolumn{5}{c}{Full sample Fama-French regression}\\
\\[-0.40cm]
$\beta_{ \mathrm{MKT}}$ & $\underset{(0.376)}{-0.454^{   }}$ & $\underset{(0.381)}{-0.267^{   }}$ & $\underset{(0.352)}{-0.274^{   }}$ & $\underset{(0.334)}{-0.390^{   }}$ & $\underset{(0.320)}{-0.293^{   }}$\\
\\[-0.40cm]
$\beta_{ \mathrm{HML}}$ & $\underset{(0.734)}{0.426^{   }}$ & $\underset{(0.761)}{0.480^{   }}$ & $\underset{(0.708)}{0.514^{   }}$ & $\underset{(0.662)}{0.482^{   }}$ & $\underset{(0.653)}{0.430^{   }}$\\
\\[-0.40cm]
$\beta_{ \mathrm{SMB}}$ & $\underset{(0.522)}{-0.321^{   }}$ & $\underset{(0.506)}{-0.271^{   }}$ & $\underset{(0.494)}{-0.213^{   }}$ & $\underset{(0.482)}{-0.276^{   }}$ & $\underset{(0.483)}{-0.258^{   }}$\\
\\[-0.40cm]
$\beta_{ \mathrm{RMW}}$ & $\underset{(0.848)}{-1.140^{   }}$ & $\underset{(0.901)}{-1.152^{   }}$ & $\underset{(0.915)}{-1.313^{   }}$ & $\underset{(0.946)}{-1.520^{   }}$ & $\underset{(0.931)}{-1.574^{*  }}$\\
\\[-0.40cm]
$\beta_{ \mathrm{CMA}}$ & $\underset{(0.945)}{-0.607^{   }}$ & $\underset{(1.048)}{-0.984^{   }}$ & $\underset{(0.994)}{-1.090^{   }}$ & $\underset{(0.934)}{-0.679^{   }}$ & $\underset{(0.919)}{-0.466^{   }}$\\
\\[-0.40cm]
$\beta_{ \mathrm{MOM}}$ & $\underset{(0.282)}{-0.185^{   }}$ & $\underset{(0.291)}{-0.079^{   }}$ & $\underset{(0.288)}{-0.078^{   }}$ & $\underset{(0.286)}{-0.087^{   }}$ & $\underset{(0.298)}{0.009^{   }}$\\
\\[-0.40cm]
$R^{2}$ & 0.028 & 0.023 & 0.032 & 0.043 & 0.040\\
\\[-0.40cm]
$P$-value ($\beta_{ \mathrm{MKT}}$ = \dots = $\beta_{ \mathrm{MOM}}$ = 0) & 0.688 & 0.771 & 0.598 & 0.406 & 0.468\\
\hline
\hline
\end{tabular}
\end{small}
\smallskip
\begin{scriptsize}
\parbox{\textwidth}{\emph{Note.} We construct a monthly return series (in percent), $r_{m} = \sum_{t=1}^{ \#\mathrm{EA}_{m}} r_{t}^{ \mathrm{EA}} + (N_{m} - \#\mathrm{EA}_{m}) r_{f,d}$, for $m = 1, \dots, M$, where $\#\mathrm{EA}_{m}$ is the number of announcement days with trading signals in month $m$, $r_{t}^{ \mathrm{EA}}$ is the return generated from each trade, $r_{f,d}$ is the daily risk-free rate, $N_{m}$ is the number of days in month $m$, and $M$ is the number of months. In Panel A, $\bar{r}_{m}$ is the sample average monthly return. We subtract the monthly risk-free rate, $r_{f,m}$, from $r_{m}$ to create an excess return. The Sharpe ratio is the average excess return divided by the sample standard deviation of the return series, which we convert to an annualized figure by multiplying with $\sqrt{12}$. We regress the excess return on an intercept, the five factors from the extended \citet*{fama-french:15a} asset pricing model, and the momentum factor of \citet*{carhart:97a}: $r_{m} - r_{f,m} = \alpha + \beta_{ \mathrm{MKT}} r_{m}^{ \mathrm{MKT}} + \beta_{ \mathrm{HML}} r_{m}^{ \mathrm{HML}} + \beta_{ \mathrm{SMB}} r_{m}^{ \mathrm{SMB}} + \beta_{ \mathrm{RMW}} r_{m}^{ \mathrm{RMW}} + \beta_{ \mathrm{CMA}} r_{m}^{ \mathrm{CMA}} + \beta_{ \mathrm{MOM}} r_{m}^{ \mathrm{MOM}} + \epsilon_{m}$, where $r_{m}^{ \mathrm{MKT}}$ is the monthly excess return on the market (MKT), $r_{m}^{ \mathrm{HML}}$ is High Minus Low (HML), $r_{m}^{ \mathrm{SMB}}$ is Small Minus Big (SML), $r_{m}^{ \mathrm{RMW}}$ is Robust Minus Weak (RMW), $r_{m}^{ \mathrm{CMA}}$ is Conservative Minus Aggressive (CMA), and $r_{m}^{ \mathrm{MOM}}$ is Momentum (MOM). The table reports parameter estimates from the regression in Panel B. Standard errors are shown in parenthesis below the parameter estimate. *, **, and *** denote statistical significance at the 10\%, 5\%, and 1\% level. The number of observations is 146. $R^{2}$ is the coefficient of determination. The $P$-value is for testing the hypothesis $H_{0} : \beta_{ \mathrm{MKT}} = \dots = \beta_{ \mathrm{MOM}} = 0$.}
\end{scriptsize}
\end{center}
\end{table}

The first row in Panel A of Table \ref{table:excess-return} reports the sample mean of the raw monthly returns series computed under our different trading scenarios. For trades executed at the transaction price (column 1 labeled "Trade"), the mean return is 10.35\% per month. This figure is both economically large and highly statistically significant. In the same vein, executing at the midquote earns a mean return of 8.61\% per month. Trading at the best bid or offer reduces the mean return to 4.15\% per month, which remains highly significant. Introducing latency of five or ten seconds brings it further down to 2.34\% and 1.62\% per month with only the former being significant at the 5\% level, however. These results are in close alignment with our trading strategy results from Table \ref{table:trading-strategy} in the main text, so aggregation to the monthly horizon does not affect our conclusions.

We also compute the Sharpe ratio as the average monthly excess returns divided by the sample standard deviation of the monthly return series, both converted to annualized figures. The Sharpe ratio (shown in the third row) follows the pattern of the mean (excess) returns, starting out at 2.26 for the transaction price and 1.90 for midquote setting. It then drops to 0.97, 0.59 and 0.41 for the BBO scenarios with a zero, five and ten second latency delay.

Panels B and C in Table \ref{table:excess-return} show subsample results with the top row reporting average raw returns and the third row showing the Sharpe ratio. The results are distinctly different. The mean return in the scenario with no trading frictions (12.75\% per month) is almost twice as high in the 2008--2015 period as compared to the 2016--2020 subsample (6.68\%). Moreover, the mean return in the early subsample is statistically significant across all trading scenarios even after accounting for bid-ask spreads and a 10-second latency delay. In contrast, in the late subsample the mean return is only 0.057\% per month in the BBO scenario, which is statistically insignificant. Our estimated Sharpe ratios are also much smaller, in some cases even negative, in the second subsample as compared to the first period. These results are again consistent with our daily returns analysis.

To adjust for risk exposure, we next subtract the monthly risk-free rate, $r_{f,m}$, from $r_{m}$ and regress the resulting excess return on an intercept, the five factors from the extended \citet*{fama-french:15a} model, and the momentum factor of \citet*{carhart:97a}:
\begin{equation} \label{risk-adj regression}
r_{m} - r_{f,m} = \alpha + \beta_{ \mathrm{MKT}} r_{m}^{ \mathrm{MKT}} + \beta_{ \mathrm{HML}} r_{m}^{ \mathrm{HML}} + \beta_{ \mathrm{SMB}} r_{m}^{ \mathrm{SMB}} + \beta_{ \mathrm{RMW}} r_{m}^{ \mathrm{RMW}} + \beta_{ \mathrm{CMA}} r_{m}^{ \mathrm{CMA}} + \beta_{ \mathrm{MOM}} r_{m}^{ \mathrm{MOM}} + \epsilon_{m},
\end{equation}
where $r_{m}^{ \mathrm{MKT}}$ is the monthly excess return on the market portfolio (MKT), and the remaining covariates are excess returns on the factor-mimicking portfolios, where $r_{m}^{ \mathrm{HML}}$ is High Minus Low (HML), $r_{m}^{ \mathrm{SMB}}$ is Small Minus Big (SML), $r_{m}^{ \mathrm{RMW}}$ is Robust Minus Weak (RMW), $r_{m}^{ \mathrm{CMA}}$ is Conservative Minus Aggressive (CMA), and $r_{m}^{ \mathrm{MOM}}$ is Momentum (MOM).\footnote{The data are downloaded from Kenneth French's website: \path{http://mba.tuck.dartmouth.edu/pages/faculty/ken.french/}. The one-month T-bill rate is from Ibbotson Associates.}

Our trading strategies generate slightly negative market betas with values ranging from -0.454 to -0.267, though none of these estimates are statistically different from zero. The only risk factor that our trading strategy returns load marginally significantly on is the profitability factor, RMW, which generates negative estimates around -1.1 to -1.6. Since earnings announcements contain information that is relevant to profitability, it is not surprising that the RMW factor stands out in terms of significance. Our trading rule goes long in stocks whose actual earnings numbers surprise analysts the most while it shorts stocks whose earnings disappoint. The negative loadings on the RMW factor is therefore consistent with firms with relatively weak profitability outperforming analyst expectations the most, whereas companies with relatively robust earnings disappoint the most. However, the relationship is weak.

The monthly alpha estimates from the risk-adjustment regression in \eqref{risk-adj regression} are bigger than the corresponding mean raw returns. The alpha estimate is positive (at 2.37\% per month) and statistically significant even for the BBO strategy with a 10-second latency delay in the full sample. Turning to the subsample analysis, Panels B and C show that alpha estimates are far bigger, and always statistically significant in the early period (2008--2015) as compared to the later sample (2016--2020). Conversely, in the second subsample alpha estimates are only statistically significant under the no-friction and midquote scenarios.

We conclude from these findings that while risk-adjusting the returns from our trading strategy slightly strengthens our performance results, the basic conclusions from the earlier analysis remain unaltered. In the early subsample, there is strong evidence that information on earnings surprises could have been exploited to generate significantly positive mean returns both on a raw and risk-adjusted basis. Conversely, market efficiency seems to have improved after 2016 to the point where our trading strategy is no longer generating significantly positive raw or risk-adjusted returns after accounting for transaction costs.

\end{document}